\documentclass[10pt]{article}
\usepackage{microtype}
\usepackage{xspace}
\usepackage{mathtools}
\usepackage{fullpage}
\usepackage{hyperref}
\usepackage{amsthm}
\usepackage{amsfonts}
\usepackage{booktabs}
\usepackage{multirow}
\usepackage{tikz}
\usepackage{pgfplots}
\usepackage{listings}
\usepackage{tikz}
\usepackage{bm}
\usepackage{enumerate}
\usepackage{enumitem}
\usepackage{todonotes}

\usetikzlibrary{fit,shapes, arrows, shadows, calc, automata, positioning,backgrounds,decorations.pathmorphing,backgrounds,petri,decorations.markings,decorations.pathreplacing}
\usetikzlibrary{topaths,calc}
 
\theoremstyle{plain}                  
\newtheorem{theorem}{Theorem}

\newtheorem{proposition}[theorem]{Proposition}
\newtheorem{conjecture}[theorem]{Conjecture}
\newtheorem{corollary}[theorem]{Corollary}         
\newtheorem{definition}[theorem]{Definition}
\newtheorem{example}[theorem]{Example}

\title{Counting Triangles under Updates in Worst-Case Optimal Time}

\author{
Ahmet Kara$^1$, Hung Q. Ngo$^2$, Milos Nikolic$^3$, Dan Olteanu$^1$, Haozhe Zhang$^1$ \\ \\
$^1$University of Oxford \enspace\enspace $^2$RelationalAI, Inc. \enspace\enspace $^3$University of Edinburgh
}

\definecolor{light-gray}{gray}{0.7.2}
\newcommand{\bigO}[1]{\mathcal{O}(#1)}

\newcommand{\eps}{\epsilon}    
\newcommand{\ivme}{\text{IVM}$^{\eps}$\xspace}
\newcommand{\db}{\mathbf{D}}
\newcommand{\dbeps}{\mathbf{P}}
\newcommand{\Dom}{\mathsf{Dom}}
\newcommand{\inst}[1]{\mathbf{#1}}
\newcommand{\state}{\mathcal{Z}}
\newcommand{\upd}{\mathit{u}}
\newcommand{\update}{\mathit{apply}}
\newcommand{\minor}{\mathit{minor}}
\newcommand{\major}{\mathit{major}}
\newcommand{\OuMv}{\textsf{OuMv}\xspace}
\newcommand{\OMv}{\textsf{OMv}\xspace}

\newcommand{\ztimes}{\cdot}

\newcommand{\vecnormal}[1]{\textnormal{\bf #1}\xspace}
\newcommand{\tuple}[1]{(#1)}
\newcommand{\floor}[1]{\left\lfloor #1 \right\rfloor}
\newcommand{\ceil}[1]{\left\lceil #1 \right\rceil}
\newcommand{\Qsj}[0]{}

\newcommand{\deltaA}{\alpha}
\newcommand{\deltaB}{\beta}
\newcommand{\bdeltaA}{\boldsymbol{\alpha}}

\newcommand{\deltaC}{\gamma}
\newcommand{\p}{\mathit{m}}

\newcommand{\threer}{3\text{R}\xspace}
\newcommand{\threerpj}[1]{3\text{R}\ensuremath{^{#1}}\xspace}

\newcommand{\atoms}[1]{\mathit{rels}(#1)}

\newcommand{\parts}[1]{\textit{parts}(#1)}
\newcommand{\SAV}{\textit{SAV}}

\newcommand{\TAB}{\makebox[2.5ex][r]{}}%
\newcommand{\STAB}{\makebox[1.5ex][r]{}}%
\newcommand{\OUTPUT}{\textbf{output}\xspace}%
\newcommand{\LET}{\textbf{let}\xspace}%
\newcommand{\IF}{\textbf{if}\xspace}%
\newcommand{\ELSE}{\textbf{else}\xspace}%
\newcommand{\FOREACH}{\textbf{foreach}\xspace}%
\newcommand{\DO}{\textbf{do}\xspace}%
\newcommand{\AND}{\textbf{and}\xspace}%
\newcommand{\OR}{\textbf{or}\xspace}%
\newcommand{\RETURN}{\textbf{return}\xspace}%
\newcommand{\nop}[1]{}

\newcounter{CommentCounter}

\newcounter{magicrownumbers}
\newcommand\rownumber{\footnotesize\stepcounter{magicrownumbers}\arabic{magicrownumbers}}

\begin{document}
\maketitle
\begin{abstract}
\noindent
We consider the problem of incrementally maintaining the triangle count query under single-tuple updates to the input relations. We introduce an approach that exhibits a space-time tradeoff such that the space-time product is quadratic in the size of the input database and the update time can be as low as the square root of this size. This lowest update time is worst-case 
optimal conditioned on the Online Matrix-Vector Multiplication conjecture.

The classical and factorized incremental view maintenance approaches are recovered as special cases of our approach within the space-time tradeoff. In particular, they require linear-time maintenance under updates, which is suboptimal. Our approach can also count all triangles in a static database in the worst-case optimal time needed for enumerating them.
\end{abstract}

\section{Introduction}\label{sec:intro}
 
We consider the problem of incrementally maintaining the result of the triangle count query
\begin{align}
\label{query:triangle}
Q() = \sum\limits_{a \in \Dom(A)}\sum\limits_{b \in \Dom(B)}\sum\limits_{c \in \Dom(C)} R(a,b)\ztimes S(b,c) \ztimes T(c,a)
\end{align}
under single-tuple updates to the relations $R$, $S$, and $T$
with schemas $(A,B)$, $(B,C)$, and $(C,A)$, respectively.
The relations are given as functions mapping tuples over relation schemas 
to tuple multiplicities. 
 A single-tuple update $\delta R = \{\,(\deltaA,\deltaB) \mapsto \p\,\}$ to relation $R$ maps the tuple 
$(\deltaA, \deltaB)$ to a nonzero multiplicity $\p$, which is positive for inserts and negative for deletes.

The triangle query and its counting variant have served as a milestone for worst-case optimality of join algorithms in the centralized and parallel settings and for randomized approximation schemes for data processing. They serve as the workhorse showcasing suboptimality of mainstream join algorithms used currently by virtually all commercial database systems. For a database $\db$ consisting of $R$, $S$, and $T$, standard binary join plans implementing these queries may take $O(|\db|^2)$ time, yet these queries can be solved in $\bigO{|\db|^{\frac{3}{2}}}$ time~\cite{AYZ:Counting:1997}.
This observation motivated a new line of work on worst-case optimal algorithms for arbitrary join queries~\cite{Ngo:SIGREC:2013}. The triangle query has also served as a yardstick for understanding the optimal communication cost for parallel query evaluation in the Massively Parallel Communication model~\cite{Koutris:FTDB:2018}. The triangle count query has witnessed the development of randomized approximation schemes with increasingly lower time and space requirements, e.g.,~\cite{Eden:approximately:FOCS:2015}.

A worst-case optimal result for incrementally maintaining the exact triangle count query has so far not been established. Incremental maintenance algorithms may benefit from a good range of processing techniques whose flexible combinations may make it harder to reason about optimality. Such techniques include algorithms for aggregate-join queries with low complexity developed for the non-incremental case~\cite{Ngo:PODS:2018}; pre-materialization of views that reduces maintenance of the query to that of simpler subqueries~\cite{KochAKNNLS14}; and delta processing that allows to only compute the change in the result instead of the entire result~\cite{Chirkova:Views:2012:FTD}.

\subsection{Existing Incremental View Maintenance (IVM) Approaches}

The problem of incrementally maintaining the triangle count has received a fair amount of attention. Existing exact approaches require at least linear time in worst case. 
After each update to a database $\db$, 
the na\"ive approach joins the relations 
$R$, $S$, and $T$ in time $\bigO{|\db|^{\frac{3}{2}}}$ 
using a worst-case optimal algorithm~\cite{AYZ:Counting:1997,Ngo:SIGREC:2013}
and counts the result tuples. 
The number of distinct tuples in the result is at most 
$|\db|^{\frac{3}{2}}$, which is a well-known result by 
Loomis and Whitney from 1949 (see recent notes on the history of this result~\cite{Ngo:PODS:2018}). 
The classical first-order IVM~\cite{Chirkova:Views:2012:FTD} computes on the 
fly a delta query $\delta Q$ per single-tuple update $\delta R$ to relation $R$ (or any other relation) and updates the query result:
\begin{align*}
\delta Q() = \delta R(\deltaA,\deltaB) \ztimes \sum\limits_{c \in \Dom(C)} S(\deltaB,c) \ztimes T(c,\deltaA),
\hspace*{6em} Q() = Q() + \delta Q().
\end{align*}
The delta computation takes $\bigO{|\db|}$ time since it needs to intersect two lists of possibly linearly many $C$-values that are paired with $\deltaB$ in $S$ and with $\deltaA$ in $T$ (i.e., the multiplicity of such pairs in $S$ and $T$ is nonzero).  
The recursive IVM~\cite{KochAKNNLS14} speeds up the delta computation by precomputing three auxiliary views representing the update-independent parts of the delta queries for updates to $R$, $S$, and $T$:

\begin{align*}
V_{ST}(b,a) &= \sum\limits_{c \in \Dom(C)} S(b,c) \ztimes T(c,a) \\
V_{TR}(c,b) &= \sum\limits_{a \in \Dom(A)} T(c,a) \ztimes R(a,b) \\
V_{RS}(a,c) &= \sum\limits_{b \in \Dom(B)} R(a,b) \ztimes S(b,c).
\end{align*}
These three views take $\bigO{|\db|^2}$ space but allow to compute the delta query for single-tuple updates to the input relations in $\bigO{1}$ time. 
Computing the delta $\delta Q() = \delta R(\deltaA,\deltaB) \ztimes  V_{ST}(\deltaB,\deltaA)$ requires just a constant-time lookup in $V_{ST}$;
however, maintaining the views $V_{RS}$ and $V_{TR}$, which refer to $R$, still requires $\bigO{|\db|}$ time.
The factorized IVM~\cite{NO18} materializes only one of the three views, for instance, $V_{ST}$. In this case, the maintenance under updates to $R$ takes $\bigO{1}$ time, but the maintenance under updates to $S$ and $T$ still takes $\bigO{|\db|}$ time.

Further exact IVM approaches focus on acyclic conjunctive queries. For free-connex acyclic conjunctive queries, the dynamic Yannakakis approach allows for enumeration of result tuples with constant delay under single-tuple updates~\cite{Idris:dynamic:SIGMOD:2017}. For databases with or without integrity constraints, it is known that a strict, small subset of the class of acyclic conjunctive queries admit constant-time update, while all other conjunctive queries have update times dependent on the size of the input database~\cite{BerkholzKS17,Berkholz:ICDT:2018}. 

Further away from our line of work is the development of dynamic descriptive complexity, starting with the DynFO complexity class and the much-acclaimed result on FO expressibility of the maintenance for graph reachability under edge inserts and deletes, cf.\@ a recent survey~\cite{Schwentick:DynamicComplexity:2016}. 
The $k$-clique query can be maintained under edge inserts by a quantifier-free update program of arity $k-1$ but not of arity $k-2$~\cite{Zeume:Clique:2017}.

A distinct line of work investigates randomized approximation schemes with an arbitrary relative error for counting triangles in a graph given as a stream of edges, e.g.,~\cite{Bar:reductions:SODA:2002,Jowhari:new:COCOON:2005,Buriol:counting:PODS:2006,Mcgregor:better:PODS:2016,Cormode:secondlook:TCS:2017}. Each edge in the data stream corresponds to 
a tuple insert, and tuple deletes are not considered. The emphasis of these approaches is on space efficiency, and they express the space utilization as a function of the number of nodes and edges in the input graph and of the number of triangles. The space utilization is generally sublinear but may become superlinear if, for instance, the number of edges is greater than the square root of the number of triangles. The update time is polylogarithmic in the number of nodes in the graph.

A complementary line of work unveils structure in the PTIME complexity class by giving lower bounds on the complexity of problems under various conjectures~\cite{Henzinger:OMv:2015,Williams:2018:finegrained}.

\begin{definition}[Online Matrix-Vector Multiplication (\OMv)~\cite{Henzinger:OMv:2015}]\label{def:OMv}
We are given an $n \times n$ Boolean matrix $\vecnormal{M}$ and  receive $n$ column vectors of size $n$, denoted by $\vecnormal{v}_1, \ldots, \vecnormal{v}_n$, one by one; after seeing each vector $\vecnormal{v}_i$, we output the product $\vecnormal{M} \vecnormal{v}_i$ before we see the next vector.
\end{definition}

\begin{conjecture}[\OMv Conjecture, Theorem 2.4 in~\cite{Henzinger:OMv:2015}]\label{conj:omv}
For any $\gamma > 0$, there is no algorithm that solves \OMv in time $\bigO{n^{3-\gamma}}$.
\end{conjecture}

The \OMv conjecture has been used to exhibit conditional lower bounds 
for many dynamic problems, including those previously based on other popular problems and conjectures, such as 3SUM and combinatorial Boolean matrix multiplication~\cite{Henzinger:OMv:2015}. This also applies to our triangle count query:
For any $\gamma > 0$ and database of domain size $n$, 
there is no algorithm that incrementally maintains
the triangle count under single-tuple updates
with arbitrary preprocessing time, $\bigO{n^{1-\gamma}}$ update time,
and $\bigO{n^{2-\gamma}}$ answer time,
unless the \OMv conjecture fails~\cite{BerkholzKS17}.

\subsection{Our Contribution}
This paper introduces \ivme, an incremental view maintenance approach that maintains the triangle count in amortized sublinear time. Our main result is as  follows:

\begin{theorem}\label{theo:main_result}
Given a database $\db$ and $\eps \in [0,1]$, 
\ivme incrementally maintains the result of Query~\eqref{query:triangle} under single-tuple updates to $\db$ with 
$\bigO{|\inst{D}|^{\frac{3}{2}}}$ preprocessing time, $\bigO{|\inst{D}|^{\max\{\eps,1-\eps\}}}$ amortized update time, 
constant answer time, and $\bigO{|\inst{D}|^{1 + \min\{\eps,1-\eps\}}}$ space.
\end{theorem}
 
The preprocessing time is for computing the triangle count on the initial database before the updates; if we start with the empty database, then this time is $\bigO{1}$.
The \ivme approach exhibits a tradeoff between space and amortized update time,
cf.\@ Figure~\ref{fig:time_space_plot}.

\ivme uses a data structure that partitions each input relation into a heavy part 
and a light part based on the degrees of data values. 
The degree of an $A$-value $a$ in relation $R$ is 
the number of $B$-values paired with $a$ in $R$.
The light part of $R$ consists of all tuples $(a,b)$ from $R$
such that the degree of $a$ in $R$ is below a certain threshold 
that depends on the database size and $\eps$. All other tuples
are included in the heavy part of $R$. Similarly, the 
relations $S$ and $T$ are partitioned based on the degrees
of $B$-values in $S$ and $C$-values in $T$, respectively.    
The maintenance is adaptive in that it uses different evaluation strategies for different heavy-light combinations of parts of the input relations that overall keep the update time sublinear. Section~\ref{sec:strategy} introduces this adaptive maintenance strategy.

As the database evolves under updates, \ivme needs to rebalance the heavy-light partitions to account for a new database size and updated degrees of data values. While this rebalancing may take superlinear time, it remains sublinear per single-tuple update. The update time is therefore amortized.
Section~\ref{sec:rebalancing} discusses the rebalancing strategy of \ivme.

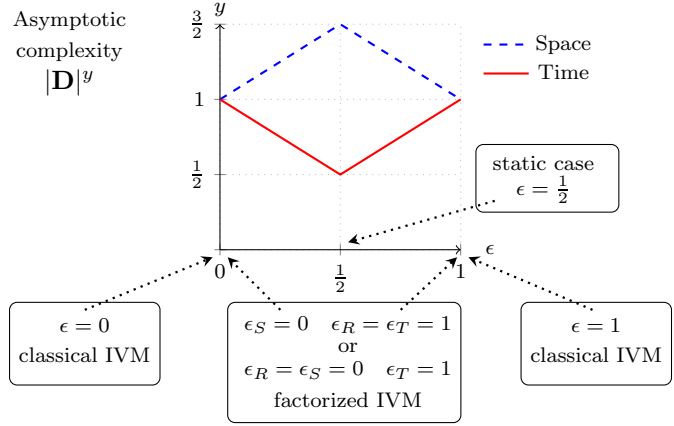
\begin{figure}[t]
\begin{center}
\begin{tikzpicture}
\begin{axis}[
grid=major,
    grid style={dotted},
xmin=0, xmax=1, ymin=0, ymax=1.5,
every axis plot post/.append style={mark=none},
  xtick ={0, 0.5, 1},
  ytick ={0, 0.5, 1, 1.5},
  xticklabels={\footnotesize{$0$},\footnotesize{$\frac{1}{2}$},\footnotesize{$1$}},
   yticklabels={$$,\footnotesize{$\frac{1}{2}$},\footnotesize{$1$},\footnotesize{$\frac{3}{2}$}},  
y=2cm,
    x=3.2cm,
axis lines=middle,
    axis line style={->},
    x label style={at={(axis description cs:1.18,-0.06)}},
    xlabel={\footnotesize{$\eps$}},
    y label style={at={(axis description cs:-0.9,1.1)},align=center},
    ylabel=\footnotesize{Asymptotic} \\  \footnotesize{complexity} \\ $|\inst{D}|^y$,
  axis x line*=bottom,
  axis y line*=left,
  legend style={at={(1.6,1)},draw=none},
  legend entries={\footnotesize{Space},\footnotesize{Time}}
 ]

  \addplot[color=blue,mark=none,domain=0:1,thick,dashed] coordinates{ 
  (0, 1) 
  (1/2, 3/2)
  (1,1) 
}; 
  \addplot[color=red,mark=none,domain=0:1,thick] coordinates{ 
  (0, 1) 
  (1/2, 1/2)
  (1,1) 
}; 
\end{axis}
\node at(0,3.2) {\footnotesize{$y$}};
\node at(5,-1) {\footnotesize{$\eps = 1$}};
\node at(5,-1.4) {\footnotesize{classical IVM}};
\draw [->,>=stealth, dotted, line width=0.3mm] (4.95,-0.8) -- (3.3,-0.1);

\node at(-1.8,-1) {\footnotesize{$\eps = 0$}};
\node at(-1.8,-1.4) {\footnotesize{classical IVM}};
\draw [->,>=stealth, dotted, line width=0.3mm] (-1.8,-0.8) -- (-0.1,-0.1);

\node at(1.7,-1) {\footnotesize{$\eps_S = 0$}\ \ \ \footnotesize{$\eps_R = \eps_T = 1$}};
\node at(1.7,-1.3) {\footnotesize{or}};
\node at(1.7,-1.6) {\footnotesize{$\eps_R = \eps_S = 0$}\ \ \ \footnotesize{$\eps_T = 1$}};
\node at(1.7,-2) {\footnotesize{factorized IVM}};
\draw [->,>=stealth, dotted, line width=0.3mm] (0.8,-0.8) -- (0.1,-0.1);
\draw [->,>=stealth,dotted, line width=0.3mm] (2.4,-0.8) -- (3.15,-0.1);

\node at(4.3,1.4) {\footnotesize{}};
\node at(4.3,1.15) {\footnotesize{static case}};
\node at(4.3,0.8) {\footnotesize{$\eps = \frac{1}{2}$}};
\draw [->,>=stealth,dotted, line width=0.3mm] (3.65,0.65) -- (1.7,0.1);

\draw[rounded corners=1mm] (-2.8,-1.8) rectangle (-0.8,-0.7);
\draw[rounded corners=1mm] (4,-1.8) rectangle (6,-0.7);

\draw[rounded corners=1mm] (0.1,-2.3) rectangle (3.2,-0.7);

\draw[rounded corners=1mm](3.4,1.4) rectangle (5.3,0.5);
\end{tikzpicture}
\caption{
\ivme's space and amortized update time parameterized by $\eps$. 
The classical IVM is recovered by setting $\eps \in\{0,1\}$. 
The factorized IVM is recovered by setting $\eps_R \in \{0,1\}$, $\eps_S = 0$, and $\eps_T = 1$ 
when $V_{ST}$ is materialized (similar treatment when $V_{RS}$ or $V_{TR}$ is materialized).
For $\eps = \frac{1}{2}$, \ivme counts all triangles in a static database in the worst-case optimal time for enumerating them.}
\label{fig:time_space_plot}
\end{center}
\end{figure} 

For $\eps=\frac{1}{2}$, \ivme achieves the lowest update time $\bigO{|\inst{D}|^{\frac{1}{2}}}$ while requiring $\bigO{|\inst{D}|^{\frac{3}{2}}}$ space. This update time is optimal conditioned on the \OMv conjecture. For this, we specialize the 
lower bound result in~\cite{BerkholzKS17}  to refer to the size $|\inst{D}|$ of the database:

\begin{proposition}\label{prop:lower_bound_triangle_count}
For any $\gamma > 0$ and database $\db$,
there is no algorithm that incrementally maintains the result of Query~\eqref{query:triangle} under single-tuple updates to $\db$ with arbitrary preprocessing time, $\bigO{|\db|^{\frac{1}{2} - \gamma}}$ amortized update time, and $\bigO{|\db|^{1 - \gamma}}$ answer time, unless the \OMv conjecture fails.
\end{proposition}

This lower bound is shown in Appendix \ref{sec:lowerbound}.
Theorem~\ref{theo:main_result} and Proposition~\ref{prop:lower_bound_triangle_count} imply that \ivme incrementally maintains the triangle count with optimal update time:

\begin{corollary}[Theorem~\ref{theo:main_result} and Proposition~\ref{prop:lower_bound_triangle_count}]\label{cor:ivme_optimal}
Given a database $\db$, \ivme incrementally maintains the result of Query~\eqref{query:triangle} under single-tuple updates to $\db$ with worst-case optimal amortized update time $\bigO{|\db|^{\frac{1}{2}}}$ and constant answer time, unless the \OMv conjecture fails. 
\end{corollary}

\ivme also applies to triangle count queries with self-joins, such as when maintaining the count of triangles in a graph given by the edge relation. The space and time complexities are the same as in Theorem~\ref{theo:main_result} (Appendix~\ref{sec:self_joins}).

\ivme defines a continuum of maintenance approaches that exhibit a space-time tradeoff based on $\epsilon$. As depicted in Figure~\ref{fig:time_space_plot},
the classical first-order IVM and the factorized IVM are specific extreme points in this continuum. To recover the former, we set $\eps \in \{0,1\}$ for $\bigO{|\db|}$ update time and $\bigO{|\db|}$ space for the input relations. 
To recover the latter, we use a distinct parameter $\epsilon$ per relation: for example, using $\eps_R \in \{0,1\}$, $\eps_S=0$, and $\epsilon_T=1$, we support updates to $R$ in $\bigO{1}$ time and updates to $S$ and $T$ in $\bigO{|\db|}$ time; the view $V_{ST}$ takes $\bigO{|\db|^2}$ space (Appendix~\ref{sec:recovery}).

We observe that at optimality, \ivme recovers the worst-case optimal time $\bigO{|\db|^{\frac{3}{2}}}$ of non-incremental algorithms for enumerating all triangles~\cite{Ngo:SIGREC:2013}.
Whereas these algorithms are monolithic and require processing the input data in bulk and all joins at the same time, \ivme achieves the same complexity by inserting $|\db|$ tuples one at a time in initially empty relations $R$, $S$, and $T$, and by using standard join plans (Appendix~\ref{sec:recovery}).

\section{Preliminaries}\label{sec:preliminaries}

\paragraph*{Data Model.}
A schema $\inst{X}$ is a tuple of variables. 
Each variable $X$ has a discrete domain 
$\Dom(X)$ of data values. A tuple 
$\inst{x}$ of data values over schema $\inst{X}$ is an element from 
$\Dom(\inst{X}) = \prod_{X \in \inst{X}}{\Dom(X)}$.
We use uppercase letters for variables and lowercase letters for data values.
Likewise, we use bold uppercase letters for schemas and bold lowercase letters for tuples of data values.

A relation $K$ over schema $\inst{X}$ is a function 
$K: \Dom(\inst{X}) \to \mathbb{Z}$ mapping tuples over
$\inst{X}$ to integers such that   
$K(\inst{x}) \neq 0$ for finitely many tuples $\inst{x}$. 
We say that a tuple $\inst{x}$ is in 
$K$, denoted by $\inst{x} \in K$, if $K(\inst{x}) \neq 0$. 
The value $K(\inst{x})$ represents the multiplicity 
of $\inst{x}$ in $K$.
The size $|K|$ of $K$ is the size of the set $\{ \inst{x} \mid \inst{x} \in K \}$. 
A database $\db$ is a set of relations, and its size $|\db|$ is the sum of the sizes of the relations in $\db$.

Given a tuple $\inst{x}$ over 
schema  $\inst{X}$ and a variable $X$ in $\inst{X}$, 
we write $\inst{x}[X]$ to denote the value of $X$ in $\inst{x}$.
For a relation $K$ over $\inst{X}$, a variable $X$ in $\inst{X}$, and a data value 
$x \in \Dom(X)$, 
we use $\sigma_{X = x} K$ to denote the set of tuples in  
$K$ whose $X$-value is $x$, that is, 
$\sigma_{X = x} K = 
\{\, \inst{x} \,\mid\, \inst{x} \in K \land \inst{x}[X] = x \,\}$.
We write $\pi_{X}K$ to denote the set of $X$-values 
in $K$, that is, 
$\pi_{X}K = \{\, \inst{x}[X] \,\mid\, \inst{x} \in K \,\}$.

\paragraph*{Query Language.}
We express queries and view definitions in the language of 
functional aggregate queries (FAQ)~\cite{FAQ:PODS:2016}. 
 Compared to the original FAQ definition that uses several 
commutative semirings, we define our queries using the single 
commutative ring $(\mathbb{Z},+,\ztimes,0,1)$ of integers   
with the usual addition and multiplication. 
A query Q has one of the two forms:

\begin{enumerate}
    \item Given a set $\{X_i\}_{i \in [n]}$ of variables and an index set $S \subseteq [n]$, 
let $\inst{X}_{S}$ denote a tuple $(X_i)_{i \in S}$ of variables and 
$\inst{x}_{S}$ denote a tuple of data values over the schema $\inst{X}_{S}$.
Then,

\begin{equation*}
Q(\inst{x}_{[f]}) = \sum\limits_{x_{f+1} \in \Dom(X_{f+1})} \cdots 
\sum\limits_{x_{n} \in \Dom(X_{n})}\  \ \prod_{S\in\mathcal{M}} K_S(\inst{x}_S),\text{ where:}
\end{equation*}

\begin{itemize}
\item $\mathcal{M}$ is a multiset of index sets.

\item For every index set $S \in \mathcal{M}$, 
$K_S : \Dom(\inst{X}_S) \rightarrow \mathbb{Z}$ 
is a relation over the schema $\inst{X}_S$.

\item $\inst{X}_{[f]}$ is the tuple of free variables of $Q$.
The variables $X_{f+1}, \ldots ,X_n$ are called bound.
\end{itemize}

\item $Q(\inst{x}) = Q_1(\inst{x}) + Q_2(\inst{x})$, where $Q_1$ and $Q_2$ are queries over the same tuple of free variables.
\end{enumerate}

In the following, we use $\sum_{x_i}$ as a shorthand for $\textstyle\sum_{x_i\in \Dom(X_i)}$.

\paragraph*{Updates and Delta Queries.}
An update $\delta K$ to a relation $K$ is a relation over the schema of $K$.
A single-tuple update, written as $\delta K = \{\inst{x} \mapsto \p\}$, maps the tuple $\inst{x}$ to the nonzero multiplicity $\p\in\mathbb{Z}$ and any other tuple to $0$; that is, $|\delta K| = 1$.
The data model and query language make no distinction between inserts and deletes -- these are updates represented as relations in which tuples have positive and negative multiplicities.

Given a query $Q$ and an update $\delta K$, the delta query $\delta Q$ defines the change in the query result after applying $\delta K$ to the database. The rules for deriving delta queries follow from the associativity, commutativity, and distributivity of the ring operations. 

\vspace{3pt}
\begin{tabular}{l@{\hskip 0.4in}l}
Query $Q(\inst{x})$ & Delta query $\delta Q(\inst{x})$ \\
\midrule
$Q_1(\inst{x}_1) \ztimes Q_2(\inst{x}_2)$ & 
$\delta Q_1(\inst{x}_1) \ztimes Q_2(\inst{x}_2) +
  Q_1(\inst{x}_1)\, \ztimes \delta Q_2(\inst{x}_2) + 
  \delta Q_1(\inst{x}_1) \ztimes \delta Q_2(\inst{x}_2)$ \\

$\textstyle\sum_{x} Q_1(\inst{x}_1)$ & 
$\textstyle\sum_{x} \delta Q_1(\inst{x}_1)$ \\

$Q_1(\inst{x}) + Q_2(\inst{x})$ & 
$\delta Q_1(\inst{x}) + \delta Q_2(\inst{x})$ \\

$K'(\inst{x})$ & 
$\delta K(\inst{x})$ when $K = K'$ and $0$ otherwise \\
\end{tabular}
\vspace{3pt}

\paragraph*{Computation Time.}
Our maintenance algorithm takes as input the triangle count query  $Q$ and a 
database $\db$ and maintains the result of $Q$ under a sequence of single-tuple updates. 
We distinguish the following computation times: 
(1) \emph{preprocessing time} is spent on initializing the algorithm using $\db$ before any update is received, 
(2) \emph{update time}  is spent on processing one single-tuple update, and
(3) \emph{answer time} is spent on obtaining the result of $Q$. 
We consider two types of bounds on the update time: 
\emph{worst-case bounds}, which limit the time each individual update takes in the worst case, and
\emph{amortized worst-case bounds}, which limit the average worst-case time taken by a sequence of updates. 
Enumerating a set of tuples with constant delay means that the time until reporting the first tuple, the time between reporting two consecutive tuples, and the time between reporting the last tuple and the end of enumeration is constant.
When referring to sublinear time, we mean $\bigO{|\db|^{1-\gamma}}$ for some $\gamma > 0$, where $|\db|$ is the database size.

\paragraph*{Computational Model.}
We consider the RAM model of computation. 
Each relation (view) $K$ over schema $\inst{X}$ is implemented by a data structure that stores key-value entries $(\inst{x},K(\inst{x}))$ for each tuple $\inst{x}$ over $\inst{X}$ with $K(\inst{x}) \neq 0$ and 
needs space linear in the number of such tuples. 
We assume that this data structure supports
(1) looking up, inserting, and deleting entries in constant time,
(2) enumerating all stored entries in $K$ with constant delay, and
(3) returning $|K|$ in constant time.
For instance, a hash table with chaining, where entries are doubly linked for efficient enumeration, 
can support these operations in constant time on average, under the assumption of simple uniform hashing. 

For each variable $X$ in the schema $\inst{X}$ of relation $K$, we further assume there is an index structure on $X$ that allows:
(4) enumerating all entries in $K$ matching $\sigma_{X=x}K$ with constant delay,
(5) checking $x \in \pi_{X}K$ in constant time, and 
(6) returning $|\sigma_{X=x}K|$ in constant time, 
for any $x \in \Dom(X)$, and
(7) inserting and deleting index entries in constant time.
Such an index structure can be realized, for instance, as a hash table with chaining
where each key-value entry stores an $X$-value $x$ and a doubly-linked list of pointers to the entries in $K$ having the $X$-value $x$.
Looking up an index entry given $x$ takes constant time on average, and its doubly-linked list enables enumeration of the matching entries in $K$ with constant delay. 
Inserting an index entry into the hash table additionally prepends a new pointer to the doubly-linked list for a given $x$; overall, this operation takes constant time on average.
For efficient deletion of index entries, each entry in $K$ also stores back-pointers to its index entries (as many back-pointers as there are index structures for $K$). When an entry is deleted from $K$, locating and deleting its index entries takes constant time per index.

\paragraph*{Data Partitioning.}
We partition each input relation into two parts based on the degrees of its values. 
Similar to common techniques used in databases to deal with data skew, our IVM approach employs different maintenance strategies for values of high and low frequency.

\begin{definition}[Relation Partition]
\label{def:loose_relation_partition}
Given a relation $K$ over schema $\inst{X}$, a variable  $X$ from the schema $\inst{X}$, 
and a threshold $\theta$, a  partition of $K$ on $X$ with threshold $\theta$ 
is a set $\{ K_h, K_l \}$ satisfying the following conditions:
\\[6pt]
\begin{tabular}{@{\hskip 0.5in}rl}
{(union)} & 
$K(\inst{x}) = K_h(\inst{x}) + K_l(\inst{x})$ for $\inst{x} \in \Dom(\inst{X})$\\[4pt]
{(domain partition)} & $(\pi_{X}K_h) \cap (\pi_{X}K_l) = \emptyset$ \\[4pt]
{(heavy part)} & for all $x \in \pi_{X}K_h:\; |\sigma_{X=x} K_h| \geq 
\frac{1}{2}\,\theta$ \\[4pt]
{(light part)} & for all $x \in \pi_{X}K_l:\; |\sigma_{X=x} K_l| < \frac{3}{2}\,\theta$
\end{tabular}\\[6pt]
The set $\{ K_h, K_l \}$ is called a strict partition of $K$ on $X$ with threshold 
$\theta$ if it satisfies the union and 
domain partition conditions and the following strict versions
of the heavy part and light part conditions: 
\\[6pt]
\begin{tabular}{@{\hskip 0.5in}rl}
{(strict heavy part)} & for all $x \in \pi_{X}K_h:\; |\sigma_{X=x} K_h| \geq 
\theta$ \\[4pt]
{(strict light part)} & for all $x \in \pi_{X}K_l:\; |\sigma_{X=x} K_l| < \theta$
\end{tabular}\\[6pt]
The relations $K_h$ and $K_l$ are called the heavy and light parts of $K$.
\end{definition}

Definition~\ref{def:loose_relation_partition} admits multiple ways to (non-strictly) partition a relation
$K$ on variable $X$ with threshold $\theta$. 
For instance, assume that $|\sigma_{X=x} K| = \theta$ for some $X$-value $x$ in $K$. Then, all tuples in $K$ with $X$-value $x$ can be in either the heavy or light part of $K$; but they cannot be in both parts because of the domain partition condition. If the partition is strict, then all such tuples are in the heavy part of $K$. 

The strict partition of a relation $K$  is unique for a given threshold and can be computed in time linear in the size of $K$.

\section{IVM\textsuperscript{$\eps$}: Adaptive Maintenance of the Triangle Count}\label{sec:strategy}
We present \ivme, our algorithm for the incremental maintenance of 
the result  of Query~\eqref{query:triangle}.
We start with a high-level overview.
Consider a database $\db$ consisting of three relations 
$R$, $S$, and $T$ with schemas $(A,B)$, $(B,C)$, and $(C,A)$, respectively.
We partition $R$, $S$, and $T$ on variables $A$, $B$, and $C$, respectively, for 
a given threshold.
We then decompose Query~\eqref{query:triangle} into eight skew-aware views expressed over these relation parts:
\begin{align*}
Q_{rst}() = \sum\limits_{a,b,c} R_r(a,b)\ztimes S_s(b,c) \ztimes T_t(c,a), \quad\text{ for } 
r,s,t \in \{h,l\}.
\end{align*}
Query~\eqref{query:triangle} is then the sum of these skew-aware views: $Q() = 
\textstyle\sum_{r,s,t \in \{h,l\}} Q_{rst}()$.

\ivme adapts its maintenance strategy to each skew-aware view $Q_{rst}$ to ensure amortized update time that is sublinear in the database size. While most of these views admit sublinear delta computation over the relation parts, few exceptions require linear-time maintenance. For these exceptions, \ivme precomputes the update-independent parts of the delta queries as \emph{auxiliary materialized views} and then exploits these views to speed up the delta evaluation. 

One such exception is the view $Q_{hhl}$. Consider a single-tuple update 
$\delta R_h = \{(\deltaA,\deltaB) \mapsto \p\}$ to the heavy part $R_h$ of relation $R$, 
where $\deltaA$ and $\deltaB$ are fixed data values.
Computing the delta view 
$\delta Q_{hhl}() = \delta R_h(\deltaA,\deltaB) \ztimes \textstyle\sum_c  S_h(\deltaB,c) \ztimes T_l(c,\deltaA)$ requires iterating over all the $C$-values $c$ paired with $\deltaB$ in $S_h$ and with $\deltaA$ in $T_l$; the number of such $C$-values can be linear in the size of the database. To avoid this iteration, \ivme precomputes the view 
$V_{ST}(b,a) = \sum_c S_h(b,c) \ztimes T_l(c,a)$  and uses this view to evaluate 
$\delta Q_{hhl}() = \delta R_h(\deltaA,\deltaB) \ztimes V_{ST}(\deltaB,\deltaA)$ in constant time.

Such auxiliary views, however, also require maintenance. All such views 
created by \ivme can be maintained in sublinear time under single-tuple updates to the input relations. Figure~\ref{fig:view_definitions} summarizes these views used by \ivme to maintain Query~\eqref{query:triangle}: $V_{RS}$, $V_{ST}$ and $V_{TR}$. They serve to avoid linear-time delta computation for updates to $T$, $R$, and $S$, respectively. \ivme also materializes the result of Query~\eqref{query:triangle}, which ensures constant answer time. 

\begin{figure}[t]
  \begin{center}
    \renewcommand{\arraystretch}{1.2}  
    \begin{tabular}{@{\hskip 0.05in}l@{\hskip 0.4in}l@{\hskip 0.05in}}
      \toprule
      Materialized View Definition & Space Complexity \\    
      \midrule
      $Q() = \sum\limits_{r,s,t \in \{h,l\}} \,
      \sum\limits_{a,b,c} R_{r}(a,b) \ztimes S_{s}(b,c) \ztimes T_{t}(c,a)$ & 
      $\bigO{1}$ \\
      $V_{RS}(a,c) = \sum_{b} R_{h}(a,b) \ztimes S_{l}(b,c)$ & 
      $\bigO{|\inst{D}|^{1+\min{\{\,\eps, 1-\eps \,\}}}}$ \\
      $V_{ST}(b,a) = \sum_{c} S_{h}(b,c) \ztimes T_{l}(c,a)$ & 
      $\bigO{|\inst{D}|^{1+\min{\{\,\eps, 1-\eps \,\}}}}$ \\
      $V_{TR}(c,b) = \sum_{a} T_{h}(c,a) \ztimes R_{l}(a,b)$ & 
      $\bigO{|\inst{D}|^{1+\min{\{\,\eps, 1-\eps \,\}}}}$ \\
      \bottomrule    
    \end{tabular}
  \end{center}
  \caption{The definition and space complexity of the materialized views in $\inst{V} = \{ Q, V_{RS}, V_{ST}, V_{TR} \}$  as part of an \ivme  state of a database $\inst{D}$
  partitioned for $\eps \in [0,1]$.} 
  \label{fig:view_definitions}
\end{figure}

\medskip
We now describe our strategy in detail. We start by defining the state that \ivme initially creates and maintains upon each update. Then, we specify the procedure for processing a single-tuple update to any input relation, followed by the space complexity analysis of \ivme. Section~\ref{sec:rebalancing} gives the procedure for processing a sequence of such updates.

\begin{definition}[\ivme State]\label{def:ivme_state}
Given a database $\db = \{ R, S, T\}$ and $\eps \in [0,1]$, an 
\ivme  state of $\db$ is a tuple $\state = (\eps,N, \dbeps, \inst{V})$, where: 
\begin{itemize}
\item $N$ is a natural number such that the size 
invariant $\floor{\frac{1}{4}N} \leq |\db| < N$ holds.
$N$ is called the threshold base.  

\item $\dbeps = \{R_h, R_l, S_h, S_l, T_h, T_l\}$ consists of the partitions of 
$R$, $S$, and $T$ on variables $A$, $B$, and $C$, respectively, with threshold 
$\theta = N^{\eps}$.

\item $\inst{V}$ is the set of materialized views $\{ Q, V_{RS}, V_{ST}, V_{TR} \}$
as defined in Figure~\ref{fig:view_definitions}. 
\end{itemize} 
The initial state $\state$ of $\db$ has $N = 2\ztimes|\db| + 1$ and the three partitions in $\inst{P}$ are strict.

\end{definition}
By construction, $|\dbeps|=|\db|$. 
The size invariant implies $|\db| = \Theta(N)$ and, 
together with the heavy and light part conditions, 
facilitates the amortized analysis of \ivme in Section~\ref{sec:rebalancing}.
Definition~\ref{def:loose_relation_partition} provides two essential upper bounds for each relation partition in an \ivme state:
The number of distinct $A$-values in $R_h$ is at most $\frac{N}{\frac{1}{2}N^{\eps}} = 2N^{1-\eps}$, i.e., $|\pi_A R_h| \leq 2N^{1-\eps}$, and the number of tuples in $R_l$ with an $A$-value $a$ is less than 
$\frac{3}{2}N^{\eps}$, i.e., $|\sigma_{A = a}R_l| < \frac{3}{2}N^{\eps}$, for any $a \in \Dom(A)$. The same bounds hold for $B$-values in $\{S_h,S_l\}$ and $C$-values in $\{T_h,T_l\}$.

\subsection{Preprocessing Stage}

The preprocessing stage constructs the initial \ivme state given a database $\db$ and $\eps\in[0,1]$.

\begin{proposition}\label{prop:preprocessing_step}
Given a database $\db$ and $\eps\in[0,1]$, constructing the initial \ivme state of $\db$ takes $\bigO{|\db|^{\frac{3}{2}}}$ time.
\end{proposition}

\begin{proof}
We analyze the time to construct the initial state $\state = (\eps, N, \inst{P}, \inst{V})$ of $\inst{D}$.
Retrieving the size $|\inst{D}|$ and computing $N = 2\ztimes|\db| + 1$ take constant time. 
Strictly partitioning the input relations from $\db$ using the threshold $N^{\eps}$, as described in Definition~\ref{def:loose_relation_partition}, takes $\bigO{|\db|}$ time.
Computing the result of the triangle count query on $\db$ (or $\dbeps$) using a worst-case optimal join algorithm~\cite{Ngo:SIGREC:2013} takes $\bigO{|\db|^{\frac{3}{2}}}$ time. Computing the auxiliary views $V_{RS}$, $V_{ST}$, and $V_{TR}$ takes $\bigO{|\db|^{1+\min\{\eps, 1-\eps\}}}$ time, as shown next.  
Consider the view 
$V_{RS}(a,c) = \textstyle\sum_b R_h(a,b) \ztimes S_l(b,c)$. 
To compute $V_{RS}$, one can iterate over all $(a,b)$ pairs in $R_h$ and then find the $C$-values in $S_l$ for each $b$. 
The light part $S_l$ contains at most $N^{\eps}$ distinct $C$-values for any $B$-value, which gives an upper bound of $|R_{h}| \ztimes N^{\eps}$ on the size of $V_{RS}$. Alternatively, one can iterate over all $(b,c)$ pairs in $S_l$ and then find the $A$-values in $R_h$ for each $b$. The heavy part $R_h$ contains at most $N^{1-\eps}$ distinct $A$-values, which gives an upper bound of $|S_{l}| \ztimes N^{1-\eps}$ on the size of $V_{RS}$. 
The number of steps needed to compute this result is upper-bounded by $\min\{\,|R_{h}| \ztimes N^{\eps},\, |S_{l}| \ztimes N^{1-\eps}\,\} < \min\{\, N \ztimes N^{\eps},\, N \ztimes N^{1-\eps}\,\} = N^{1+\min\{\eps,1-\eps\}}$. From $|\db| = \Theta(N)$ follows that computing $V_{RS}$ on the database partition $\inst{P}$ takes $\bigO{|\db|^{1+\min\{\eps,1-\eps\}}}$ time; the analysis for $V_{ST}$ and $V_{TR}$ is analogous.  
Note that $\max_{\eps\in[0,1]}\{1+\min\{\eps, 1-\eps\}\} = \frac{3}{2}$.
Overall, the initial state $\state$ of $\db$ can be constructed in $\bigO{|\db|^{\frac{3}{2}}}$ time.
\end{proof}

The preprocessing stage of \ivme happens \emph{before} any update is received. In case we start from an empty database, the preprocessing cost of \ivme is $\bigO{1}$.

\subsection{Processing a Single-Tuple Update}\label{sec:single-tuple-update}

We describe the \ivme strategy for maintaining the result of Query~\eqref{query:triangle} under a single-tuple update to the relation $R$. This update can affect either the heavy or light part of $R$, hence we write $\delta R_r$, where $r$ stands for $h$ or $l$. 
We assume that checking whether the update affects the heavy or light part of $R$ takes constant time. 
The update is represented as a relation $\delta R_r=\{\, (\deltaA,\deltaB) \mapsto \p \,\}$, where $\deltaA$ and $\deltaB$ are data values and $\p\in\mathbb{Z}$. 
Due to the symmetry of the triangle query and auxiliary views, updates to $S$ and $T$ are handled similarly.

\begin{figure}[t]
\begin{center}
\renewcommand{\arraystretch}{1.2}
\setcounter{magicrownumbers}{0}
\begin{tabular}{ll@{\hskip 0.25in}l@{\hspace{0.6cm}}c}
\toprule
\multicolumn{2}{l}{\textsc{ApplyUpdate}($\delta R_r,\state$)}& & Time \\
\cmidrule{1-2} \cmidrule{4-4}
\rownumber & \LET $\delta R_r = \{(\deltaA,\deltaB) \mapsto \p\}$ \\
\rownumber & \LET $\state = (\eps, N, 
\{R_h, R_l, S_h, S_l, T_h, T_l\},
\{Q,V_{RS}, V_{ST}, V_{TR}\})$ \\

\rownumber & $\delta Q_{rhh}() = \delta{R_{r}(\deltaA,\deltaB)} \ztimes \textstyle\sum_c 
S_{h}(\deltaB,c) \ztimes T_{h}(c,\deltaA)$
&& $\bigO{|\inst{D}|^{1-\eps}}$ \\

\rownumber & 
$\delta Q_{rhl}() = \delta{R_{r}(\deltaA,\deltaB)} \ztimes V_{ST}(\deltaB,\deltaA)$ & &
$\bigO{1}$ \\

\rownumber & $\delta Q_{rlh}() = \delta R_{r}(\deltaA,\deltaB) \ztimes \textstyle\sum_c 
S_{l}(\deltaB,c) \ztimes T_{h}(c,\deltaA)$  & &
$\bigO{|\inst{D}|^{\min{\{\eps, 1-\eps\}}}}$ \\ 

\rownumber & $\delta Q_{rll}() = \delta R_{r}(\deltaA,\deltaB) \ztimes \textstyle\sum_c 
S_{l}(\deltaB,c) \ztimes T_{l}(c,\deltaA)$  & &
$\bigO{|\inst{D}|^{\eps}}$ \\

\rownumber & 
$Q() = Q() + \delta Q_{rhh}() + \delta Q_{rhl}() + \delta Q_{rlh}() + 
\delta Q_{rll}()$ & &
$\bigO{1}$ \\

\rownumber & \IF ($r$ is $h$) &\\  
\rownumber & \TAB  
$V_{RS}(\deltaA,c) = V_{RS}(\deltaA,c) + \delta R_h(\deltaA,\deltaB) \ztimes 
S_{l}(\deltaB,c)$  & &
$\bigO{|\inst{D}|^{\eps}}$ \\

\rownumber & \ELSE &\\
\rownumber & \TAB 
 $V_{TR}(c,\deltaB) = V_{TR}(c,\deltaB) + T_{h}(c,\deltaA) \ztimes \delta R_l(\deltaA,\deltaB)$ & &
$\bigO{|\inst{D}|^{1-\eps}}$ \\

\rownumber & $R_r(\deltaA,\deltaB) = R_r(\deltaA,\deltaB) + \delta{R}_r(\deltaA,\deltaB)$ & &
$\bigO{1}$ \\

\rownumber & \RETURN 
$\state$ & &
 \\
\midrule
\multicolumn{2}{r}{Total update time:} & &
$\bigO{|\inst{D}|^{\max\{\eps, 1-\eps\}}}$ \\
\bottomrule
\end{tabular}
\end{center}\vspace{-1em}
\caption{
 (left) Counting triangles under a single-tuple update.
 \textsc{ApplyUpdate} takes as input an update $\delta R_r$ to the heavy or light part of $R$, hence $r \in \{h,l\}$, and 
  the current \ivme state $\state$ of a database $\inst{D}$ partitioned using $\eps\in[0,1]$.
It returns a new state that results from applying $\delta R_r$ to $\state$. 
Lines 3-6 compute the deltas of the affected skew-aware views, and Line 7 
maintains $Q$.
Lines 9 and 11 maintain the auxiliary views $V_{RS}$ and $V_{TR}$, respectively. Line 12 maintains the affected part $R_r$.
(right) The time complexity of computing and applying deltas.  
The evaluation strategy for computing $\delta Q_{rlh}$ in Line 5 may choose either 
$S_l$ or $T_h$ to bound $C$-values, depending on $\eps$. The total time is the maximum of all individual times. The maintenance procedures for $S$ and $T$ are similar.
}
\label{fig:applyUpdate}
\vspace{-5pt}
\end{figure}

Figure~\ref{fig:applyUpdate} shows the procedure {\textsc{ApplyUpdate}} that 
takes as input a current \ivme state $\state$ and the update 
$\delta R_r$, and returns a new state that results from applying $\delta R_r$ to $\state$.
The procedure computes the deltas of the skew-aware views referencing $R_{r}$, which are   
$\delta Q_{rhh}$ (Line 3), $\delta Q_{rhl}$ (Line 4), 
$\delta Q_{rlh}$ (Line 5), and $\delta Q_{rll}$ (Line 6), 
and uses these deltas to maintain the triangle count (Line 7).
These skew-aware views are not materialized, but their 
deltas facilitate the maintenance of the triangle count.  
If the update affects the heavy part $R_h$ of $R$, the procedure maintains 
$V_{RS}$ (Line 9) and $R_h$ (Line 12); otherwise, it maintains $V_{TR}$ (Line 11) and
$R_l$ (Line 12).
The view $V_{ST}$ remains unchanged as it has no reference to $R_h$ or $R_l$.

Figure~\ref{fig:applyUpdate} also gives the time complexity of computing these deltas and applying them to $\state$. This complexity is either constant or dependent on the number of $C$-values for which matching tuples in the parts of $S$ and $T$ have nonzero multiplicities.  

\begin{proposition}\label{prop:single_step_time}
Given a state $\state$ constructed from a database $\db$ for $\eps\in[0,1]$, \ivme maintains $\state$ under a single-tuple update to any input relation in $\bigO{|\db|^{\max\{\eps,1-\eps\}}}$ time.
\end{proposition} 
\begin{proof}
We analyze the running time of the procedure from 
Figure~\ref{fig:applyUpdate} given a single-tuple update $\delta R_r = \{(\deltaA,\deltaB) \mapsto \p\}$ and a state 
$\state=(\eps, N, \dbeps, \inst{V})$ of $\db$. Since the query and auxiliary views are symmetric, the analysis for updates to $S$ and $T$ is similar.

We first analyze the evaluation strategies for the deltas of the skew-aware views $Q_{rst}$:
\begin{itemize}
\item (Line 3) Computing $\delta Q_{rhh}$ requires summing over $C$-values ($\deltaA$ and $\deltaB$ are fixed). The minimum degree of each $C$-value in $T_h$ is $\frac{1}{2}N^{\eps}$, which means the number of distinct $C$-values in $T_h$ is at most $\frac{N}{\frac{1}{2}N^{\eps}} = 2N^{1-\eps}$. Thus, this delta evaluation takes $\bigO{N^{1-\eps}}$ time.

\item (Line 4) Computing $\delta Q_{rhl}$ requires constant-time lookups in $\delta R_r$ and $V_{ST}$.

\item (Line 5) Computing $\delta Q_{rlh}$ can be done in two ways, depending on $\eps$: either sum over at most $2N^{1-\eps}$ $C$-values in $T_h$ for the given $\deltaA$ or sum over at most $\frac{3}{2}N^{\eps}$ $C$-values in $S_l$ for the given $\deltaB$. This delta computation takes at most $\min\{2N^{1-\eps}, \frac{3}{2}N^{\eps}\}$ constant-time operations, thus $\bigO{N^{\min{\{\eps, 1-\eps\}}}}$ time.

\item (Line 6) Computing $\delta Q_{rll}$ requires summing over at most $\frac{3}{2}N^{\eps}$ $C$-values in $S_l$ for the given $\deltaB$. This delta computation takes $\bigO{N^{\eps}}$ time. 
\end{itemize}
Maintaining the result of Query~\eqref{query:triangle} using these deltas takes constant time (Line 7).
The views $V_{RS}$ and $V_{TR}$ are maintained for updates to distinct parts of R.
Maintaining $V_{RS}$ requires iterating over at most $\frac{3}{2}N^{\eps}$ $C$-values in $S_l$ for the given $\deltaB$ (Line 9);
similarly, maintaining $V_{TR}$ requires iterating over at most $2N^{1-\eps}$ $C$-values in $T_h$ for the given $\deltaA$ (Line 11).
Finally, maintaining the (heavy or light) part of $R$ affected by $\delta R_{r}$ takes constant time (Line 12).  
The total update time is 
$\bigO{\max\{1,N^{\eps}, N^{1-\eps}, N^{\min\{\eps,1-\eps\}}\}} = \bigO{N^{\max\{\eps,1-\eps\}}}$. 
From the invariant $|\db| = \Theta(N)$ follows the claimed time complexity $\bigO{|\db|^{\max\{\eps,1-\eps\}}}$.
\end{proof}

\subsection{Space Complexity}
We next analyze the space complexity of the \ivme maintenance strategy.
  
\begin{proposition}\label{prop:space_complexity}
Given a database $\db$ and $\eps\in[0,1]$, the \ivme state 
constructed from $\db$ to support the maintenance of the result of Query~\eqref{query:triangle} takes $\bigO{|\db|^{1 +\min\{\eps,1-\eps\}}}$ space.
\end{proposition}  
\begin{proof}
We consider a state $\state = (\eps, N, \dbeps, \inst{V})$ of database $\db$. 
$N$ and $\eps$ take constant space and $|\dbeps| = |\inst{D}|$. 
Figure~\ref{fig:view_definitions} summarizes the space complexity of the materialized views $Q$, $V_{RS}$, $V_{ST}$, and $V_{TR}$ from $\inst{V}$. 
The result of $Q$ takes constant space.
 As discussed in the proof of Proposition~\ref{prop:preprocessing_step}, to compute 
 the auxiliary view $V_{RS}(a,c) = \textstyle\sum_b R_h(a,b) \ztimes 
S_{l}(b,c)$, we can use either $R_h$ or $S_l$ as the outer relation:
\begin{align*}
\!\!\!\!\!|V_{RS}|
\,\leq\, \min\{\, |R_h| \ztimes\!\! \max_{b\in\pi_B S_l}\!|\sigma_{B=b}S_l|,\, |S_l| \ztimes\!\! \max_{b\in\pi_B R_h}\!|\sigma_{B=b}R_h| \,\} 
\,<\, \min\{\, N \ztimes \frac{3}{2}N^{\eps}, N \ztimes 2N^{1-\eps} \,\}
\end{align*} 
The size of $V_{RS}$ is thus $\bigO{N^{1+\min\{\eps, 1-\eps\}}}$. From $|\db|=\Theta(N)$ follows that $V_{RS}$ takes $\bigO{|\db|^{1+\min\{\eps,1-\eps\}}}$ space; the space analysis for $V_{ST}$ and $V_{TR}$ is analogous. Overall, the state $\state$ of $\db$ takes $\bigO{|\db|^{1+\min\{\eps,1-\eps\}}}$ space.
\end{proof}

\section{Rebalancing Partitions}\label{sec:rebalancing}
The partition of a relation may change after updates.
For instance,  an insert 
$\delta R_l = \{(\deltaA,\deltaB) \mapsto 1\}$  may violate the size invariant $\floor{\frac{1}{4}N} \leq |\db| < N$ or may violate the light part condition $|\sigma_{A=\deltaA}R_l| < \frac{3}{2}N^{\eps}$ and require moving all tuples with the $A$-value $\deltaA$ from $R_l$ to $R_h$. 
As the database evolves under updates, \ivme performs \emph{major} and \emph{minor} rebalancing steps to ensure the size invariant and the conditions for heavy and light parts of each partition always hold. This rebalancing also ensures that the upper bounds on the number of data values, such as the number of $B$-values paired with $\deltaA$ in $R_l$ and the number of distinct $A$-values in $R_h$, are valid. The rebalancing cost is amortized over multiple updates. 

\begin{figure}[t]
\begin{center}
\begin{tikzpicture}
\node at(0,0)[anchor=north] {
\renewcommand{\arraystretch}{1.1}
\begin{tabular}{@{\hskip 0.02in}l@{\hskip 0.02in}}
\toprule
{\textsc{OnUpdate}}($\delta R,\state$)\\
\midrule
\LET $\delta R = \{(\deltaA,\deltaB) \mapsto \p\}$  \\
\LET $\state = (\eps, N, \{R_h,R_l\} \cup \dbeps, \inst{V})$ \\
\IF ($\deltaA \in \pi_A R_h$ \OR $\eps=0$) \\
\TAB {$\state$ = \textsc{ApplyUpdate}}($\delta R_h = \{(\deltaA,\deltaB) \mapsto \p\},\state$)\\
\ELSE \\
\TAB $\state$ = {\textsc{ApplyUpdate}}($\delta R_l = \{(\deltaA,\deltaB) \mapsto \p\},\state$)\\
\IF($|\inst{D}| = N$)\\
\TAB $N = 2N$ \\
\TAB $\state$ = {\textsc{MajorRebalancing}}($\state$) \\
\ELSE \IF($|\inst{D}| < \floor{\frac{1}{4}N}$) \\
\TAB $N = \floor{\frac{1}{2}N}-1$ \\
\TAB $\state$ = {\textsc{MajorRebalancing}}($\state$) \\
\ELSE \IF($\deltaA \in \pi_{A}R_l$ \AND $|\sigma_{A=\deltaA} R_l| \geq \frac{3}{2}N^{\eps}$) \\
\TAB $\state$ = {\textsc{MinorRebalancing}}($R_l, R_h, A, \deltaA, \state$) \\
\ELSE \IF($\deltaA \in \pi_{A}R_h$ \AND $|\sigma_{A=\deltaA} R_h| < \frac{1}{2}N^{\eps}$) \\
\TAB $\state$ =  {\textsc{MinorRebalancing}}($R_h, R_l, A, \deltaA, \state$)\\
\RETURN $\state$\\
\bottomrule
\end{tabular}
};

\node at(4.2,0)[anchor=north west] {
\renewcommand{\arraystretch}{1.1}
    \begin{tabular}{@{\hskip 0.02in}l@{\hskip 0.02in}}
\toprule
{\textsc{MajorRebalancing}}($\state$)\\
\midrule 
\LET $\state = (\eps,N, \{R_h, R_l, S_h, S_l, T_h, T_l\},$ \\
\TAB\TAB\TAB\TAB $\{ Q, V_{RS}, V_{ST}, V_{TR} \})$ \\
$\{ R_h, R_l \} = \textsc{StrictPartition}(R_h,R_l, A, N^{\eps})$ \\
$\{ S_h, S_l \} = \textsc{StrictPartition}(S_h, S_l, B, N^{\eps})$ \\
$\{ T_h, T_l \} = \textsc{StrictPartition}(T_h, T_l, C, N^{\eps})$ \\
$V_{RS}(a,c) = \textstyle\sum_b R_{h}(a,b) \ztimes S_{l}(b,c)$\\
$V_{ST}(b,a) = \textstyle\sum_c S_{h}(b,c) \ztimes T_{l}(c,a)$\\
$V_{TR}(c,b) = \textstyle\sum_a T_{h}(c,a) \ztimes R_{l}(a,b)$\\
\RETURN $\state$ \\
\bottomrule
\end{tabular}
};

\node at(4.2,-8.99)[anchor=south west] {
\renewcommand{\arraystretch}{1.1}
\begin{tabular}{@{\hskip 0.02in}l@{\hskip 0.02in}}
\toprule
{\textsc{MinorRebalancing}}($K_{\mathit{src}}, K_{\mathit{dst}}, X, x, \state$)\\
\midrule
\FOREACH $\vecnormal{t} \in \sigma_{X=x} K_{\mathit{src}}$ \DO \\
\TAB $m = K_{\mathit{src}}(\vecnormal{t})$ \\
\TAB $\state$ = {\textsc{ApplyUpdate}}($\delta K_{\mathit{src}} = \{\, \vecnormal{t} \mapsto -m \,\},\state$)\\ 
\TAB $\state$ = {\textsc{ApplyUpdate}}($\delta K_{\mathit{dst}} = \{\, \vecnormal{t} \mapsto m \,\},\state$)\\
\RETURN $\state$\\
\bottomrule
\end{tabular}
};
\end{tikzpicture}
\end{center}
\vspace{-1.2em}
\caption{
Counting triangles under a single-tuple update with rebalancing. 
\textsc{OnUpdate} takes as input an update $\delta R$ and 
the current \ivme state $\state$ of a database $\inst{D}$. 
It returns a new state that results from applying 
$\delta R$ to $\state$ and, if necessary, rebalancing partitions.
The condition $\eps=0$ in the third line ensures that all tuples are in $R_h$ when $\eps=0$.
\textsc{ApplyUpdate} is given in Figure~\ref{fig:applyUpdate}.
\textsc{MinorRebalancing}($K_{\mathit{src}}, K_{\mathit{dst}}, X, x, \state)$ 
moves all tuples with the $X$-value $x$ from $K_{\mathit{src}}$ to $K_{\mathit{dst}}$. 
\textsc{MajorRebalancing}($\state$) recomputes the relation partitions and views
in $\state$. 
\textsc{StrictPartition}($K_h,K_l,X,\theta$) constructs a strict partition of relation $K$
on variable $X$ with threshold $\theta$ 
(see Definition~\ref{def:loose_relation_partition}).
The \textsc{OnUpdate} procedures for updates to relations $S$ and $T$ are analogous. 
}
\label{fig:updateProgram}
\vspace{-0.45em}
\end{figure}

\paragraph*{Major Rebalancing}
If an update causes the database size to fall below $\lfloor \frac{1}{4} N \rfloor$ or reach $N$, \ivme halves or, respectively, doubles $N$, followed by strictly repartitioning the database with the new threshold $N^{\eps}$ and recomputing the materialized views, as shown in Figure~\ref{fig:updateProgram}. 

\begin{proposition}\label{prop:major_cost}
Given $\eps\in[0,1]$, major rebalancing of an \ivme state constructed from a database $\db$ takes $\bigO{|\db|^{1+\min\{\eps,1-\eps\}}}$ time.
\end{proposition}
\begin{proof}
We consider the major rebalancing procedure from Figure~\ref{fig:updateProgram}.
Strictly partitioning the input relations takes $\bigO{|\db|}$ time. From the proof of Proposition~\ref{prop:preprocessing_step} and $|\db|=\Theta(N)$ follow that recomputing $V_{RS}$, $V_{ST}$, and $V_{TR}$ takes $\bigO{|\db|^{1+\min\{\eps,1-\eps\}}}$ time.
\end{proof}

The (super)linear time of major rebalancing is amortized over $\Omega{(N)}$ updates. After a major rebalancing step, it holds that $|\db| = \frac{1}{2}N$ (after doubling), or $|\db| = \frac{1}{2}N - \frac{1}{2}$ or $|\db| = \frac{1}{2}N - 1$ (after halving, i.e., setting $N$ to $\floor{\frac{1}{2}N}-1$; the two options are due to the floor functions in the size invariant and halving expression). To violate the size invariant $\floor{\frac{1}{4}N} \leq |\db| < N$ and trigger another major rebalancing, the number of required updates is at least $\frac{1}{4}N$.
Section~\ref{sec:main_proof} proves the amortized $\bigO{|\db|^{\min\{\eps,1-\eps\}}}$ time of major rebalancing.

\paragraph*{Minor Rebalancing}
After each update $\delta R = \{(\deltaA,\deltaB) \mapsto \p\}$, \ivme checks whether the two conditions $|\sigma_{A=\deltaA} R_h| \geq \frac{1}{2}N^{\eps}$ and $|\sigma_{A=\deltaA} R_l| < \frac{3}{2}N^{\eps}$ still hold. 
If the first condition is violated, all tuples in $R_h$ with the $A$-value $\deltaA$ are moved to $R_l$ and the affected views are updated; similarly, if the second condition is violated, all tuples with the $A$-value $\deltaA$ are moved from $R_l$ to $R_h$, followed by updating the affected views. Figure~\ref{fig:updateProgram} shows the procedure for minor rebalancing, which deletes affected tuples from one part and inserts them into the other part. 

\begin{proposition}\label{prop:minor_cost}
Given $\eps\in[0,1]$, minor rebalancing of an \ivme state constructed from a database $\db$ takes $\bigO{|\db|^{\eps+\max\{\eps,1-\eps\}}}$ time.
\end{proposition}
\begin{proof}
Consider a state $\state=(\eps, N, \dbeps, \inst{V})$. Minor rebalancing moves fewer than $\frac{1}{2}N^{\eps}$ tuples (from heavy to light) or fewer than $\frac{3}{2}N^{\eps} + 1$ tuples (from light to heavy). 
Each tuple move performs one delete and one insert and costs $\bigO{|\db|^{\max\{\eps,1-\eps\}}}$ by Proposition~\ref{prop:single_step_time}. Since there are $\bigO{N^{\eps}}$ such operations and $|\db|=\Theta(N)$, the total time is $\bigO{|\db|^{\eps+\max\{\eps,1-\eps\}}}$.
\end{proof}

The (super)linear time of minor rebalancing is amortized over $\Omega(N^{\eps})$ updates. 
This lower bound on the number of updates comes from the heavy and light part conditions (cf. Definition~\ref{def:loose_relation_partition}), namely from the gap between the two thresholds in these conditions. Section~\ref{sec:main_proof} proves the amortized $\bigO{|\db|^{\max\{\eps,1-\eps\}}}$ time of minor rebalancing.

Figure~\ref{fig:updateProgram} gives the trigger procedure \textsc{OnUpdate} that maintains Query~\eqref{query:triangle} under a single-tuple update to relation $R$ and, if necessary, rebalances partitions; the procedures for updates to $S$ and $T$ are analogous.
Given an update $\delta R = \{(\deltaA,\deltaB) \mapsto \p\}$ and an \ivme state of a database $\inst{D}$, the procedure first checks in constant time whether the update affects the heavy or light part of $R$. 
The update targets $R_h$ if there exists a tuple with the same $A$-value $\deltaA$ already in $R_h$, or $\eps$ is set to $0$; otherwise, the update targets $R_l$.
When $\eps = 0$, all tuples are in $R_h$, while $R_l$ remains empty.
Although this behavior is not required by \ivme (without the $\eps=0$ condition, $R_l$ would contain only tuples whose $A$-values have the degree of $1$, and $R_h$ would contain all other tuples), it allows us to recover existing IVM approaches, such as classical IVM and factorized IVM, which do not partition relations; by setting $\eps$ to $0$ or $1$, \ivme ensures that all tuples are in $R_h$ or respectively $R_l$.
The procedure \textsc{OnUpdate} then invokes \textsc{ApplyUpdate} from Figure~\ref{fig:applyUpdate}.
If the update causes a violation of the size invariant $\floor{\frac{1}{4}N} \leq |\db| < N$, the procedure invokes \textsc{MajorRebalancing} to recompute the relation partitions and auxiliary views (note that major rebalancing has no effect on the triangle count). 
Otherwise, if the heavy or light part condition is violated, \textsc{MinorRebalancing} moves all tuples with the given $A$-value $\deltaA$ from the source part to the destination part of $R$. 

\subsection{Proof of Theorem~\ref{theo:main_result}}\label{sec:main_proof}
We are now ready to prove Theorem~\ref{theo:main_result} that states the complexity of \ivme.

\begin{proof}
The preprocessing stage constructs the initial \ivme state from a database $\inst{D}$ in $\bigO{|\inst{D}|^{\frac{3}{2}}}$ time, as shown in Proposition~\ref{prop:preprocessing_step}.
Materializing the query result ensures constant answer time. 
The space complexity $\bigO{|\inst{D}|^{1 + \min\{\eps,1-\eps\}}}$ follows from Proposition~\ref{prop:space_complexity}.

We next analyze the amortized update time complexity.
Let $\state_0 = (\eps,N_{0},\inst{P}_{0},\inst{V}_{0})$ be the initial \ivme state of a database $\db_0$ and 
$\upd_0, \upd_1, \ldots ,\upd_{n-1}$ a sequence of arbitrary single-tuple updates. 
The application of this update sequence to $\state_0$ yields a sequence 
$\state_0 \overset{\upd_0}{\longrightarrow} \state_1 \overset{\upd_1}{\longrightarrow} \ldots \overset{\upd_{n-1}}{\longrightarrow} \state_{n}$
of \ivme states, 
where $\state_{i+1}$ is the result of executing the procedure $\textsc{OnUpdate}(\upd_{i},\state_{i})$ from Figure~\ref{fig:updateProgram}, for $0\leq i < n$.
Let $c_i$ denote the actual execution cost of $\textsc{OnUpdate}(\upd_i,\state_{i})$. 
For some $\Gamma>0$, we can decompose each $c_i$ as:
\begin{align*}
c_i = c_i^{\update} + c_i^{\major} + c_i^{\minor} + \Gamma, \text{\qquad for } 0 \leq i < n,
\end{align*} 
where $c_i^{\update}$, $c_i^{\major}$, and $c_i^{\minor}$ are the actual costs of the subprocedures \textsc{ApplyUpdate}, \textsc{MajorRebalancing}, and
\textsc{MinorRebalancing}, respectively, in \textsc{OnUpdate}.
If update $\upd_i$ causes no major rebalancing, then $c_i^{\major} = 0$; 
similarly, if $\upd_i$ causes no minor rebalancing, then $c_i^{\minor} = 0$. 
These actual costs admit the following worst-case upper bounds: 
\\[6pt]
\begin{tabular}{r@{\hskip 0.05in}l@{\hskip 0.1in}l}
$c^{\update}_i$ & $\leq \gamma N_{i}^{\max\{\eps,1-\eps\}}$ & (by Proposition~\ref{prop:single_step_time}), \\[4pt]
$c^{\major}_i$ & $\leq \gamma N_{i}^{1 + \min\{\eps,1-\eps\}}$& {(by Proposition~\ref{prop:major_cost})}, and \\[4pt]
$c^{\minor}_i$ & $\leq \gamma N_{i}^{\eps+\max\{\eps,1-\eps\}}$ & {(by Proposition~\ref{prop:minor_cost})}, 
\end{tabular}\\[6pt]
where $\gamma$ is a constant derived from their asymptotic bounds, and $N_i$ is the threshold base of $\state_i$.
The actual costs of major and minor rebalancing can be superlinear in the database size. 

The crux of this proof is to show that assigning a {\em sublinear amortized cost} $\hat{c}_{i}$ to each update $\upd_{i}$ accumulates enough budget to pay for such expensive but less frequent rebalancing procedures.
For any sequence of $n$ updates, our goal is to show that the accumulated amortized cost is no smaller than the accumulated actual cost:
\begin{align}\label{eq:inequality_proof}
\sum_{i=0}^{n-1} \hat{c}_i \geq \sum_{i=0}^{n-1}c_i.
\end{align}  
The amortized cost assigned to an update $\upd_i$ is $\hat{c}_i = \hat{c}^{\update}_i + \hat{c}^{\major}_i + \hat{c}^{\minor}_i + \Gamma$, where
\begin{align*}
\hat{c}^{\update}_i = \gamma N_{i}^{\max\{\eps,1-\eps\}}, \quad
\hat{c}^{\major}_i = 4\gamma N_{i}^{\min\{\eps,1-\eps\}}, \quad
\hat{c}^{\minor}_i = 2\gamma N_{i}^{\max\{\eps,1-\eps\}}, \quad \text{ and }
\end{align*}
$\Gamma$ and $\gamma$ are the constants used to upper bound the actual cost of \textsc{OnUpdate}.
In contrast to the actual costs $c^{\major}_i$ and $c^{\minor}_i$, the amortized costs $\hat{c}^{\major}_i$ and $\hat{c}^{\minor}_i$ are always nonzero. 

We prove that such amortized costs satisfy Inequality~\eqref{eq:inequality_proof}.
Since $\hat{c}^{\update}_i \geq c^{\update}_i $ for $0 \leq  i < n$, it suffices to show that the following inequalities hold:

\vspace{-12pt}
\begin{align}
\label{ineq:major}
\text{\em (amortizing major rebalancing) \TAB} & \sum_{i=0}^{n-1} \hat{c}_i^{\major}  \geq \sum_{i=0}^{n-1}c_i^{\major} \qquad\text{ and}\\
\label{ineq:minor}
\text{\em (amortizing minor rebalancing) \TAB} & \sum_{i=0}^{n-1} \hat{c}_i^{\minor}  \geq \sum_{i=0}^{n-1}c_i^{\minor}. 
\end{align}
We prove Inequalities~\eqref{ineq:major} and \eqref{ineq:minor} by induction on the length $n$ of the update sequence.

\vspace{-10pt}
\subparagraph*{Major rebalancing.}
\begin{itemize}
\item {\em Base case}:
We show that Inequality~\eqref{ineq:major} holds for $n = 1$. 
The preprocessing stage sets $N_{0} = 2\ztimes|\inst{D}_0| + 1$. 
If the initial database $\inst{D}_0$ is empty, then $N_{0}=1$ and $\upd_0$ triggers major rebalancing (and no minor rebalancing).  
The amortized cost $\hat{c}_0^{\major} = 4 \gamma N_0^{\min\{\eps,1-\eps\}} = 4\gamma$ suffices to cover the actual cost $c_0^{\major}\leq \gamma N_0^{1 + \min\{\eps,1-\eps\}} = \gamma$.
If the initial database is nonempty, $\upd_0$ cannot trigger major rebalancing (i.e., violate the size invariant) because 
$\floor{\frac{1}{4}N_{0}} = \floor{\frac{1}{2}|\inst{D}_0|} \leq |\inst{D}_0| -1$ (lower threshold) and 
$|\inst{D}_0| + 1 < N_{0} = 2\ztimes|\inst{D}_0| + 1$ (upper threshold); 
then, $\hat{c}_0^{\major} \geq c_0^{\major} = 0$. Thus, Inequality~\eqref{ineq:major} holds for $n = 1$. 

\item {\em Inductive step}:
Assumed that Inequality~\eqref{ineq:major} holds for all update sequences of length up to $n-1$, we show it holds for update sequences of length $n$.
If update $\upd_{n-1}$ causes no major rebalancing, then 
$\hat{c}_{n-1}^{\major} = 4 \gamma N_{n-1}^{\min\{\eps, 1-\eps\}} \geq 0$ and $c_{n-1}^{\major} = 0$, thus Inequality~\eqref{ineq:major} holds for $n$. 
Otherwise, if applying $\upd_{n-1}$ violates the size invariant, 
the database size $|\inst{D}_{n}|$ is either 
$\floor{\frac{1}{4}N_{n-1}}-1$ or $N_{n-1}$. 
Let $\state_j$ be the state created after the previous major rebalancing or, 
if there is no such step, the initial state. 
For the former ($j>0$), the major rebalancing step ensures
$|\inst{D}_j| = \frac{1}{2}N_j$ after doubling and 
$|\inst{D}_j| = \frac{1}{2}N_j - \frac{1}{2}$ or $|\inst{D}_j| = \frac{1}{2}N_j - 1$ after halving the threshold base $N_j$;
for the latter ($j=0$), the preprocessing stage ensures $|\inst{D}_j| = \frac{1}{2}N_j - \frac{1}{2}$.
The threshold base $N_j$ changes only with major rebalancing, thus $N_j = N_{j+1}=\ldots=N_{n-1}$.
The number of updates needed to change the database size from 
$|\inst{D}_j|$ to $|\inst{D}_{n}|$ (i.e., between two major rebalancing) is at least $\frac{1}{4}N_{n-1}$ since
$\min\{ \frac{1}{2}N_j - 1 - (\floor{\frac{1}{4}N_{n-1}} - 1), N_{n-1} - \frac{1}{2}N_j \} \geq \frac{1}{4}N_{n-1}$. 
Then,

\vspace{-13pt}
\begin{align*}
  \sum_{i=0}^{n-1} \hat{c}_i^{\major} &\geq \sum_{i=0}^{j-1} c_i^{\major} + \sum_{i=j}^{n-1} \hat{c}_i^{\major} \qquad\qquad\quad && \textit{(by induction hypothesis)}\\
  & = \sum_{i=0}^{j-1} c_i^{\major} + \sum_{i=j}^{n-1} 4\gamma 
  N_{n-1}^{\min\{\eps,1-\eps\}} && \textit{($N_j = \ldots = N_{n-1}$)}\\
  & \geq \sum_{i=0}^{j-1} c_i^{\major} + \frac{1}{4}N_{n-1} \,4 \gamma
  N_{n-1}^{\min\{\eps,1-\eps\}} && \textit{(at least $\frac{1}{4}N_{n-1}$ updates)}\\
  & = \sum_{i=0}^{j-1} c_i^{\major} + \gamma N_{n-1}^{1 + \min\{\eps,1-\eps\}} \\
  & \geq \sum_{i=0}^{j-1} c_i^{\major} + c_{n-1}^{\major} = \sum_{i=0}^{n-1} c_i^{\major} && \textit{($c_{j}^{\major} = \ldots = c_{n-2}^{\major} = 0$)}.
\end{align*}
\end{itemize}

\vspace{-1pt}
\noindent Thus, Inequality~\eqref{ineq:major} holds for update sequences of length $n$.

\subparagraph*{Minor rebalancing.}
When the degree of a value in a partition changes such that the heavy or light part condition no longer holds, minor rebalancing moves the affected tuples between the heavy and light parts of the partition.
To prove Inequality~\eqref{ineq:minor}, we decompose the cost of minor rebalancing per relation and data value of its partitioning variable. 
\begin{align*}
c_i^{\minor} & =  \sum_{a \in \Dom(A)} c_i^{R,a} +  \sum_{b \in \Dom(B)}  c_i^{S,b} +  \sum_{c \in \Dom(C)}  c_i^{T,c} \quad\text{ and } \\[3pt]
\hat{c}_i^{\minor} & =  \sum_{a \in \Dom(A)}  \hat{c}_i^{R,a} +  \sum_{b \in \Dom(B)}  \hat{c}_i^{S,b} +  \sum_{c \in \Dom(C)}  \hat{c}_i^{T,c}
\end{align*}

We write $c_i^{R,a}$ and $\hat{c}_i^{R,a}$ to denote the actual and respectively amortized costs of minor rebalancing caused by update $\upd_i$, for relation $R$ and an $A$-value $a$.
If update $\upd_i$ is of the form $\delta R = \{(\deltaA,\deltaB) \mapsto \p\}$ and causes minor rebalancing, then $c_i^{R,\deltaA} = c_i^\minor$; otherwise, $c_i^{R,\deltaA} = 0$.
If update $\upd_i$ is of the form $\delta R = \{(\deltaA,\deltaB) \mapsto \p\}$, then $\hat{c}_i^{R,\deltaA} = \hat{c}_i^\minor$ regardless of whether $\upd_i$ causes minor rebalancing or not; otherwise, $\hat{c}_i^{R,\deltaA} = 0$.  
The actual costs $c_i^{S,b}$ and $c_i^{T,c}$ and the amortized costs $\hat{c}_i^{S,b}$ and $\hat{c}_i^{T,c}$ are defined similarly.

We prove that for the partition of $R$ and any $\deltaA \in \Dom(A)$ the following inequality holds:
\begin{align}
\label{ineq:minor_decomposed}
\sum_{i=0}^{n-1} \hat{c}_i^{R,\deltaA}  \geq \sum_{i=0}^{n-1}c_i^{R,\deltaA}.
\end{align}
Due to the symmetry of the triangle query, Inequality~\eqref{ineq:minor} follows directly from Inequality~\eqref{ineq:minor_decomposed}.

We prove Inequality~\eqref{ineq:minor_decomposed} for an arbitrary $\deltaA \in \Dom(A)$ by induction on the length $n$ of the update sequence.
\begin{itemize}
\item {\em Base case}:
We show that Inequality~\eqref{ineq:minor_decomposed} holds for $n = 1$. 
Assume that update $\upd_0$ is of the form $\delta R = \{(\deltaA,\deltaB) \mapsto \p\}$; 
otherwise, $\hat{c}_0^{R,\deltaA} = c_0^{R,\deltaA} = 0$, and Inequality~\eqref{ineq:minor_decomposed} follows trivially for $n=1$.  
If the initial database is empty, $\upd_0$ triggers major rebalancing but no minor rebalancing, thus $\hat{c}_0^{R,\deltaA} = 2\gamma N_{0}^{\max\{\eps,1-\eps\}} \geq  c_0^{R,\deltaA} = 0$.
If the initial database is nonempty, each relation is partitioned using the threshold $N_0^{\eps}$. 
For update $\upd_0$ to trigger minor rebalancing, the degree of the $A$-value $\deltaA$ in $R_h$ or $R_l$ 
has to either decrease from $\ceil{N_0^{\eps}}$ to $\ceil{\frac{1}{2}N_0^{\eps}}-1$ (heavy to light) or increase from $\ceil{N_0^{\eps}}-1$ to $\ceil{\frac{3}{2}N_0^{\eps}}$ (light to heavy). The former happens only if $\ceil{N_0^{\eps}}=1$ and update $\upd_0$ removes the last tuple with the $A$-value $\deltaA$ from $R_h$, thus no minor rebalancing is needed; 
the latter cannot happen since update $\upd_0$ can increase $|\sigma_{A=\deltaA}R_l|$ to at most $\ceil{N_0^{\eps}}$, and $\ceil{N_0^{\eps}} < \ceil{\frac{3}{2}N_0^{\eps}}$. 
In any case, $\hat{c}_0^{R,\deltaA} \geq c_0^{R,\deltaA}$, which implies that Inequality~\eqref{ineq:minor_decomposed} holds for $n = 1$. 

\item {\em Inductive step}:
Assumed that Inequality~\eqref{ineq:minor_decomposed} holds for all update sequences of length up to $n-1$, we show it holds for update sequences of length $n$.
Consider that update $u_{n-1}$ is of the form $\delta R = \{(\deltaA,\deltaB) \mapsto \p\}$ and causes minor rebalancing; otherwise, $\hat{c}_{n-1}^{R,\deltaA} \geq 0$ and $c_{n-1}^{R,\deltaA} = 0$, and Inequality~\eqref{ineq:minor_decomposed} follows trivially for $n$.  
Let $\state_j$ be the state created after the previous major rebalancing or, if there is no such step, the initial state. The threshold changes only with major rebalancing, thus $N_j=N_{j+1}=\ldots=N_{n-1}$. 
Depending on whether there exist minor rebalancing steps since state $\state_j$, we distinguish two cases:

  \begin{itemize}
    \item[Case 1:] There is no minor rebalancing caused by an update of the form
    $\delta R = \{(\deltaA,\deltaB') \mapsto \p'\}$ since state $\state_j$; thus, 
    $c_{j}^{R,\deltaA} = \ldots = c_{n-2}^{R,\deltaA} = 0$.
    From state $\state_j$ to state $\state_{n}$, the number of tuples with the
    $A$-value $\deltaA$ either decreases from at least $\ceil{N_{j}^{\eps}}$ to $\ceil{\frac{1}{2}N_{n-1}^{\eps}}-1$ (heavy to light) or increases from at most $\ceil{N_{j}^{\eps}}-1$ to $\ceil{\frac{3}{2}N_{n-1}^{\eps}}$ (light to heavy). 
    For this change to happen, the number of updates needs to be greater than $\frac{1}{2}N_{n-1}^{\eps}$ since 
    $N_j=N_{n-1}$ and
    $\min\{ \ceil{N_{j}^{\eps}} - (\ceil{\frac{1}{2}N_{n-1}^{\eps}} - 1),
      \ceil{\frac{3}{2}N_{n-1}^{\eps}} - (\ceil{N_{j}^{\eps}} -1)
     \} > \frac{1}{2}N_{n-1}^{\eps}$. Then,

    \begin{align*}
      \sum_{i=0}^{n-1} \hat{c}_i^{R,\deltaA} &\geq \sum_{i=0}^{j-1} c_i^{R,\deltaA} + \sum_{i=j}^{n-1} \hat{c}_i^{R,\deltaA} \qquad\qquad\qquad && \textit{(by induction hypothesis)}\\
      & = \sum_{i=0}^{j-1} c_i^{R,\deltaA} + \sum_{i=j}^{n-1} 2 \gamma N_{n-1}^{\max\{\eps,1-\eps\}} && \textit{($N_j = \ldots = N_{n-1}$)}\\
      & > \sum_{i=0}^{j-1} c_i^{R,\deltaA} + \frac{1}{2}N_{n-1}^{\eps} 2 \gamma N_{n-1}^{\max\{\eps,1-\eps\}} && \textit{(more than $\frac{1}{2}N_{n-1}^{\eps}$ updates)}\\
      & \geq \sum_{i=0}^{j-1} c_i^{R,\deltaA} + c_{n-1}^{R,\deltaA} = \sum_{i=0}^{n-1} c_i^{R,\deltaA} && \textit{($c_{j}^{R,\deltaA} = \ldots = c_{n-2}^{R,\deltaA} = 0$)}.\\
    \end{align*}

    \vspace{-4pt}
    \item[Case 2:] There is at least one minor rebalancing step caused by an update of the form $\delta R = \{(\deltaA,\deltaB') \mapsto \p'\}$ since state $\state_j$. 
    Let $\state_\ell$ denote the state created after the previous minor rebalancing caused by an update of this form; thus, $c_{\ell}^{R,\deltaA} = \ldots = c_{n-2}^{R,\deltaA} = 0$.
    The minor rebalancing steps creating $\state_\ell$ and $\state_{n}$ move tuples with the $A$-value $\deltaA$ between $R_h$ and $R_l$ in opposite directions.
    From state $\state_\ell$ to state $\state_n$, the number of such tuples either decreases from $\ceil{\frac{3}{2}N_{l}^{\eps}}$ to $\ceil{\frac{1}{2}N_{n-1}^{\eps}}-1$ (heavy to light) or increases from $\ceil{\frac{1}{2}N_{l}^{\eps}}-1$ to $\ceil{\frac{3}{2}N_{n-1}^{\eps}}$ (light to heavy). 
    For this change to happen, the number of updates needs to be greater than $N_{n-1}^{\eps}$ since
    $N_l = N_{n-1}$ and
    $\min\{\ceil{\frac{3}{2}N_{l}^{\eps}} - (\ceil{\frac{1}{2}N_{n-1}^{\eps}}-1), 
      \ceil{\frac{3}{2}N_{n-1}^{\eps}} - (\ceil{\frac{1}{2}N_{l}^{\eps}} - 1)
    \} > N_{n-1}^{\eps}$. Then,
    \begin{align*}
      \hspace{-0.2cm}
      \sum_{i=0}^{n-1} \hat{c}_i^{R,\deltaA} &\geq \sum_{i=0}^{\ell-1} c_i^{R,\deltaA} + \sum_{i=\ell}^{n-1} \hat{c}_i^{R,\deltaA} \qquad\qquad\qquad && \textit{(by induction hypothesis)}\\
      & = \sum_{i=0}^{\ell-1} c_i^{R,\deltaA} + \sum_{i=\ell}^{n-1} 2 \gamma N_{n-1}^{\max\{\eps,1-\eps\}} && \textit{($N_j = \ldots = N_{n-1}$)} \\
      & > \sum_{i=0}^{\ell-1} c_i^{R,\deltaA} + N_{n-1}^{\eps} 2 \gamma N_{n-1}^{\max\{\eps,1-\eps\}} && \textit{(more than $N_{n-1}^{\eps}$ updates)} \\
      & > \sum_{i=0}^{\ell-1} c_i^{R,\deltaA} + c_{n-1}^{R,\deltaA} = \sum_{i=0}^{n-1} c_i^{R,\deltaA} && \textit{($c_{\ell}^{R,\deltaA} = \ldots = c_{n-2}^{R,\deltaA} = 0$)}.
    \end{align*}
  \end{itemize}
Cases 1 and 2 imply that Inequality~\eqref{ineq:minor_decomposed} holds for update sequences of length $n$.
\end{itemize}

This shows that Inequality~\eqref{eq:inequality_proof} holds when the amortized cost of $\textsc{OnUpdate}(\upd_i, \state_i)$ is
$$\hat{c}_{i} = 
\gamma N_{i}^{\max\{\eps,1-\eps\}} + 
4\gamma N_{i}^{\min\{\eps,1-\eps\}} +
2\gamma N_{i}^{\max\{\eps,1-\eps\}} + 
\Gamma, \text{\quad  for } 0\leq i < n,$$
where $\Gamma$ and $\gamma$ are constants. 
The amortized cost $\hat{c}^{\major}_i$ 
of major rebalancing is $4\gamma N_{i}^{\min\{\eps,1-\eps\}}$, 
and the amortized cost $\hat{c}^{\minor}_i$ 
 of minor rebalancing is $2\gamma N_{i}^{\max\{\eps,1-\eps\}}$. 
From the size invariant $\floor{\frac{1}{4}N_{i}} \leq |\db_{i}| < N_{i}$ follows that $|\db_i| < N_i < 4(|\db_i|+1)$ for $0\leq i < n$, where $|\db_i|$ is the database size before update $\upd_i$. 
This implies that for any database $\inst{D}$,
the amortized major rebalancing time is $\bigO{|\db|^{\min\{\eps, 1-\eps\}}}$,
the amortized minor rebalancing time is $\bigO{|\db|^{\max\{\eps, 1-\eps\}}}$, and
the overall amortized update time of \ivme is $\bigO{|\inst{D}|^{\max\{\eps, 1-\eps\}}}$.
\end{proof}

Given $\eps \in [0,1]$, \ivme maintains the triangle count query in 
$\bigO{|\inst{D}|^{\max\{\epsilon,1-\epsilon\}}}$ amortized update time while using $\bigO{|\inst{D}|^{1 + \min\{\epsilon,1-\epsilon\}}}$ space. 
It thus defines a tradeoff between time and space parameterized by $\eps$, as shown in Figure~\ref{fig:time_space_plot}.
\ivme achieves the optimal amortized update time $\bigO{|\inst{D}|^{\frac{1}{2}}}$ at $\eps = \frac{1}{2}$, for which the space is $\bigO{|\inst{D}|^{\frac{3}{2}}}$.

 \section{Conclusion and Future Work}
 
This paper introduces \ivme, an incremental maintenance approach to counting 
triangles under updates that exhibits a space-time tradeoff such that the space-time product is quadratic in the size of the database. \ivme can trade space for update time. The amortized update time can be as low as the square root of the database size, which is worst-case optimal conditioned on the Online Matrix-Vector Multiplication (\OMv) conjecture.
The space requirements of  \ivme can be improved to linear while keeping the amortized update time optimal by using a refined partitioning that takes into account the degrees of data values for both variables (instead of one variable only) in each relation 
(Appendix~\ref{sec:part_all_var}).
 \ivme captures classical and factorized IVM as special cases with suboptimal, linear update time (Appendix~\ref{sec:recovery}).

There are worst-case optimal algorithms for {\em join} queries in the {\em static} setting~\cite{Ngo:PODS:2018}. In contrast, \ivme is worst-case optimal for the {\em count} aggregate over the triangle join query in the {\em dynamic} setting. The latter setting poses challenges beyond the former. 
First, the optimality argument for static join algorithms follows from their runtime being linear(ithmic) in their output size; this argument does not apply to our triangle count query, since its output is a scalar and hence of constant size. 
Second, optimality in the dynamic setting requires a more fine-grained argument that exploits the skew in the data for different evaluation strategies, view materialization, and delta computation; 
in contrast, there are static worst-case optimal join algorithms 
that do not need to exploit skew, materialize views, nor delta computation. 

This paper opens up a line of work on {\em dynamic worst-case optimal query evaluation algorithms}. The goal is a complete characterization of the complexity of incremental maintenance for arbitrary functional aggregate queries over various rings~\cite{FAQ:PODS:2016}. 
Different rings can be used as the domain of tuple multiplicities (or payloads). We used here the ring $(\mathbb{Z},+,\ztimes,0,1)$ of integers to support counting. The relational data ring supports payloads with listing and factorized representations of relations, and the degree-$m$ matrix ring supports payloads with gradients used for learning linear regression models~\cite{NO18}.

Towards the aforementioned goal, we would first like to find a syntactical characterization of all queries that admit incremental maintenance in (amortized) sublinear time. Using known (first-order, fully recursive, or factorized) incremental maintenance techniques, cyclic and even acyclic joins require at least linear update time. Our intuition is that this characterization is given by a notion of diameter of the query hypergraph. This class strictly contains the q-hierarchical queries, which admit constant-time updates~\cite{BerkholzKS17}. 

Minor variants of \ivme can be used to maintain the counting versions of any query built using three relations (Appendix~\ref{sec:inc_maint_count_three_rels}), 
the 4-path query (Appendix~\ref{sec:4_path}), and the Loomis-Whitney queries
(Appendix~\ref{sec:loomis-whitney})
 in worst-case optimal time.
The same conditional lower bound on the update time shown for the triangle count applies for most of the mentioned queries, too.
This leads to the striking realization that, while in the static setting the counting versions of the cyclic query computing triangles and the acyclic query computing paths of length $3$ have different complexities, $\bigO{|\db|^{\frac{3}{2}}}$ and $\bigO{|\db|}$, and pose distinct computational challenges, they share the same complexity and can use a very similar approach in the dynamic setting. 
A further \ivme variant 
allows the constant-delay enumeration of all triangles 
after each update, while preserving the same optimal amortized update time as for counting 
triangles (Appendix~\ref{sec:enum_triangles}). These variants exploit the fact that
our amortization technique is agnostic to the query to maintain and the update mechanism. It relies on two prerequisites. First, rebalancing is performed by moving tuples between relation parts. Second, the number of moved tuples per rebalancing is asymptotically no more than the number of updates performed since the previous rebalancing.

\bibliographystyle{abbrv}
\bibliography{bibliography}
  
\appendix
\section{Recovering Existing Dynamic and Static Approaches}\label{sec:recovery}
In sections \ref{sec:classicalIVM} and \ref{sec:facIVM} we show how 
\ivme recovers the classical first order IVM~\cite{Chirkova:Views:2012:FTD}
and the factorized  IVM~\cite{NO18} for maintaining 
the triangle count.  
Section \ref{sec:staticCount} describes how \ivme 
counts all triangles in a static database in the worst-case optimal time 
to list them~\cite{Ngo:SIGREC:2013}.

We assume that $\inst{D}= \{R,S,T\}$ is the input database.

\subsection{The Classical First-Order IVM} 
\label{sec:classicalIVM}
The classical first order IVM materializes only the input relations
and the result of Query \eqref{query:triangle}.
Given a single-tuple update $\delta R = \{(\deltaA,\deltaB) \mapsto \p\}$
to relation $R$, 
it maintains Query \eqref{query:triangle} under $\delta R$ by computing 
the delta query $\delta Q() = \delta R(\deltaA,\deltaB)\ztimes \textstyle\sum_{c}  S(\deltaB,c) \ztimes T(c,\deltaA)$ and setting 
$Q() = Q() + \delta Q()$. The delta computation
needs $\bigO{|\inst{D}|}$ time, since it requires
the iteration over possibly linearly many $C$-values paired with $\deltaB$
in $S$ and with $\deltaA$ in $T$.  
The evaluation of updates to the relations 
$S$ and $T$ is analogous. 

To recover the classical IVM by \ivme, 
we choose $\epsilon = 0$ or $\epsilon = 1$.
In the former case, all tuples stay in the 
heavy relation parts, and in the latter case, they stay 
in the light relation parts.
In both cases, the auxiliary views stay empty, since 
each of them uses one light and one heavy relation part. 
Likewise, almost all skew-aware views 
stay empty as well. 
The only possibly nonempty 
skew-aware view 
returns the full triangle count. 
We explain both cases in more detail. 

We first  
consider the case  
$\eps=1$. 
The preprocessing stage 
sets the threshold base 
$N_0$ of the initial state $\state_0$ 
to $2 \cdot |\inst{D}| +1$
 and strictly 
 partitions 
each relation with threshold $N_0^{\eps} = N_0$.
Since, $|\sigma_{A=a}R| < N_0$ for all $A$-values
$a$, all tuples in $R$ end up in  the light part 
of $R$. Similarly, 
all tuples in $S$ or $T$ are assigned 
to the light parts of their 
relations.
Consequently, all auxiliary views 
in $\state_0$ are empty, since each of them
refers to at least one heavy relation part. 
The only possibly  
nonempty skew-aware view is
 $\delta Q_{lll}()$, which returns the full triangle count.
 Given an initial update $\delta R = \{(\deltaA, \deltaB) \mapsto \p\}$, the procedure  
 \textsc{OnUpdate} in Figure \ref{fig:updateProgram} 
 invokes \textsc{ApplyUpdate}
 from Figure \ref{fig:applyUpdate}  for the light part 
 of $R$, since $\deltaA \notin \pi_A R_h$
 and $\eps \neq 0$. 
  The degree of $\deltaA$ in 
 $R_l$ cannot reach  $\frac{3}{2} N_0^{\eps}$
 as a result of the update, hence,
\textsc{OnUpdate} does 
 not invoke minor rebalancing.
The procedure 
\textsc{MajorRebalancing}, which might be 
invoked by \textsc{OnUpdate},
does not move 
tuples to the heavy relation parts, 
since the partition threshold is always greater 
than  the database size.   
Similarly, all subsequent updates are propagated 
to the light relation parts and do not trigger movements 
between relation parts.
This means that 
the materialized auxiliary views $V_{RS}$ 
$V_{ST}$, and $V_{TR}$ as well as 
all skew-aware views referring to 
heavy relation parts  stay empty. 
The only materialized view maintained by 
 \textsc{ApplyUpdate} is the triangle count
 $Q$. 
 Hence, the space complexity is dominated 
 by the size of the input relations and is therefore 
 linear like for the classical IVM.
 Given any update $\delta R_l = \{(\deltaA,\deltaB) \mapsto \p\}$,
 \textsc{ApplyUpdate} computes the delta   
 $\delta Q_{lll}() = \delta R_l(\deltaA,\deltaB) \ztimes
 \textstyle\sum_{c} S_l(\deltaB,c)\ztimes T_l(c,\deltaA)$
 and sets $Q() = Q() + \delta Q_{lll}()$.
 Following the proof of 
 Proposition \ref{prop:single_step_time}, this 
 delta computation requires $\bigO{|\inst{D}|^{\eps}}
 = \bigO{|\inst{D}|}$ time. 
 Together with the linear 
 time needed for major rebalancing, the overall 
 worst-case update time is linear.

The case for  
$\eps = 0$ is analogous.
All tuples stay in the heavy relation parts, and the 
light parts are kept empty. 
This means
that all auxiliary views 
 $V_{RS}$, $V_{ST}$, and $V_{TR}$
 as well as 
all skew-aware views besides 
 $Q_{hhh}$  
 stay empty.
 Hence, the space complexity
 is linear. 
 Given an update
 $\delta R_h = \{(\deltaA,\deltaB) \mapsto \p\}$,
the procedure \textsc{ApplyUpdate}
maintains the result of $Q$ 
by computing  the delta 
 $\delta Q_{hhh}() = \delta R_h(\deltaA,\deltaB) \ztimes
 \textstyle\sum_{c} S_h(\deltaB,c) \ztimes T_h(c,\deltaA)$
 and setting $\delta Q() = Q() + \delta Q_{hhh}()$.
 Due to the proof of 
 Proposition \ref{prop:single_step_time}, this 
 delta computation 
requires $\bigO{|\inst{D}|^{1-\eps}}
= \bigO{|\inst{D}|}$ time in the worst case.

\subsection{The Factorized IVM}
\label{sec:facIVM}
The factorized IVM  materializes besides 
the input relations and the result of Query
\eqref{query:triangle}, an additional view to speed up 
updates to one relation. 
Here we choose this relation to be 
$R$; the cases for $S$ and $T$  are analogous. 
The materialized view is 
$\hat{V}_{ST}(b,a) = \textstyle\sum_c S(b,c) \ztimes T(c,a)$, which is of 
size $\bigO{\inst{|D}|^2}$. 
Updates $\delta R = \{(\deltaA, \deltaB) \mapsto \p\}$  to relation $R$
can be processed by computing 
$\delta Q() = \delta R(\deltaA, \deltaB) \ztimes 
V_{ST}(\deltaB,\deltaA)$
and setting $Q() = Q() + \delta Q()$  in overall 
constant time.
Updates to  
relations $S$ and $T$, however, affect not only $Q$, but also 
$V_{ST}$. Given an update 
$\delta S = \{(\deltaB, \deltaC) \mapsto \p\}$, the factorized IVM
maintains $Q$  by computing the delta
$\delta Q() = \delta S(\deltaB,\deltaC)\ztimes 
\textstyle\sum_{a}  R(a,\deltaB) \ztimes T(\deltaC,a)$ 
and setting $Q() = Q() + \delta Q$. The view
 $V_{ST}$ is maintained by computing  
$V_{ST}(\deltaB,a) = \delta S(\deltaB,\deltaC) \ztimes T(\deltaC,a)$.
Both computations need 
$\bigO{|\inst{D}|}$ time, as they require the iteration 
over possibly linearly many $A$-values.
Updates to $T$ are handled analogously.

To recover factorized IVM, we use a different parameter $\eps_K$ 
for each relation $K$: 
$\eps_R  \in \{0,1\}$, $\eps_S = 0$, and $\eps_T = 1$.
This means that in each \ivme state with threshold base 
$N$, relations $R$, $S$, and $T$ are partitioned with thresholds
$N^{\eps_R}$, $N^{\eps_S}$, and $N^{\eps_T}$, respectively. 
We consider the setting $\eps_R = \eps_S = 0$ and $\eps_T = 1$.
The other setting $\eps_R = \eps_T = 1$ and $\eps_S = 0$ is analogous. 
Similar to the case of recovering the classical IVM 
in Section~\ref{sec:classicalIVM}, 
the relations $R$ and $S$ are completely included in their heavy parts and $T$ is contained in its light part. 
The only possibly nonempty auxiliary and skew-aware views 
are  
$V_{ST}(b,a) = \textstyle\sum_c S_h(b,c) \ztimes T_l(c,a)$ 
and
 $\delta Q_{hhl}$.
 Following the space analysis in the proof of 
 Proposition \ref{prop:space_complexity},
 the view $V_{ST}$ is of size 
$\bigO{|\inst{D}|^{1+ \min\{\eps_T,1-\eps_S\}}} = \bigO{|\inst{D}|^{2}}$. 
Query $Q$
is maintained under updates 
$\delta R_h = \{(\deltaA,\deltaB) \mapsto \p\}$
by computing the delta 
$\delta Q_{hhl}() = \delta R_h(\deltaA, \deltaB)
\ztimes V_{ST}(\deltaB,\deltaA)$ 
and setting $Q() = Q() + \delta Q_{hhl}()$
in constant time. 
Updates to $S$ and $T$ affect besides Query $Q$,
the
view $V_{ST}$. Let  
$\delta S_h = \{(\deltaB, \deltaC) \mapsto \p\}$
be an update to $S_h$.
\ivme maintains $Q$ and 
$V_{ST}$ under $\delta S_h$ by computing:
\begin{itemize}
\item $\delta Q_{hhl}() = \delta S_h(\deltaB,\deltaC) \ztimes
\textstyle\sum_a R_h(a, \deltaB) \ztimes T_l(\deltaC,a)$,
\TAB  $Q() = Q() + \delta Q_{hhl}()$, and
\item $V_{ST}(\deltaB,a) = V_{ST}(\deltaB,a) + 
\delta S_h(\deltaB,\deltaC) \ztimes
T_l(\deltaC,a)$.
\end{itemize}
According to the analysis  
in the proof of  
Proposition \ref{prop:single_step_time},
the computation in 
the first line requires  
$\bigO{|\inst{D}|^{\min\{\eps_T,1-\eps_R\}}}$
time
and the second line needs 
$\bigO{|\inst{D}|^{\eps_T}}$ time.
Since $\eps_R  = 0$ and 
$\eps_T = 1$, the computation time amounts to  
$\bigO{|\inst{D}|}$.
In case of an update  
$\delta T_l = \{(\deltaC, \deltaA) \mapsto \p\}$,
\ivme perform the following computations:
\begin{itemize}
\item $\delta Q_{hhl}() = \delta T_l(\deltaC,\deltaA) \ztimes
\textstyle\sum_b R_h(\deltaA, b) \ztimes S_h(b, \deltaC)$,
\TAB  $Q() = Q() + \delta Q_{hhl}()$, and
\item $V_{ST}(b,\deltaA) = V_{ST}(b,\deltaA) + 
\delta T_l(\deltaC,\deltaA) \ztimes
S_h(b,\deltaC)$.
\end{itemize}
Following the proof of  
Proposition \ref{prop:single_step_time},
both lines need  
$\bigO{|\inst{D}|^{1-\eps_S}}$ time.
As $\eps_S = 0$, the computation time is  
$\bigO{|\inst{D}|}$.

\subsection{Counting Triangles in  Static Databases}\label{sec:static_count} 
\label{sec:staticCount}
All triangles in the database $\inst{D} = \{R,S,T\}$ can 
be counted by computing the join of 
$R$, $S$, and $T$ using a worst-case optimal algorithm and 
summing up the multiplicities of the result tuples in overall
time $\bigO{|\db|^{\frac{3}{2}}}$~\cite{Ngo:SIGREC:2013}.
\ivme allows to recover this computation time for counting  triangles
in the static case.
We fix $\eps = \frac{1}{2}$ and insert all tuples from $\db$, one at a time, into a
database $\inst{D}'$ with initially empty relations.  
The preprocessing time is constant. 
By Theorem~\ref{theo:main_result}, \ivme guarantees  
$\bigO{M^{\frac{1}{2}}}$ amortized update time, where $M$
is the size of $\inst{D}'$ at update time.
Thus, the total time to count the triangles in $\inst{D}$ is
$\bigO{\sum_{M=0}^{|\db|-1} M^{\frac{1}{2}}} = \bigO{|\db|\ztimes |\db|^{\frac{1}{2}}} = \bigO{|\db|^{\frac{3}{2}}}$.

To avoid rebalancing, we can preprocess the input relations in $\inst{D}$ 
to decide for each tuple its ultimate relation part.
More precisely, each tuple $(a,b)$ from $R$ with $|\sigma_{A=a} R| \geq |\db|^{\frac{1}{2}}$ is inserted to the heavy part of $R$ while all other tuples from $R$ are inserted to the light part. The distribution of the tuples from $S$ and $T$ is analogous.
Since we do not perform any rebalancing,
by Proposition~\ref{prop:single_step_time},
the worst-case (and not only amortized) time of each insert is $\bigO{|\db|^{\frac{1}{2}}}$.

\section{Counting Triangles with Self-Joins}\label{sec:self_joins}

\ivme also applies to the triangle count query with self-joins over one relation. This query models the common case of counting triangles in a graph given by its edge relation $R$.
\begin{align*}
Q^{\Qsj}() = \sum\limits_{a,b,c} R(a,b) \ztimes R(b,c) \ztimes R(c,a)
\end{align*}
A trivial way to maintain $Q^{\Qsj}$ is to create three distinct copies of $R$ and treat each update to $R$ as a sequence of three updates to its copies. The asymptotic time and space complexities remain unchanged (see Theorem~\ref{theo:main_result}). 
Alternatively, we can maintain only $R$ and apply the delta rules from Section~\ref{sec:preliminaries} on the self-joins in $Q^{\Qsj}$. We describe this latter approach next.

We partition the relation $R$ with threshold $N^{\eps}$ for a fixed $\eps \in [0,1]$
into $R_h$ and $R_l$ . The query $Q^{\Qsj}$ is the sum of all the partial counts: $Q^{\Qsj}() = \textstyle\sum_{u,v,w \in \{h,l\}} 
\textstyle\sum_{a,b,c} R_u(a,b) \ztimes R_v(b,c) \ztimes R_w(c,a)$.
Now, consider an update $\delta R_{r} = \{(\deltaA, \deltaB) \mapsto \p\}$ to $R$, where 
$r$ is fixed to either $h$ or $l$. For simplicity, we omit the function arguments 
 in what follows. 
The delta query $\delta Q^{\Qsj}$ is:
\begin{align*}
\delta Q^{\Qsj}
= \delta\big(\textstyle\sum\limits_{u,v,w \in \{h,l\}}\; \textstyle\sum\limits_{a,b,c} R_u \ztimes R_v \ztimes R_w \big)
= \textstyle\sum\limits_{u,v,w \in \{h,l\}}\; \textstyle\sum\limits_{a,b,c}\delta\big(R_u \ztimes R_v \ztimes R_w \big)
\end{align*}

We simplify the analysis using algebraic transformations. From the delta rules, we have:
\begin{align*}
& \delta(R_u \ztimes R_v \ztimes R_w) \\
&\qquad = 
\left(\delta R_{u} \ztimes R_v \ztimes R_w \right) \,+\, \left( R_u \ztimes \delta(R_v \ztimes R_w) \right) \,+\, \left(\delta R_u \ztimes \delta(R_v \ztimes R_w) \right) \\
& \qquad =
\left(\delta R_{u} \ztimes R_v \ztimes R_w \right) \,+\,
\left( R_u \ztimes \delta R_{v} \ztimes R_w \right) \,+\,
\left( R_u \ztimes R_v \ztimes \delta R_{w} \right) \,+\,
\left( R_u \ztimes \delta R_{v} \ztimes \delta R_{w} \right) \,+\, \\
& \qquad\quad\,
\left( \delta R_{u} \ztimes \delta R_{v} \ztimes R_w \right) \,+\,
\left( \delta R_{u} \ztimes R_v \ztimes \delta R_{w} \right) \,+\,
\left( \delta R_{u} \ztimes \delta R_{v} \ztimes \delta R_{w} \right)
\end{align*}
 
Since the query $Q^{\Qsj}$ is symmetric, we have the following equalities:
\begin{align*}
&\textstyle\sum\limits_{u,v,w \in \{h,l\}} \; \textstyle\sum\limits_{a,b,c}
 \delta R_{u} \ztimes R_v \ztimes R_w = 
\textstyle\sum\limits_{u,v,w \in \{h,l\}} \; \textstyle\sum\limits_{a,b,c}
R_u  \ztimes \delta R_{v} \ztimes   R_w =
\textstyle\sum\limits_{u,v,w \in \{h,l\}} \; \textstyle\sum\limits_{a,b,c}
R_u \ztimes   R_v \ztimes \delta R_{w} \;\text{ and}\\[3pt]
&\textstyle\sum\limits_{u,v,w \in \{h,l\}} \; \textstyle\sum\limits_{a,b,c}
\delta R_{u} \ztimes \delta R_{v} \ztimes   R_w 
=
\textstyle\sum\limits_{u,v,w \in \{h,l\}} \; \textstyle\sum\limits_{a,b,c}
\delta R_{u}  \ztimes R_v \ztimes \delta R_{w}
= 
\textstyle\sum\limits_{u,v,w \in \{h,l\}} \; \textstyle\sum\limits_{a,b,c}
R_u \ztimes \delta R_{v}  \ztimes \delta R_{w}
\end{align*}

Let ${\bf 3}$ be a relation mapping the empty tuple to multiplicity $3$, that is,  ${\bf 3} = \{\, \tuple{} \mapsto 3 \,\}$. Then,
\begin{align*}
\hspace{-1.2em}
\delta Q^{\Qsj} = &\hspace{0.1cm} {\bf 3} \ztimes 
\textstyle\sum\limits_{u,v,w \in \{h,l\}}\;
\textstyle\sum\limits_{a,b,c} \delta R_u \ztimes  R_v \ztimes R_w + \\
& \hspace{0.1cm}{\bf 3} \ztimes 
\textstyle\sum\limits_{u,v,w \in \{h,l\}} \; \textstyle\sum\limits_{a,b,c} \delta R_{u} \ztimes \delta R_{v} \ztimes R_w  + 
\textstyle\sum\limits_{u,v,w \in \{h,l\}}\;
\textstyle\sum\limits_{a,b,c}\delta R_{u} \ztimes \delta R_{v} \ztimes \delta R_{w}
\end{align*}

We can further simplify $\delta Q$ because each $\delta R_u$ with $u \neq r$ evaluates to the empty relation.
\begin{align*}
\hspace{-1.2em}
\delta Q^{\Qsj} = &\hspace{0.1cm} {\bf 3} \ztimes 
\textstyle\sum\limits_{v,w \in \{h,l\}}\;
\textstyle\sum\limits_{c} \delta R_r \ztimes  R_v \ztimes R_w +  \hspace{0.1cm}{\bf 3} \ztimes 
\textstyle\sum\limits_{w \in \{h,l\}} \; \delta R_{r} \ztimes \delta R_{r} \ztimes R_w  + 
\delta R_{r} \ztimes \delta R_{r} \ztimes \delta R_{r}
\end{align*}

The last two summands having multiple occurrences of $\delta R_r$ yield non-empty relations only when $\deltaA = \deltaB$, that is, when $\delta R_r(\deltaA, \deltaA)$ (or $\delta R_r(\deltaB, \deltaB)$) returns a nonzero multiplicity. Thus, we can write the delta query as:
\begin{align*}
\hspace{-1.2em}
\delta Q^{\Qsj}() = &\hspace{0.1cm} {\bf 3} \ztimes  
\delta R_r(\deltaA,\deltaB) \ztimes \textstyle\sum\limits_{v,w \in \{h,l\}}\;
 \textstyle\sum\limits_{c}    R_v(\deltaB,c) \ztimes R_w(c,\deltaA) + \\
& \hspace{0.1cm}{\bf 3} \ztimes 
\delta R_{r}(\deltaA,\deltaA) \ztimes \delta R_{r}(\deltaA,\deltaA) \ztimes \textstyle\sum\limits_{w \in \{h,l\}}
 R_w(\deltaA,\deltaA)  + 
\delta R_{r}(\deltaA,\deltaA) \ztimes \delta R_{r}(\deltaA,\deltaA) \ztimes \delta R_{r}(\deltaA,\deltaA)
\end{align*}

\begin{figure}[t]
\begin{center}
\renewcommand{\arraystretch}{1.2}
\setcounter{magicrownumbers}{0}
\begin{tabular}{ll@{\hskip 0.25in}l@{\hspace{0.6cm}}c}
\toprule
\multicolumn{2}{l}{\textsc{ApplyUpdateSelfJoin}($\delta R_r,\state$)}& & Time \\
\cmidrule{1-2} \cmidrule{4-4}

\rownumber & \LET $\delta R_r = \{(\deltaA,\deltaB) \mapsto \p\}$ \\
\rownumber & \LET $\state = (\eps, N, \{R_l,R_h\}, \{Q,V\})$ \\

\rownumber & $Q() = Q() + {\bf 3} \ztimes 
\delta{R_{r}(\deltaA,\deltaB)} \ztimes \textstyle\sum_c  R_{h}(\deltaB,c) \ztimes R_{h}(c,\deltaA)$ & &
$\bigO{|\inst{D}|^{1-\eps}}$ \\

\rownumber & $Q() = Q() + {\bf 3} \ztimes \delta{R_{r}(\deltaA,\deltaB)} \ztimes  
V(\deltaB,\deltaA)$ & &
$\bigO{1}$ \\

\rownumber & $Q() = Q() + {\bf 3} \ztimes \delta R_{r}(\deltaA,\deltaB) \ztimes\textstyle\sum_c  R_{l}(\deltaB,c) \ztimes R_{h}(c,\deltaA)$ & &
$\bigO{|\inst{D}|^{\min{\{\eps, 1-\eps\}}}}$ \\ 

\rownumber & $Q() = Q() + {\bf 3} \ztimes \delta R_{r}(\deltaA,\deltaB) \ztimes\textstyle\sum_c  R_{l}(\deltaB,c) \ztimes R_{l}(c,\deltaA)$ & &
$\bigO{\inst{D}^{\eps}}$ \\

\rownumber & $Q() = Q() + {\bf 3} \ztimes \delta R_{r}(\deltaA,\deltaA) \ztimes \delta R_{r}(\deltaA,\deltaA) \ztimes R_{h}(\deltaA,\deltaA)$ & &   
$\bigO{1}$ \\

\rownumber & $Q() = Q() + {\bf 3} \ztimes \delta R_{r}(\deltaA,\deltaA) \ztimes \delta R_{r}(\deltaA,\deltaA) \ztimes R_{l}(\deltaA,\deltaA)$ & &   
$\bigO{1}$ \\

\rownumber & $Q() = Q() + \delta R_{r}(\deltaA,\deltaA) \ztimes 
\delta R_{r}(\deltaA,\deltaA) \ztimes \delta R_{r}(\deltaA,\deltaA)$ & &
$\bigO{1}$ \\


\rownumber & \IF ($r$ is $h$) &\\  
\rownumber & \TAB $V(\deltaA,c) = V(\deltaA,c) +  \delta R_{h}(\deltaA,\deltaB) \ztimes  
R_{l}(\deltaB,c)$  & &
$\bigO{\inst{D}^{\eps}}$ \\

\rownumber & \ELSE &\\
\rownumber & \TAB $V(a,\deltaB) = V(a,\deltaB) +  R_{h}(a,\deltaB) \ztimes \delta R_{l}(\deltaA,\deltaB)$  & &
$\bigO{\inst{D}^{1-\eps}}$ \\

\rownumber & $R_{r}(\deltaA,\deltaB) = R_{r}(\deltaA,\deltaB) + \delta{R}_{r}(\deltaA,\deltaB)$ & &
$\bigO{1}$ \\

\rownumber & \RETURN $\state$ \\
\midrule
\multicolumn{2}{r}{Total update time:} & &
$\bigO{|\inst{\inst{D}}|^{\max\{\eps, 1-\eps\}}}$ \\
\bottomrule
\end{tabular}
\end{center}\vspace{-1em}
\caption{
 (left) Counting triangles over a self-join under a single-tuple update. 
\textsc{ApplyUpdateSelfJoin} takes as input an update $\delta R_r$ to the heavy or light part of $R$, hence $r \in \{h,l\}$, and the current \ivme state $\state$ of a database $\inst{D}$ partitioned using $\eps\in[0,1]$. 
It returns a new state that arises from applying $\delta R_r$ to $\state$.
Notation: ${\bf 3} = \{\, \tuple{} \mapsto 3 \,\}$. 
Lines 7-9 compute non-empty deltas only when $\deltaA = \deltaB$.  
(right) The time complexity of computing and applying deltas. 
}
\label{fig:applyUpdate-selfjoin}
\end{figure}

Figure~\ref{fig:applyUpdate-selfjoin} shows the procedure for maintaining the result of the triangle count query with self-joins under an update 
$\delta R_{r} = \{(\deltaA,\deltaB) \mapsto \p\}$.
Lines $3-6$ compute the first summand of $\delta Q$ similarly as in the self-join free case except that now each delta count is multiplied by $3$.
As before, we materialize an auxiliary view 
$V(a,c) = \textstyle\sum_b R_h(a,b) \ztimes R_l(b,c)$ to achieve sublinear delta computation for $Q^{\Qsj}$. 

Whereas for maintaining the result of the triangle count query without self-joins \ivme materializes three distinct views (see Figure~\ref{fig:view_definitions}), here these three views are equivalent; thus, only $V$ is needed. 
In the remaining two summands of $\delta Q$, $\delta R_{r}$ binds all three variables 
$A$, $B$, $C$ such that the count results are nonzero if $\deltaA = \deltaB$. 
Lines $7-9$ compute these two summands in $\bigO{1}$ time. 
Overall, the complexity of maintaining the result of the triangle count query with self-joins under a single-tuple update is the same as that without self-joins: $\bigO{N^{\max\{\eps, 1-\eps\}}}$ time and $\bigO{N^{1+\min\{\eps, 1-\eps\}}}$ space. The amortized analysis is also similar.

\section{Worst-Case Optimality of \texorpdfstring{\ivme}{IVM-eps}}\label{sec:lowerbound}
In this section, we prove Proposition 
\ref{prop:lower_bound_triangle_count} that states
a lower bound on the incremental maintenance
of the triangle count, conditioned on the \OMv conjecture (Conjecture \ref{conj:omv}).
The proof is inspired by recent work~\cite{BerkholzKS17}.
It relies on the Online Vector-Matrix-Vector Multiplication (\OuMv) conjecture, which is implied by the \OMv conjecture. We introduce the \OuMv
problem and state the corresponding conjecture.

\begin{definition}[Online Vector-Matrix-Vector Multiplication (\OuMv)~\cite{Henzinger:OMv:2015}]\label{def:OuMv}
We are given an $n \times n$ Boolean matrix $\vecnormal{M}$ and  receive $n$ pairs of Boolean column-vectors of size $n$, denoted by $(\vecnormal{u}_1,\vecnormal{v}_1), \ldots, (\vecnormal{u}_n,\vecnormal{v}_n)$; after seeing each pair $(\vecnormal{u}_i,\vecnormal{v}_i)$, we output the product 
$\vecnormal{u}_i^{\text{T}} \vecnormal{M} \vecnormal{v}_i$ before we see the next pair.
\end{definition}

\begin{conjecture}[\OuMv Conjecture, Theorem 2.7 in~\cite{Henzinger:OMv:2015}]\label{conj:OuMv}
For any $\gamma > 0$, there is no algorithm that solves \OuMv in time 
$\bigO{n^{3-\gamma}}$.
\end{conjecture}

\subsection{Reduction from the \OuMv Problem}
The following proof of Proposition 
\ref{prop:lower_bound_triangle_count}
reduces the \OuMv problem to  
the problem of incrementally maintaining the triangle count.
This reduction implies that 
if there is an algorithm that incrementally maintains the triangle count
under single-tuple updates 
with arbitrary preprocessing time, $\bigO{|\inst{D}|^{\frac{1}{2}-\gamma}}$
amortized update time, and  $\bigO{|\inst{D}|^{1-\gamma}}$ enumeration delay
for some $\gamma > 0$ and database $\inst{D}$, then  
the \OuMv problem can be solved in subcubic time. 
This contradicts the \OuMv conjecture and, consequently, the \OMv conjecture.

\begin{figure}[t]
\begin{center}
\renewcommand{\arraystretch}{1.0}
\begin{tabular}{l}
\toprule 
{\textsc{SolveOuMv}}($\vecnormal{M}, \vecnormal{u}_1,\vecnormal{v}_1, \ldots, \vecnormal{u}_n,\vecnormal{v}_n$)\\
\midrule
initial database state $\state = (\eps,1,  \dbeps_0, \inst{V}_0)$ \\
$\FOREACH\STAB (i,j) \in \vecnormal{M} \STAB\DO$ \\
$\TAB\STAB\delta S = \{\, (i,j) \mapsto \vecnormal{M}(i,j) \,\}$ \\
$\TAB\STAB\textsc{OnUpdate}(\delta S, \state)$\\
$\FOREACH\STAB r = 1, \ldots ,n \STAB\DO$ \\
$\TAB\STAB\FOREACH\STAB i = 1, \ldots ,n \STAB\DO$ \\
$\TAB\STAB\TAB\STAB\delta R = \{\, (a,i) \mapsto (\vecnormal{u}_r(i) - R(a,i)) \,\}$ \\
$\TAB\STAB\TAB\STAB\textsc{OnUpdate}(\delta R, \state)$ \\
$\TAB\STAB\TAB\STAB\delta T = \{\, (i,a) \mapsto (\vecnormal{v}_r(i) - T(i,a)) \,\}$ \\
$\TAB\STAB\TAB\STAB\textsc{OnUpdate}(\delta T, \state)$ \\
$\TAB\STAB \textbf{output}\STAB(Q(\,) \neq 0)$ \\
\bottomrule
\end{tabular}

\end{center}
\caption{The procedure \textsc{SolveOuMv} solves the \OuMv problem using an incremental maintenance algorithm that counts triangles under single-tuple updates.
The state $\state = (\eps,1, \dbeps_0, \inst{V}_0)$ is the initial \ivme state of a database  
with empty relations $R$, $S$ and $T$.
The procedure \textsc{OnUpdate} maintains the triangle count query under single-tuple updates
 as given in Figure~\ref{fig:updateProgram}.}
\label{fig:lower_bound_reduction}
\end{figure}

\begin{proof}[Proof of Proposition~\ref{prop:lower_bound_triangle_count}](inspired by~\cite{BerkholzKS17}) 
For the sake of contradiction, assume that there is an incremental maintenance algorithm $\mathcal{A}$
with arbitrary preprocessing time, amortized update time $\bigO{|\inst{D}|^{\frac{1}{2}-\gamma}}$, 
and answer time $\bigO{|\inst{D}|^{1-\gamma}}$  that counts triangles under single-tuple updates.
We show that this algorithm can be used to design an algorithm $\mathcal{B}$ that solves the $\OuMv$ problem in subcubic time, which contradicts the \OuMv conjecture.

Figure~\ref{fig:lower_bound_reduction} gives the pseudocode of $\mathcal{B}$ processing an $\OuMv$ input $(\vecnormal{M}, (\vecnormal{u}_1,\vecnormal{v}_1), \ldots ,(\vecnormal{u}_n,\vecnormal{v}_n))$.
We denote the entry of $\vecnormal{M}$ in row $i$ and column $j$ by $\vecnormal{M}(i,j)$ and the $i$-th component of $\vecnormal{v}$ by $\vecnormal{v}(i)$.
The algorithm first constructs the initial state $\state$ from an empty database $\db = \{R,S,T\}$.
Then, it executes at most $n^2$ updates to the relation $S$ such that 
$S = \{\, (i,j) \mapsto \vecnormal{M}(i,j) \,\mid\, i,j \in \{1,\ldots, n\} \,\}$. 
In each round $r \in \{1, \ldots , n\}$, the algorithm executes at most $2n$ updates to the relations $R$ and $T$ 
such that $R = \{\, (a, i) \mapsto \vecnormal{u}_r(i) \,\mid\, i \in \{1,\ldots, n\} \,\}$
and $T = \{\, (i, a) \mapsto \vecnormal{v}_r(i) \,\mid\, i \in \{1,\ldots, n\} \,\}$, where $a$
is some constant.
By construction, $\vecnormal{u}_r^{\text{T}}\vecnormal{M}\vecnormal{v}_r = 1$ if and only if there exist 
$i,j\in\{1,\ldots, n\}$ such that $\vecnormal{u}_r(i) = 1$, $\vecnormal{M}(i,j)=1$, and $\vecnormal{v}_r(j) = 1$, which is equivalent to 
$R(a,i) \ztimes S(i,j) \ztimes T(j,a) = 1$. Thus, the algorithm outputs 
$1$ at the end of round $r$ if and only if the triangle count is nonzero.

Constructing the initial state from an empty database takes constant time. 
The construction of relation $S$ from $\vecnormal{M}$ requires at most $n^2$ updates. 
Given that the amortized time for each update is $\bigO{|\db|^{\frac{1}{2}-\gamma}}$
and the database size $|\db|$ is $\bigO{n^2}$, the overall time for this phase is $\bigO{n^2 \ztimes n^{1-2\gamma}} = \bigO{n^{3-2\gamma}}$. 
In each round, the algorithm performs at most $2n$ updates and outputs the result in $\bigO{|\inst{D}|^{1-\gamma}}$ time. The overall execution time is $\bigO{n^{2-2\gamma}}$ per round and $\bigO{n^{3-2 \gamma}}$ for $n$ rounds. Thus, algorithm $\mathcal{B}$ needs $\bigO{n^{3-2 \gamma}}$ time to solve the \OuMv problem, which contradicts the \OuMv conjecture and, consequently, the \OMv conjecture. 
\end{proof}

Theorem~\ref{theo:main_result} and Proposition~\ref{prop:lower_bound_triangle_count} imply that for $\eps = \frac{1}{2}$, \ivme incrementally maintains the count of triangles under single-tuple updates to a database $\db$ with optimal amortized update time $\bigO{|\db|^{\frac{1}{2}}}$ and constant answer time, unless the \OMv conjecture fails (Corollary~\ref{cor:ivme_optimal}).

Note that the proof of Proposition~\ref{prop:lower_bound_triangle_count}
can easily be turned into a lower bound proof for maintaining the result 
of any triangle query with free variables.  In the reduction described 
in the proof of Proposition~\ref{prop:lower_bound_triangle_count}, we output 1 at the end of a round if and only if the triangle count is nonzero.
To turn the proof into lower bound proof for the maintenance
of any triangle query with free variables, we can use the same encoding
for the matrix
$\vecnormal{M}$ and the vectors and ask at the end of each round 
whether the query result 
contains at least one tuple. This reduction gives us a conditional lower bound on the update time 
and enumeration delay for triangle queries with free variables: 

\begin{corollary}[Proof of Proposition \ref{prop:lower_bound_triangle_count}]\label{prop:lower_bound_triangle_enum}
For any $\gamma > 0$ and database $\db$,
there is no algorithm that incrementally maintains the result of any triangle query with free variables under single-tuple updates to $\db$ with arbitrary preprocessing time, $\bigO{|\db|^{\frac{1}{2} - \gamma}}$ amortized update time, and $\bigO{|\db|^{1 - \gamma}}$ enumeration delay, 
unless the \OMv conjecture fails.
\end{corollary}

\section{Improving Space by Refining Relation Partitions}
\label{sec:part_all_var}
The \ivme algorithm 
 for maintaining the result of the triangle count, as presented in Sections 
\ref{sec:strategy} and \ref{sec:rebalancing},  
partitions each relation on its first variable
and achieves a space-time tradeoff 
such that 
the space-time product is quadratic.
This tradeoff is visualized in 
Figure \ref{fig:time_space_plot}.
In this section, we design a variant 
of this algorithm that achieves better 
space complexity  
by using more refined 
relation partitions.
The partitioning 
takes the degrees of 
data values of both variables for each relation into account. 
While the preprocessing time, amortized
update time, and answer time stay as before, the 
space complexity becomes 
$\bigO{|\inst{D}|^{\max\{\min\{1 + \eps,2-2\eps\}, 1\}}}$.
Figure \ref{fig:time_space_plot_part_both_var}
shows the obtained amortized update time and space complexity 
parameterized by 
$\eps \in [0,1]$.  
For $\eps = 0.5$, the algorithm achieves 
optimal amortized update time $\bigO{\inst{|D|^{\frac{1}{2}}}}$ while using 
linear space, hence the space-time product becomes 
$\bigO{\inst{|D|^{\frac{3}{2}}}}$. 
\begin{figure}[t]
\begin{center}
\begin{tikzpicture}
\begin{axis}[
grid=major,
    grid style={dotted},
xmin=0, xmax=1, ymin=0, ymax=1.5,
every axis plot post/.append style={mark=none},
  xtick ={0,0.333, 0.5, 1},
  ytick ={0, 0.5, 1,1.333},
  xticklabels={\footnotesize{$0$},\footnotesize{$\frac{1}{3}$},\footnotesize{$\frac{1}{2}$},\footnotesize{$1$}},
   yticklabels={$$,\footnotesize{$\frac{1}{2}$},\footnotesize{$1$},
   \footnotesize{$\frac{4}{3}$}},  
y=2cm,
    x=3.2cm,
axis lines=middle,
    axis line style={->},
    x label style={at={(axis description cs:1.18,-0.06)}},
    xlabel={\footnotesize{$\eps$}},
    y label style={at={(axis description cs:-0.9,1.1)},align=center},
    ylabel=\footnotesize{Asymptotic} \\  \footnotesize{complexity} \\ $|\inst{D}|^y$,
  axis x line*=bottom,
  axis y line*=left,
  legend style={at={(1.6,1)},draw=none},
  legend entries={\footnotesize{Space},\footnotesize{Time}}
 ]

  \addplot[color=blue,mark=none,domain=0:1,thick,dashed] coordinates{ 
  (0, 1) 
  (1/3, 4/3)
  (1/2, 1)
  (1,1) 
}; 
  \addplot[color=red,mark=none,domain=0:1,thick] coordinates{ 
  (0, 1) 
  (1/2, 1/2)
  (1,1) 
}; 
\end{axis}
\node at(0,3.2) {\footnotesize{$y$}};

\end{tikzpicture}
\caption{
The space and amortized update time parameterized by $\eps$
of \ivme with 
refined partitions.}
\label{fig:time_space_plot_part_both_var}
\end{center}
\end{figure}
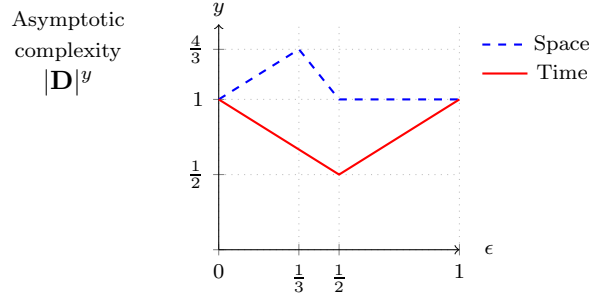 
We call this \ivme version, \ivme with refined
partitions.  
Theorem \ref{theo:main_result_part_both_var}
summarizes the main results we obtain under
refined partitions. 

 \begin{theorem}\label{theo:main_result_part_both_var}
Given a database $\db$ and $\eps \in [0,1]$, 
\ivme with refined partitions 
 incrementally maintains the triangle count under single-tuple updates to $\db$ with 
$\bigO{|\inst{D}|^{\frac{3}{2}}}$ preprocessing time, $\bigO{|\inst{D}|^{\max\{\eps,1-\eps\}}}$ amortized update time, 
constant answer time, and 
$\bigO{|\inst{D}|^{\max\{\min\{1 + \eps,2-2\eps\}, 1\}}}$ space.
\end{theorem}
The rest of Appendix \ref{sec:part_all_var} is dedicated to the 
introduction of \ivme with refined partitions 
and the proof of Theorem  
\ref{theo:main_result_part_both_var}. We emphasize on the
main differences to the original \ivme algorithm that 
partitions each relation on a single variable. 

\subsection{Refined Partitions}
We generalize the notion of a
relation partition from 
Definition 
 \ref{def:loose_relation_partition}. 
Let $K$ be a relation with schema $\inst{X}$
and let $S$ be a set of variables in $\inst{X}$.
A partition of $K$ on the variables 
in $S$  
depends on the degrees of the 
values of each variable in $S$.
The partition contains a relation $K_{\rho}$
for each function 
$\rho$ mapping each variable
in $S$ to $h$ or $l$.
For instance, 
in case $S = \{X,Y\}$ and $\rho$ is the function that maps 
$X$ to $l$ and $Y$ to $h$, $K_{\rho}$ intuitively 
consists of all tuples $\inst{x}$ from $K$ such that 
 the degree of $\inst{x}[X]$  is low
and that of $\inst{x}[Y]$ is high in $K$.

We now formalize the above intuition.  
Given a variable set $S$, 
we denote the set of all mappings
$\rho: S \rightarrow \{h,l\}$ 
from $S$ to $\{h,l\}$
by $\{h,l\}^{S}$.

\begin{definition}[Relation Partition on Multiple Variables]\label{def:loose_relation_partition_sets}
Given a relation $K$ over schema $\inst{X}$,
 a set $S$ of variables from $\inst{X}$, 
 and a threshold $\theta$, a  partition of $K$ on 
the variables in $S$ with threshold $\theta$ 
is a set $\{K_\rho\}_{\rho \in \{h,l\}^{S}}$ satisfying the following conditions:
\\[6pt]
\begin{tabular}{@{\hskip 0.5in}rl}
{(union)} & $K(\inst{x}) = \sum\limits_{\rho \in \{h,l\}^{S}}
K_{\rho}(\inst{x})$ for $\inst{x} \in \Dom(\inst{X})$ \\[4pt]
{(domain partition)} & $(\pi_{S}R_\rho) \cap (\pi_{S}R_{\rho'}) = \emptyset$ 
for each pair $\rho,\rho' \in \{h,l\}^{S}$ \\[8pt]
 & for all $\rho \in \{h,l\}^{S}$
and  $Y \in S$: \\[4pt]
(heavy part) & \TAB
if $\rho(Y)=h$, then $|\sigma_{Y=y} K| \geq 
\frac{1}{2}\,\theta$ for all $y \in \pi_{Y}K_\rho$, \\[4pt]
(light part) & \TAB
if $\rho(Y)=l$, then $|\sigma_{Y=y} K| < \frac{3}{2}\,\theta$
for all $y \in \pi_{Y}K_\rho$
\end{tabular}\\[6pt]
The set $\{K_\rho\}_{\rho \in \{h,l\}^{S}}$ is called a 
 strict partition of $K$ on the variables in $S$ with 
threshold $\theta$ if it satisfies the union and 
domain partition conditions and the following strict versions
of the heavy part and light part conditions: 
\\[6pt]
\begin{tabular}{@{\hskip 0.5in}rl}

& for all $\rho \in \{h,l\}^{S}$ and $Y \in S$: \\[4pt]
(strict heavy part)  & 
\TAB if $\rho(Y)=h$, then $|\sigma_{Y=y} K| \geq 
\theta$ for all $y \in \pi_{Y}K_\rho$ and \\[4pt]
(strict light part) & 
\TAB if $\rho(Y)=l$, then $|\sigma_{Y=y} K| < \theta$
for all $y \in \pi_{Y}K_\rho$
\end{tabular}\\[6pt]
A relation $K_\rho$ is called heavy on a variable $Y \in S$ if 
$\rho(Y) = h$, otherwise it is light on $Y$.
\end{definition}

\subsection{Adaptive Maintenance under Refined Partitions}
We first give an intuitive description of
\ivme with refined 
partitions. 
The algorithm partitions each input 
relation to the triangle count query on 
both variables in its schema.   
Given an input relation 
$K \in \{R,S,T\}$ with schema $(X,Y)$ 
and a function $\rho: \{X,Y\} \rightarrow \{h,l\}$,
we denote $K_{\rho}$ simply by $K_{\rho(X)\rho(Y)}$. 
The algorithm  decomposes the triangle count query 
into skew-aware views of the form 
$$Q_{r_Ar_Bs_Bs_Ct_Ct_A}() = 
\sum\limits_{a,b,c} R_{r_Ar_B}(a,b) 
\ztimes S_{s_Bs_C}(b,c) \ztimes
T_{t_Ct_A}(c,a),$$
with 
$r_A,r_B,s_B,s_C,t_C,t_A \in \{h,l\}$.
The triangle count query can be rewritten 
as a sum of such views:
$$Q() = \sum\limits_{r_A,r_B,s_B,s_C,t_C,t_A \in \{h,l\}} \;\,
\sum\limits_{a,b,c} R_{r_Ar_B}(a,b) \ztimes S_{s_Bs_C}(b,c) \ztimes
T_{t_Ct_A}(c,a).$$ 
The main difference to the original \ivme  
algorithm is that each materialized auxiliary view is
composed of relations that are heavy on the 
variables exported by the view.
This implies better bounds on the view 
sizes.
For instance, consider 
a single-tuple update   
$\delta R_{r_Ar_B}$
for some $r_A,r_B \in \{h,l\}$ 
to relation $R$. 
\ivme with refined 
partitions 
uses an auxiliary view
only for computing the delta skew-aware 
view $\delta Q_{r_Ar_Bhllh}$.  
The view is of the form 
$V_{ST}(b,a) = 
\sum_{c} S_{hl}(b,c) \ztimes T_{lh}(c,a)$.
The relations $S_{hl}$ and $T_{lh}$
are heavy on the exported variables 
$A$ and $B$.
\nop{
Given a single-tuple update 
$\delta R_{r_A,r_B} = \{(\deltaA,\deltaB) \mapsto \p\}$
for some $r_A,r_B \in \{h,l\}$,
the \ivme algorithm computes the deltas 
of almost all  skew-aware views 
by using the same strategies 
as the \ivme algorithm  that partitions only on the first variables. 
The only cases that differ  
are the computations 
of the delta skew-aware views of form 
$\delta Q_{r_Ar_Bhs_Clt_A}() = 
\delta R_{r_Ar_B}(\deltaA,\deltaB) \ztimes 
\textstyle\sum_{c}  S_{h s_C}(\deltaB,c) \ztimes
T_{lt_A}(c,\deltaA)$. 
As seen in Section \ref{sec:single-tuple-update},
the \ivme algorithm that considers partitioning on
only the first variables  uses  in these case  
an auxiliary view consisting of the product of $S_h$
and $T_l$. 
The \ivme algorithm that considers partitioning 
on both variables of each relation 
further distinguishes   on $s_C$ and $t_A$.
In case $s_C = h$ or $t_A = l$, sublinear computation 
time can be obtained without using
any auxiliary view. 
}
\begin{figure}[t]
  \begin{center}
    \renewcommand{\arraystretch}{1.2}  
    \begin{tabular}{@{\hskip 0.05in}l@{\hskip 0.4in}l@{\hskip 0.05in}}
      \toprule
      Materialized View Definition & Space Complexity \\    
      \midrule
      $Q() = \sum\limits_{r_A, r_B,s_B,s_C,t_C,t_A \in \{h,l\}} \,\,
      \sum\limits_{a,b,c} R_{r_Ar_B}(a,b) \ztimes S_{s_Bs_C}(b,c) \ztimes 
      T_{t_Ct_A}(c,a)$ & 
      $\bigO{1}$ \\
      $V_{RS}(a,c) = \sum_{b} R_{hl}(a,b) \ztimes S_{lh}(b,c)$ & 
      $\bigO{|\inst{D}|^{\min\{\,2-2\eps,\,1+\eps\}}}$ \\
      $V_{ST}(b,a) = \sum_{c} S_{hl}(b,c) \ztimes T_{lh}(c,a)$ & 
      $\bigO{|\inst{D}|^{\min\{\,2-2\eps,\,1+\eps\}}}$ \\
      $V_{TR}(c,b) = \sum_{a} T_{hl}(c,a) \ztimes R_{lh}(a,b)$ & 
       $\bigO{|\inst{D}|^{\min\{\,2-2\eps,\,1+\eps\}}}$ \\
      \bottomrule    
    \end{tabular}
  \end{center}
  \caption{The definitions and space complexities of the 
  views in $\inst{V} = \{ Q, V_{RS}, V_{ST}, V_{TR} \}$ 
  maintained by \ivme with refined partitions as part of 
  an \ivme   state of a database $\inst{D}$ for a fixed $\eps \in [0,1]$.} 
  \label{fig:view_definitions_part_all}
\end{figure}
\nop{
Only in case  
$s_C = l$ and $t_A = h$, an auxiliary  
view of the form 
$V_{ST}(b,a) = \sum_{c} S_{hl}(b,c) \ztimes T_{lh}(c,a)$
is needed.
Using this view, the computation 
of the delta skew-aware view  
$\delta Q_{r_Ar_Bhllh}$
reduces to a constant-time lookup
of the multiplicity of $(\deltaB,\deltaA)$ 
in $V_{ST}$.
} 
Figure~\ref{fig:view_definitions_part_all}
gives all three auxiliary views materialized by
\ivme with refined partitions. 
The views $V_{TR}$ and $V_{RS}$ are used 
to facilitate delta computation under updates to 
the parts of 
relations $S$ and $T$.

\paragraph{\ivme States with Refined Partitions.} 
We slightly extend Definition \ref{def:ivme_state}. 
 Given a fixed $\eps \in [0,1]$,
a state of a database $\inst{D}$ maintained by 
\ivme with refined partitions 
is of the form 
$(\eps,N, \inst{P}, \inst{V})$,
where: 
\begin{itemize}
\item $N$ is a natural number such that the
invariant $\floor{\frac{1}{4}N} \leq |\db| < N$ holds.

\item $\dbeps= 
\{R_{r_Ar_B}\}_{r_A,r_B \in \{h,l\}} \cup
\{S_{s_Bs_C}\}_{s_B,s_C \in \{h,l\}} \cup
\{T_{t_Ct_A}\}_{t_C,t_A \in \{h,l\}}$ 
consists of the refined partitions
of $R$, $S$, and T
with threshold $N^{\eps}$.

\item $\inst{V}$ consists of the materialized views in
Figure \ref{fig:view_definitions_part_all}. 
\end{itemize}
In the initial state of $\inst{D}$ it holds $N = 2\ztimes|\inst{D}| + 1$ and the 
partitions in $\inst{P}$ are strict.

The bounds
on the frequencies of 
data values
in heavy and light relation parts
mentioned after Definition \ref{def:ivme_state}
carry over to refined partitions:  
The number of distinct values 
of a variable in a part of relation $R$ that is 
heavy on that variable is at most 
$\frac{N}{\frac{1}{2}N^{\eps}} = 
2N^{1-\eps}$, i.e.,
$|\pi_{A} R_{hr_B}| \leq   
2N^{1-\eps} \geq 
|\pi_{B} R_{r_Ah}|$ 
for any $r_A,r_B \in \{h,l\}$;
the number of values 
paired with a single data value
of a variable in a part of $R$ 
that is light on that variable
is less than $\frac{3}{2}N^{\eps}$, i.e.,    
$|\sigma_{A = a} R_{lr_B}| < \frac{3}{2}N^{\eps}
>|\sigma_{B = b} R_{r_Al}|$
for any $r_A,r_B \in \{h,l\}$, $A$-value $a$ and $B$-value $b$. 
Analogous bounds hold for the parts of $S$ and 
$T$.

\subsection{Preprocessing Time Under Refined Partitions}
\label{sec:pre_time_refined}
In the preprocessing stage, \ivme with refined 
partitions computes the initial
state 
$\state = (\eps, N, \inst{P}, \inst{V})$ 
of a given database $\inst{D} = \{R,S,T\}$ such that 
$N =2 |\inst{D}|+1$ and
$\inst{P}$ consists of the strict partitions of 
$R$, $S$, and $T$ on $\{A,B\}$, $\{B,C\}$, and 
$\{C,A\}$, respectively, with threshold $N^{\eps}$.
The preprocessing time is the same as for the original 
\ivme algorithm. 

\begin{proposition}\label{prop:preprocessing_step_refined}
Given a database $\db$ and $\eps\in[0,1]$, \ivme 
with refined partitions constructs the initial \ivme state
of $\db$ in $\bigO{|\db|^{\frac{3}{2}}}$ time.
\end{proposition}

\begin{proof}
We analyze the time to construct the 
initial \ivme state $\state = (\eps, N, \inst{P}, \inst{V})$
of $\inst{D}$.
Setting the value of $N$ is a constant-time operation. 
Strictly partitioning the relations can be 
accomplished in linear time.  
The triangle count  
$Q$ can be computed in time 
$|\inst{D}|^{\frac{3}{2}}$
by using a
worst-case optimal join algorithm~\cite{Ngo:SIGREC:2013}.
We now analyze the time to compute the view
$V_{RS}(a,c) = \sum_{b} R_{hl}(a,b) \ztimes S_{lh}(b,c)$.
The analysis for the other two auxiliary views in $\inst{V}$ is analogous. 
The view $V_{RS}$ can be computed in two ways. One option is to iterate 
over the tuples $(a,b)$ in $R_{hl}$ and, for each such tuple,
to go over all $C$-values paired with $b$ in $S_{lh}$.
For each $(a,b)$ in $R_{hl}$ and 
$(b,c)$ in $S_{lh}$, 
the multiplicity of $(a,c)$  in $V_{RS}$ is increased 
by $R_{hl}(a,b) \ztimes S_{lh}(b,c)$. 
Since $S_{lh}$ is light on $B$ and heavy on $C$, 
it contains at most $N^{\min\{\eps,1-\eps\}}$ $C$-values for each
$B$-value. 
Hence, 
the computation 
time is $\bigO{|R_{hl}| \ztimes N^{\min\{\eps,1-\eps\}}}$, which,
due to $N = \Theta(|\inst{D}|)$, is  
$\bigO{|\inst{D}|^{1 + \min\{\eps,1-\eps\}}}$.
Alternatively, one can   iterate 
over the tuples $(b,c)$ in $S_{lh}$ and, for each such tuple,
go over all $A$-values paired with $b$ in $R_{hl}$.
Since $R_{hl}$ is light on $B$ and heavy on $A$, 
the computation time is the same.

Thus, the overall computation time 
for the initial state of $\inst{D}$ is dominated by 
the computation time for $Q$, which is 
$\bigO{|\inst{D}|^{\frac{3}{2}}}$.
\end{proof}

\subsection{Space Complexity  Under Refined Partitions}
While the original \ivme algorithm 
needs $\bigO{|\inst{D}|^{1 + \min\{\eps,1-\eps\}}}$,
the space complexity of \ivme with refined partitions 
is improved to 
$\bigO{|\inst{D}|^{\max\{\min\{1 + \eps,2-2\eps\}, 1\}}}$

\begin{proposition}\label{prop:space_complexity_ref_part}
Given a database $\db$ and $\eps\in[0,1]$, the state 
of $\db$ maintained by \ivme with refined partitions 
to support the maintenance of the result of Query~\eqref{query:triangle} takes 
$\bigO{|\inst{D}|^{\max\{\min\{1 + \eps,2-2\eps\}, 1\}}}$ space.
\end{proposition}  

\begin{proof}
We analyze the space complexity of a state 
$\state = (\eps,N, \inst{P}, \inst{V})$ of a database $\inst{D}$ 
maintained by \ivme with refined  
partitions. 
The space occupied by $\eps$ and $N$ is constant, and the size of the 
relation partitions in $\inst{P}$ is linear.  

Figure  \ref{fig:view_definitions_part_all} gives the 
sizes of the views in $\inst{V}$. The size of $Q$ is constant,
since it consists of an empty tuple mapped to the triangle count. 
We investigate the space complexity of the auxiliary view
$V_{RS}(a,c) = \sum_{b} R_{hl}(a,b) \ztimes S_{lh}(b,c)$.
The analysis for the other auxiliary views in $\inst{V}$  is  analogous.
The view $V_{RS}$ admits two size bounds. The first bound 
is the product of the size of $R_{hl}$ 
and the maximum number of tuples in $S_{lh}$ 
that match with a single tuple in $R_{hl}$, that is, 
the size of $V_{RS}$ is bounded by 
$|R_{hl}| \cdot \max_{b \in \Dom(B)}\{|\sigma_{B = b} S_{lh}|\}$.
Since $S$ is light on $B$ and heavy on $C$, the latter
expression simplifies to 
$N \cdot \min\{ \frac{3}{2}N^\eps, 2N^{1-\eps}\}
=
\bigO{N^{1 + \min\{\eps,1-\eps\}}}$. Note that 
swapping the roles of $R_{hl}$ and $S_{lh}$ in the above 
calculation does not give a better size upper bound. 
The second bound 
is obtained by taking the product of the number of all 
$A$-values in $R_{hl}$ and the number of all 
$C$-values in $S_{lh}$. 
Since $R_{hl}$ is heavy on $A$ and
$S_{lh}$ is heavy on $C$,  
we obtain the size bound $2N^{1-\eps} \cdot 2N^{1-\eps} = 
\bigO{N^{2-2\eps}}$. 
Hence, the size of  
$V_{RS}$ is 
$\bigO{N^{\min\{1 + \min\{\eps, 1-\eps\},2-2\eps\}}}
= \bigO{N^{\min\{1 + \eps,2-2\eps\}}}$ 
which, due to $N = \Theta(|\inst{D}|)$,
is $\bigO{|\inst{D}|^{\min\{1 + \eps,2-2\eps\}}}$.

Thus, taking the linear space of the relation partitions 
in $\dbeps$  into account,
 the overall space complexity is 
$\bigO{|\inst{D}|^{\max\{\min\{1 + \eps,2-2\eps\}, 1\}}}$.
\end{proof}
 
\subsection{Processing a Single-Tuple Update Under Refined Partitions}
\label{sec:proc_single_tuple_refined}
Figure \ref{fig:applyUpdate_part_all}
gives the procedure
\textsc{ApplyUpdateRP} of \ivme with refined partitions 
that takes as input an update
$\delta R_{r_Ar_B}= \{(\deltaA,\deltaB) \mapsto \p\}$
with $r_A,r_B \in \{h,l\}$ 
and a state   
$\state = (\eps, N, \dbeps,\inst{V})$ 
of a database $\inst{D}$ and
returns a new state that 
results from applying $\delta R_{r_Ar_B}$ to $\state$.
The procedure is a straightforward extension of  the procedure  
\textsc{ApplyUpdate} described in  
  Figure~\ref{fig:applyUpdate}.
\begin{figure}[t]
\begin{center}
\renewcommand{\arraystretch}{1.2}
\setcounter{magicrownumbers}{0}
\begin{tabular}{ll@{\hskip 0.25in}l@{\hspace{0.6cm}}c}
\toprule
\multicolumn{2}{l}{\textsc{ApplyUpdateRP}($\delta R_{r_Ar_B},\state$)}& & Time \\
\cmidrule{1-2} \cmidrule{4-4}
\rownumber & \LET $\delta R_{r_Ar_B} = \{(\deltaA,\deltaB) \mapsto \p\}$ \\

\rownumber & \LET $\state = (\eps, N, 
\{R_{r_Ar_B}\}_{r_A,r_B \in \{h,l\}}\cup \{S_{s_Bs_C}\}_{s_B,s_C \in \{h,l\}} \cup  
\{T_{t_Ct_A}\}_{t_C,t_A\in \{h,l\}},$ \\

&\TAB\TAB\TAB\TAB $\{Q,V_{RS}, V_{ST}, V_{TR}\})$ \\

\rownumber & $\delta Q_{r_Ar_Bhs_Cht_A}() =  \delta{R_{r_Ar_B}(\deltaA,\deltaB)} \ztimes \textstyle\sum_c 
S_{hs_C}(\deltaB,c) \ztimes T_{ht_A}(c,\deltaA)$ 
for $s_C,t_A \in \{h,l\}$
& & $\bigO{|\inst{D}|^{1-\eps}}$ \\

\rownumber & $\delta Q_{r_Ar_Bhhlt_A}() =  \delta{R_{r_Ar_B}(\deltaA,\deltaB)} \ztimes \textstyle\sum_c 
S_{hh}(\deltaB,c) \ztimes T_{lt_A}(c,\deltaA)$ 
for $t_A \in \{h,l\}$
& & $\bigO{|\inst{D}|^{1-\eps}}$ \\

\rownumber & $\delta Q_{r_Ar_Bhs_Cll}() =  \delta{R_{r_Ar_B}(\deltaA,\deltaB)} \ztimes \textstyle\sum_c 
S_{hs_C}(\deltaB,c) \ztimes T_{ll}(c,\deltaA)$ 
for $s_C \in \{h,l\}$ & & $\bigO{|\inst{D}|^{\eps}}$ \\

\rownumber & $\delta Q_{r_Ar_Bhllh}() =  \delta{R_{r_Ar_B}(\deltaA,\deltaB)} \ztimes 
V_{ST}(\deltaB, \deltaA)$ 
& & $\bigO{1}$ \\

\rownumber & $\delta Q_{r_Ar_Bls_Cht_A}() =  \delta{R_{r_Ar_B}(\deltaA,\deltaB)} \ztimes \textstyle\sum_c 
S_{ls_C}(\deltaB, c) \ztimes T_{ht_A}(c,\deltaA)$ 
for $s_C,t_A \in \{h,l\}$ & &
$\bigO{|\inst{D}|^{\min{\{\eps, 1-\eps\}}}}$ \\ 

\rownumber & $\delta Q_{r_Ar_Bls_Clt_A}() \mathrel{+{=}} \delta{R_{r_Ar_B}(\deltaA,\deltaB)} \ztimes \textstyle\sum_c 
S_{ls_C}(\deltaB,c) \ztimes T_{lt_A}(c,\deltaA)$ 
for $s_C,t_A \in \{h,l\}$
& &
$\bigO{|\inst{D}|^{\eps}}$ \\

\rownumber & 
$Q() = Q() + \sum_{r_A,r_B,s_B,s_C,t_C,t_A \in \{h,l\}} \delta Q_{r_Ar_Bs_Bs_Ct_Ct_A}()$ & &
$\bigO{1}$ \\

\rownumber & \IF ($r_A$ is  $h$ and $r_B$ is  $l$) &\\  
\rownumber & \TAB $V_{RS}(\deltaA,c) = 
V_{RS}(\deltaA,c) + \delta R_{hl}(\deltaA,\deltaB) \ztimes 
S_{lh}(\deltaB,c)$ & &
$\bigO{|\inst{D}|^{\min\{\eps,1-\eps\}}}$ \\

\rownumber & \ELSE\IF ($r_A = l$ and $r_B = h$) &\\  
\rownumber & \TAB $V_{TR}(c,\deltaB) = 
V_{TR}(c,\deltaB) + \delta R_{lh}(\deltaA,\deltaB) \ztimes 
 T_{hl}(c,\deltaA)$ & &
$\bigO{|\inst{D}|^{\min\{\eps,1-\eps\}}}$ \\

\rownumber & $R_{r_Ar_B}(\deltaA,\deltaB) = 
R_{r_Ar_B}(\deltaA,\deltaB) + \delta{R}_{r_Ar_B}(\deltaA,\deltaB)$ & &
$\bigO{1}$ \\

\rownumber & \RETURN 
$\state$ & &\\
\midrule
\multicolumn{2}{r}{Total update time:} & &
$\bigO{|\inst{D}|^{\max\{\eps, 1-\eps\}}}$ \\
\bottomrule
\end{tabular}
\end{center}\vspace{-1em}
\caption{
 \textsc{ApplyUpdateRP} adapts
  the procedure \textsc{ApplyUpdate} from 
  Figure~\ref{fig:applyUpdate} to refined partitions.
 It takes as input an update
  $\delta R_{r_Ar_B} = \{(\deltaA,\deltaB) \mapsto \p\}$
 with $r_Ar_B \in \{h,l\}$
 and the current \ivme state 
 of a database $\inst{D}$ 
 partitioned using $\eps\in[0,1]$ and returns  
a new state that results from applying $\delta R_{r_Ar_B}$ to $\state$. 
Lines 3-8 compute the deltas of the affected skew-aware views, and Line 9 
maintains $Q$.
Lines 11 and 13 maintain the auxiliary views $V_{RS}$ and $V_{TR}$, respectively. 
Line 14 maintains the affected part $R_{r_Ar_B}$.
The maintenance procedures for updates to $S$ and $T$ are similar.
}
\label{fig:applyUpdate_part_all}
\end{figure}
\begin{proposition}\label{prop:single_step_time_ref_part}
Given a state $\state$ constructed from a database $\db$ for $\eps\in[0,1]$, 
\ivme with refined partitions maintains $\state$ under a single-tuple update to any input relation in $\bigO{|\db|^{\max\{\eps,1-\eps\}}}$ time.
\end{proposition} 
\begin{proof}
We investigate the runtime of the 
procedure  \textsc{ApplyUpdateRP}.
We first analyze the evaluation strategies for 
computing the delta
skew-aware views
$\delta Q_{r_Ar_Bs_Bs_Ct_Ct_A}$
with $s_B,s_C,t_C,t_A \in \{h,l\}$.
 The case
$s_B = t_C = h$ is handled in Line 3,
the case
$s_B = l$ and $t_C = h$ in Line 7, and
the case  
$s_B =t_C = l$ in Line 8. 
In all these cases we ignore 
the partitioning of the 
relations on their second variables  
and mimic 
the evaluation strategies for computing the delta
skew-aware views 
$\delta Q_{r_A hh}$, 
$\delta Q_{r_A lh}$, and 
$\delta Q_{r_A ll}$ from the procedure
 \textsc{ApplyUpdate} in  
  Figure  \ref{fig:applyUpdate}. 
Hence, we get the same time complexities. 
The remaining delta skew-aware views, namely
those where $s_B$ is  $h$ and  $t_C$ is $l$, 
are handled in Lines 4-6.
The computation of such delta views 
requires the iteration over all 
$C$-values paired with $\deltaB$ in $S_{hs_C}$ and with $\deltaA$
in $T_{lt_A}$. 
We make a case distinction
on whether
$S_{hs_C}$ and $T_{lt_A}$ are heavy or light on 
 $s_C$ and $t_A$, respectively. 
 
\begin{itemize}
\item Case $s_C = h$ (Line 4): The delta view is of the form 
$\delta Q_{r_Ar_B h h l t_A}() = 
\delta{R_{r_Ar_B}(\deltaA,\deltaB)} \ztimes \textstyle\sum_c 
S_{hh}(\deltaB,c) \ztimes T_{lt_A}(c,\deltaA)$
with $t_A \in \{h,l\}$.
Since $S_{hh}$ is heavy on $C$, the number 
of distinct $C$-values in $S_{hh}$ is at most $2N^{1-\eps}$,
which means that 
the delta computation takes $\bigO{N^{1-\eps}}$ 
time.   

\item Case $t_A = l$ (Line 5): The delta view is of the form
$\delta Q_{r_Ar_B h s_C l l}() =
\delta{R_{r_Ar_B}(\deltaA,\deltaB)} \ztimes \textstyle\sum_c 
S_{hs_C}(\deltaB,c) \ztimes T_{ll}(c,\deltaA)$
with $s_C \in \{h,l\}$.
Since $T_{ll}$ is light on $A$, there are less than 
$\frac{3}{2} N^{\eps}$ $C$-values paired with $\deltaA$ in $T_{ll}$.
Hence, the computation time is $\bigO{N^{\eps}}$. 

\item Case $s_C = l$ and $t_A = h$ (Line 6): The delta view
is of the form $\delta Q_{r_Ar_B h l l h}()=
\delta{R_{r_Ar_B}(\deltaA,\deltaB)} \ztimes \textstyle\sum_c 
S_{hl}(\deltaB,c) \ztimes T_{lh}(c,\deltaA)$. 
Delta computation amounts to a constant-time 
lookup in $V_{ST}$.
\end{itemize} 
Summing up the deltas in Line 9 takes constant time. 

We now turn to the maintenance 
of the auxiliary views. 
Computing $\delta V_{RS}$ (Line 11) requires the iteration 
over all $C$-values paired with  $\deltaB$ in $S_{lh}$. Similarly, 
computing $\delta V_{TR}$ (Line 13) requires the iteration 
over all $C$-values paired with  $\deltaA$ in $T_{hl}$. Both 
relations $S_{lh}$ and $T_{hl}$ are heavy on $C$
and light on the other variable. 
Hence, the number of $C$-values iterated over 
is bounded by 
$\min\{2N^{1-\eps}, \frac{3}{2}N^{\eps}\}$. This implies
that the computation time is $\bigO{N^{\min\{\eps, 1-\eps\}}}$.
The update of $R_{r_Ar_B}$ in Line 14 can be done in constant 
time. 

It follows that the overall  time  of the 
procedure \textsc{ApplyUpdateRP} 
in Figure \ref{fig:applyUpdate_part_all}
is $\bigO{N^{\max\{\eps, 1-\eps}\}}$, which,   
by $N = \Theta(|\inst{D}|)$, is 
$\bigO{|\inst{D}|^{\max\{\eps, 1-\eps}\}}$. 
\end{proof}

\subsection{Rebalancing Refined Partitions}
\begin{figure}[t]
\begin{center}
\begin{tikzpicture}
\hspace{-0.4cm}
\node at(-1,0)[anchor=north] {
\renewcommand{\arraystretch}{1.2}
\setcounter{magicrownumbers}{0}
\begin{tabular}{l}
\toprule
\textsc{OnUpdateRP}($\delta R,\state$) \\
\midrule
\begin{minipage}{5.2cm}
\renewcommand{\arraystretch}{1.2}
\begin{tabular}{@{\hskip -0.05in}ll@{\hskip 0.25in}l@{\hspace{0.6cm}}c}
&\LET $\delta R = \{(\deltaA,\deltaB) \mapsto \p\}$  \\
&\LET $\state = (\eps, N, \{R_{r_Ar_B}\}_{r_A,r_B \in \{h,l\}} \cup \dbeps, \inst{V})$ \\[0.1cm]
\rownumber & \IF ($(\deltaA,\deltaB) \in R_{hh}$ \OR $\eps=0$) \\
\rownumber & \TAB {\textsc{ApplyUpdateRP}}($\delta R_{hh},\state$)\\
\rownumber & \ELSE\IF ($(\deltaA,\deltaB) \in R_{hl}$) \\
\rownumber & \TAB {\textsc{ApplyUpdateRP}}($\delta R_{hl},\state$)\\
\rownumber & \ELSE\IF ($(\deltaA,\deltaB) \in R_{lh}$) \\
\rownumber & \TAB {\textsc{ApplyUpdateRP}}($\delta R_{lh},\state$)\\
\rownumber & \ELSE \\
\rownumber & \TAB {\textsc{ApplyUpdateRP}}($\delta R_{ll},\state$)\\
\rownumber & \IF($|\inst{D}| = N$)\\
\rownumber & \TAB $N = 2N$ \\
\rownumber & \TAB {\textsc{MajorRebalancingRP}}($\state$) \\
\rownumber & \ELSE \IF($|\inst{D}| < \floor{\frac{1}{4}N}$) \\
\rownumber & \TAB $N = \floor{\frac{1}{2}N}-1$ \\
\rownumber & \TAB {\textsc{MajorRebalancingRP}}($\state$) \\
\end{tabular}
\end{minipage}
\hspace{1.3cm}
\begin{minipage}{9.5cm}
\vspace{0.5cm}
\renewcommand{\arraystretch}{1.2}
\begin{tabular}{ll@{\hskip 0.25in}l@{\hspace{0.6cm}}c}
\rownumber & \ELSE \\

\rownumber & \TAB \IF ($\deltaA \in (\pi_A R_{ll} \cup \pi_A R_{lh})$ \AND\\
& \TAB\TAB\TAB  
$|\sigma_{A = \deltaA} R_{ll}| + |\sigma_{A = \deltaA} R_{lh}| \geq \frac{3}{2} N^{\eps}$) \\
\rownumber & \TAB \TAB {\textsc{MinorRebalancingRP}}($R_{ll}, R_{hl},R_{lh},R_{hh}, A,\deltaA, \state$) \\

\rownumber & \TAB \ELSE\IF ($\deltaA \in  (\pi_A R_{hl} \cup \pi_A R_{hh})$ \AND \\
& \TAB\TAB\TAB\TAB\TAB$|\sigma_{A = \deltaA} R_{hl}| +  |\sigma_{A = \deltaA}R_{hh}| < \frac{1}{2} N^{\eps}$) \\
\rownumber & \TAB \TAB {\textsc{MinorRebalancingRP}}($R_{hl}, R_{ll},R_{hh},R_{lh}, A,\deltaA, \state$) \\

\rownumber & \TAB \IF ($\deltaB \in  (\pi_B R_{ll} \cup \pi_B R_{hl})$ \AND \\
&\TAB\TAB\TAB $|\sigma_{B = \deltaB}
R_{ll}| + |\sigma_{B = \deltaB} R_{hl}| \geq \frac{3}{2} N^{\eps}$) \\
\rownumber & \TAB \TAB {\textsc{MinorRebalancingRP}}($R_{ll}, R_{lh},R_{hl},R_{hh}, A,\deltaA, \state$) \\

\rownumber & \TAB \ELSE\IF ($\deltaB \in (\pi_B R_{lh} \cup \pi_B R_{hh})$ \AND \\
& \TAB\TAB\TAB\TAB\TAB $|\sigma_{B = \deltaB}
R_{lh}| + |\sigma_{B = \deltaB} R_{hh}| < \frac{1}{2} N^{\eps}$) \\
\rownumber & \TAB \TAB {\textsc{MinorRebalancingRP}}($R_{lh}, R_{ll},R_{hh},R_{hl}, A,\deltaA, \state$) \\
\end{tabular}
\end{minipage}\\
\bottomrule
\end{tabular}
};

\end{tikzpicture}
\end{center}
\vspace{-1.2em}
\caption{
\textsc{OnUpdateRP} adapts the procedure 
\textsc{OnUpdate} from Figure~\ref{fig:updateProgram}
to refined partitions. 
It takes as input an update $\delta R = \{(\deltaA,\deltaB) \mapsto \p\}$ and 
a current \ivme state $\state$ of a database $\inst{D}$
and returns a state that results from applying $\delta R$
to $\state$.
\textsc{ApplyUpdateRP} is given in  Figure~\ref{fig:applyUpdate_part_all}.
\textsc{MinorRebalancingRP}($K_{\mathit{src}}, K_{\mathit{dst}}, K_{\mathit{src}}', K_{\mathit{dst}}', X, x, \state$)
is defined in Figure \ref{fig:minor_rebal_part_both_var} and 
moves all tuples with $X$-value $x$ from 
$K_{\mathit{src}}$ and $K_{\mathit{src}}'$ to 
$K_{\mathit{dst}}$ and $K_{\mathit{dst}}'$, respectively. 
\textsc{MajorRebalancingRP}, which is defined similarly
to \textsc{MajorRebalancing} in Figure 
\ref{fig:updateProgram}, strictly repartitions
the relations in $\dbeps$ with threshold $N^{\eps}$ and  
recomputes the views in $\inst{V}$. 
The \textsc{OnUpdateRP} procedures for updates to $S$
and $T$ are analogous.
}
\label{fig:updateProgram_Part_both_var}
\end{figure}

\begin{figure}[t]
\begin{center}
\begin{tikzpicture}
\node at(-8.3,0)[anchor=west] {
\renewcommand{\arraystretch}{1.1}
\begin{tabular}{@{\hskip 0.02in}l@{\hskip 0.02in}}
\toprule
{\textsc{MinorRebalancingRP}}($K_{\mathit{src}}, K_{\mathit{dst}}, K_{\mathit{src}}', K_{\mathit{dst}}', X, x, \state$)\\
\midrule
\FOREACH $\vecnormal{t} \in \sigma_{X=x} K_{\mathit{src}}$ \DO \\
\TAB $m = K_{\mathit{src}}(\vecnormal{t})$ \\
\TAB $\state$ = {\textsc{ApplyUpdateRP}}($\delta K_{\mathit{src}} = \{\, \vecnormal{t} \mapsto -m \,\},\state$)\\ 
\TAB $\state$ = {\textsc{ApplyUpdateRP}}($\delta K_{\mathit{dst}} = \{\, \vecnormal{t} \mapsto m \,\},\state$)\\
\FOREACH $\vecnormal{t} \in \sigma_{X=x} K_{\mathit{src}}'$ \DO \\
\TAB $m = K_{\mathit{src}}'(\vecnormal{t})$ \\
\TAB $\state$ = {\textsc{ApplyUpdateRP}}($\delta K_{\mathit{src}}' = \{\, \vecnormal{t} \mapsto -m \,\},\state$)\\ 
\TAB $\state$ = {\textsc{ApplyUpdateRP}}($\delta K_{\mathit{dst}}' = \{\, \vecnormal{t} \mapsto m \,\},\state$)\\
\RETURN $\state$\\
\bottomrule
\end{tabular}
};

\end{tikzpicture}
\end{center}
\vspace{-1.2em}
\caption{
\textsc{MinorRebalancingRP} adapts the procedure 
\textsc{MinortRebalancing} from Figure~\ref{fig:updateProgram}
to refined partitions. 
It moves all tuples with $X$-value $x$ from 
$K_{\mathit{src}}$ to 
$K_{\mathit{dst}}$ and 
from $K_{\mathit{src}}'$ to
$K_{\mathit{dst}}'$. }
\label{fig:minor_rebal_part_both_var}
\vspace{-0.45em}
\end{figure}

The trigger procedure \textsc{OnUpdateRP} in
Figure \ref{fig:updateProgram_Part_both_var} 
adapts the procedure  
\textsc{OnUpdate} in Figure 
\ref{fig:updateProgram} to refined partitions.   
It takes as input an update 
$\delta R = \{(\deltaA,\deltaB) \mapsto \p\}$ 
and a current state $\state = (\eps, N, \dbeps,\inst{V})$ 
with refined partitions of relations of 
a database $\inst{D}$, 
 maintains $\state$
under the update and, if necessary, performs 
major and minor rebalancing.   
The trigger procedures for updates 
relations $S$ and $T$ are defined
analogously. 

In Lines 1-8, the procedure identifies the part 
$R_{r_Ar_B}$ of relation $R$ that is affected
by the update.  
Then, it calls the procedure
\textsc{ApplyUpdateRP}($\delta R_{r_Ar_B} = \{(\deltaA,\deltaB) \mapsto \p\}, \state$), which is defined in Figure
  \ref{fig:applyUpdate_part_all}.
  If the update causes the violation of 
  the size invariant
$\floor{\frac{1}{4}N} \leq |\db| < N$,
the procedure halves or doubles $N$ and 
performs major rebalancing by calling 
\textsc{MajorRebalancingRP}, which,  similar 
to \textsc{MajorRebalancing}
in Figure  \ref{fig:updateProgram},   
strictly repartitions the relations in $\dbeps$ 
with threshold $N^{\eps}$
and recomputes the auxiliary views
in $\inst{V}$.  As explained in Appendix \ref{sec:pre_time_refined},
recomputation of partitions and views, and hence, 
major rebalancing    
takes $\bigO{|\inst{D}|^{1 + \min\{\eps, 1-\eps\}}}$
time.

In Lines 16-23, the procedure checks whether 
the update causes a violation of the 
light or heavy part conditions from
Definition \ref{def:loose_relation_partition_sets}
and, if so, performs minor rebalancing
by invoking \textsc{MinorRebalancingRP}($K_{\mathit{src}}, K_{\mathit{dst}}, K_{\mathit{src}}', K_{\mathit{dst}}', X, x, \state$) described in 
Figure \ref{fig:minor_rebal_part_both_var}.
In contrast to the procedure 
\textsc{MinorRebalancing} 
in Figure \ref{fig:updateProgram}, the parameter list 
of
\textsc{MinorRebalancingRP} contains  
two source relations 
$K_{\mathit{src}}$ and $K_{\mathit{src}}'$ 
and two target relations 
$K_{\mathit{dst}}$ and $K_{\mathit{dst}}'$.
\textsc{MinorRebalancingRP} moves all tuples in
$K_{\mathit{src}}$ and $K_{\mathit{src}}'$
with $X$-value $x$ to
$K_{\mathit{dst}}$ and $K_{\mathit{dst}}'$, respectively. 
While the procedure 
\textsc{OnUpdate}
in Figure \ref{fig:updateProgram} 
invokes  \textsc{MinorRebalancing} at most once per update,
the procedure \textsc{OnUpdateRP} might call 
\textsc{MinorRebalancingRP} up to two times
per update.
We consider the case where the update affects 
the part $R_{ll}$ and 
implies that the latter relation becomes heavy on $\alpha$
as well as on $\beta$. In this case,
the conditions in both lines 16 and 20 hold.
Note that both parts $R_{ll}$ and $R_{lh}$ can contain 
tuples from relation $R$ with
  $A$-value $\alpha$.
The procedure \textsc{MinorRebalancingRP}
called in Line 17 moves such tuples
from $R_{ll}$ to $R_{hl}$ and from
$R_{lh}$ to $R_{hh}$.
Likewise, both $R_{ll}$ and $R_{hl}$ can contain 
$R$-tuples with $B$-value $\beta$.
The procedure \textsc{MinorRebalancingRP}
called in Line 21 moves these tuples 
from $R_{ll}$ to $R_{lh}$ and from
$R_{hl}$ to $R_{hh}$. 
In general, the number of tuples moved from one relation part to the other 
is less than $\frac{3}{2}N^{\eps}+1$. A minor rebalancing step 
moves tuples between at most between two pairs of relation 
parts. Since each tuple move needs 
$\bigO{N^{\max\{\eps,1-\eps\}}}$  
time and $|\inst{D}| = \Theta(N)$, minor rebalancing 
needs overall $\bigO{|\inst{D}|^{\eps + \max\{\eps,1-\eps\}}}$ 
time. 

We now put all pieces together to 
prove the main theorem of this section. 

\begin{proof}[Proof of Theorem~\ref{theo:main_result_part_both_var}]
As shown in Proposition~\ref{prop:preprocessing_step_refined},
the preprocessing stage takes $\bigO{|\inst{D}|^{\frac{3}{2}}}$ time.
Constant answer time is ensured by the fact that 
the triangle count is materialized and maintained.  
By Proposition~\ref{prop:space_complexity_ref_part},
the space complexity is 
$\bigO{|\inst{D}|^{\max\{\min\{1 + \eps,2-2\eps\}, 1\}}}$.

The analysis of the amortized update time follows the 
the proof of Theorem \ref{theo:main_result}. 
The time $\bigO{|\inst{D}|^{1 + \min\{\eps, 1-\eps\}}}$ needed to perform major rebalancing 
is amortized over 
$\Omega(N)$ updates following the previous major 
rebalancing step. 
The minor rebalancing time $\bigO{|\inst{D}|^{\eps + \max\{\eps,1-\eps\}}}$  
needed to move all tuples
with a specific value between relation parts is amortized
over $\Omega(N^{\eps})$ updates to tuples with the same
values following the previous minor rebalancing step for that value. 
Hence, the amortized rebalancing time is 
$\bigO{|\inst{D}|^{\max\{\eps,1-\eps\}}}$.
Since, by Proposition~\ref{prop:single_step_time_ref_part},
 the time to process a single-tuple update 
is $\bigO{|\db|^{\max\{\eps,1-\eps\}}}$,
this implies 
that the overall amortized single-tuple update time is
$\bigO{|\inst{D}|^{\max\{\eps,1-\eps\}}}$.
\end{proof}

\section{Enumerating Triangles under Updates}
\label{sec:enum_triangles}
We focus on the problem of enumerating all triangles in 
the result of the full triangle query 
\begin{align*}
Q(a,b,c) = R(a,b) \cdot S(b,c) \cdot T(c,a)
\end{align*}
under single-tuple updates to the input relations. 
We present an \ivme variant that requires 
$\bigO{|\inst{D}|^{\max{\eps,1-\eps}}}$ update time and ensures 
constant-delay enumeration after each update. 
Due to Corollary \ref{prop:lower_bound_triangle_enum}
in Appendix \ref{sec:lowerbound},
this is worst-case optimal  for $\eps = 0.5$.
The enumeration process reports only distinct tuples and their multiplicities in the query result.

The rest of this section is dedicated to the proof of the following theorem.

\begin{theorem}\label{theo:main_result_fulljoin}
Given a database $\db$ and $\eps\in[0,1]$, \ivme maintains 
the result of the full triangle query under single-tuple updates 
to $\db$ with $\bigO{|\db|^{\frac{3}{2}}}$ preprocessing 
time, $\bigO{|\db|^{\max\{\eps,1-\eps\}}}$ amortized update time,
constant enumeration delay,  
and $\bigO{|\db|^{\frac{3}{2}}}$ space.
\end{theorem}

\subsection{Adaptive Maintenance for the Full Triangle Query}
\label{sec:adaptieve_maint_enum}
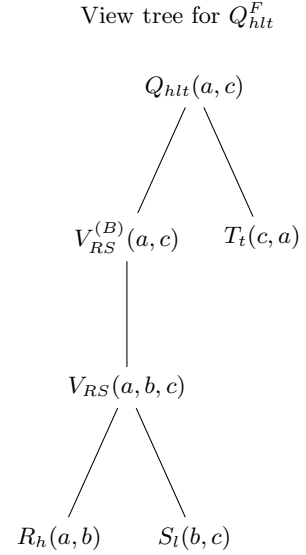
\begin{figure}[t]
  \begin{minipage}[b]{0.74\textwidth}
  \begin{center}
    \small
    \renewcommand{\arraystretch}{1.3}  
    \begin{tabular}{@{\hskip 0.05in}l@{\hskip 0.4in}l@{\hskip 0.05in}}
      \toprule
      Materialized View Definition & Space Complexity \\    
      \midrule
      $Q^{L}(a,b,c) = \sum_{u \in \{h,l\}} 
      R_{u}(a,b) \cdot S_{u}(b,c) \cdot T_{u}(c,a)$ & 
      $\bigO{|\inst{D}|^{\frac{3}{2}}}$ \\[3pt]

      $Q^{F}_{hlt}(a,b,c) = 
      R_{h}(a,b) \cdot S_{l}(b,c) \cdot T_{t}(c,a)$ & \\

      $\TAB\TAB V_{RS}(a,b,c) = R_{h}(a,b) \cdot S_{l}(b,c)$ & 
      $\bigO{|\inst{D}|^{1+\min{\{\,\eps, 1-\eps \,\}}}}$ \\

      $\TAB\TAB V^{(B)}_{RS}(a,c) = \sum_b V_{RS}(a,b,c)$ & 
      $\bigO{|\inst{D}|^{1+\min{\{\,\eps, 1-\eps \,\}}}}$ \\

      $\TAB\TAB Q_{hlt}(a,c) =
      V^{(B)}_{RS}(a,c) \cdot T_{t}(c,a)$ & 
      $\bigO{|\inst{D}|}$\\[3pt]

      $Q^{F}_{rhl}(a,b,c) =  
      R_{r}(a,b) \cdot S_{h}(b,c) \cdot T_{l}(c,a)$ & \\

      $\TAB\TAB V_{ST}(b,c,a) = S_{h}(b,c) \cdot T_{l}(c,a)$ & 
      $\bigO{|\inst{D}|^{1+\min{\{\,\eps, 1-\eps \,\}}}}$ \\

      $\TAB\TAB V^{(C)}_{ST}(b,a) = \sum_{c} V_{ST}(b,c,a)$ & 
      $\bigO{|\inst{D}|^{1+\min{\{\,\eps, 1-\eps \,\}}}}$ \\

      $\TAB\TAB Q_{rhl}(a,b) =
      R_{r}(a,b) \cdot V^{(C)}_{ST}(b,a)$ & 
      $\bigO{|\inst{D}|}$\\[3pt]

      $Q^{F}_{lsh}(a,b,c) =  
       R_{l}(a,b) \cdot S_{s}(b,c) \cdot T_{h}(c,a)$ & \\

      $\TAB\TAB V_{TR}(c,a,b) = T_{h}(c,a) \cdot R_{l}(a,b)$ & 
      $\bigO{|\inst{D}|^{1+\min{\{\,\eps, 1-\eps \,\}}}}$ \\

      $\TAB\TAB V^{(A)}_{TR}(c,b) = \sum_a V_{TR}(c,a,b)$ & 
      $\bigO{|\inst{D}|^{1+\min{\{\,\eps, 1-\eps \,\}}}}$ \\

      $\TAB\TAB Q_{lsh}(b,c) =  
      S_{s}(b,c) \cdot V^{(A)}_{TR}(c,b)$ & 
      $\bigO{|\inst{D}|}$\\

      \bottomrule    
    \end{tabular}\vspace*{-0.5em}
  \end{center}
  \end{minipage}  
  \begin{minipage}[b]{0.25\textwidth}
    \small
    \begin{tikzpicture}[xscale=0.45, yscale=1]
      \node at (0, 0) (A) {$Q_{hlt}(a,c)$};
      \node at (2, -2) (B) {$T_{t}(c,a)$} edge[-] (A);
      \node at (-2, -2) (C) {$V^{(B)}_{RS}(a,c)$} edge[-] (A);
      \node at (-2, -4) (D) {$V_{RS}(a,b,c)$} edge[-] (C);
      \node at (-4, -6) (E) {$R_{h}(a,b)$} edge[-] (D);
      \node at (0, -6) (F) {$S_{l}(b,c)$} edge[-] (D);

      \node at (-0.5, 0.95) (G) {View tree for $Q^{F}_{hlt}$};
    \end{tikzpicture}
  \end{minipage}
  \caption{
  (left) The materialized views 
  $\inst{V} = \{ Q^{L}, Q_{hlt}, Q_{rhl}, Q_{lsh}, V_{RS}, V^{(B)}_{RS}, V_{ST}, V^{(C)}_{ST}, V_{TR}, V^{(A)}_{TR} \}$ supporting 
  the constant-delay enumeration of the result of the full triangle query.
  $r$, $s$, and $t$ stand for $h$ or $l$.
  $\inst{D}$ is the input database. The input relations are partitioned 
  for a fixed $\eps \in [0,1]$. The superscripts $L$ and $F$ denote listing and factorized forms of query results. The result of $Q^L$ is materialized, while the results of $Q^F_{hlt}$, $Q^F_{rhl}$, and $Q^F_{lsh}$ are enumerable with constant delay using other auxiliary views (denoted by indentation). 
  (right) The view tree for maintaining the result of $Q^F_{hlt}$ in factorized form. } 
  \label{fig:view_definitions_full_join}
\end{figure}

The \ivme variant employs a similar adaptive maintenance strategy as with the triangle count query. It first partitions the relations $R$, $S$, and $T$ on the variables $A$, $B$, and $C$, respectively, with the same thresholds as in the original \ivme algorithm. It 
then decomposes $Q$ into eight skew-aware views defined over the relation parts:
\begin{align*}
Q_{rst}(a,b,c) = R_r(a,b) \cdot S_s(b,c) \cdot T_t(c,a), \quad\text{ for } 
r,s,t \in \{h,l\}.
\end{align*}
Enumerating the result of $Q$ is equivalent to enumerating the result of each $Q_{rst}$.
As with the triangle count query, the \ivme variant customizes the maintenance strategy 
for each of these views and relies on auxiliary views to speed up the view maintenance.  

The original \ivme algorithm, however, fails to achieve sublinear maintenance time for most of these skew-aware views. 
Consider for instance the view $Q_{hhl}$ and a single-tuple update 
$\delta{R_h} = \{(\alpha,\beta) \mapsto \p\}$ to the heavy part $R_h$ of relation $R$. 
The delta $\delta Q_{hhl}(\alpha,\beta,c) = \delta R_h(\alpha,\beta) \cdot 
S_h(\beta,c) \cdot T_l(c,\alpha)$ iterates over linearly many $C$-values in the worst case. 
Precomputing the view $V_{ST}(b,c,a) = S_h(b,c) \cdot T_l(c,a)$ and 
rewriting the delta as $\delta Q_{hhl}(\alpha,\beta,c) = \delta R_h(\alpha,\beta) \cdot V_{ST}(\beta,c,\alpha)$ makes no improvement in the worst-case running time. In contrast, for the triangle count query, the view $V_{ST}(b,a) = S_h(b,c) \cdot T_l(c,a)$ enables computing $\delta Q_{hhl}$ in constant time. 

The skew-aware views of the full triangle query can be maintained in sublinear time by avoiding the listing (tabular) form of the view results. For that purpose, 
the result of a skew-aware view can be maintained in {\em factorized form}:
 Instead of using one materialized view, a hierarchy of materialized views
 is created  such that each of them admits sublinear maintenance time and all of them together guarantee constant-delay enumeration of the result of the skew-aware view. 
This technique of factorized evaluation appears in recent publications studying incremental view maintenance~\cite{BerkholzKS17,Idris:dynamic:SIGMOD:2017,NO18}.

Figure~\ref{fig:view_definitions_full_join}(left) presents the views used by \ivme to maintain the result of the full triangle query under updates to the base relations. 
The top-level view $Q^{L}$ materializes the union of the results of $Q_{hhh}$ and $Q_{lll}$ in listing form. The remaining top-level views, 
$Q^{F}_{hlt}$, $Q^{F}_{rhl}$, and $Q^{F}_{lsh}$ with $r,s,t \in \{h,l\}$, avoid materialization altogether but ensure constant-delay enumeration of their results using other auxiliary materialized views (denoted by indentation).

Figure~\ref{fig:view_definitions_full_join}(right) shows the materialized views needed to maintain the result of $Q^{F}_{hlt}$ for $t \in \{h,l\}$ in factorized form. 
These views make a view tree with input relations as leaves and updates propagating in a bottom-up manner. 
The result of $Q^{F}_{hlt}$ is distributed among two auxiliary views, $Q_{hlt}$ and $V_{RS}$. The former stores all $(a,c)$ pairs that would appear in the result of $Q^{F}_{hlt}$, while the latter provides the matching $B$-values for each $(a,c)$ pair. 
The two views together provide constant-delay enumeration of the result of 
$Q^{F}_{hlt}$. In addition to them, the view $V^{(B)}_{RS}$ serves to support constant-time updates to $T_{t}$. The view trees for 
$Q^{F}_{rhl}$ and $Q^{F}_{lsh}$ are analogous.

\begin{figure}[t]
\begin{center}
\renewcommand{\arraystretch}{1.3}
\setcounter{magicrownumbers}{0}
\begin{tabular}{ll@{\hskip 0.25in}l@{\hspace{0.6cm}}c}
\toprule
\multicolumn{2}{l}{\textsc{ApplyUpdateEnum}($\delta R_{r}(\deltaA,\deltaB),\state$)}& & Time \\
\cmidrule{1-2} \cmidrule{4-4}

\rownumber & \LET $\delta R_r = \{(\deltaA,\deltaB) \mapsto \p\}$ \\
\rownumber & \LET $\state = (\eps, N, 
\{R_h, R_l, S_h, S_l, T_h, T_l\},$ \\
& \TAB \TAB \TAB \TAB $\{ Q^{L}, Q_{hlt}, Q_{rhl}, Q_{lsh}, V_{RS}, V^{(B)}_{RS}, V_{ST}, V^{(C)}_{ST}, V_{TR}, V^{(A)}_{TR} \})$ \\

\rownumber & \IF ($r$ is $h$) &\\  
\rownumber & \TAB $Q^{L}(\deltaA,\deltaB,c) = Q^{L}(\deltaA,\deltaB,c) + \delta{R_{h}(\deltaA,\deltaB)} \cdot S_{h}(\deltaB,c) \cdot T_{h}(c,\deltaA)$ & & $\bigO{|\inst{D}|^{1-\eps}}$ \\

\rownumber & \TAB $V_{RS}(\deltaA,\deltaB,c) = \delta{R_h(\deltaA,\deltaB)} \cdot S_l(\deltaB,c)$ & & $\bigO{|\inst{D}|^\eps}$ \\

\rownumber & \TAB $V^{(B)}_{RS}(\deltaA,c) = \delta{V_{RS}}(\deltaA,\deltaB,c)$ & & $\bigO{|\inst{D}|^\eps}$ \\

\rownumber & \TAB $Q_{hlt}(c,\deltaA) =  \delta{V^{(B)}_{RS}}(\deltaA,c) 
\cdot T_t(c,\deltaA)$ & & $\bigO{|\inst{D}|^\eps}$ \\

\rownumber & \ELSE &\\
\rownumber & \TAB $Q^{L}(\deltaA,\deltaB,c) = \delta{R_{l}(\deltaA,\deltaB)} \cdot S_{l}(\deltaB,c) \cdot T_{l}(c,\deltaA)$ & & $\bigO{|\inst{D}|^{\eps}}$ \\

\rownumber & \TAB $V_{TR}(c,\deltaA,\deltaB) = \delta{R_l(\deltaA,\deltaB)} \cdot T_h(c,\deltaA)$ & & $\bigO{|\inst{D}|^{1-\eps}}$ \\

\rownumber & \TAB $V^{(A)}_{TR}(c,\deltaB) = V^{(A)}_{TR}(c,\deltaB) + \delta{V_{TR}}(c,\deltaB)$ & & $\bigO{|\inst{D}|^{1-\eps}}$ \\

\rownumber & \TAB $Q_{lsh}(\deltaB,c) = Q_{lsh}(\deltaB,c) +  \delta{V^{(A)}_{TR}}(c,\deltaB) \cdot S_s(\deltaB,c)$ & & $\bigO{|\inst{D]|^{1-\eps}}}$ \\

\rownumber & $Q_{rhl}(a,b) = Q_{rhl}(a,b) + \delta{R_{r}(\deltaA,\deltaB)} \cdot V^{(C)}_{ST}(\deltaB,\deltaA)$ & & $\bigO{1}$ \\

\rownumber & $R_{r}(\deltaA,\deltaB) = R_{r}(\deltaA,\deltaB) + 
\delta{R}_{r}(\deltaA,\deltaB)$ & &
$\bigO{1}$ \\

\midrule
\multicolumn{2}{r}{Total update time:} & &
$\bigO{N^{\max\{\eps, 1-\eps\}}}$ \\
\bottomrule
\end{tabular}
\end{center}\vspace{-1em}
\caption{
 (left) Maintaining an \ivme state under a single-tuple update to 
 support constant-delay enumeration of the result 
 of the full triangle query. 
 \textsc{ApplyUpdateEnum} takes as input an update 
 $\delta R_r$ to the heavy or light part of $R$ and the current 
 \ivme state of a database $\inst{D}$ partitioned using $\eps \in [0,1]$. 
 $r$, $s$, and $t$ stand for $h$ or $l$. 
(right) The time complexity of computing deltas. 
The procedures for updates to $S$ and $T$ are similar.
}
\label{fig:applyUpdate_full_join}
\end{figure}

\subsection{Preprocessing Time  for the Full Triangle Query} 
The preprocessing stage builds the initial \ivme state $\state = (\dbeps, \inst{V}, N)$ for a given database $\db$ and $\eps\in[0,1]$. This step partitions the input relations and computes the views in $\inst{V}$ from Figure~\ref{fig:view_definitions_full_join} before processing any update. If $\db$ is empty, the preprocessing cost is $\bigO{1}.$

\begin{proposition}\label{prop:preprocessing_step_fulljoin}
Given a database $\db$ and $\eps\in[0,1]$, constructing the 
initial \ivme state of $\db$ to support the maintenance 
of the result of the full triangle query 
takes $\bigO{|\db|^{\frac{3}{2}}}$ time.
\end{proposition}
\begin{proof}
Partitioning the input relations takes $\bigO{|\db|}$ time. Computing $Q^{L}$ using worst-case optimal join algorithms takes $\bigO{|\db|^{\frac{3}{2}}}$ time~\cite{Ngo:SIGREC:2013}. The remaining top-level views $Q^{F}_{hlt}$, $Q^{F}_{rhl}$, and  $Q^{F}_{lsh}$ with $r,s,t \in \{h,l\}$ are not materialized. Computing the auxiliary views $V_{RS}$ and $V^{(B)}_{RS}$ takes $\bigO{|\db|^{1+\min\{\eps,1-\eps\}}}$ time, as explained in the proof of Proposition~\ref{prop:preprocessing_step}. The view $Q_{hlt}$ intersects $T_{t}$ and $V^{(B)}_{RS}$ in linear time. The same reasoning applies to the other views. Thus, the overall preprocessing time is $\bigO{|\db|^{\frac{3}{2}}}$.
\end{proof}

\subsection{Space Complexity of Maintaining the Full Triangle Query}
We next analyze the space complexity of the \ivme algorithm.

\begin{proposition}\label{prop:space_complexity_fulljoin}
Given a database $\db$, the \ivme state constructed from $\db$ to support 
the maintenance of the result of the full triangle 
query takes $\bigO{|\db|^{\frac{3}{2}}}$ space.
\end{proposition}  
\begin{proof}
Let $\state=(\dbeps, \inst{V}, N)$ be a state of $\db$.
Figure~\ref{fig:view_definitions_full_join} summarizes the space complexity of the materialized views in $V$.
The size of $\dbeps$ and the result of $Q^{L}$ is upper-bounded by $N^{\frac{3}{2}}$, the maximum number of triangles in a database of size $N$. 
The space complexity of $V_{RS}$, $V^{(B)}_{RS}$, $V_{ST}$, $V^{(C)}_{ST}$, $V_{TR}$, and $V^{(A)}_{TR}$ is $\bigO{N^{1+\min\{\eps,1-\eps\}}}$, as discussed in the proof of Proposition~\ref{prop:preprocessing_step}.
The views $Q_{hlt}$, $Q_{rhl}$, and $Q_{lsh}$ with 
$r,s,t \in \{h,l\}$ take space linear in the database size.
From the invariant $|\db|=\Theta(N)$ follows the claimed space complexity $\bigO{|\db|^{\frac{3}{2}}}$.
\end{proof}

\subsection{Processing a Single-Tuple Update to the Full Triangle Query}
Figure~\ref{fig:applyUpdate_full_join} shows the procedure for maintaining a current state $\state$ of the full triangle query under an update $\delta R_{r}(a,b)$. 
If the update affects the heavy part $R_h$ of $R$, the procedure maintains the listing form of $Q^{L}$ (line 4) and the factorized form of $Q^{F}_{hlt}$ by propagating $\delta{R_h}$ through the view tree from Figure~\ref{fig:view_definitions_full_join} (lines 5-7). If the update affects the light part $R_l$ of $R$, the procedure maintains $Q^{L}$ (line 9) and propagates 
$\delta{R_l}$ through the view tree for $Q^{F}_{lsh}$ (lines 10-12). Finally, it updates $Q_{rhl}$ (line 13) and the part of $R$ affected by $\delta{R_{r}}$ (line 14).
The views $V_{ST}$ and $V^{(C)}_{ST}$ remain unchanged as they do not refer to $R_h$ or $R_l$.

\begin{proposition}\label{prop:single_step_time_fulljoin}
Given a state $\state$ constructed from a 
database $\db$ for a fixed $\eps\in[0,1]$ to support the maintenance 
of the result of the full triangle query, \ivme maintains $\state$ under 
a single-tuple update to any input 
relation in $\bigO{|\db|^{\max\{\eps,1-\eps\}}}$ time.
\end{proposition}
\begin{proof}
Figure~\ref{fig:applyUpdate_full_join} shows the time complexity of each maintenance statement of the \textsc{ApplyUpdateEnum} procedure, for a given single-tuple update $\delta R_{r}(a,b)$ and a state 
$\state=(\dbeps, \inst{V}, N)$ of $\db$. 
This complexity is determined by the number of $C$-values that need to be iterated over during delta computation. 

We first analyze the case when $\delta{R_{r}}$ affects the heavy part $R_h$ of $R$.
Updating the view $Q^{L}$ (line 4) requires iterating over $C$-values. The number of distinct $C$-values in the heavy part $T_h$ is at most $2N^{1-\eps}$.
Computing $\delta V_{RS}$ (line 5) needs to iterate over at most $\frac{3}{2}N^{\eps}$ $C$-values in $S_l$ for the given $b$.
Propagating $\delta V_{RS}$ through the view tree for $Q^{F}_{hlt}$ shown in Figure~\ref{fig:view_definitions_full_join} (lines 5-7) takes time linear in the size of this delta, that is, $\bigO{N^{\eps}}$. 

We now consider the case when $\delta{R_{r}}$ affects the light part $R_l$ of $R$. Updating $Q^{L}$ (line 9) requires iterating over at most $\frac{3}{2}N^{\eps}$ $C$-values in $S_l$ for the given $b$.
Computing $\delta V_{TR}$ (line 10) touches at most $2N^{1-\eps}$ distinct $C$-values in $T_h$. Propagating this delta through the view tree for $Q^{F}_{lsh}$ (lines 10-12) takes $\bigO{N^{1-\eps}}$ time.   

Finally, updating $Q_{rhl}$ and the part of $R$ affected by $\delta R_{r}$ takes constant time. 
The total execution time of the procedure from Figure~\ref{fig:applyUpdate_full_join} is $\bigO{N^{\max\{\eps,1-\eps\}}}$. 
From the invariant $|\db| = \Theta(N)$ follows the claimed time complexity $\bigO{|\db|^{\max\{\eps,1-\eps\}}}$. The analysis for updates to $S$ and $T$ is similar due to the symmetry of the triangle query and materialized views.
\end{proof}

\subsection{Rebalancing Partitions  for the Full Triangle Query}
The major rebalancing procedure 
for the full triangle query recomputes the materialized views from Figure~\ref{fig:view_definitions_full_join} in 
$\bigO{|\db|^{\frac{3}{2}}}$ time. 
Minor rebalancing is accomplished by 
moving tuples between relation parts,
as in Section \ref{sec:rebalancing}. Hence,
its time complexity is   
$\bigO{|\db|^{\eps + \max\{\eps, 1-\eps\}}}$.

We now prove the main theorem of this section.

\begin{proof}[Proof of Theorem \ref{theo:main_result_fulljoin}]
The preprocessing time $\bigO{|\db|^{\frac{3}{2}}}$ follows from Proposition~\ref{prop:preprocessing_step_fulljoin}, while the space complexity $\bigO{|\db|^{\frac{3}{2}}}$ follows from Proposition~\ref{prop:space_complexity_fulljoin}.
The materialized views used by \ivme (cf. Figure~\ref{fig:view_definitions_full_join}) provide constant-delay enumeration of the full triangle query result. 
Materializing $Q^{L}$ in listing form achieves this goal.
The views $Q^{F}_{hlt}$, $Q^{F}_{rhl}$, and $Q^{F}_{lsh}$ are stored in factorized form but can enumerate their results with constant delay. 
Figure~\ref{fig:view_definitions_full_join} shows the view tree for $Q^{F}_{hlt}$ where the views $V_{RS}$, $V^{(B)}_{RS}$, and $Q_{hlt}$ are constructed and maintained in a bottom-up manner. The root $Q_{hlt}$ contains only the $(a,c)$ pairs that would appear in the result of $Q^{F}_{hlt}$, and for each $(a,c)$ pair, $V_{RS}$ provides only the $B$-values that exist in both $R_h$ and $S_l$. Thus, iterating over $Q_{hlt}$ and $V_{RS}$ can enumerate the result of $Q^{F}_{hlt}$ with constant delay. Similar analysis applies to $Q^{F}_{rhl}$ and $Q^{F}_{lsh}$.

As in the proof of Theorem~\ref{theo:main_result} in Section~\ref{sec:main_proof}, 
the time $\bigO{|\db|^{\frac{3}{2}}}$ to process major rebalancing is amortized over $\Omega(|\db|)$ updates.
Likewise,  the time to do minor rebalancing  is amortized over 
$\Omega(|\db|^{\eps})$ updates. Hence, the amortized 
major and minor rebalancing times are   
$\bigO{|\db|^{\frac{1}{2}}}$
and 
$\bigO{|\db|^{\max\{\eps, 1-\eps\}}}$, respectively. 
Since by Proposition \ref{prop:single_step_time_fulljoin}, 
the time to process a single-tuple update is 
$\bigO{|\db|^{\max\{\eps, 1-\eps\}}}$, 
the overall amortized update  is
$\bigO{|\db|^{\max\{\eps, 1-\eps\}}}$.
\end{proof}

\section{Loomis-Whitney Count Queries}
\label{sec:loomis-whitney}
In this section, we consider the incremental maintenance of Loomis-Whitney count queries, which 
generalize  the triangle count query. 
A Loomis-Whitney query of degree $n \geq 3$ is of the form
\begin{align*}
Q() = \sum\limits_{a_1, \ldots, a_n}  
R_1(a_1,a_2,...,a_{n-1}) \ztimes R_2(a_2,a_3,...,a_n) \ztimes ... \ztimes 
R_n(a_n, a_1,...,a_{n-2}),
\end{align*}
where each relation $R_i$ is over the schema $\inst{A}_i = (A_j)_{j \in \{1,\ldots, n\} \setminus \{k\}}$ with $i = (k\bmod{n})+1$.
The triangle count query is a Loomis-Whitney query of degree $3$.

It turns out that any Loomis-Whitney query can be maintained with the same 
complexities as the triangle count query.
 
\begin{theorem}\label{theo:main_result_LW}
Given a database $\db$ and $\eps\in[0,1]$, \ivme maintains any Loomis-Whitney count query under single-tuple updates to $\db$ with $\bigO{|\inst{D}|^{1 + \max\{\min\{\eps, 1-\eps\}, \frac{1}{n-1}\}}}$ preprocessing time, 
$\bigO{|\db|^{\max\{\eps,1-\eps\}}}$ amortized update time, constant answer time, and $\bigO{|\db|^{1 + \min\{\eps,1-\eps\}}}$ space.
\end{theorem}

\noindent
Given that the answer time is constant, 
the update time stated in Theorem~\ref{theo:main_result_LW}
is worst-case optimal, conditioned on the \OMv conjecture.

\begin{proposition}\label{prop:lower_bound_LW}
For any $\gamma > 0$, Loomis-Whitney count query $Q$, and database $\db$,
there is no algorithm that incrementally maintains the result of $Q$ under single-tuple updates to $\db$ with arbitrary preprocessing time, $\bigO{|\db|^{\frac{1}{2} - \gamma}}$ amortized update time, and $\bigO{|\db|^{1 - \gamma}}$ answer time, unless the \OMv conjecture fails.
\end{proposition}

In the rest of Appendix \ref{sec:loomis-whitney}
we prove Theorem \ref{theo:main_result_LW}
and Proposition \ref{prop:lower_bound_LW}.
The proofs are similar to the triangle count case. 

\subsection{Adaptive Maintenance for Loomis-Whitney Queries}\label{sec:maintenance_LW}
\nop{
\begin{figure}[t]
  \begin{center}
    \renewcommand{\arraystretch}{1.2}  
    \begin{tabular}{@{\hskip 0.05in}l@{\hskip 0.4in}l@{\hskip 0.05in}}
      \toprule
      Materialized View Definition & Space Complexity \\    
      \midrule
      $Q() = \sum\limits_{u_1, \ldots, u_n\in \{h,l\}}\; \sum\limits_{a_1, \ldots, a_n} R_1^{u_1}(a_1,a_2,\ldots,a_{n-1}) \ztimes \ldots \ztimes R_n^{u_n}(a_n, a_1,\ldots,a_{n-2})$  & 
      $\bigO{1}$ \\
   $V_i(a_i, \ldots, a_n, a_1, \ldots, a_{i-2}) = \sum_{a_{i-1}} \left( R_{i+1}^h \ztimes R_{i+2}^l \ztimes \ldots \ztimes R_n^l \ztimes R_1^l \ztimes \ldots \ztimes R_{i-1}^l \right)$ & \\
      \bottomrule    
    \end{tabular}
  \end{center}
  \caption{The definitions and space complexities of the 
  views in $\inst{V} = \{ Q, V_{RS}, V_{ST}, V_{TR} \}$ 
  maintained by \ivme with refined partitions as part of 
  an \ivme   state of a database $\inst{D}$ for a fixed $\eps \in [0,1]$.} 
  \label{fig:view_definitions_loomis_whitney}
\end{figure}
}
Consider a database $\db=\{R_1, \ldots, R_n\}$, a Loomis-Whitney count query $Q$ of some degree $n \geq 3$, and a fixed $\eps\in[0,1]$. 
The \ivme variant we introduce here is very similar to the 
original \ivme algorithm maintaining the triangle count.   
It partitions each relation $R_i$ on $A_i$
into a light part $R_i^l$ and a heavy part $R_i^h$.
Then, it decomposes $Q$ into skew-aware views expressed over the relation parts: 
\begin{align*}
Q_{u_1 \ldots u_n}() =  \sum\limits_{a_1, \ldots, a_n} R_1^{u_1}(a_1,a_2,\ldots,a_{n-1}) \ztimes R_2^{u_2}(a_2,a_3,\ldots,a_n) \ztimes \ldots \ztimes R_n^{u_n}(a_n, a_1,\ldots,a_{n-2})
\end{align*} 
for $u_1,\ldots, u_n \in \{h,l\}$. The query $Q$ is the sum of these views: 
$Q() = \textstyle\sum_{u_1,\ldots, u_n \in \{h,l\}}Q_{u_1 \ldots u_n}()$. 

The \ivme variant materializes, besides $Q$, auxiliary views  
$V_i$ with $i \in \{1, \ldots ,n\}$ that join the heavy relation part $R_{i+1}^h$ with the light parts of all other relations but $R_{i}$. For convenience, we 
skip the relation schemas in the following expression:
\begin{align*}
    V_i(a_i, \ldots, a_n, a_1, \ldots, a_{i-2}) = \sum_{a_{i-1}} \left( R_{i+1}^h \ztimes R_{i+2}^l \ztimes \ldots \ztimes R_n^l \ztimes R_1^l \ztimes \ldots \ztimes R_{i-1}^l \right) 
\end{align*}
for $i \in \{1, \ldots ,n\}$.  In the above view definition, $i -1$ stands for $n$ in case of $i = 1$ and $i +1$ stands for $1$ in case of $i = n$.

An \ivme  state $(\eps, N, \dbeps, \inst{V})$ of $\inst{D}$ is defined as usual: 
$N$ is the threshold base with $\floor{\frac{1}{4}N} \leq |\db| < N$, 
$\dbeps$ contains for each $R_i$
its heavy-light partitions on variable $A_i$ with threshold $N^{\eps}$,
and $\inst{V}$ consists of $Q$ and the auxiliary views $V_1, \ldots,  V_n$.

\begin{proof}[Proof of Theorem \ref{theo:main_result_LW}]
Let $Q$ be a Loomis-Whitney query of some degree $n \geq 3$.
The analysis of the preprocessing time and 
space complexity follows the same reasoning
 as in the proofs of Propositions 
 \ref{prop:preprocessing_step} and 
 \ref{prop:space_complexity}.
 
\paragraph{Preprocessing Time and Space Complexity.}
The FAQ-width of $Q$ is $\frac{n}{n-1}$, hence, $Q$ can be computed in time $\bigO{N^{\frac{n}{n-1}}}$.
We analyze the time to compute an auxiliary view 
$V_i$. 
Each tuple in  $R_{i+1}^h$ fixes the data values of 
all but the variable $A_{i-1}$. Each tuple that results from the 
join of $R_{i+1}^h$ and $R_{i+2}^l$ fixes the values
of all variables in the query. 
Hence,  $V_i$ can be computed 
by iterating first over all tuples in $R_{i+1}^h$, and for each 
such tuple $\inst{t}$, iterating over all matching tuples $\inst{t}'$ in 
$R_{i+2}^l$. For each matching pair $\inst{t}$ and $\inst{t}'$,
we additionally look up the tuples in the other relations that agree 
with $\inst{t}$ and $\inst{t}'$ on the common variables. 
Since $R_{i+2}^l$ is light on $A_{i+2}$ and the
data value of this variable is fixed by each tuple in 
$R_{i+1}^h$, the computation time is 
$\bigO{|R_{i+1}^h| \ztimes N^{\eps}}$.
Alternatively, we can iterate over the tuples in 
$R_{i+2}^l$ and search for matching tuples in $R_{i+1}^h$.
Since the only variable whose value is not fixed by tuples from 
$R_{i+2}^l$ is $A_{i+1}$ and $R_{i+1}^h$ is heavy on that variable, 
the computation time in this case is 
$\bigO{|R_{i+2}^l| \ztimes N^{1-\eps}}$.
Hence, the overall computation time is 
$\bigO{N^{1 + \max\{\min\{\eps, 1-\eps\}, \frac{1}{n-1}\}}}$, which, due to 
$N = \Theta(|\inst{D}|)$, is 
$\bigO{|\inst{D}|^{1 + \max\{\min\{\eps, 1-\eps\}, \frac{1}{n-1}\}}}$.
The argumentation for the space complexity follows a similar reasoning. 

\nop{The space complexity of each view $V_i$ can be computed following the same reasoning an in Proposition~\ref{prop:space_complexity}. The maximum size of $V_i$ on $\dbeps$ is determined by the join size for any pair of relation parts. For instance, 
\begin{align*}
|R_{i+1}^h \ztimes R_{i+2}^l| \leq \min\{ |R_{i+1}^h| \ztimes \frac{3}{2}N^{\eps}, |R_{i+2}^l| \ztimes 2N^{1-\eps} \} < \min\{ N \ztimes \frac{3}{2}N^{\eps}, N \ztimes 2N^{1-\eps} \}.
\end{align*}
The size of the join result is thus $\bigO{N^{1+\min\{\eps, 1-\eps\}}}$. From $|\db|=\Theta(N)$ follows that $V_{i}$ on $\dbeps$ takes $\bigO{|\db|^{1+\min\{\eps,1-\eps\}}}$ space. Overall, the state $\state$ of $\db$ takes $\bigO{|\db|^{1+\min\{\eps,1-\eps\}}}$ space.
}

\paragraph{Maintaining the Result of $Q$ Under Single-Tuple Updates.}
Given an \ivme state 
$\state$
of database $\inst{D}$, we  
 analyze the computation of $\delta Q$ under a single-tuple update
  $\delta R_1^{r}= \{(\deltaA_1, \ldots , \deltaA_{n-1}) \mapsto \p\}$ 
  to relation $R_1$.
  The strategies for updates to the other relations are analogous. 
The update affects either the heavy or the light part of $R$.
It fixes the values of all variables but variable $A_n$ in the deltas of the skew-aware views:
\begin{align*}
\delta Q_{r u_2 \ldots u_n}() =  
\delta R_1^{r}(\deltaA_1,\deltaA_2,\ldots,\deltaA_{n-1}) \ztimes 
\sum_{a_n} R_2^{u_2}(\deltaA_2,\deltaA_3,\ldots,a_n) \ztimes \ldots
\ztimes R_n^{u_n}(a_n, \deltaA_1,\ldots, \deltaA_{n-2})
\end{align*} 
where $u_2,\ldots, u_n \in \{h,l\}$.
We distinguish four cases when computing a delta $\delta Q_{r u_2 \ldots u_n}$:

\begin{enumerate}
    \item Case where $u_2 = \ldots  = u_n = h$.\;
    The number of distinct $A_n$-values in $R_n^h$ is upper-bounded by $\frac{N}{\frac{1}{2}N^{\eps}} = 2N^{1-\eps}$. 
    Thus, computing $\delta Q_{r h\ldots h}$ takes $\bigO{N^{1-\eps}}$ time.

    \item Case where $u_2 = \ldots = u_{n-1} = h$ and $u_n = l$.\;
    To compute this delta, we use the following auxiliary view:
    \begin{align*}
        \hspace{-0.55cm}V_1(a_1, \ldots, a_{n-1}) = 
        \sum_{a_n} R_2^{h}(a_2,a_3,\ldots,a_n) \ztimes \ldots \ztimes R_n^{l}(a_n, a_1,\ldots,a_{n-2}).
    \end{align*}
    Then, computing $\delta Q_{rhh \ldots l}() = \delta R_1^{r}(\deltaA_1,\deltaA_2,\ldots,\deltaA_{n-1}) \ztimes  V_1(\deltaA_1, \ldots, \deltaA_{n-1})$ takes constant time. 

    \item Case where $\exists j\in [2,n) : u_j = l$ and $u_n = h$. 
    Computing the delta requires either iterating over at most $\frac{3}{2}N^{\eps}$ $A_n$-values in $R_j^l$ for the given $(\deltaA_1, \ldots, \deltaA_{n-1})$ or iterating over at most $2 N^{1-\eps}$ distinct $A_n$-values in $R_n^h$. 
    The time complexity of computing such deltas is $\bigO{\min\{\frac{3}{2}N^{\eps}, 2N^{1-\eps}\}} = \bigO{N^{\min\{\eps, 1-\eps\}}}$.

    \item Case where $\exists j\in [2,n) : u_j = l$ and $u_n = l$. 
    Computing such deltas requires iterating over at most $\frac{3}{2}N^{\eps}$ $A_n$-values in $R_j^l$ for the given $(\deltaA_1, \ldots, \deltaA_{n-1})$ and looking up in the other relations for each $(\deltaA_1, \ldots,\deltaA_{n-1}, a_{n})$. The time complexity of computing these deltas is $\bigO{N^{\eps}}$.    
\end{enumerate}

The overall time to compute the results of delta skew-aware views is 
$\bigO{N^{\max\{\eps, 1-\eps\}}} = \bigO{|\inst{D}|^{\max\{\eps, 1-\eps\}}}$.

\paragraph*{Maintaining Auxiliary Views Under Single-Tuple Updates.}
We analyze the time to maintain the views $V_2, \ldots, V_n$ under an update 
$\delta R_1^{r}= \{(\deltaA_1, \ldots , \deltaA_{n-1}) \mapsto \p\}$ to the heavy or light part of $R_1$. The analysis for updates to the other relations is similar.

We distinguish two cases when computing the deltas of these views under the 
update $\delta R_1^{r}$:

\begin{enumerate}
    \item Case where $r = h$. 
    The only view using $R_1^h$ is $V_{n}$. Computing $\delta{V_n}$ means iterating over at most $\frac{3}{2}N^{\eps}$ $A_n$-values of any of the relations 
    $\{ R_i^l \}_{i \in \{3,n\}}$, for the given $(\deltaA_1, \ldots, \deltaA_{n-1})$ and looking up in the remaining relations for each $(\deltaA_1, \ldots, \deltaA_{n-1}, a_{n})$. This delta computation requires  $\bigO{N^{\eps}}$ time.

    \item Case where $r = l$.
    The update affects the views $V_2, \ldots, V_{n-1}$, which reference 
    $R_1^l$. Let $V_i$ be a view affected by the update. 
    In case the degree of $Q$ is at least $n = 4$, we can 
    compute $\delta V_i$ by iterating over at most 
    $\frac{3}{2}N^{\eps}$ $A_n$-values $a_n$ of any of the relations $R_j^l$ with 
    $j \in \{2,n-1\} \setminus \{i\}$
    and looking up each tuple 
    $(\deltaA_1, \ldots, \deltaA_{n-1}, a_{n})$ in the remaining relations. 
    This delta computation needs $\bigO{N^{\eps}}$ time. 
    In case  the degree of $Q$ is $3$, the materilized views as well as their 
    delta computation is the same as for the  
    triangle count query. 
    Hence, in this case, delta computation needs 
    $\bigO{|\inst{D}|^{\max\{\eps, 1-\eps\}}}$ time.  
\end{enumerate}

It follows that the preprocessing time as well as the  
space complexity are $\bigO{|\inst{D}|^{1 + \min\{\eps, 1-\eps\}}}$.
The time to process a single-tuple update is
$\bigO{|\inst{D}|^{\max\{\eps, 1-\eps\}}}$.
Since the result of $Q$ is materialized, the answer time is
constant.  Rebalancing procedures  and amortization of
rebalancing costs is similar to the triangle count case   
in Section~\ref{sec:rebalancing}. Thus, the overall 
amortized update time is $\bigO{|\inst{D}|^{\max\{\eps, 1-\eps\}}}$. 
\end{proof}

\subsection{Worst-Case Optimality of \ivme for Loomis-Whitney Queries}\label{sec:lowerbound_loomis_whitney}
We show that the amortized update time stated in 
Theorem \ref{theo:main_result_LW} is worst-case optimal, conditioned 
on the \OMv conjecture.
The proof is a slight extension of the proof of 
Proposition  \ref{prop:lower_bound_triangle_count}.
in Appendix~\ref{sec:lowerbound}.
\nop{
To prove this, we reduce the \OuMv problem to maintaining 
Loomis-Whitney count queries. 
The hardness of the \OuMv problem relies on the hardness 
of the \OMv
The proof is a simple extension  
   
The lower bound proof is a slight extension of the lower bound 
proof for the triangle count case.

We state a conditional lower bound for incrementally maintaining a Loomis-Whitney count query of any degree $n\geq3$ under updates to the underlying database. The lower bound is conditioned on the \OMv conjecture, and the proof is similar to that from Appendix~\ref{sec:lowerbound}.
}

\begin{proof}[Proof of Proposition \ref{prop:lower_bound_LW}]
Let $Q$  be Loomis-Whitney count query of some degree $d \geq 3$
and assume 
that there is an algorithm $\mathcal{A}$ with amortized 
update time $\bigO{|\inst{D}|^{\frac{1}{2}-\gamma}}$ and answer time $\bigO{|\inst{D}|^{1-\gamma}}$  that maintains the result of $Q$ under single-tuple updates.
We show that $\mathcal{A}$ can be used to design an algorithm $\mathcal{B}$ that solves the $\OuMv$ problem in subcubic time, which contradicts the \OuMv conjecture (Conjecture \ref{conj:OuMv}). Since the latter 
conjecture relies on the \OMv conjecture 
(Conjecture \ref{conj:omv}), this also contracts the
  \OMv conjecture.

Let $(\vecnormal{M}, (\vecnormal{u}_1,\vecnormal{v}_1), \ldots ,(\vecnormal{u}_n,\vecnormal{v}_n))$ be an input to the \OuMv problem.
In the proof of Proposition \ref{prop:lower_bound_triangle_count}
in Appendix~\ref{sec:lowerbound},
we used the relation $S$ to encode matrix $\vecnormal{M}$ and 
the relations $R$
and $T$ to encode the vectors  
$\vecnormal{u}_i$ and $\vecnormal{v}_i$, respectively. Here, the matrix 
$\vecnormal{M}$ is encoded by the relations $R_3, \ldots R_d$
and the vectors $\vecnormal{u}_i$ and $\vecnormal{v}_i$
are encoded by the relations $R_1$ and $R_2$, respectively. 
We denote the entry of $\vecnormal{M}$ in row $i$ and column $j$ by $\vecnormal{M}(i,j)$ and the $i$-th component of $\vecnormal{v}$ by $\vecnormal{v}(i)$. 
Let $a$ be some constant.
Algorithm $\mathcal{B}$ starts from an empty database $\db=\{R_1, \ldots, R_d\}$.
Then, it executes at most $n^2$ updates to each of the relations $R_3,\ldots,R_d$ such that
$R_k = \{\, (A_k:a,\ldots, A_{d-1}:a, A_d:i, A_1:j, A_2:a, A_2:a, \ldots)
\mapsto \vecnormal{M}(i,j) \,\mid\, i,j \in \{1,\ldots, n\} \,\}$, for $3\leq k\leq d$. 
That is, for each $i,j \in \{1,\ldots, n\}$ and each relation $R_k$, 
the algorithm inserts a tuple $\vecnormal{t}$ with  multiplicity $M(i,j)$ such that 
$\vecnormal{t}[A_d]=i$, $\vecnormal{t}[A_1]=j$, and all other tuple values are fixed to $a$.
In each round $r \in \{1, \ldots , n\}$, it executes at most $2n$ updates to the relations $R_1$ and $R_2$ such that
$R_1 = \{\, (A_1:i,A_2:a,\ldots,A_{d-1}:a) \mapsto \vecnormal{v}_r(i) \,\mid\, i \in \{1,\ldots, n\} \,\}$ and
$R_2 = \{\, (A_2:a,\ldots,A_{d-1}:a,A_d:i) \mapsto \vecnormal{u}_r(i) \,\mid\, i \in \{1,\ldots, n\} \,\}$.
By construction, $\vecnormal{u}_r^{\text{T}}\vecnormal{M}\vecnormal{v}_r = 1$ if and only if there exist 
$i,j\in\{1,\ldots, n\}$ such that $\vecnormal{u}_r(i) = 1$, $\vecnormal{M}(i,j)=1$, and $\vecnormal{v}_r(j) = 1$, which is equivalent to 
$R_1(A_1:j,A_2:a,\ldots,A_{d-1}:a) \ztimes 
R_2(A_2:a,\ldots,A_{d-1}:a,A_d:i) \ztimes 
R_3(A_3:a,\ldots, A_{d-1}:a, A_{d}:i, A_{1}:j) \ztimes \ldots \ztimes
R_d(A_d:i, A_1:j, A_2:a, \ldots, A_{d-2}:a)  = 1$. Thus, the algorithm outputs $1$ at the end of round $r$ if and only if the result of $Q$ after round $r$ is nonzero. 

The following cost analysis is similar to the analysis in the 
proof of Proposition \ref{prop:lower_bound_triangle_count}.
The construction of the relations $R_3,\ldots,R_d$ from $\vecnormal{M}$ requires at most $n^2$ updates per relation. 
Given that the amortized time for each update is $\bigO{|\db|^{\frac{1}{2}-\gamma}}$
and the database size $|\db|$ is $\bigO{n^2}$, this phase takes overall $\bigO{n^2 \ztimes n^{1-2\gamma}} = \bigO{n^{3-2\gamma}}$ time. 
In each round, the algorithm performs at most $2n$ updates and outputs the result in $\bigO{|\inst{D}|^{1-\gamma}}$ time. The overall execution time is $\bigO{n^{2-2\gamma}}$ per round and $\bigO{n^{3-2 \gamma}}$ for $n$ rounds. Thus, algorithm $\mathcal{B}$ needs $\bigO{n^{3-2 \gamma}}$ time to solve the \OuMv problem, which contradicts the \OuMv conjecture  and, consequently, the \OMv conjecture.
\end{proof}
\nop{
Theorem~\ref{theo:main_result_LW} and Proposition~\ref{prop:lower_bound_LW} imply that for $\eps = \frac{1}{2}$, \ivme incrementally maintains a Loomis-Whitney count query under single-tuple updates to a database $\db$ with optimal amortized update time $\bigO{|\db|^{\frac{1}{2}}}$ and constant answer time, unless the \OMv conjecture fails.
}

\section{4-Path Count Query}
\label{sec:4_path}
We consider the problem of incrementally maintaining the result of the following
4-path count query 
\begin{align*}
Q() = \sum\limits_{a,b,c} R(a)\ztimes S(a,b) \ztimes T(b,c) \ztimes U(c)
\end{align*}
under single-tuple updates to the relations $R$, $S$, $T$, and $U$
with schemas $(A)$, $(A,B)$, $(B,C)$, and $(C)$, respectively.
\ivme maintains the 4-path count with the same 
complexities as the triangle count.  
 
\begin{theorem}\label{theo:main_result_path4}
Given a database $\db$ and $\eps \in [0,1]$, 
\ivme  incrementally maintains the 4-path count
 under single-tuple updates to $\db$ with 
$\bigO{|\inst{D}|^{1 + \min\{\eps,1-\eps\}}}$ preprocessing time, $\bigO{|\inst{D}|^{\max\{\eps,1-\eps\}}}$ amortized update time, 
constant answer time, and 
$\bigO{|\inst{D}|^{1 + \min\{\eps, 1-\eps\}}}$ space.
\end{theorem}

The update time in Theorem \ref{theo:main_result_path4}
is worst-case optimal for $\eps = 0.5$, conditioned 
on the \OMv conjecture. This is implied by the following proposition. 

\begin{proposition}\label{prop:lower_bound_path_4}
For any $\gamma > 0$ and database $\db$,
there is no algorithm that incrementally maintains 
the 4-path count under single-tuple updates to $\db$ with arbitrary preprocessing time, $\bigO{|\db|^{\frac{1}{2} - \gamma}}$ amortized update time, and $\bigO{|\db|^{1 - \gamma}}$ answer time, unless the \OMv conjecture fails.
\end{proposition}

The \ivme maintenance strategy introduced in the following sections 
is easily extendible to queries that result from $Q$ by
extending each relation by unboundedly many non-join variables. 
Such extended relations can be replaced by views that 
aggregate  away the non-join variables and can be updated 
in constant time.

In the sequel, we first prove Theorem \ref{theo:main_result_path4}
and then Proposition \ref{prop:lower_bound_path_4}.

\subsection{Adaptive Maintenance of the 4-Path Count}
We present an \ivme variant that partitions relations $S$ and $T$ on both 
variables (cf.~Definition \ref{def:loose_relation_partition_sets}) 
and leaves relations $R$ and $U$ unpartitioned.  
Hence, the partitions of $S$ and $T$ are of the form  
$\{S_{s_A,s_B}\}_{s_A,s_B \in \{h,l\}}$
and 
$\{T_{t_B,t_C}\}_{t_B,t_C \in \{h,l\}}$, respectively.
For instance, the relation part $S_{lh}$ is light on
$A$ and heavy on $B$.  
The \ivme variant 
decomposes the query into skew-aware views of the form 
\begin{align*}
Q_{s_As_Bt_Bt_C}() = \sum\limits_{a,b,c} R(a)\ztimes S_{s_A,s_B}(a,b) \ztimes T_{t_B,t_C}(b,c) \ztimes U(c)
\end{align*}
with $s_A,s_B,t_B,t_C \in \{h,l\}$.
The 4-path count can then be expressed 
as the sum of these skew-aware views:
\begin{align*}
Q() = \sum\limits_{s_A,s_B,t_B,t_C \in \{h,l\}} Q_{s_As_Bt_Bt_C}()
\end{align*}

Figure \ref{fig:view_definitions_path_4} lists all 
auxiliary views materialized to facilitate the computation of delta 
skew-aware views under updates. 
The \ivme variant uses some views containing indicator projections
of relations, which we introduce next. An indicator projection
$\displaystyle\exists_X K$ of a relation $K$ on one of 
its variables $X$ projects the tuples in $K$ onto $X$ and maps each 
value in the projection  to multiplicity $1$.

\begin{definition}[Indicator Projection]
Given a relation $K$ with schema $\inst{X}$ and a variable 
$X$ in $\inst{X}$,  an indicator projection $\displaystyle\exists_X K$
of $K$ on $X$ is defined as
$$\displaystyle\exists_X K(x) = 
\begin{cases}
1 & \text{ if } \exists \inst{x} \in \Dom(\inst{X}), \text{ with } K(\inst{x}) \neq 0
\text{ and } \inst{x}[X] = x\\
0 & \text{ otherwise }\\
\end{cases}
$$
\end{definition}

\begin{figure}[t] 
 \begin{center}
    \renewcommand{\arraystretch}{1.2}  
    \begin{tabular}{@{\hskip 0.05in}l@{\hskip 0.4in}l@{\hskip 0.05in}}
      \toprule
      Materialized View Definition & Space Complexity \\    
      \midrule
      $Q() = \sum\limits_{s_A, s_B,t_B,t_C \in \{h,l\}} \,
      \sum\limits_{a,b,c} R(a) \ztimes S_{s_As_B}(a,b) \ztimes T_{t_Bt_C}(b,c) \ztimes U(c)$ & 
      $\bigO{1}$ \\


      $V_{RS_{ll}}(b) = \sum_{a} R(a) \ztimes S_{ll}(a,b)$ & 
      $\bigO{|\inst{D}|}$ \\
      $V_{RS_{s_Ah}}(b) = \sum_{a} R(a) \ztimes S_{s_Ah}(a,b)$ & 
      $\bigO{|\inst{D}|^{1-\eps}}$ \\


      $V_{S_{ll}T_{lh}}(a,c) = \sum_{b} S_{ll}(a,b) \ztimes T_{lh}(b,c)$ & 
      $\bigO{|\inst{D}|^{1+\min{\{\,\eps, 1-\eps \,\}}}}$ \\

      $V_{S_{hl}T_{ll}}(a,c) = \sum_{b} S_{hl}(a,b) \ztimes T_{ll}(b,c)$ & 
      $\bigO{|\inst{D}|^{1+\min{\{\,\eps, 1-\eps \,\}}}}$ \\
      $V_{S_{hl}T_{t_Bh}}(a,c) = \sum_{b} S_{hl}(a,b) \ztimes T_{t_Bh}(b,c)$ & 
      $\bigO{|\inst{D}|^{\min{\{\,1 + \eps, 2-2\eps \,\}}}}$ \\


      $V_{S_{hh}T_{lh}}(a,c) = \sum_{b} S_{hh}(a,b) \ztimes T_{lh}(b,c)$ & 
      $\bigO{|\inst{D}|^{\min{\{\,1 + \eps, 2-2\eps \,\}}}}$ \\

      $V_{T_{ll}U}(b) = \sum_{c} T_{ll}(b,c) \ztimes U(c)$ & 
      $\bigO{|\inst{D}|}$ \\

      $V_{T_{ht_C}U}(b) = \sum_{c} T_{ht_C}(b,c) \ztimes U(c)$ & 
      $\bigO{|\inst{D}|^{1-\eps}}$ \\


      $V_{T_{h}}(b) = \displaystyle\exists_B T_{hl}(b)$ & 
      $\bigO{|\inst{D}|^{1-\eps}}$ \\
      $V_{RS_{hl}T_{hl}}(b) = \sum_{a} R(a) \ztimes S_{hl}(a,b) \ztimes V_{T_{h}}(b)$ &
      $\bigO{|\inst{D}|^{1-\eps}}$ \\
      $V_{RS_{ll}T_{lh}}(c) = \sum_{a} R(a) \ztimes V_{S_{ll}T_{lh}}(a,c)$ &
      $\bigO{|\inst{D}|^{1-\eps}}$ \\

      $V_{S_{h}}(b) =  \displaystyle\exists_B S_{lh}(b)$ & 
      $\bigO{|\inst{D}|^{1-\eps}}$ \\
      $V_{S_{lh}T_{lh}U}(b) = \sum_{c} V_{S_{h}}(b) \ztimes T_{lh}(b,c) \ztimes U(c)$ & 
      $\bigO{|\inst{D}|^{1-\eps}}$ \\
      
      $V_{S_{hl}T_{ll}U}(a) = \sum_{c} V_{S_{hl}T_{ll}}(a,c) \ztimes U(c)$ & 
      $\bigO{|\inst{D}|^{1-\eps}}$ \\

      \bottomrule    
    \end{tabular}
  \end{center}
  \caption{The definitions and space complexities of the materialized views in 
  $\inst{V}$  as part of a  state of a database $\inst{D} = \{R,S,T,U\}$
  partitioned for a fixed $\eps \in [0,1]$.} 
  \label{fig:view_definitions_path_4}
\end{figure}
The views $V_{T_{h}}(b)$ and $V_{S_{h}}(b)$
in Figure \ref{fig:view_definitions_path_4}
are defined as indicator projections of the relations 
$T_{hl}$ and $S_{lh}$ on $B$. The purpose
 of $V_{T_{hl}}$ in 
$V_{RS_{hl}T_{hl}}$
is to put a sublinear bound on the
number of $B$-values $b$ with  
$V_{RS_{hl}T_{hl}}(b) \neq 0$.
The role of $V_{S_{lh}}$
in
$V_{S_{lh}T_{lh}U}$
is analogous.

Given a fixed $\eps \in [0,1]$,
an \ivme state 
$(\eps,N, \inst{P}, \inst{V})$
of a database $\inst{D} = \{R,S,T,U\}$ is defined 
as in Definition \ref{fig:view_definitions}
with the only differences that  
\begin{itemize}
\item $\dbeps=  
\{R,
\{S_{s_As_B}\}_{s_A,s_B \in \{h,l\}},
\{T_{t_Bt_C}\}_{t_B,t_C \in \{h,l\}},
U\}$, where 
$\{S_{s_As_B}\}_{s_A,s_B \in \{h,l\}}$
and 
$\{T_{t_Bt_C}\}_{t_B,t_C \in \{h,l\}}$ 
are the partitions of $S$ and $T$, respectively, 
with threshold $N^{\eps}$, and 
\item $\inst{V}$ consists of the views given in
Figure \ref{fig:view_definitions_path_4}. 
\end{itemize}

\subsection{Preprocessing Time for the 4-Path Count}
In the preprocessing stage, \ivme computes the initial
state 
$\state = (\eps, N, \inst{P}, \inst{V})$ 
of a given database $\inst{D} = \{R,S,T,U\}$ such that 
$N =2 |\inst{D}|+1$ and
$\inst{P}$ consists of the relations $R$ and $U$ and the strict partitions of 
$S$ and $T$ with threshold $N^{\eps}$.

\begin{proposition}\label{prop:preprocessing_step_4_path}
Given a database $\db$ and $\eps\in[0,1]$, 
constructing the initial \ivme state
of $\db$ to support the maintenance of the 4-path count takes
$\bigO{|\inst{D}|^{1 + \min{\{\eps, 1-\eps\}}}}$ time.
\end{proposition}

\begin{proof}
Setting the value of $N$ is a constant-time operation. 
Strictly partitioning the relations can be 
accomplished in linear time.  
The FAQ-width of $Q$ is one, hence, it can be computed in 
linear time \cite{FAQ:PODS:2016}.

We can compute $V_{RS_{ll}}$ and $V_{RS_{s_Ah}}$ by iterating over 
the tuples $(a,b)$ in the $S$-part and, for each such tuple, doing a 
lookup of $a$ in $R$. 
This takes  $\bigO{|\inst{D}|}$ time.
The analysis for $V_{T_{ll}U}$ and $V_{T_{ht_C}U}$ is analogous. 

The views $V_{S_{ll}T_{lh}}$, $V_{S_{hl}T_{ll}}$, $V_{S_{hl}T_{t_Bh}}$ and $V_{S_{hh}T_{lh}}$ can be computed in two ways. 
One option is to iterate over the tuples $(a,b)$ in the $S$-part and, for each such tuple, to go over all $C$-values paired with $b$ in the $T$-part. Alternatively, 
we can iterate over the tuples $(b,c)$ in the $T$-part and, for each 
such tuple, go over all $A$-values paired with $b$ in the $S$-part. 
Computing the views $V_{S_{ll}T_{lh}}$ and $V_{S_{hl}T_{ll}}$ takes  $\bigO{|\inst{D}|^{1 + \min{\{\eps, 1-\eps\}}}}$ time: 
$T_{lh}$ is light on $B$ and heavy on $C$, so there are $\bigO{|\inst{D}|^{1 + \min{\{\eps, 1-\eps\}}}}$ $C$-values paired with each $(a,b)$ from the $S$-part; likewise,
$S_{hl}$ is light on $B$ and heavy on $A$, so there are $\bigO{|\inst{D}|^{1 + \min{\{\eps, 1-\eps\}}}}$ $A$-values paired with each $(b,c)$ from the $T$-part.
We can compute the views $V_{S_{hl}T_{t_Bh}}$ and $V_{S_{hh}T_{lh}}$ in time $\bigO{|\inst{D}|^{1 + \eps}}$: 
$S_{hl}$ is light on $B$, so there are $\bigO{|\inst{D}|^{\eps}}$ $A$-values paired with each tuple $(b,c)$ from the $T$-part; likewise, $T_{lh}$ is light on $B$, so there are $\bigO{|\inst{D}|^{\eps}}$ $C$-values paired with each tuple $(a,b)$ from $S$. 

We next analyze  the computation times for $V_{T_h}$, $V_{RS_{hl}T_{hl}}$ and $V_{RS_{ll}T_{lh}}$. The analysis for $V_{S_h}$, $V_{S_{lh}T_{lh}U}$ and $V_{S_{hl}T_{ll}U}$ is analogous.
The view $V_{T_h}$ can be computed in time $\bigO{|\inst{D}|}$ by performing 
a single pass through the $T$-part.  
The view $V_{RS_{hl}T_{hl}}$ can be computed in time $\bigO{|\inst{D}|}$ by iterating 
over the tuples $(b,c)$ in $T_{hl}$ and, for each such tuple,  
looking up $b$ in $V_{T_h}$ and 
$a$ in $R$. 
Since the view $V_{S_{ll}T_{lh}}$ is materialized, the view $V_{RS_{ll}T_{lh}}$ can be computed in time $\bigO{|\inst{D}|}$ by iterating over the tuples $(a,c)$ in $V_{S_{ll}T_{lh}}$ and doing lookups for $a$ in $R$.

Overall, the initial state of $\inst{D}$ can be computed in time $\bigO{|\inst{D}|^{1 + \min{\{\eps, 1-\eps\}}}}$. 
\end{proof}


\begin{figure}[htbp]
  \begin{center}
  \renewcommand{\arraystretch}{1.5}
  \setcounter{magicrownumbers}{0}
  \begin{tabular}{ll@{\hskip 0.25in}l@{\hspace{0.6cm}}c}
  \toprule
  \multicolumn{2}{l}{\textsc{ApplyUpdate4Path}($\delta R,\state$)}& & Time \\
  \cmidrule{1-2} \cmidrule{4-4}
  \rownumber & \LET $\delta R = \{(\deltaA) \mapsto \p\}$ \\
  \rownumber & \LET $\state = (\eps, N, 
  \{R, U\} \cup \{S_{s_As_B}\}_{s_A,s_B \in \{h,l\}} \cup \{T_{t_At_B}\}_{t_A,t_B \in \{h,l\}}, $ \\
  & \TAB \TAB $\{Q,V_{RS_{ll}},V_{RS_{lh}},V_{RS_{hh}}, V_{S_{ll}T_{lh}}, V_{S_{hl}T_{ll}},V_{S_{hl}T_{lh}},V_{S_{hl}T_{hh}},V_{S_{hh}T_{lh}},\})$ \\
  & \TAB \TAB $V_{T_{ll}U}, V_{T_{hl}U},V_{T_{hh}U},V_{T_h},V_{RS_{hl}T_{hl}},V_{RS_{ll}T_{lh}},V_{S_h},V_{S_{lh}T_{lh}U},V_{S_{hl}T_{ll}U}\})$ \\
  \rownumber & $\delta Q_{llll}() = \delta{R(\deltaA)} \ztimes \textstyle\sum_b S_{ll}(\deltaA,b) \ztimes V_{T_{ll}U}(b)$
  && $\bigO{|\inst{D}|^{\eps}}$ \\

\rownumber & $\delta Q_{lllh}() = \delta{R(\deltaA)} \ztimes \textstyle\sum_c V_{S_{ll}T_{lh}}(\deltaA,c) \ztimes U(c)$
&& $\bigO{|\inst{D}|^{1-\eps}}$ \\

  \rownumber & $\delta Q_{llht_C}() = \delta{R(\deltaA)} \ztimes \textstyle\sum_b S_{ll}(\deltaA,b) \ztimes V_{T_{ht_C}U}(b)$
  && $\bigO{|\inst{D}|^{\min{\{\eps, 1-\eps\}}}}$ \\


  \rownumber & $\delta Q_{lhll}() = \delta{R(\deltaA)} \ztimes \textstyle\sum_b S_{lh}(\deltaA,b) \ztimes V_{T_{ll}U}(b)$
  && $\bigO{|\inst{D}|^{\min{\{\eps, 1-\eps\}}}}$ \\ 

  \rownumber & $\delta Q_{lhlh}() = \delta{R(\deltaA)} \ztimes \textstyle\sum_b S_{lh}(\deltaA,b) \ztimes V_{S_{lh}T_{lh}U}(b)$
  && $\bigO{|\inst{D}|^{\min{\{\eps, 1-\eps\}}}}$ \\ 

  \rownumber & $\delta Q_{lhht_C}() = 
  \delta{R(\deltaA)} \ztimes \textstyle\sum_b S_{lh}(\deltaA, b) \ztimes 
  V_{T_{ht_C}U}(b)$
  && $\bigO{|\inst{D}|^{\min{\{\eps, 1-\eps\}}}}$ \\


  \rownumber & $\delta Q_{hlll}() = \delta{R(\deltaA)} \ztimes \textstyle V_{S_{hl}T_{ll}U}(\deltaA)$
  && $\bigO{1}$ \\

\rownumber & $\delta Q_{hlhl}() = \delta{R(\deltaA)} \ztimes \textstyle\sum_b S_{hl}(\deltaA,b) \ztimes V_{T_{hl}U}(b)$
&& $\bigO{|\inst{D}|^{1-\eps}}$ \\ 

  \rownumber & $\delta Q_{hlt_Bh}() = \delta{R(\deltaA)} \ztimes \textstyle\sum_c V_{S_{hl}T_{t_Bh}}(\deltaA,c) \ztimes U(c)$
  && $\bigO{|\inst{D}|^{1-\eps}}$ \\


  \rownumber & $\delta Q_{hhll}() = \delta{R(\deltaA)} \ztimes \textstyle\sum_b S_{hh}(\deltaA,b) \ztimes V_{T_{ll}U}(b)$
  && $\bigO{|\inst{D}|^{1-\eps}}$ \\ 

  \rownumber & $\delta Q_{hhlh}() = \delta{R(\deltaA)} \ztimes \textstyle\sum_c V_{S_{hh}T_{lh}}(\deltaA, c) \ztimes U(c)$
  && $\bigO{|\inst{D}|^{1-\eps}}$ \\

  \rownumber & $\delta Q_{hhht_C}() = \delta{R(\deltaA)} \ztimes \textstyle\sum_b S_{hh}(\deltaA, b) \ztimes V_{T_{ht_C}U}(b)$
  && $\bigO{|\inst{D}|^{1-\eps}}$ \\

  \rownumber & 
  $Q() = Q() + \sum\limits_{s_A,s_B,t_A,t_B \in \{h,l\}} \,\delta Q_{s_As_Bt_At_B}()$ & &
  $\bigO{1}$ \\

  \rownumber & $V_{RS_{ll}}(b) = V_{RS_{ll}}(b) + \delta R(\deltaA) \ztimes S_{ll}(\deltaA,b)$ & & $\bigO{|\inst{D}|^{\eps}}$ \\
  \rownumber & $V_{RS_{lh}}(b) = V_{RS_{lh}}(b) + \delta R(\deltaA) \ztimes S_{lh}(\deltaA,b)$ & & $\bigO{|\inst{D}|^{\min\{1-\eps\}}}$ \\
  \rownumber & $V_{RS_{hh}}(b) = V_{RS_{hh}}(b) + \delta R(\deltaA) \ztimes S_{hh}(\deltaA,b)$ & & $\bigO{|\inst{D}|^{1-\eps}}$ \\

  \rownumber & $V_{RS_{hl}T_{hl}}(b) = V_{RS_{hl}T_{hl}}(b) + \delta R(\deltaA) \ztimes S_{hl}(\deltaA,b) \ztimes V_{T_{h}}(b)$ & & $\bigO{|\inst{D}|^{1-\eps}}$ \\
  \rownumber & $V_{RS_{ll}T_{lh}}(c) = V_{RS_{ll}T_{lh}}(c) + \delta R(\deltaA) \ztimes V_{S_{ll}T_{lh}}(\deltaA,c)$ & & $\bigO{|\inst{D}|^{1-\eps}}$ \\
  
  \rownumber & $R(\deltaA) = R(\deltaA) + \delta{R}(\deltaA)$ & &
  $\bigO{1}$ \\
  
  \rownumber & \RETURN 
  $\state$ & &
   \\
  \midrule
  \multicolumn{2}{r}{Total update time:} & &
  $\bigO{|\inst{D}|^{\max\{\eps, 1-\eps\}}}$ \\
  \bottomrule
  \end{tabular}
  \end{center}\vspace{-1em}
  \caption{
   (left) Maintaining the result of the 4-path count under a single-tuple update to $R$. 
   \textsc{ApplyUpdate4Path} takes as input an update
    $\delta R$ and 
    the current \ivme state $\state$ of a database $\inst{D}$ partitioned using $\eps\in[0,1]$.
  It returns a new state that results from applying $\delta R$ to $\state$. 
  $s_A$, $s_B$, $t_B$, and $t_C$ can be $l$ or $h$.
  Lines 5-16 compute the deltas of the affected skew-aware views. Line 17
  maintains $Q$.
  Lines 18-22 maintain the auxiliary views affected by this update. Line 23 maintains the affected $R$.
  (right) The time complexity of computing and applying deltas.  
 The maintenance procedures for updates to $U$ is analogous.
  }
  \label{fig:applyUpdate_path_4_delta_R}
\end{figure}


\begin{figure}[htbp]
  \begin{center}
  \renewcommand{\arraystretch}{1.4}
  \setcounter{magicrownumbers}{0}
  \begin{tabular}{ll@{\hskip 0.25in}l@{\hspace{0.6cm}}c}
  \toprule
  \multicolumn{2}{l}{\textsc{ApplyUpdate4Path}($\delta S_{s_As_B},\state$)}& & Time \\
  \cmidrule{1-2} \cmidrule{4-4}
  \rownumber & \LET $\delta S_{s_As_B} = \{(\deltaA, \deltaB) \mapsto \p\}$ \\
  \rownumber & \LET $\state = (\eps, N, 
  \{R, U\} \cup \{S_{s_As_B}\}_{s_A,s_B \in \{h,l\}} \cup \{T_{t_At_B}\}_{t_A,t_B \in \{h,l\}}, $ \\
  & \TAB \TAB $\{Q,V_{RS_{ll}},V_{RS_{lh}},V_{RS_{hh}}, V_{S_{ll}T_{lh}}, V_{S_{hl}T_{ll}},V_{S_{hl}T_{lh}},V_{S_{hl}T_{hh}},V_{S_{hh}T_{lh}},\})$ \\
  & \TAB \TAB $V_{T_{ll}U}, V_{T_{hl}U},V_{T_{hh}U},V_{T_h},V_{RS_{hl}T_{hl}},V_{RS_{ll}T_{lh}},V_{S_h},V_{S_{lh}T_{lh}U},V_{S_{hl}T_{ll}U}\})$ \\

  \rownumber & $\delta Q_{s_As_Bll}() = R(\deltaA) \ztimes \delta{S_{s_As_B}(\deltaA,\deltaB)} \ztimes \textstyle\sum_c T_{ll}(\deltaB,c) \ztimes U(c)$
  && $\bigO{|\inst{D}|^{\eps}}$ \\ 
  \rownumber & $\delta Q_{s_As_Blh}() = R(\deltaA) \ztimes \delta{S_{s_As_B}(\deltaA,\deltaB)} \ztimes \textstyle\sum_c T_{lh}(\deltaB,c) \ztimes U(c)$
  && $\bigO{|\inst{D}|^{\min{\{\eps, 1-\eps\}}}}$ \\ 
  \rownumber & $\delta Q_{s_As_Bhl}() = R(\deltaA) \ztimes \delta{S_{s_As_B}(\deltaA,\deltaB)} \ztimes V_{T_{hl}U}(\deltaB)$ 
  && $\bigO{1}$ \\
  \rownumber & $\delta Q_{s_As_Bhh}() = R(\deltaA) \ztimes \delta{S_{s_As_B}(\deltaA,\deltaB)} \ztimes \textstyle\sum_c T_{hh}(\deltaB,c) \ztimes U(c)$
  && $\bigO{|\inst{D}|^{1-\eps}}$ \\

  \rownumber & 
  $Q() = Q() + \sum\limits_{s_A,s_B,t_A,t_B \in \{h,l\}} \,\delta Q_{s_As_Bt_At_B}()$ & &
  $\bigO{1}$ \\

  \rownumber & \IF ($s_A$ is $h$ \AND $s_B$ is $h$) &\\  
  \rownumber & \TAB
  $V_{RS_{hh}}(\deltaB) = V_{RS_{hh}}(\deltaB) + R(\deltaA) \ztimes \delta{S_{hh}(\deltaA,\deltaB)}$ 
  & &
  $\bigO{1}$ \\
  \rownumber & \TAB  
  $V_{S_{hh}T_{lh}}(\deltaA,c) = V_{S_{hh}T_{lh}}(\deltaA,c) + \delta S_{hh}(\deltaA, \deltaB) \ztimes T_{lh}(\deltaB, c)$
  & &
  $\bigO{|\inst{D}|^{\min\{\eps, 1-\eps\}}}$ \\
  

  \rownumber & \ELSE \IF ($s_A$ is $l$ \AND $s_B$ is $h$) &\\  
  \rownumber & \TAB
  $V_{RS_{lh}}(\deltaB) = V_{RS_{lh}}(\deltaB) + R(\deltaA) \ztimes \delta{S_{lh}(\deltaA,\deltaB)}$ 
  & &
  $\bigO{1}$ \\
  \rownumber & \TAB
  $V_{S_{h}}(\deltaB) = V_{S_{h}}(\deltaB) + \delta{S_{lh}(\deltaA,\deltaB)}$  
  & &
  $\bigO{1}$ \\

  \rownumber & \ELSE \IF ($s_A$ is $h$ \AND $s_B$ is $l$) &\\  
  \rownumber & \TAB
  $V_{S_{hl}T_{ll}}(\deltaA,c) = V_{S_{hl}T_{ll}}(\deltaA,c) + \delta S_{hl}(\deltaA, \deltaB) \ztimes T_{ll}(\deltaB, c)$
  & &
  $\bigO{|\inst{D}|^{\eps}}$ \\
  \rownumber & \TAB
  $V_{S_{hl}T_{lh}}(\deltaA,c) = V_{S_{hl}T_{lh}}(\deltaA,c) + \delta S_{hl}(\deltaA, \deltaB) \ztimes T_{lh}(\deltaB, c)$
  & &
  $\bigO{|\inst{D}|^{\min\{\eps, 1-\eps\}}}$ \\
  \rownumber & \TAB
  $V_{S_{hl}T_{hh}}(\deltaA,c) = V_{S_{hl}T_{hh}}(\deltaA,c) + \delta S_{hl}(\deltaA, \deltaB) \ztimes T_{hh}(\deltaB, c)$  
  & &
  $\bigO{|\inst{D}|^{1-\eps}}$ \\
  \rownumber & \TAB
  $V_{S_{hl}T_{ll}U}(\deltaA) = V_{S_{hl}T_{ll}U}(\deltaA) + \delta S_{hl}(\deltaA, \deltaB) \ztimes \textstyle\sum_c T_{ll}(\deltaB, c) \ztimes U(c)$  
  & &
  $\bigO{|\inst{D}|^{\eps}}$ \\

  \rownumber & \ELSE &\\ 
  \rownumber & \TAB
  $V_{RS_{ll}}(\deltaB) = V_{RS_{ll}}(\deltaB) + R(\deltaA) \ztimes \delta{S_{ll}(\deltaA,\deltaB)}$ 
  & &
  $\bigO{1}$ \\
  \rownumber & \TAB
  $V_{S_{ll}T_{lh}}(\deltaA,c) = V_{S_{ll}T_{lh}}(\deltaA,c) + S_{ll}(\deltaA, \deltaB) \ztimes T_{lh}(\deltaB, c)$  
  & &
  $\bigO{|\inst{D}|^{\min\{\eps, 1-\eps\}}}$ \\

  \rownumber & \TAB
  $V_{RS_{ll}T_{lh}}(c) = V_{RS_{ll}T_{lh}}(c) + R(\deltaA) \ztimes \delta{S_{ll}(\deltaA,\deltaB)} \ztimes T_{lh}(\deltaB, c)$  
  & &
  $\bigO{|\inst{D}|^{\min\{\eps, 1-\eps\}}}$ \\

  \rownumber & $S_{s_As_B}(\deltaA,\deltaB) = S_{s_As_B}(\deltaA,\deltaB) + \delta{S_{s_As_B}}(\deltaA,\deltaB)$ & &
  $\bigO{1}$ \\
  
  \rownumber & \RETURN 
  $\state$ & &
   \\
  \midrule
  \multicolumn{2}{r}{Total update time:} & &
  $\bigO{|\inst{D}|^{\max\{\eps, 1-\eps\}}}$ \\
  \bottomrule
  \end{tabular}
  \end{center}\vspace{-1em}
  \caption{
     (left) Maintaining the result of the 4-path count under a single-tuple update to relation $S$.
   \textsc{ApplyUpdate4Path} takes as input an update
    $\delta S_{s_As_B}$ with $s_A,s_B \in \{h,l\}$ to a part $S$ and 
    the current \ivme state $\state$ of a database $\inst{D}$ partitioned using $\eps\in[0,1]$.
  It returns a new state that results from applying $\delta S_{s_As_B}$ to $\state$.
      $s_A$, $s_B$, $t_B$, and $t_C$ can be $l$ or $h$. 
  Lines 5-8 compute the deltas of the affected skew-aware views. Line 9
  maintains $Q$.
  Lines 10-25 maintain the auxiliary views affected by the update. Line 26 maintains the affected $S$.
  (right) The time complexity of computing and applying deltas.  
  The maintenance procedures for updates to $T$ is analogous.
  }
  \label{fig:applyUpdate_path_4_delta_S}
\end{figure}

\subsection{Space Complexity of Maintaining the 4-Path Count}
\begin{proposition}\label{prop:space_complexity_4_path}
Given a database $\db$ and $\eps\in[0,1]$, the \ivme state 
constructed from $\db$ to support the maintenance of the 4-path count
takes $\bigO{|\db|^{1 +\min\{\eps,1-\eps\}}}$ space.
\end{proposition}  

\begin{proof}
We analyze the space complexity of a state 
$\state = (\eps,N, \inst{P}, \inst{V})$ of a database $\inst{D}$. 
The space occupied by $\eps$ and $N$ is constant. The size 
of the relation partitions in $\inst{P}$ is linear.  

Figure~\ref{fig:view_definitions_path_4} gives the 
sizes of the views in $\inst{V}$. 
The size of $Q$ is constant, since it consists of an empty tuple mapped to the result of the query. 
The views $V_{S_{ll}T_{lh}}$, $V_{S_{hl}T_{ll}}$, $V_{S_{hl}T_{t_Bh}}$ and $V_{S_{hh}T_{lh}}$ admit two size bounds. 
The first bound is the product of the size of the $S$-part and the maximum number of tuples in the $T$-part that match with a single tuple in the $S$-part, i.e., 
$|S_{s_As_B}| \cdot \max_{b \in \Dom(B)}\{|\sigma_{B=b} T_{t_Bt_C}|\}$,
or symmetrically, 
the product of the size of the $T$-part and the maximum number of tuples in the 
$S$-part that match with a single tuple in the $T$-part,
i.e., $|T_{t_At_B}| \cdot \max_{b \in \Dom(B)}\{|\sigma_{B=b} S_{s_As_B}|\}$.
It follows that these views admit the size 
bound $\bigO{N^{1 + \min\{\eps,1-\eps\}}}$, because in $V_{S_{ll}T_{lh}}$ and $V_{S_{hh}T_{lh}}$, the relation part $T_{lh}$ is heavy on $C$ and light on $B$, 
and in $V_{S_{hl}T_{ll}}$ and $V_{S_{hl}T_{t_Bh}}$, the relation part $S_{hl}$ is heavy on $A$ and light on $B$.
The second bound is obtained by taking the product of the number of all possible $A$-values in the $S$-part and the number of all possible 
$C$-values in the $T$-part. 
Hence, the 
views $V_{S_{hh}T_{lh}}$ and $V_{S_{hl}T_{t_Bh}}$ admit the 
size bound $\bigO{N^{2-2\eps}}$, because the $S$-part is heavy on 
$A$ and the $T$-part is heavy on $C$. 
Thus, the size bound of these views is the minimum of these two bounds.

We next analyze the space complexity of the views $V_{RS_{ll}}$ and $V_{RS_{s_A}h}$. 
The analysis for the views $V_{T_{ll}U}$ and $V_{T_{ht_C}}$ is analogous. 
The space complexity of the views $V_{RS_{ll}}$ and $V_{RS_{s_Ah}}$ are 
$\bigO{N}$ and $\bigO{N^{1-\eps}}$, respectively, because there could be linearly many $B$-values in $S_{ll}$ and at most $\bigO{N^{1-\eps}}$ $B$-values in $S_{s_Ah}$.

Finally, we analyze the space complexity of the views $V_{T_h}$, $V_{RS_{hl}T_{hl}}$, and $V_{RS_{ll}T_{lh}}$.
The analysis for the views $V_{S_h}, V_{S_{lh}T_{lh}U}$ and $V_{S_{hl}T_{ll}U}$ is analogous. 
The size of the view $V_{T_h}$ is $\bigO{N^{1-\eps}}$ because there are 
$\bigO{N^{1-\eps}}$ distinct $B$-values in $T_{hl}$.
The size of the view $V_{RS_{hl}T_{hl}}$ is bounded by the number of $B$-values in $S_{hl}$ and $V_T$. Since there are $\bigO{N^{1-\eps}}$ $B$-values in $V_{T_h}$, the space complexity of $V_{RS_{hl}T_{hl}}$ is $\bigO{N^{1-\eps}}$.
Similarly, the size of the view $V_{RS_{ll}T_{lh}}$ is bounded by the number of $C$-values in $T_{lh}$. Since $T_{lh}$ is heavy on $C$, there are
$\bigO{N^{1-\eps}}$ $C$-values in $T_{lh}$, thus, the space complexity of $V_{RS_{ll}T_{lh}}$ is $\bigO{N^{1-\eps}}$.

Hence, taking the linear space of the relation partitions 
in $\dbeps$  into account,
 the overall space complexity is 
$\bigO{|\inst{D}|^{1 + \max\{\eps, 1-\eps\}}}$.
\end{proof}

\subsection{Processing a Single-Tuple Update to the 4-Path Count}
Figure~\ref{fig:applyUpdate_path_4_delta_R} 
presents the procedure \textsc{ApplyUpdate4Path} that takes as input an update
to $R$  and
a state $\state = (\eps, N, \dbeps,\inst{V})$ of a database $\inst{D}$
with $\dbeps = 
\{R\} \cup
\{S_{s_As_B}\}_{s_A,s_B \in \{h,l\}} \cup
\{T_{t_Bt_C}\}_{t_B,t_C \in \{h,l\}} \cup 
\{U\}$ 
and materialized views $\inst{V}$ as defined in Figure~\ref{fig:view_definitions_path_4}. The procedure 
maintains $\state$ under the update.
Figure~\ref{fig:applyUpdate_path_4_delta_S}
gives the procedure for updates to $S$. 
The procedures for updates to $U$ and $T$ are analogous to the maintenance procedures for updates to $R$ and $S$, respectively. 

\begin{proposition}\label{prop:single_step_time_4_path}
Given a state $\state$ constructed from a database $\db$ for $\eps\in[0,1]$
to support the maintenance of the 4-path count, 
\ivme maintains $\state$ under a single-tuple update to any input relation in 
$\bigO{|\db|^{\max\{\eps,1-\eps\}}}$ time.
\end{proposition} 

\begin{proof}
We  first consider updates to $R$
and then to $S$.
\paragraph{Updates to $R$.}
We analyze the computation time 
of the procedure \textsc{ApplyUpdate4Path} for an update
$\delta R= \{(\deltaA) \mapsto m\}$ as given in 
Figure~\ref{fig:applyUpdate_path_4_delta_R}. 
The time to maintain the materialized views 
 is determined by the number of $B$-
 or $C$-values needed to be iterated over during delta computation.

We first analyze the computation of the deltas of the views 
that use views composed 
of a $T$-part and $U$, i.e., the computation of the deltas
 $\delta Q_{llll}$, $\delta Q_{llht_C}$, $\delta Q_{lhll}$, $\delta Q_{lhht_C}$, $\delta Q_{hlhl}$, $\delta Q_{hhll}$, and $\delta Q_{hhht_C}$ (in 
 lines 3, 5, 6, 8, 10, 12, and 14). 
We can choose between  two options when computing 
these deltas.  
We can iterate over the $B$-values paired with 
$\deltaA$ in the $S$-part and look up these $B$-values 
in the view 
that joins the $T$-part with $U$. Alternatively, we can iterate 
over the $B$-values $b$ in the view that joins the $T$-part with
$U$ and look up the tuples $(\deltaA, b)$ in the $S$-part.
The complexity is determined by the number of $B$-values we iterate 
over. 
Computing $\delta Q_{llll}$ (line 3) needs $\bigO{N^{\eps}}$ time, since 
$S_{ll}$ is light on $A$. 
Computing $\delta Q_{llht_C}$, $\delta Q_{lhll}$, and $\delta Q_{lhht_C}$ (lines 5, 6 and, 8) needs $\bigO{N^{\min\{\eps, 1-\eps\}}}$ time, 
since the $S$-part is light on $A$ and 
either the $S$-part is heavy on $B$ or there are 
$\bigO{N^{1-\eps}}$ distinct $B$-values in the view that joins the $T$-part with $U$. 
Computing $\delta Q_{hhll}$ and $\delta Q_{hhht_C}$ (lines 12 and 14) needs 
$\bigO{N^{1-\eps}}$ time, because $S_{hh}$ is heavy on $B$. 

We next analyze the time to compute the 
delta views that use views that join the $S$-part and the $T$-part, i.e., the computation times for $\delta Q_{lllh}$, $\delta Q_{hlt_Bh}$, and $\delta Q_{hhlh}$ (lines 4, 11, and 13). 
To compute these delta views we need to iterate over the $C$-values 
paired with $\deltaA$ in the view that joins the $S$-part and 
the $T$-part and to look up the $C$-values in $U$.
The complexity is determined by the number of $C$-values we need 
to iterate over. 
Since there are at most $\bigO{N^{1-\eps}}$ distinct $C$-values in the views 
$V_{S_{ll}T_{lh}}$, $V_{S_{hl}T_{t_Bh}}$, and $V_{S_{hh}T_{lh}}$, 
the time complexity of computing these deltas is $\bigO{N^{1-\eps}}$.

Computing $\delta Q_{lhlh}$ (line 7) needs the iteration 
 over the $B$-values paired with $\deltaA$ in $V_{S_{lh}T_{lh}U}$
  and the lookups of these $B$-values  in $V_{S_{lh}T_{lh}U}$. 
  Since $S_{lh}$ is heavy on $B$ and light on $A$, the computation time is $\bigO{N^{\min\{\eps, 1-\eps\}}}$.
  Note that in the definition of $\delta Q_{lhlh}$,
   the $S$-part occurs twice: once outside the view $V_{S_{lh}T_{lh}U}$
   and once within $V_{S_{lh}T_{lh}U}$. However, 
   $V_{S_{lh}T_{lh}U}$ uses just an indicator 
  projection of the $S$-part. Hence, the multiplicities in the
   $S$-part do not contribute twice to the computation of 
   $\delta Q_{lhlh}$.
Computing $\delta Q_{hlll}$ (line 9) requires a single lookup of $\deltaA$ in $V_{S_{hl}T_{ll}U}$, which takes constant time.

We now analyze the maintenance time 
for the auxiliary views. 
Maintaining $V_{RS_{ll}}$, $V_{RS_{lh}}$, and $V_{RS_{hh}}$ (lines 16-18) requires the iteration over all $B$-values paired with $\deltaA$ in the $S$-part.
Maintaining $V_{RS_{ll}}$ takes $\bigO{N^{\eps}}$ time, since $S_{ll}$ is light on $A$.
Maintaining $V_{RS_{hh}}$ takes $\bigO{N^{1-\eps}}$ time, since $S_{hh}$ is heavy on $B$.
Maintaining $V_{RS_{lh}}$ takes $\bigO{N^{\min\{\eps, 1-\eps}\}}$ time, 
since $S_{lh}$ is light on $A$ and heavy on $B$.

Maintaining $V_{RS_{hl}T_{hl}}$ requires to iterate over all distinct 
$B$-values  in $V_{T_h}$ and to do lookups of $(\deltaA, b)$ in $S_{hl}$. 
Since there are $\bigO{N^{1-\eps}}$ $B$-values in $V_{T_h}$, the computation 
time is $\bigO{N^{1-\eps}}$.
Maintaining $V_{RS_{ll}T_{lh}}$ requires the iteration over all 
$B$-values paired with $\deltaA$ in $V_{S_{ll}T_{lh}}$.
Since there are $\bigO{N^{1-\eps}}$ $B$-values in $V_{S_{ll}T_{lh}}$, the computation time is $\bigO{N^{1-\eps}}$.

It follows that the overall computation time of the 
procedure \textsc{ApplyUpdate4Path} 
in Figure~\ref{fig:applyUpdate_path_4_delta_R} 
is $\bigO{N^{\max\{\eps, 1-\eps}\}}$, which,   
by $N = \Theta(|\inst{D}|)$, is 
$\bigO{|\inst{D}|^{\max\{\eps, 1-\eps}\}}$. 

\paragraph{Updates to $S$.}
Given an update
$\delta S_{s_As_B}= \{(\deltaA,\deltaB) \mapsto \p\}$ with $s_A,s_B \in \{h,l\}$, 
we analyze the computation time of 
the procedure \textsc{ApplyUpdate4Path} in
Figure~\ref{fig:applyUpdate_path_4_delta_S}

Computing $\delta Q_{s_As_Bll}$, $\delta Q_{s_As_Blh}$,
 and $\delta Q_{s_As_Bhh}$ (lines 3, 4 and 6) requires the 
 iteration over all $C$-values in the $T$-part and the lookups
  of these $C$-values in $U$. 
The computation time is determined by the number of $C$-values needed to be iterated over.
Computing $\delta Q_{s_As_Bll}$ takes $\bigO{N^{\eps}}$ time, since $T_{ll}$ is light on $B$.
Computing $\delta Q_{s_As_Bhh}$ takes $\bigO{N^{1-\eps}}$ time, since $T_{hh}$ is heavy on $C$.
Computing $\delta Q_{s_As_Blh}$ takes $\bigO{N^{\min\{\eps, 1-\eps}\}}$ time, 
 since $T_{lh}$ is light on $B$ and heavy on $C$.
Computing $\delta Q_{s_As_Bhl}$ (line 5) takes constant time, because it only requires the lookup of $\deltaA$ in $R$ and $\deltaB$ in $V_{T_{hl}U}$.

We next analyze the maintenance 
of the auxiliary views. 
Maintaining $V_{RS_{hh}}$, $V_{RS_{lh}}$, and $V_{RS_{ll}}$ (lines 9, 12 and 20) takes constant time, because it only requires the lookup of $\deltaA$ in $R$. 
Maintaining $V_{S_{hh}T_{lh}}$ ,$V_{S_{hl}T_{ll}}$, $V_{S_{hl}T_{lh}}$, $V_{S_{hl}T_{hh}}$, and $V_{S_{ll}T_{lh}}$ (lines 10, 15-17 and 21) requires the iteration
over all $C$-values paired with $\deltaB$ in the $T$-part.
The computation time is determined by the number of $C$-values needed to be iterated over.
Maintaining $V_{S_{hl}T_{ll}}$ takes $\bigO{N^{\eps}}$ time,
 because $T_{ll}$ is light on $B$.
Maintaining $V_{S_{hl}T_{hh}}$ and $V_{S_{ll}T_{lh}}$ takes $\bigO{N^{1-\eps}}$
time, because $T_{hh}$ and $T_{lh}$ are heavy on $C$. 
Maintaining $V_{S_{hh}T_{lh}}$ and $V_{S_{hl}T_{lh}}$ takes $\bigO{N^{\min\{\eps, 1-\eps\}}}$ time, since  $T_{lh}$ is heavy on $C$ and light on $B$.

Maintaining  $V_{S_h}$ (line 13) takes constant time. 
Maintaining $V_{S_{hl}T_{ll}U}$ (line 18) needs the iteration over all 
$C$-values paired with $\deltaB$ in $T_{ll}$ and the lookups of these 
$C$-values in $U$. The computation time is $\bigO{N^{\eps}}$, 
because $T_{ll}$ is light on $B$.
Maintaining  $V_{RS_{ll}T_{lh}}$ (line 22) requires the iteration over all $C$-values paired with $\deltaB$ in $T_{lh}$. 
The computation time is $\bigO{N^{\min\{\eps, 1-\eps\}}}$, because $T_{lh}$ is heavy on $C$ and light on $B$.

We derive that the computation time of the 
procedure \textsc{ApplyUpdate4Path}
for updates to relation $S$  
is $\bigO{N^{\max\{\eps, 1-\eps}\}}$, which,   
by $N = \Theta(|\inst{D}|)$, is 
$\bigO{|\inst{D}|^{\max\{\eps, 1-\eps}\}}$. 

Hence, the overall computation time of the 
procedure \textsc{ApplyUpdate4Path}
for updates to $R$ and $S$ 
is $\bigO{|\inst{D}|^{\max\{\eps, 1-\eps}\}}$. 
\end{proof}

\subsection{Rebalancing Partitions for the 4-Path Count}
The rebalancing strategy for the 
4-path count follows the rebalancing strategy 
for the triangle count given in Section~\ref{sec:rebalancing}. 
A major rebalancing step repartitions the relations strictly 
according to the new threshold and recomputes all materialized auxiliary 
views. Repartitioning relations can be done in linear time. 
It follows from Proposition~\ref{prop:preprocessing_step_4_path}
that the auxiliary views can be computed in time 
$\bigO{|\inst{D}|^{1 + \min{\{\eps, 1-\eps\}}}}$.
This computation time is amortized over $\Omega(|\inst{D}|)$
updates.
A minor rebalancing step deletes  
$\bigO{|\inst{D}|^{\eps}}$ tuples in a
relation part and inserts them into the other  
relation part. By Proposition 
\ref{prop:single_step_time_4_path},
this takes 
$\bigO{|\inst{D}|^{\eps + \max{\{\eps, 1-\eps\}}}}$ time.
This computation time is amortized over 
$\Omega(|\inst{D}|^\eps)$ updates. 
This means that a rebalancing step needs 
$\bigO{|\inst{D}|^{\max{\{\eps, 1-\eps\}}}}$
amortized time. 

Now, we can prove the main theorem of this section.  
\begin{proof}[Proof of Theorem \ref{theo:main_result_path4}]
It follows from 
Propositions~\ref{prop:preprocessing_step_4_path}
and \ref{prop:space_complexity_4_path} that the
 preprocessing time
and the space complexity are
$\bigO{|\inst{D}|^{1 + \min\{\eps,1-\eps\}}}$. 
By Proposition~\ref{prop:single_step_time_4_path}, the 
time to process a single-tuple update is  
$\bigO{|\inst{D}|^{\max\{\eps,1-\eps\}}}$.
Since the amortized rebalancing time
is $\bigO{|\inst{D}|^{\max\{\eps,1-\eps\}}}$, 
the overall amortized update time is
 $\bigO{|\inst{D}|^{\max\{\eps,1-\eps\}}}$.
 Since the 4-path count is materialized and maintained, 
 the answer time is constant.
\end{proof}

\subsection{Worst-Case Optimality of \ivme for the 4-Path Count}
For $\eps = 0.5$, \ivme maintains the 4-path count 
with amortized $\bigO{|\inst{D}|^{\max{\{\eps, 1-\eps\}}}}$
update time and constant answer time.
By the lower bound given in 
Proposition~\ref{prop:lower_bound_path_4}, this is worst-case 
optimal, conditioned on the \OMv conjecture    
(Conjecture \ref{conj:omv}). We next prove this proposition.
  
\nop{
 The worst-case optimality of the update time of the 
\ivme strategy for $\threerpj{2} \cup \threerpj{3}$ queries, 
conditioned on the \OMv conjecture (Conjecture \ref{conj:omv}) 
follows from Proposition 
\ref{prop:lower_bound_3rel_count}. The proof of the proposition is a 
straightforward adaption of the proof of Proposition 
\ref{prop:lower_bound_triangle_count}.
}

\begin{proof}[Proof of Proposition \ref{prop:lower_bound_path_4}] 
The proof is a simple extension of the lower bound 
proof for the triangle count 
(Proposition~\ref{prop:lower_bound_triangle_count}). 
We reduce the \OuMv problem given in Definition 
\ref{def:OuMv} to the incremental maintenance of 
the 4-path count. 
We emphasize on the differences 
to the reduction in the proof of 
Proposition~\ref{prop:lower_bound_triangle_count}.

Assume that there is a dynamic algorithm 
maintaining the 4-path count 
with arbitrary preprocessing time, amortized update time 
$\bigO{|\inst{D}|^{\frac{1}{2}-\gamma}}$, 
and answer time $\bigO{|\inst{D}|^{1-\gamma}}$.
This algorithm can be used 
to solve the $\OuMv$ problem in subcubic time. This
 contradicts the \OuMv conjecture.

Let  $(\vecnormal{M}, (\vecnormal{u}_1,\vecnormal{v}_1), \ldots ,(\vecnormal{u}_n,\vecnormal{v}_n))$ be an input to the $\OuMv$ problem.
The idea is to use relation $S$ to encode the matrix
$\inst{M}$ and to use relations $R$ and $T$ to encode 
the vectors $\vecnormal{u}_i$ and $\vecnormal{v}_i$, respectively. 
We fill relation $U$ with a ``dummy'' tuple not contributing to the 
final count.  
 After the construction of the initial \ivme state for an empty database 
$\db = \{R,S,T,U\}$, we execute at most $n^2$ updates to relation $S$ such that 
$S = \{\, (i,j) \mapsto \vecnormal{M}(i,j) \,\mid\, i,j \in \{1,\ldots, n\} \,\}$.
Then, we execute an additional update to $U$ such that 
$U = \{\, (a) \mapsto 1\}$. 
In each round $r \in \{1, \ldots , n\}$, we execute at most $2n$ updates to the relations $R$ and $T$ 
such that $R = \{\, (i) \mapsto \vecnormal{u}_r(i) \,\mid\, i \in \{1,\ldots, n\} \,\}$
and $T = \{\,(i, a) \mapsto \vecnormal{v}_r(i) \,\mid\, i \in \{1,\ldots, n\} \,\}$.
At the end of round $r$, the 
algorithm outputs 
$1$ if and only if  the 4-path count is nonzero. 
The time analysis of the reduction follows the proof of 
Proposition \ref{prop:lower_bound_triangle_count}. 
\end{proof}

\section{Count Queries with Three Relations}
\label{sec:inc_maint_count_three_rels}
We investigate the incremental maintenance 
of \threer queries, i.e., 
 count queries with three relations.
 We design \ivme strategies that 
maintain \threer queries in worst-case optimal time
(conditioned on the \OMv  conjecture).

We consider queries composed of three relations 
$R$, $S$, and $T$ with schemas 
$(\inst{A}_R,\inst{A}_{RT},\inst{A}_{RS},\inst{A}_{RST})$,
$(\inst{A}_S,\inst{A}_{RS},\inst{A}_{ST},\inst{A}_{RST})$, and 
$(\inst{A}_T,\inst{A}_{ST},\inst{A}_{RT},\inst{A}_{RST})$,
respectively, where 
$\inst{A}_R$, $\inst{A}_S$, $\inst{A}_T$,
$\inst{A}_{RT}$, $\inst{A}_{RS}$, $\inst{A}_{ST}$, and $\inst{A}_{RST}$
are possibly empty tuples of
variables. The tuples are pairwise disjoint, i.e., the
same variable does not occur in two distinct tuples.
The index of each tuple indicates 
in which relations the variables in the tuple occur. 
For instance, 
all variables in $\inst{A}_{RS}$ occur in $R$ and in $S$, but none 
of them occurs in $T$.
Given such a tuple $\inst{A}_I$ with some index $I$, 
we use 
$\inst{a}_I$ to denote a tuple of data values 
over $\inst{A}_I$.
For simplicity, we 
skip the indices under the summation symbols in queries. 
By convention, the sum in a query goes over all 
data values in the domains of the non-free variables of the query.   
In count queries, all variables are bound. 
  
\begin{definition}[\threer Queries]
\label{def:threer_queries}
A \threer query is of the form 
$$Q() = \sum_{}
R(\inst{a}_R,\inst{a}_{RT},\inst{a}_{RS},\inst{a}_{RST}) \ztimes 
S(\inst{a}_S,\inst{a}_{RS},\inst{a}_{ST},\inst{a}_{RST}) \ztimes 
T(\inst{a}_T,\inst{a}_{ST},\inst{a}_{RT},\inst{a}_{RST}).$$
\end{definition}

By skipping a tuple $\inst{a}_I$ in a \threer query, 
we indicate that the corresponding variable tuple
$\inst{A}_I$ is empty.
We call the tuples 
$\inst{A}_R$, $\inst{A}_S$, and $\inst{A}_T$ non-join tuples
and the tuples  
$\inst{A}_{RT}$, $\inst{A}_{RS}$, and $\inst{A}_{ST}$ 
pair-join variable tuples.

\begin{example}
\upshape
To simplify presentation, we often show the hypergraphs 
of \threer queries. The hypergraph of a query
contains a node for each variable tuple
and a hyperedge 
 for each relation in the query.  A hyperedge corresponding 
 to a relation 
 includes those nodes in the hypergraph 
 that represent
 variable tuples in the schema of the relation. 
Consider the \threer query 
$$Q() = \sum
R(\inst{a}_R,\inst{a}_{RT},\inst{a}_{RS},\inst{a}_{RST}) \ztimes 
S(\inst{a}_S,\inst{a}_{RS},\inst{a}_{RST}) \ztimes 
T(\inst{a}_{T},\inst{a}_{RT},\inst{a}_{RST}).$$
The tuple $\inst{a}_{ST}$ is skipped in the query, which means that 
the corresponding variable tuple $\inst{A}_{ST}$ is empty. 
The following figure shows on the left-hand side 
the hypergraph of $Q$ that indicates 
that the variable tuples $\inst{A}_{T}$, $\inst{A}_R$, 
and $\inst{A}_{RS}$ are nonempty and 
the tuples $\inst{A}_{RT}$, $\inst{A}_S$, 
and $\inst{A}_{RST}$ are {\em possibly} empty.
The hypergraph of $\delta Q$ under an update 
$\delta R(\bdeltaA_R,\bdeltaA_{RT},\bdeltaA_{RS},\bdeltaA_{RST})$
to $R$ is given on the right-hand side. 
  
 
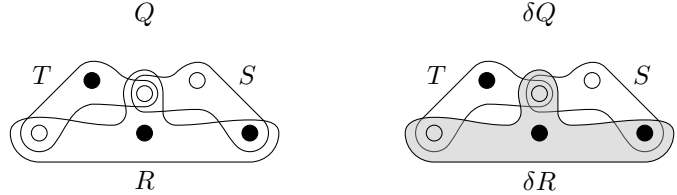
\begin{center}
\begin{tikzpicture}[scale=0.35]

  \node at (0, 0.5) (Q) {$Q$};
 \node at (-4, -4) (A) {};
  \node at (4, -4) (B) {};
  \node at (0, -4) (D) {};
  \node at (2, -2) (E) {};
  \node at (-2, -2) (F) {};
 \node at (0, -2.5) (X) {};  

      \draw
       ($(F)+(-0.5,0.5)$) 
        to[out=225,in=45] ($(A) + (-0.5,0.5)$)
         to[out=225,in=135] ($(A) + (-0.5,-0.5)$)
         to[out=315,in=225] ($(A) + (0.5,-0.5)$)
        to[out=50,in=190] ($(F) + (0,-0.9)$)    
        to[out=0,in=180] ($(X) + (0,-0.5)$) 
        to[out=0,in=270] ($(X) + (0.5,0)$)
          to[out=90,in=0] ($(X) + (0,0.5)$)               
          to[out=180,in=315] ($(F) + (1,0.3)$)
          to[out=135,in=45] ($(F) + (-0.5,0.5)$);

  \draw
       ($(E)+(0.5,0.5)$) 
        to[out=315,in=135] ($(B) + (0.5,0.5)$)
         to[out=315,in=45] ($(B) + (0.5,-0.5)$)
        to[out=225,in=315] ($(B) + (-0.5,-0.5)$)
        to[out=135,in=0] ($(E) + (0,-1.1)$)  
        to[out=180,in=0] ($(X) + (0,-0.7)$)
        to[out=180,in=270] ($(X) + (-0.5,0)$)
      to[out=90,in=180] ($(X) + (0,0.7)$) 
       to[out=0,in=225] ($(E) + (-0.8,0.3)$)
       to[out=45,in=135] ($(E) + (0.5,0.5)$);

\draw
       ($(A)+(-1.1,0)$) 
        to[out=270,in=180] ($(A) + (0,-1.1)$)
         to[out=0,in=180] ($(B) + (0,-1.1)$)
         to[out=0,in=270] ($(B) + (1.1,0)$)
         to[out=90,in=0] ($(B) + (0,0.5)$)
         to[out=180,in=0] ($(D) + (1.2,0.3)$)  
         to[out=180,in=270] ($(D) + (0.7,0.8)$)  
         to[out=90,in=270] ($(X) + (0.7,0)$)
        to[out=90,in=0] ($(X) + (0,0.9)$)
        to[out=180,in=90] ($(X) + (-0.8,0)$)      
         to[out=270,in=90] ($(D) + (-0.7,0.8)$)  
         to[out=270,in=0] ($(D) + (-1.2,0.3)$)  
         to[out=180,in=0] ($(A) + (0,0.6)$)  
       to[out=180,in=90] ($(A)+(-1.1,0)$);

   \draw [color=black] (-4, -4) circle (.3);
    \draw [color=black, fill=black] (4, -4) circle (.3);
   \draw [color=black, fill = black] (0, -4) circle (.3);     
   \draw [color=black] (2, -2) circle (.3);       
   \draw [color=black, fill = black] (-2, -2) circle (.3);       
   \draw [color=black] (0, -2.5) circle (.3); 
   
   \node at (-3.9, -1.8) (T) {$T$};
  \node at (3.9, -1.8) (S) {$S$};
  \node at (0, -5.8) (R) {$R$}; 


\begin{scope}[xshift = 15cm]
 \node at (0, 0.5) (Q) {$\delta Q$};
  \node at (-4, -4) (A) {};
  \node at (4, -4) (B) {};
  \node at (0, -4) (D) {};
  \node at (2, -2) (E) {};
  \node at (-2, -2) (F) {};
 \node at (0, -2.5) (X) {};  

      \draw
       ($(F)+(-0.5,0.5)$) 
        to[out=225,in=45] ($(A) + (-0.5,0.5)$)
         to[out=225,in=135] ($(A) + (-0.5,-0.5)$)
         to[out=315,in=225] ($(A) + (0.5,-0.5)$)
        to[out=50,in=190] ($(F) + (0,-0.9)$)    
        to[out=0,in=180] ($(X) + (0,-0.5)$) 
        to[out=0,in=270] ($(X) + (0.5,0)$)
          to[out=90,in=0] ($(X) + (0,0.5)$)               
          to[out=180,in=315] ($(F) + (1,0.3)$)
          to[out=135,in=45] ($(F) + (-0.5,0.5)$);

  \draw
       ($(E)+(0.5,0.5)$) 
        to[out=315,in=135] ($(B) + (0.5,0.5)$)
         to[out=315,in=45] ($(B) + (0.5,-0.5)$)
        to[out=225,in=315] ($(B) + (-0.5,-0.5)$)
        to[out=135,in=0] ($(E) + (0,-1.1)$)  
        to[out=180,in=0] ($(X) + (0,-0.7)$)
        to[out=180,in=270] ($(X) + (-0.5,0)$)
      to[out=90,in=180] ($(X) + (0,0.7)$) 
       to[out=0,in=225] ($(E) + (-0.8,0.3)$)
       to[out=45,in=135] ($(E) + (0.5,0.5)$);

\draw[fill= light-gray,fill opacity = 0.4]
       ($(A)+(-1.1,0)$) 
        to[out=270,in=180] ($(A) + (0,-1.1)$)
         to[out=0,in=180] ($(B) + (0,-1.1)$)
         to[out=0,in=270] ($(B) + (1.1,0)$)
         to[out=90,in=0] ($(B) + (0,0.5)$)
         to[out=180,in=0] ($(D) + (1.2,0.3)$)  
         to[out=180,in=270] ($(D) + (0.7,0.8)$)  
         to[out=90,in=270] ($(X) + (0.7,0)$)
        to[out=90,in=0] ($(X) + (0,0.9)$)
        to[out=180,in=90] ($(X) + (-0.8,0)$)      
         to[out=270,in=90] ($(D) + (-0.7,0.8)$)  
         to[out=270,in=0] ($(D) + (-1.2,0.3)$)  
         to[out=180,in=0] ($(A) + (0,0.6)$)  
       to[out=180,in=90] ($(A)+(-1.1,0)$);

  \draw [color=black] (-4, -4) circle (.3);
    \draw [color=black, fill=black] (4, -4) circle (.3);
   \draw [color=black, fill = black] (0, -4) circle (.3);     
   \draw [color=black] (2, -2) circle (.3);       
   \draw [color=black, fill = black] (-2, -2) circle (.3);       
   \draw [color=black] (0, -2.5) circle (.3);
   
   \node at (-3.9, -1.8) (T) {$T$};
  \node at (3.9, -1.8) (S) {$S$};
  \node at (0, -5.8) (R) {$\delta R$}; 

\end{scope}

\end{tikzpicture}
 \end{center}

As shown in the figure, we represent a nonempty variable tuple
by a filled circle, and a possibly empty variable tuple by a non-filled circle. 
A hyperedge representing an update is depicted in gray.  
\null\hfill\qedsymbol
\end{example}

Next, we define classes of \threer queries 
with a bound on the number of nonempty 
pair-join variable tuples.

\begin{definition}[\threer Query Classes]
Given $i\in \{0,1,2,3\}$ 
the class of \threer queries with exactly  
$i$ nonempty pair-join variable tuples is denoted by \threerpj{i}.
\end{definition}  

Observe that $\threer = \bigcup_{i = 0}^{3} \threerpj{i}$.

\begin{example}
\upshape
Let 
\begin{itemize}
\item $Q_1() = \textstyle\sum R(\inst{a}_R, \inst{a}_{RT}, \inst{a}_{RS}, \inst{a}_{RST}) \ztimes 
S(\inst{a}_S,\inst{a}_{RS},\inst{a}_{ST},\inst{a}_{RST}) \ztimes 
T(\inst{a}_T,\inst{a}_{ST},\inst{a}_{RT},\inst{a}_{RST})$,

\item $Q_2() = \textstyle\sum R(\inst{a}_{RT},\inst{a}_{RS}) \ztimes 
S(\inst{a}_S,\inst{a}_{RS}) \ztimes 
T(\inst{a}_T,\inst{a}_{RT})$, 

\item $Q_3() = \textstyle\sum  
R(\inst{a}_R) \ztimes 
S(\inst{a}_S,\inst{a}_{ST}) \ztimes 
T(\inst{a}_T,\inst{a}_{ST})$, and 

\item $Q_4() = \textstyle\sum 
R(\inst{a}_R,\inst{a}_{RST}) \ztimes 
S(\inst{a}_S,\inst{a}_{RST}) \ztimes 
T(\inst{a}_T,\inst{a}_{RST})$, 
\end{itemize}
where all data value tuples given in the queries are over nonempty
variable sets.   
The hypergraphs of the queries are depicted below.  
For instance, in $Q_2$, the variable tuples    
$\inst{A}_T$, $\inst{A}_{RT}$, $\inst{A}_{RS}$, 
and $\inst{A}_{S}$ are nonempty
while 
$\inst{A}_R$, $\inst{A}_{ST}$, and $\inst{A}_{RST}$
are empty. 

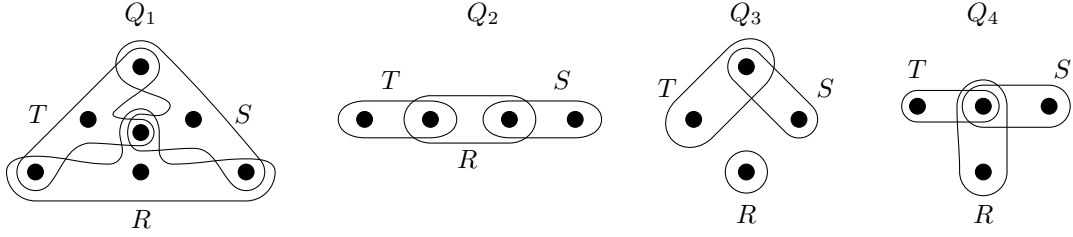
\begin{center}
\begin{tikzpicture}[scale=0.35]

  \node at (0, 2) (Q) {$Q_1$};
  \node at (0, 0) (C) {};
  \node at (-4, -4) (A) {};
  \node at (4, -4) (B) {};
  \node at (0, -4) (D) {};
  \node at (2, -2) (E) {};
  \node at (-2, -2) (F) {};
 \node at (0, -2.5) (X) {};  

  \draw
       ($(C)+(0.5,0.5)$) 
        to[out=135,in=45] ($(C) + (-0.5,0.5)$)
        to[out=225,in=45] ($(A) + (-0.5,0.5)$)
         to[out=225,in=135] ($(A) + (-0.5,-0.5)$)
         to[out=315,in=225] ($(A) + (0.5,-0.5)$)
        to[out=50,in=190] ($(F) + (0,-0.9)$)    
        to[out=0,in=180] ($(X) + (0,-0.5)$) 
        to[out=0,in=270] ($(X) + (0.5,0)$)
          to[out=90,in=0] ($(X) + (0,0.5)$)               
          to[out=180,in=225] ($(F) + (1,0.3)$)         
          to[out=45,in=225] ($(C) + (0.5,-0.5)$)  
          to[out=45,in=315] ($(C) + (0.5,0.5)$);

  \draw
       ($(C)+(-0.7,0.7)$) 
        to[out=45,in=135] ($(C) + (0.7,0.7)$)
        to[out=315,in=130] ($(B) + (0.5,0.5)$)
         to[out=315,in=45] ($(B) + (0.5,-0.5)$)
        to[out=225,in=315] ($(B) + (-0.5,-0.5)$)
        to[out=135,in=0] ($(E) + (0,-1.1)$)  
        to[out=180,in=0] ($(X) + (0,-0.7)$)
        to[out=180,in=270] ($(X) + (-0.5,0)$)
      to[out=90,in=180] ($(X) + (0,0.7)$) 
       to[out=0,in=225] ($(E) + (-1,0.3)$)        
       to[out=45,in=315] ($(C) + (-0.7,-0.5)$)  
       to[out=135,in=225] ($(C) + (-0.7,0.7)$);

  \draw
       ($(A)+(-1.1,0)$) 
        to[out=270,in=180] ($(A) + (0,-1.1)$)
         to[out=0,in=180] ($(B) + (0,-1.1)$)
         to[out=0,in=270] ($(B) + (1.1,0)$)
         to[out=90,in=0] ($(B) + (0,0.5)$)
         to[out=180,in=0] ($(D) + (1.2,0.3)$)  
         to[out=180,in=270] ($(D) + (0.7,0.8)$)  
         to[out=90,in=270] ($(X) + (0.7,0)$)
        to[out=90,in=0] ($(X) + (0,0.9)$)
        to[out=180,in=90] ($(X) + (-0.8,0)$)      
         to[out=270,in=90] ($(D) + (-0.7,0.8)$)  
         to[out=270,in=0] ($(D) + (-1.2,0.3)$)  
         to[out=180,in=0] ($(A) + (0,0.6)$)  
       to[out=180,in=90] ($(A)+(-1.1,0)$);                        
      
       \draw [color=black, fill=black] (0, 0) circle (.3);
\draw [color=black, fill=black] (-4, -4) circle (.3);
        \draw [color=black, fill=black] (4, -4) circle (.3);
        \draw [color=black, fill=black] (0, -4) circle (.3);
        \draw [color=black, fill=black] (2, -2) circle (.3);
\draw [color=black, fill=black] (-2, -2) circle (.3);
\draw [color=black, fill=black] (0, -2.5) circle (.3);
         
   \node at (-3.9, -1.8) (T) {$T$};
  \node at (3.9, -1.8) (S) {$S$};
  \node at (0, -5.8) (R) {$R$}; 


\begin{scope}[xshift = 11cm]
 \node at (2, 2) (Q) {$Q_2$};
  \node at (-2.5, -2) (F) {};
  \node at (0, -2) (A) {};
  \node at (3, -2) (B) {};
  \node at (5.5, -2) (E) {};

  \draw
       ($(F)+(-1,0)$) 
        to[out=270,in=180] ($(F) + (0,-0.7)$)
         to[out=0,in=180] ($(A) + (0,-0.7)$)  
         to[out=0,in=270] ($(A) + (1,0)$)  
         to[out=90,in=0] ($(A) + (0,0.7)$)  
         to[out=180,in=0] ($(F) + (0,0.7)$)  
       to[out=180,in=90] ($(F)+(-1,0)$);                        
       
     \draw
       ($(A)+(-1,0)$) 
        to[out=270,in=180] ($(A) + (0,-0.9)$)
         to[out=0,in=180] ($(B) + (0,-0.9)$)  
         to[out=0,in=270] ($(B) + (1,0)$)  
         to[out=90,in=0] ($(B) + (0,0.9)$)  
         to[out=180,in=0] ($(A) + (0,0.9)$)  
       to[out=180,in=90] ($(A)+(-1,0)$);                            

     \draw
       ($(B)+(-1,0)$) 
        to[out=270,in=180] ($(B) + (0,-0.7)$)
         to[out=0,in=180] ($(E) + (0,-0.7)$)  
         to[out=0,in=270] ($(E) + (1,0)$)  
         to[out=90,in=0] ($(E) + (0,0.7)$)  
         to[out=180,in=0] ($(B) + (0,0.7)$)  
       to[out=180,in=90] ($(B)+(-1,0)$);

  \draw [color=black, fill=black] (-2.5, -2) circle (.3);
  \draw [color=black, fill=black] (0, -2) circle (.3);
  \draw [color=black, fill=black] (3, -2) circle (.3);
  \draw [color=black, fill=black] (5.5, -2) circle (.3);

   \node at (-1.5, -0.5) (T) {$T$};
  \node at (5, -0.5) (S) {$S$};
  \node at (1.4, -3.5) (R) {$R$}; 
\end{scope}

\begin{scope}[xshift = 23cm]
   \node at (0, 2) (Q) {$Q_3$};
   
  \node at (0, 0) (C) {};
  \node at (0, -4) (D) {};
  \node at (2, -2) (E) {};
  \node at (-2, -2) (F) {};
  
  \draw
       ($(C)+(0.8,0.8)$) 
         to[out=315,in=45] ($(C) + (0.7,-0.7)$)
        to[out=225,in=45] ($(F) + (0.7,-0.7)$)          
          to[out=225,in=315] ($(F) + (-0.8,-0.8)$)         
          to[out=135,in=225] ($(F) + (-0.7,0.7)$)     
         to[out=45,in=225] ($(C) + (-0.7,0.7)$)      
       to[out=45,in=135] ($(C) + (0.8,0.8)$);         
          
   \draw
       ($(C)+(-0.5,0.5)$) 
        to[out=45,in=135] ($(C) + (0.5,0.5)$)
        to[out=315,in=135] ($(E) + (0.5,0.5)$)
       to[out=315,in=45] ($(E) + (0.5,-0.5)$)
            to[out=225,in=315] ($(E) + (-0.5,-0.5)$)
            to[out=135,in=315] ($(C) + (-0.5,-0.5)$)                             
           to[out=135,in=225] ($(C) + (-0.5,0.5)$);

  \draw
       ($(D)+(0,0.8)$) 
        to[out=0,in=90] ($(D) + (0.8,0)$)
         to[out=270,in=0] ($(D) + (0,-0.8)$)  
         to[out=180,in=270] ($(D) + (-0.8,0)$)  
        to[out=90,in=180] ($(D) + (0,0.8)$);

       \draw [color=black, fill=black] (0, 0) circle (.3);
  \draw [color=black, fill=black] (0, -4) circle (.3);
  \draw [color=black, fill=black] (2, -2) circle (.3);
  \draw [color=black, fill=black] (-2, -2) circle (.3);
   
   \node at (-3, -0.8) (T) {$T$};
  \node at (3, -0.8) (S) {$S$};
  \node at (0, -5.6) (R) {$R$}; 

\end{scope}

\begin{scope}[xshift = 32cm]
 \node at (0, 2) (Q) {$Q_4$};

  \node at (0, -4) (D) {};
  \node at (2.5, -1.5) (E) {};
  \node at (-2.5, -1.5) (F) {};
  \node at (0, -1.5) (X) {};
  
  \draw
       ($(F)+(-0.6,0)$) 
         to[out=90,in=180] ($(F) + (0,0.6)$)
        to[out=0,in=180] ($(X) + (0,0.6)$)          
          to[out=0,in=90] ($(X) + (0.6,0)$)         
          to[out=270,in=0] ($(X) + (0,-0.6)$)     
         to[out=180,in=0] ($(F) + (0,-0.6)$)         
         to[out=180,in=270] ($(F) + (-0.6,0)$);        
    
   \draw
       ($(X)+(0,1)$) 
         to[out=0,in=90] ($(X) + (0.9,0)$)
        to[out=270,in=90] ($(D) + (0.9,0)$)          
          to[out=270,in=0] ($(D) + (0,-0.9)$)         
          to[out=180,in=270] ($(D) + (-0.9,0)$)     
         to[out=90,in=270] ($(X) + (-1,0)$)         
         to[out=90,in=180] ($(X) + (0,1)$);

   \draw
       ($(X)+(-0.8,0)$) 
         to[out=90,in=180] ($(X) + (0,0.8)$)
        to[out=0,in=180] ($(E) + (0,0.8)$)          
          to[out=0,in=90] ($(E) + (0.8,0)$)         
          to[out=270,in=0] ($(E) + (0,-0.8)$)     
         to[out=180,in=0] ($(X) + (0,-0.8)$)         
         to[out=180,in=270] ($(X) + (-0.8,0)$);           
         
 \draw [color=black, fill=black] (0, -4) circle (.3);
  \draw [color=black, fill=black] (2.5, -1.5) circle (.3);
 \draw [color=black, fill=black] (-2.5, -1.5) circle (.3);
  \draw [color=black, fill=black] (0, -1.5) circle (.3);

   \node at (-2.5, -0.1) (T) {$T$};
  \node at (3, -0.1) (S) {$S$};
  \node at (0, -5.6) (R) {$R$}; 

\end{scope}

\end{tikzpicture}
 \end{center}

The query 
$Q_1$
is a \threerpj{3} query, since all three 
pair-join variable 
tuples
$\inst{A}_{RT}$, $\inst{A}_{RS}$, and $\inst{A}_{ST}$
are nonempty.
The query is cyclic and non-hierarchical (for a definition of 
hierarchical queries, see Appendix \ref{sec:maintenance_3-rel_1-pair-join}. 
The query $Q_2$
with its two nonempty pair-join variable  
tuples 
$\inst{A}_{RT}$ and $\inst{A}_{RS}$
belongs to the class \threerpj{2}.
It is a non-hierarchical path query. 
The pair-join variable tuple $\inst{A}_{ST}$
is nonempty in $Q_3$. The other two 
 pair-join variable tuples are empty.  
Therefore, it is a 
\threerpj{1} query.
The query 
$Q_4$ is a 
 \threerpj{0} query, since it does not have any nonempty pair-join variable
 tuple. 
 The latter two queries are acyclic and hierarchical.
\null\hfill\qedsymbol
\end{example}

The maintenance complexity 
of a \threer  query 
depends on the number 
of its nonempty pair-join variable tuples. 
The relationship 
between the number of such tuples and the simplicity 
of the structure of a query is evident:
\threerpj{3} queries 
are cyclic with the triangle count being 
the simplest query of this class. 
The queries in the class 
\threerpj{2}
are acyclic and it is well-known 
that such queries admit good 
computational behaviour \cite{BeeriFMY83}.
Finally, the class $\threerpj{0} \cup \threerpj{1}$ 
consists of all hierarchical \threer queries 
 (Proposition \ref{prop:hierarchical_3_rel_queries}),
 for which it is already known that they can be maintained 
 with constant update and answer time
 \cite{BerkholzKS17}.
The following theorem summarizes our main results 
on \threer queries.

\begin{theorem}\label{theo:maintain_3rel_3pair-join}
Given a \threer query, a
database $\db$ and $\eps \in [0,1]$, 
\ivme maintains the query under single-tuple updates 
to $\inst{D}$ with the complexities given in Table 
\ref{complexity_table}.

\begin{table}[h]
\begin{center}
\renewcommand{\arraystretch}{1.2}
\begin{tabular}{@{\hskip 0.05in}l@{\hskip 0.2in}l@{\hskip 0.2in}l@{\hskip 0.1in}l@{\hskip 0.05in}}

Query Class & & Maintenance Complexities \\
\bottomrule
\threerpj{3} & preprocessing time & $\bigO{|\inst{D}|^{\frac{3}{2}}}$\\
& update time & amortized $\bigO{|\inst{D}|^{\max\{\eps,1-\eps\}}}$ \\
& answer time & constant \\
& space &  $\bigO{|\inst{D}|^{1 + \min\{\eps,1-\eps\}}}$\\

\midrule
\threerpj{2} & preprocessing time & $\bigO{|\inst{D}|}$\\
& update time & amortized $\bigO{|\inst{D}|^{\max\{\eps,1-\eps\}}}$ \\
& answer time & constant \\
& space &  $\bigO{|\inst{D}|}$\\

\midrule
$\threerpj{0} \cup \threerpj{1}$ & preprocessing time & $\bigO{|\inst{D}|}$\\
& update time & (non-amortized) constant \\
& answer time & constant \\
& space &  $\bigO{|\inst{D}|}$\\
\end{tabular}
\end{center}
\caption{The time and space complexities of maintaining 
\threer queries with \ivme.}
\label{complexity_table}
\end{table}
\end{theorem}


Theorem \ref{theo:maintain_3rel_3pair-join} 
states that \ivme recovers 
the result   that hierarchical 
queries can be maintained with constant-time updates \cite{BerkholzKS17}. 
This result is obviously optimal with respect to update and answer times. 
The following Proposition 
\ref{prop:lower_bound_3rel_count}
implies that for $\eps = \frac{1}{2}$, 
the update and answer times  
for the other query classes
in Theorem \ref{theo:maintain_3rel_3pair-join}
are optimal as well. 

\begin{proposition}
\label{prop:lower_bound_3rel_count}
For any $\gamma > 0$ and database $\db$,
there is no algorithm that incrementally maintains a 
\threerpj{i} query with $i \in \{2,3\}$ 
under single-tuple updates to $\db$ with arbitrary preprocessing time, $\bigO{|\db|^{\frac{1}{2} - \gamma}}$ amortized update time, and $\bigO{|\db|^{1 - \gamma}}$ answer time, unless the \OMv conjecture fails.
\end{proposition}

As a corollary of Theorem \ref{theo:maintain_3rel_3pair-join}  and Proposition \ref{prop:lower_bound_3rel_count} we derive that \ivme achieves for each \threer
query the worst-case optimal update time with constant answer time:

\begin{corollary}[Theorem \ref{theo:maintain_3rel_3pair-join} and Proposition~\ref{prop:lower_bound_3rel_count}]\label{cor:ivme_optimal_3Rel}
Given a \threer query $Q$ and a database $\db$, \ivme incrementally maintains 
$Q$ under single-tuple updates to $\db$ with worst-case optimal update time 
 and constant answer time. 
If $Q$ is a $\threerpj{i}$ query with $i \in \{2,3\}$, the update time is amortized 
$\bigO{|\db|^{\frac{1}{2}}}$ and worst-case optimality is conditioned
on the \OMv conjecture.
If $Q$ is a \threerpj{i} query with $i \in \{0,1\}$, the update time is 
(non-amortized) constant.  
\end{corollary}

In Appendix \ref{sec:meta_strategy} we introduce a general 
\ivme strategy for the incremental maintenance of arbitrary \threer queries. 
In Appendices 
\ref{sec:maintenance_3-rel_3-pair-join} - \ref{sec:maintenance_3-rel_1-pair-join} we
show that this generic strategy  
admits the optimal results 
for the query classes
in Theorem \ref{theo:maintain_3rel_3pair-join}.
In Appendix \ref{sec:maintenance_3-rel_1-pair-join}
we additionally show that 
$\threerpj{0}\cup \threerpj{1}$
consists of all hierarchical \threer  queries.

\subsection{A General \ivme Strategy for \threer Queries}
\label{sec:meta_strategy}
In this section we generalize 
the \ivme strategy for the maintenance of the triangle count 
to arbitrary \threer queries. 
In Sections \ref{sec:strategy} and 
\ref{sec:rebalancing} we have seen
that the \ivme approach 
consists of three basic components: 
relation partitioning, view materialization, and
delta evaluation. 
In Appendices
\ref{sec:3r_rel_part}, \ref{sec:3r_mat_views}, and 
\ref{sec:3r_eval_strats}) we introduce  
the relation partitioining, views and delta evaluation 
strategies for \threer queries. 
In case of \threerpj{i} queries with $i \leq \{0,1\}$,
the evaluation strategies undergo an optimization phase,
which we explain in Appendix \ref{sec:3r_opt_phase}. 
Afterwards, we give 
the definition of an \ivme state 
for \threer queries (Appendix 
\ref{sec:3r_ivme_states}).     
Then, 
we analyze the space and time complexities of the introduced
views and delta evaluation strategies 
(Appendices \ref{sec:3r_space} and \ref{sec:3r_time}).
 In Appendices 
\ref{sec:maintenance_3-rel_3-pair-join} - \ref{sec:maintenance_3-rel_1-pair-join} 
we show that restricting the general \ivme 
strategy to specific \threer classes 
results in specialized strategies admitting the complexity results 
given in Table \ref{complexity_table}.

Note that each relation in a \threer query can be replaced 
by a view that aggregates away the non-join variables.
Such views admit linear space and 
constant update time. Hence, 
without loss of generality, we assume in the following 
that in all considered queries  
 the non-join variable tuples are empty. 
For the rest of Appendix \ref{sec:inc_maint_count_three_rels} 
we fix an $\eps \in [0,1]$.

\subsubsection{Relation Partitioning}
\label{sec:3r_rel_part}
Unlike the triangle case
where the partition of a relation is based on the degrees of values over a {\em single} variable, 
the partition of a relation in a \threer query 
can depend on the 
degrees of value tuples.
The following definition slightly extends Definition
\ref{def:loose_relation_partition}
in this respect.
Note the difference between the partitioning introduced in 
Definition 
\ref{def:loose_relation_partition_sets}
 and the one introduced in the following 
 Definition \ref{def:loose_relation_partition_3R}.
 In the former case, a relation is partitioned on each variable 
 from a variable set. In the latter case, a tuple of variables 
 is treated like a single variable.
 
\begin{definition}[Relation Partition for \threer Queries]\label{def:loose_relation_partition_3R}
Given a relation $K$ over schema $X$, a tuple $\inst{A}$ of variables 
from the schema of $K$, and a threshold $\theta$, 
a partition of $K$ on 
$\inst{A}$ with threshold $\theta$ is a set $\{ K_h, K_l \}$ satisfying the following conditions:
\\[6pt]
\begin{tabular}{@{\hskip 0.5in}rl}
{(union)} & 
$K(\inst{x}) = K_h(\inst{x}) + K_l(\inst{x})$ for $\inst{x} \in \Dom(\inst{X})$ \\[4pt]
{(domain partition)} & $(\pi_{\inst{A}}K_h) \cap (\pi_{\inst{A}}K_l) = \emptyset$ \\[4pt]
{(heavy part)} & for all $\inst{a} \in \pi_{\inst{A}}R_h:\; 
|\sigma_{\inst{A} =\inst{a}} 
R_h| \geq \frac{1}{2}\,\theta$ \\[4pt]
{(light part)} & for all $\inst{a} \in \pi_{\inst{A}}R_l:\; |\sigma_{\inst{A}=\inst{a}} R_l| < 
\frac{3}{2} \,\theta$
\end{tabular}\\[6pt]
The set $\{ K_h, K_l \}$ is called a strict partition of $K$ on $\inst{A}$ with threshold 
$\theta$ if it satisfies the union and 
domain partition conditions and the following strict versions
of the heavy part and light part conditions: 
\\[6pt]
\begin{tabular}{@{\hskip 0.5in}rl}
{(strict heavy part)} & for all $\inst{a} \in \pi_{\inst{A}}K_h:\; |\sigma_{\inst{A}=\inst{a}} K_h| \geq 
\theta$ \\[4pt]
{(strict light part)} & for all $\inst{a} \in \pi_{\inst{A}}K_l:\; |\sigma_{\inst{A}=\inst{a}} K_l| < \theta$
\end{tabular}
\end{definition}

As usual, we call the relations $K_h$ and $K_l$ the \emph{heavy} and the 
\emph{light} part of $K$, respectively. 
\ivme does not necessarily partition
all relations of \threer queries. 
Given a \threer query as in Definition 
\ref{def:threer_queries}, 
we assign each relation 
$K \in \{R,S,T\}$ a 
tuple of partition variables.
We define the partition variable tuples 
of $R$, $S$, and $T$ as   
$\inst{A}_{RT}$, 
$\inst{A}_{RS}$, and 
$\inst{A}_{ST}$, respectively. 
\ivme partitions a
relation $K$
on its partition variable tuple  
if the latter tuple is nonempty. 
Otherwise, relation $K$ is not partitioned. 
In the former case we define
the parts of $K$
as  $\parts{K} = \{K_l,K_h\}$. 
In case $K$ is not partitioned, 
we set $\parts{K} = \{K\}$.
Just like in the triangle case, 
the precise threshold of the partitions
depends on the current database
in an \ivme state, which will be defined 
in Appendix \ref{sec:3r_ivme_states}.

\subsubsection{Materialized Views}
\label{sec:3r_mat_views}
\begin{figure}[t]
\begin{center}
\begin{tikzpicture}[scale=0.35]

  \node at (0, 2) (Q) {$Q_{R_lS_hT_l}$};
  \node at (0, 0) (C) {};
  \node at (-4, -4) (A) {};
  \node at (4, -4) (B) {};
  \node at (0, -4) (D) {};
  \node at (2, -2) (E) {};
  \node at (-2, -2) (F) {};
 \node at (0, -2.5) (X) {};

  \draw
       ($(C)+(0.5,0.5)$) 
        to[out=135,in=45] ($(C) + (-0.5,0.5)$)
        to[out=225,in=45] ($(A) + (-0.5,0.5)$)
         to[out=225,in=135] ($(A) + (-0.5,-0.5)$)
         to[out=315,in=225] ($(A) + (0.5,-0.5)$)
        to[out=50,in=190] ($(F) + (0,-0.9)$)    
        to[out=0,in=180] ($(X) + (0,-0.5)$) 
        to[out=0,in=270] ($(X) + (0.5,0)$)
          to[out=90,in=0] ($(X) + (0,0.5)$)               
          to[out=180,in=225] ($(F) + (1,0.3)$)         
          to[out=45,in=225] ($(C) + (0.5,-0.5)$)  
          to[out=45,in=315] ($(C) + (0.5,0.5)$);

  \draw
       ($(C)+(-0.7,0.7)$) 
        to[out=45,in=135] ($(C) + (0.7,0.7)$)
        to[out=315,in=130] ($(B) + (0.5,0.5)$)
         to[out=315,in=45] ($(B) + (0.5,-0.5)$)
        to[out=225,in=315] ($(B) + (-0.5,-0.5)$)
        to[out=135,in=0] ($(E) + (0,-1.1)$)  
        to[out=180,in=0] ($(X) + (0,-0.7)$)
        to[out=180,in=270] ($(X) + (-0.5,0)$)
      to[out=90,in=180] ($(X) + (0,0.7)$) 
       to[out=0,in=225] ($(E) + (-1,0.3)$)        
       to[out=45,in=315] ($(C) + (-0.7,-0.5)$)  
       to[out=135,in=225] ($(C) + (-0.7,0.7)$);

  \draw
       ($(A)+(-1.1,0)$) 
        to[out=270,in=180] ($(A) + (0,-1.1)$)
         to[out=0,in=180] ($(B) + (0,-1.1)$)
         to[out=0,in=270] ($(B) + (1.1,0)$)
         to[out=90,in=0] ($(B) + (0,0.5)$)
         to[out=180,in=0] ($(D) + (1.2,0.3)$)  
         to[out=180,in=270] ($(D) + (0.7,0.8)$)  
         to[out=90,in=270] ($(X) + (0.7,0)$)
        to[out=90,in=0] ($(X) + (0,0.9)$)
        to[out=180,in=90] ($(X) + (-0.8,0)$)      
         to[out=270,in=90] ($(D) + (-0.7,0.8)$)  
         to[out=270,in=0] ($(D) + (-1.2,0.3)$)  
         to[out=180,in=0] ($(A) + (0,0.6)$)  
       to[out=180,in=90] ($(A)+(-1.1,0)$);

  \draw [color=black, fill=black] (0, 0) circle (.3);
  \draw [color=black, fill=black] (-4, -4) circle (.3);
  \draw [color=black, fill=black] (4, -4) circle (.3);  
  \draw [color=black, fill=black] (0, -2.5) circle (.3);
   
   \node at (-3.9, -1.8) (T) {$T_l$};
  \node at (3.9, -1.8) (S) {$S_h$};
  \node at (0, -5.8) (R) {$R_l$};


\begin{scope}[xshift = 13cm]
 \node at (0, 2) (Q) {$\delta Q_{R_lS_hT_l}$};
  \node at (0, 0) (C) {};
  \node at (-4, -4) (A) {};
  \node at (4, -4) (B) {};
  \node at (0, -4) (D) {};
  \node at (2, -2) (E) {};
  \node at (-2, -2) (F) {};
 \node at (0, -2.5) (X) {};

  \draw
       ($(C)+(0.5,0.5)$) 
        to[out=135,in=45] ($(C) + (-0.5,0.5)$)
        to[out=225,in=45] ($(A) + (-0.5,0.5)$)
         to[out=225,in=135] ($(A) + (-0.5,-0.5)$)
         to[out=315,in=225] ($(A) + (0.5,-0.5)$)
        to[out=50,in=190] ($(F) + (0,-0.9)$)    
        to[out=0,in=180] ($(X) + (0,-0.5)$) 
        to[out=0,in=270] ($(X) + (0.5,0)$)
          to[out=90,in=0] ($(X) + (0,0.5)$)               
          to[out=180,in=225] ($(F) + (1,0.3)$)         
          to[out=45,in=225] ($(C) + (0.5,-0.5)$)  
          to[out=45,in=315] ($(C) + (0.5,0.5)$);

  \draw
       ($(C)+(-0.7,0.7)$) 
        to[out=45,in=135] ($(C) + (0.7,0.7)$)
        to[out=315,in=130] ($(B) + (0.5,0.5)$)
         to[out=315,in=45] ($(B) + (0.5,-0.5)$)
        to[out=225,in=315] ($(B) + (-0.5,-0.5)$)
        to[out=135,in=0] ($(E) + (0,-1.1)$)  
        to[out=180,in=0] ($(X) + (0,-0.7)$)
        to[out=180,in=270] ($(X) + (-0.5,0)$)
      to[out=90,in=180] ($(X) + (0,0.7)$) 
       to[out=0,in=225] ($(E) + (-1,0.3)$)        
       to[out=45,in=315] ($(C) + (-0.7,-0.5)$)  
       to[out=135,in=225] ($(C) + (-0.7,0.7)$);

  \draw[fill= light-gray,fill opacity = 0.4]
       ($(A)+(-1.1,0)$) 
        to[out=270,in=180] ($(A) + (0,-1.1)$)
         to[out=0,in=180] ($(B) + (0,-1.1)$)
         to[out=0,in=270] ($(B) + (1.1,0)$)
         to[out=90,in=0] ($(B) + (0,0.5)$)
         to[out=180,in=0] ($(D) + (1.2,0.3)$)  
         to[out=180,in=270] ($(D) + (0.7,0.8)$)  
         to[out=90,in=270] ($(X) + (0.7,0)$)
        to[out=90,in=0] ($(X) + (0,0.9)$)
        to[out=180,in=90] ($(X) + (-0.8,0)$)      
         to[out=270,in=90] ($(D) + (-0.7,0.8)$)  
         to[out=270,in=0] ($(D) + (-1.2,0.3)$)  
         to[out=180,in=0] ($(A) + (0,0.6)$)  
       to[out=180,in=90] ($(A)+(-1.1,0)$);

  \draw [color=black, fill=black] (0, 0) circle (.3);       
  \draw [color=black, fill=black] (0, 0) circle (.3);
  \draw [color=black, fill=black] (-4, -4) circle (.3);
  \draw [color=black, fill=black] (4, -4) circle (.3);
   \draw [color=black, fill=black] (0, -2.5) circle (.3);
   
   \node at (-3.9, -1.8) (T) {$T_l$};
  \node at (3.9, -1.8) (S) {$S_h$};
  \node at (0, -5.8) (R) {$\delta R_l$}; 
\end{scope}

\begin{scope}[xshift = 26cm]
   \node at (0, 2.2) (Q) {replacing $S_h$ and $T_l$ };
   \node at (0, 1) (Q) {by a  view};
  \node at (-4, -4) (A) {};
  \node at (4, -4) (B) {};
  \node at (0, -4) (D) {};
  \node at (2, -2) (E) {};
  \node at (-2, -2) (F) {};
 \node at (0, -2.5) (X) {};

  \draw
    ($(A) + (-0.7,0)$)
         to[out=270,in=180] ($(A) + (0,-0.7)$)
         to[out=0,in=270] ($(A) + (0.7,0)$)
        to[out=90,in=180] ($(F) + (0,-0.9)$)    
        to[out=0,in=180] ($(X) + (0,-0.5)$) 
        to[out=0,in=180] ($(E) + (0,-0.9)$)
         to[out=0,in=90] ($(B) + (-0.7,0)$)
          to[out=270,in=180] ($(B) + (0,-0.7)$)
          to[out=0,in=270] ($(B) + (0.7,0)$)
         to[out=90,in=270] ($(B) + (0.7,2)$)  
         to[out=90,in=0] ($(B) + (0,3)$)                               
         to[out=180,in=0] ($(A) + (0,3)$)                  
         to[out=180,in=90] ($(A) + (-0.7,2)$)                           
         to[out=270,in=90] ($(A) + (-0.7,0)$);

  \draw[fill= light-gray,fill opacity = 0.4]
       ($(A)+(-1.1,0)$) 
        to[out=270,in=180] ($(A) + (0,-1.1)$)
         to[out=0,in=180] ($(B) + (0,-1.1)$)
         to[out=0,in=270] ($(B) + (1.1,0)$)
         to[out=90,in=0] ($(B) + (0,0.5)$)
         to[out=180,in=0] ($(D) + (1.2,0.3)$)  
         to[out=180,in=270] ($(D) + (0.7,0.8)$)  
         to[out=90,in=270] ($(X) + (0.7,0)$)
        to[out=90,in=0] ($(X) + (0,0.9)$)
        to[out=180,in=90] ($(X) + (-0.8,0)$)      
         to[out=270,in=90] ($(D) + (-0.7,0.8)$)  
         to[out=270,in=0] ($(D) + (-1.2,0.3)$)  
         to[out=180,in=0] ($(A) + (0,0.6)$)  
       to[out=180,in=90] ($(A)+(-1.1,0)$);

  \draw [color=black, fill=black] (-4, -4) circle (.3);
  \draw [color=black, fill=black] (4, -4) circle (.3);
   \draw [color=black, fill=black] (0, -2.5) circle (.3);
   
   \node at (0, -0.3) (T) {$V_{S_hT_l}$};
  \node at (0, -5.8) (R) {$\delta R_l$}; 
\end{scope}

\end{tikzpicture}
 \end{center}
 \caption{
 (left) The hypergraph of 
a skew-aware view 
$Q_{R_lS_hT_l}$. (center) 
The hypergraph the delta view 
$\delta Q_{R_lS_hT_l}$ under an 
update 
$\delta R_l$. 
(right) 
The hypergraph the delta view
obtained from
$\delta Q_{R_lS_hT_l}$ by replacing relations $S_h$ and $T_l$
by a view.
} 
  \label{fig:example_single_relation_view}
\end{figure}
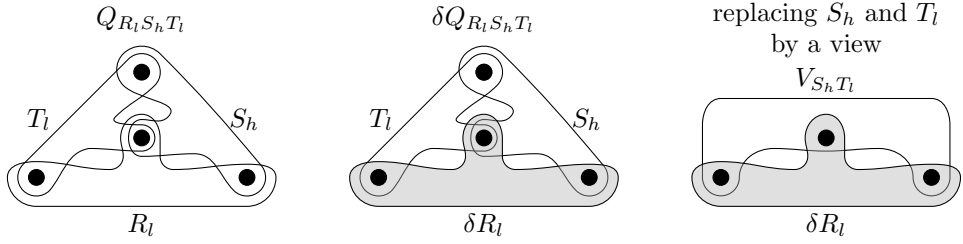
\begin{figure}[hbtp]
  \begin{center}
    \renewcommand{\arraystretch}{1.2}  
    \begin{tabular}{@{\hskip 0.0in}l@{\hskip 0.3in}l@{\hskip 0.0in}}
      \toprule
      Materialized View & Space\\    
      \midrule \\[-0.5cm]
   
   $Q()=\sum\limits_{Q' \in SAV}Q'()$ &$\bigO{1}$\\[0.3cm]    
          \hline \\[-0.4cm]
          

          if $\inst{A}_{ST} \neq \emptyset$, $\inst{A}_{RS} \neq \emptyset$,
          and $\inst{A}_{RT} \neq \emptyset$, 
          materialize 
    &  \\
  
$V_{S_hT_l}(\inst{a}_{RT},\inst{a}_{RS},\inst{a}_{RST}) = 
\sum 
S_{h} (\inst{a}_{RS}, \inst{a}_{ST}, \inst{a}_{RST}) \cdot 
T_{l}(\inst{a}_{ST}, \inst{a}_{RT}, \inst{a}_{RST})$ 
  &  $\bigO{|\inst{D}|^{1 + \min\{\eps, 1- \eps\}}}$\\[0.3cm]
  
         if $\inst{A}_{ST} \neq \emptyset$, $\inst{A}_{RS} \neq \emptyset$, 
         and $\inst{A}_{RT} = \emptyset$, 
         materialize 
    &  \\
  
$V_{S_hT_l}(\inst{a}_{RS},\inst{a}_{RST}) = 
\sum 
S_{h} (\inst{a}_{RS}, \inst{a}_{ST}, \inst{a}_{RST}) \cdot 
T_{l}(\inst{a}_{ST}, \inst{a}_{RST})$ 
  &  $\bigO{|\inst{D}|}$\\[0.3cm]

      if $\inst{A}_{ST} \neq \emptyset$, $\inst{A}_{RS} = \emptyset$,
      and $\inst{A}_{RT} \neq  \emptyset$ 
      materialize 
      & \\
     
$V_{ST_l}(\inst{a}_{RT},\inst{a}_{RST}) = 
\sum
S (\inst{a}_{ST}, \inst{a}_{RST}) \cdot 
T_{l}(\inst{a}_{ST}, \inst{a}_{RT}, \inst{a}_{RST})$ 
      &    $\bigO{|\inst{D}|}$\\[0.3cm]

if $\inst{A}_{ST} \neq \emptyset$, 
$\inst{A}_{RS} = \emptyset$, and $\inst{A}_{RT} = \emptyset$, 
materialize 
&
\\

$V_{ST_{t}}(\inst{a}_{RST}) = 
\sum 
S (\inst{a}_{ST}, \inst{a}_{RST}) \cdot 
T_{t}(\inst{a}_{ST}, \inst{a}_{RST})$ 
for each  $T_t \in \{T_l,T_h\}$     &    $\bigO{|\inst{D}|}$\\[0.3cm]     
 
   \hline \\[-0.4cm]


          if $\inst{A}_{RT} \neq \emptyset$, $\inst{A}_{ST} \neq \emptyset$,
          and $\inst{A}_{RS} \neq \emptyset$, 
          materialize 
    &  \\
  
$V_{R_lT_h}(\inst{a}_{RS},\inst{a}_{ST},\inst{a}_{RST} ) = 
\sum 
T_{h} (\inst{a}_{ST}, \inst{a}_{RT}, \inst{a}_{RST}) \cdot 
R_{l}(\inst{a}_{RT}, \inst{a}_{RS}, \inst{a}_{RST})$ 
  &  $\bigO{|\inst{D}|^{1 + \min\{\eps, 1- \eps\}}}$ \\[0.3cm]
  
      if $\inst{A}_{RT} \neq \emptyset$, 
      $\inst{A}_{ST} \neq \emptyset$, and
      $\inst{A}_{RS} = \emptyset$, 
      materialize 
    &  \\
  
$V_{R_lT_h}(\inst{a}_{ST},\inst{a}_{RST} ) = 
\sum
T_{h} (\inst{a}_{ST}, \inst{a}_{RT}, \inst{a}_{RST}) \cdot 
R_{l}(\inst{a}_{RT}, \inst{a}_{RST})$ 
  &  $\bigO{|\inst{D}|}$ \\[0.3cm]

      if $\inst{A}_{RT} \neq \emptyset$, 
      $\inst{A}_{ST} = \emptyset$, and 
      $\inst{A}_{RS} \neq \emptyset$, 
      materialize 
      & \\
     
$V_{R_lT}(\inst{a}_{RS},\inst{a}_{RST}) = 
\sum 
T (\inst{a}_{RT}, \inst{a}_{RST}) \cdot 
R_{l}(\inst{a}_{RT}, \inst{a}_{RS}, \inst{a}_{RST})$ 
      &    $\bigO{|\inst{D}|}$ \\[0.3cm]

if $\inst{A}_{RT} \neq \emptyset$, 
$\inst{A}_{ST} = \emptyset$, and $\inst{A}_{RS} = \emptyset$, 
materialize 
&
\\

$V_{R_rT}(\inst{a}_{RST}) = 
\sum
T (\inst{a}_{RT}, \inst{a}_{RST}) \cdot 
R_{r}(\inst{a}_{RT}, \inst{a}_{RST})$ 
 for each  $R_r \in \{R_l,R_h\}$
     &    $\bigO{|\inst{D}|}$ \\[0.3cm]     
   \hline \\[-0.4cm]
   
    
          if $\inst{A}_{RS} \neq \emptyset$, $\inst{A}_{RT} \neq \emptyset$,
          and $\inst{A}_{ST} \neq \emptyset$, 
          materialize 
    &  \\
  
$V_{R_hS_l}(\inst{a}_{ST},\inst{a}_{RT},\inst{a}_{RST} ) = 
\sum
R_{h} (\inst{a}_{RT}, \inst{a}_{RS}, \inst{a}_{RST}) \cdot 
S_{l}(\inst{a}_{RS}, \inst{a}_{ST}, \inst{a}_{RST})$ 
  &  $\bigO{|\inst{D}|^{1 + \min\{\eps, 1- \eps\}}}$ \\[0.3cm]

          if $\inst{A}_{RS} \neq \emptyset$, $\inst{A}_{RT} \neq \emptyset$,
          and $\inst{A}_{ST} = \emptyset$, 
          materialize 
    &  \\
  
$V_{R_hS_l}(\inst{a}_{RT},\inst{a}_{RST} ) = 
\sum
R_{h} (\inst{a}_{RT}, \inst{a}_{RS}, \inst{a}_{RST}) \cdot 
S_{l}(\inst{a}_{RS},  \inst{a}_{RST})$ 
  &  $\bigO{|\inst{D}|}$ \\[0.3cm]

          if $\inst{A}_{RS} \neq \emptyset$, $\inst{A}_{RT} = \emptyset$,
          and $\inst{A}_{ST} \neq \emptyset$, 
          materialize  
          & \\
     
$V_{RS_l}(\inst{a}_{ST},\inst{a}_{RST}) = 
\sum 
R (\inst{a}_{RS}, \inst{a}_{RST}) \cdot 
S_{l}(\inst{a}_{RS}, \inst{a}_{ST}, \inst{a}_{RST})$ 
      &    $\bigO{|\inst{D}|}$ \\[0.3cm]

         if $\inst{A}_{RS} \neq \emptyset$, $\inst{A}_{RT} = \emptyset$,
          and $\inst{A}_{ST} = \emptyset$, 
          materialize  
          & \\

$V_{RS_{s}}(\inst{a}_{RST}) = 
\sum
R (\inst{a}_{RS}, \inst{a}_{RST}) \cdot 
S_{s}(\inst{a}_{RS}, \inst{a}_{RST})$ 
for each  $S_s \in \{S_l,S_h\}$     &    $\bigO{|\inst{D}|}$ \\[0.3cm]

      \bottomrule   
    \end{tabular}
  \end{center}
  \caption{The definitions and space complexities 
  of all views materialized 
  by \ivme for the maintenance of a
  \threer  query. 
 $\SAV$ is the set of all skew-aware views. 
    The views are defined over a database $\inst{D}$ partitioned 
     for a fixed $\eps \in [0,1]$.} 
  \label{fig:meta_view_definitions}
\end{figure}
We define skew-aware views similarly to the triangle case. 
Given $R' \in \parts{R}$, $S' \in \parts{S}$, and $T' \in \parts{T}$,
a skew-aware view
$Q_{R'S'T'}$ 
results from a \threer query
$Q$ by replacing the relations $R$, $S$, and $T$
by $R'$, $S'$, and $T'$, respectively.
Let $\SAV$ be the set of all skew-aware views 
of $Q$. 
The query $Q$ can be written 
as a union of skew-aware views: 
$Q() = \sum_{Q' \in \SAV} Q'()$.  
Thanks to the delta rules in Section 
\ref{sec:preliminaries}, 
 the delta of $Q$ under 
a single-tuple update is the sum 
of the deltas of the skew-aware views.
In addition to the query $Q$, which returns the final count result, \ivme 
maintains auxiliary views that help 
to compute the deltas of skew-aware views in sublinear time.
These views are simple generalizations of the views 
used for the triangle count. 
Figure 
\ref{fig:meta_view_definitions}
lists all auxiliary views maintained by
\ivme. The index of a  view identifier 
 indicates 
the relation parts the view is composed of. 
A view is identified by the set of relation parts occurring in the view 
and its tuple of free variables. 
For instance, the view $V_{S_hT_l}(\inst{a}_{RT}, 
\inst{a}_{RS}, \inst{a}_{RST})$
is composed of  $S_h$ and $T_l$
and its tuple of free variables 
is $(\inst{A}_{RT}, \inst{A}_{RS}, \inst{A}_{RST})$. 
In Figure \ref{fig:meta_view_definitions},  each set of four auxiliary views between two lines
builds a group. Observe that the views in the 
second group result from the first group by replacing 
$R$, $S$, and $T$ by $S$, $T$, and $R$, respectively. 
The third group is obtained from the second one by 
applying the same replacement rules.

We give an intuitive explanation 
in which cases auxiliary views 
are useful.  
Assume that each relation is strictly partitioned 
with threshold $|\inst{D}|^{\eps}$
where $\inst{D} = \{R,S,T\}$. 
We derive two 
bounds on the partition $\{R_l,R_h\}$
similar to the triangle count case:
given any tuple $\inst{a}_{RT}$, 
the number of tuples $(\inst{a}_{RS}, \inst{a}_{RST})$  
paired with $\inst{a}_{RT}$ in $R_l$
 is less than $|\inst{D}|^{\eps}$; the number of tuples
in $\pi_{\inst{A}_{RT}}R_h$ is at most $\frac{|\inst{D}|}{|\inst{D}|^{\eps}} = |\inst{D}|^{1-\eps}$. 
The bounds for the partitions $\{S_l,S_h\}$
and $\{T_l,T_h\}$ are analogous. 
Consider the skew-aware view
$Q_{R_lS_lT_h}() = 
\textstyle\sum R_l(\inst{a}_{RT},\inst{a}_{RS},\inst{a}_{RST}) \ztimes
S_l(\inst{a}_{RS},\inst{a}_{ST},\inst{a}_{RST}) \ztimes 
T_h(\inst{a}_{ST},\inst{a}_{RT},\inst{a}_{RST})$.
Assume that we would like 
to compute the delta 
$\delta Q_{R_lS_lT_h}() = 
\delta R_l(\bdeltaA_{RT},\bdeltaA_{RS},\bdeltaA_{RST}) \ztimes
\textstyle\sum 
S_l(\bdeltaA_{RS},\inst{a}_{ST},\bdeltaA_{RST}) \ztimes 
T_h(\inst{a}_{ST},\bdeltaA_{RT},\bdeltaA_{RST})$
under an update 
$\delta R_l = \{(\bdeltaA_{RT},\bdeltaA_{RS},\bdeltaA_{RST}) \mapsto \p\}$.
The sum in the delta view is defined 
 over all tuples 
 $\inst{a}_{ST}$ paired 
 with $(\bdeltaA_{RS}, \bdeltaA_{RST})$ in $S_l$ 
 and with 
$(\bdeltaA_{RT},\bdeltaA_{RST})$ in $T_h$.
Since $S_l$ is light, the number 
of $\inst{a}_{ST}$-tuples 
paired with $(\bdeltaA_{RS},\bdeltaA_{RST})$ in $S_l$ 
is less than $|\inst{D}|^{\eps}$. 
Furthermore, 
since $T_h$ is heavy, the number 
of distinct tuples over $\inst{A}_{ST}$ in 
$T_h$
can be at most $|\inst{D}|^{1-\eps}$.
Therefore, the sum in 
 $\delta Q_{R_lS_lT_h}$ 
 iterates over at most 
 $|\inst{D}|^{\min\{\eps, 1-\eps\}}$
 relevant 
$\inst{A}_{ST}$-tuples.
In Appendix \ref{sec:3r_eval_strats} 
it will be shown 
that  $\delta Q_{R_lS_lT_h}$ can indeed be computed in
time $\bigO{|\inst{D}|^{\min\{\eps, 1-\eps\}}}$.

One delta query where 
sublinear delta computation requires 
the exploitation of additional auxiliary views is
$\delta Q_{R_lS_hT_l}() = 
\delta R_l(\bdeltaA_{RT},\bdeltaA_{RS},\bdeltaA_{RST}) \ztimes
\textstyle\sum
S_h(\bdeltaA_{RS},\inst{a}_{ST},\bdeltaA_{RST}) \ztimes 
T_l(\inst{a}_{ST},\bdeltaA_{RT},\bdeltaA_{RST})$.
This case is demonstrated  in
Figure \ref{fig:example_single_relation_view}.
In this case, there is no sublinear bound on the number 
of $\inst{A}_{ST}$-tuples that are paired with
$(\bdeltaA_{RS},\bdeltaA_{RST})$ in $S_h$ 
 and with 
$(\bdeltaA_{RT}, \bdeltaA_{RST})$ in 
$T_l$.
Therefore, \ivme materializes 
the auxiliary view    
$V_{S_hT_l}(\inst{a}_{RT},\inst{a}_{RS},\inst{a}_{RST}) = 
\textstyle\sum 
S_{h} (\inst{a}_{RS}, \inst{a}_{ST}, \inst{a}_{RST}) \ztimes 
T_{l}(\inst{a}_{ST}, \inst{a}_{RT}, \inst{a}_{RST})$ 
over the update-independent part of the delta query. 
Then, the computation of 
$\delta Q_{R_lS_hT_l}$ amounts to a constant-time 
look-up in that view. The second 
row of Figure \ref{fig:hypergraphs_evaluation_strategies} 
visualizes the hypergraphs of all delta skew-aware views
for which \ivme uses auxiliary views 
to enable delta computation under an update $\delta R_r$
in sublinear time.

\subsubsection{Delta Evaluation Strategies}
\label{sec:3r_eval_strats}
For each skew-aware or auxiliary view 
given in Figure \ref{fig:meta_view_definitions}
and a single-tuple update to a relation occurring in the 
view, we define an evaluation strategy to compute 
the delta view. We express an evaluation 
strategy as a count query where the order of the relations
from left to right
defines the nesting structure of the loops iterating over the relations.
Figure \ref{fig:best_update_strategy} gives for each view including
a relation (part) $R_r \in \parts{R}$, 
the evaluation 
strategy to compute the delta of the view under an update 
$\delta R_r = \{(\bdeltaA_{RT}, \bdeltaA_{RS}, \bdeltaA_{RST}) \mapsto 
\p\}$.
Figure \ref{fig:hypergraphs_evaluation_strategies} categorizes 
the hypergraphs of delta skew-aware views according 
to the type of evaluation  strategies. 
If we replace the relation symbols $R$, $S$ and $T$ in Figure 
\ref{fig:best_update_strategy} by $S$, $T$ and $R$, respectively, 
we obtain the evaluation strategies  under updates to relation $S$.
Replacing $R$, $S$ and $T$  by $T$, $R$ and $S$, respectively, 
results in the evaluation strategies under updates  to $T$.

\begin{figure}[t]
  \begin{center}
    \renewcommand{\arraystretch}{1.3}  
    \begin{tabular}{@{\hskip 0.0in}l@{\hskip 0.07in}l
    @{\hskip 0.07in}l@{\hskip 0.0in}}
      \toprule
      Delta evaluation strategy & Note &Time\\    
      \midrule

      $\delta Q_{R_rS_lT_l} ()= \delta R_r(\bdeltaA_{RT},\bdeltaA_{RS},\bdeltaA_{RST}) \cdot
      \textstyle\sum S_l(\bdeltaA_{RS},\inst{a}_{ST},\bdeltaA_{RST}) \cdot 
      T_l(\inst{a}_{ST},\bdeltaA_{RT},\bdeltaA_{RST})$  
      &  &
      $\bigO{|\inst{D}|^{\eps}}$ \\
 

      $\delta Q_{R_rS_lT_h} ()=      
       \delta R_r(\bdeltaA_{RT},\bdeltaA_{RS},\bdeltaA_{RST}) \cdot \textstyle\sum S_l(\bdeltaA_{RS},\inst{a}_{ST},\bdeltaA_{RST}) \cdot 
      T_h(\inst{a}_{ST},\bdeltaA_{RT},\bdeltaA_{RST})$ & 
      $\eps \leq \frac{1}{2}$ &$\bigO{|\inst{D}|^{\eps}}$ \\


$\delta Q_{R_rS_hT_h} ()= \delta R_r(\bdeltaA_{RT},\bdeltaA_{RS},\bdeltaA_{RST}) \cdot
\textstyle\sum T_h(\inst{a}_{ST},\bdeltaA_{RT},\bdeltaA_{RST}) \cdot 
      S_h(\bdeltaA_{RS},\inst{a}_{ST},\bdeltaA_{RST})$ & &$\bigO{|\inst{D}|^{1-\eps}}$ 
      \\


    $\delta Q_{R_rS_lT_h} ()= \delta R_r(\bdeltaA_{RT},\bdeltaA_{RS},\bdeltaA_{RST}) \cdot
       \textstyle\sum T_h(\inst{a}_{ST},\bdeltaA_{RT},\bdeltaA_{RST}) \cdot 
      S_l(\bdeltaA_{RS},\inst{a}_{ST},\bdeltaA_{RST})$ & $\eps > \frac{1}{2}$ &
      $\bigO{|\inst{D}|^{1-\eps}}$ \\
      
    $\delta Q_{R_rST_h} ()= \delta R_r(\bdeltaA_{RT},\bdeltaA_{RST}) \cdot \textstyle\sum T_h(\inst{a}_{ST},\bdeltaA_{RT},\bdeltaA_{RST}) \cdot 
      S(\inst{a}_{ST},\bdeltaA_{RST})$ & $R_r \in \{R_l,R_h\}$  & $\bigO{|\inst{D}|^{1-\eps}}$ \\

      $\delta Q_{R_rS_sT} ()= 
      \delta R_r(\bdeltaA_{RT},\bdeltaA_{RS},\bdeltaA_{RST}) \cdot
      S_s(\bdeltaA_{RS},\bdeltaA_{RST}) \cdot 
      T(\bdeltaA_{RT},\bdeltaA_{RST})$ & & $\bigO{1}$ \\


    $\delta Q_{R_rS_hT_l} () = \delta R_r(\bdeltaA_{RT},\bdeltaA_{RS},\bdeltaA_{RST}) \cdot$ 
   $V_{S_hT_l}(\bdeltaA_{RT},\bdeltaA_{RS},\bdeltaA_{RST})$ & & $\bigO{1}$ \\

    $\delta Q_{R_rST_l} ()= 
    \delta R_r(\bdeltaA_{RT},\bdeltaA_{RST}) \cdot$ 
   $V_{ST_l}(\bdeltaA_{RT},\bdeltaA_{RST})$ & & $\bigO{1}$ \\

    $\delta Q_{RST_h} ()= \delta R(\bdeltaA_{RST}) \cdot$ 
   $V_{ST_h}(\bdeltaA_{RST})$ & & $\bigO{1}$ \\
      \hline \\[-0.5cm]          
          

$\delta V_{R_hS_l}(\inst{a}_{ST},\bdeltaA_{RT},\bdeltaA_{RST}) =
\delta R_{h} (\bdeltaA_{RT}, \bdeltaA_{RS}, \bdeltaA_{RST}) \cdot  
 S_{l}(\bdeltaA_{RS}, \inst{a}_{ST}, \bdeltaA_{RST})$  & &
      $\bigO{|\inst{D}|^{\eps}}$\\     
      
  
  $\delta V_{R_hS_l}(\bdeltaA_{RT},\bdeltaA_{RST}) = 
  \delta R_{h} (\bdeltaA_{RT}, \bdeltaA_{RS}, \bdeltaA_{RST}) \cdot  
 S_{l}(\bdeltaA_{RS}, \bdeltaA_{RST})$  & & 
      $\bigO{1}$\\    
  
      
$\delta V_{RS_l}(\inst{a}_{ST},\bdeltaA_{RST}) = \delta R (\bdeltaA_{RS}, \bdeltaA_{RST}) \cdot 
S_{l}(\bdeltaA_{RS}, \inst{a}_{ST}, \bdeltaA_{RST})$  & &
      $\bigO{|\inst{D}|^{\eps}}$\\     


$\delta V_{RS_s}(\bdeltaA_{RST}) = 
\delta R (\bdeltaA_{RS}, \bdeltaA_{RST}) \cdot 
S_{s}(\bdeltaA_{RS}, \bdeltaA_{RST})$  & &
      $\bigO{1}$\\     
      

       $\delta V_{R_lT_h}(\bdeltaA_{RS},\inst{a}_{ST},\bdeltaA_{RST}) =
       \delta R_{l}(\bdeltaA_{RT}, \bdeltaA_{RS}, \bdeltaA_{RST}) \cdot 
T_{h} (\inst{a}_{ST}, \bdeltaA_{RT}, \bdeltaA_{RST})$ & &
      $\bigO{|\inst{D}|^{1-\eps}}$\\     
 
 
        $\delta V_{R_lT_h}(\inst{a}_{ST},\bdeltaA_{RST}) =
        \delta R_{l}(\bdeltaA_{RT}, \bdeltaA_{RST}) \cdot 
T_{h} (\inst{a}_{ST}, \bdeltaA_{RT}, \bdeltaA_{RST})$ & &
      $\bigO{|\inst{D}|^{1-\eps}}$\\     
 

       $\delta V_{R_lT}(\bdeltaA_{RS},\bdeltaA_{RST}) = 
       \delta R_{l}(\bdeltaA_{RT}, \bdeltaA_{RS}, \bdeltaA_{RST}) \cdot 
T (\bdeltaA_{RT}, \bdeltaA_{RST})$ & &
      $\bigO{1}$\\
      
       $\delta V_{R_rT}(\bdeltaA_{RST}) = \delta R_{r}(\bdeltaA_{RT}, \bdeltaA_{RST}) \cdot 
T (\bdeltaA_{RT}, \bdeltaA_{RST})$ & &
      $\bigO{1}$\\ 
      
      \bottomrule 
    \end{tabular}
  \end{center}
   \caption{The evaluation strategies and time complexities to
  compute the deltas of the views in 
  Figure  \ref{fig:meta_view_definitions} under an update 
  $\delta R_r = \{(\inst{a}_{RT}, \inst{a}_{RS}, \inst{a}_{RST})\mapsto \p\}$. 
 It holds 
  $R_r \in \parts{R}$ and $S_s \in \parts{S}$.} 
  \label{fig:best_update_strategy}
\end{figure}

We give the interpretation of some exemplary 
evaluation strategies from Figure \ref{fig:best_update_strategy}.
\begin{itemize} 
\item The evaluation strategy to compute the delta 
of $Q_{RrS_lT_l}$ under  
$\delta R_r = \{(\bdeltaA_{RT}, \bdeltaA_{RS}, \bdeltaA_{RST}) \mapsto \p\}$
with $R_r \in {\parts{R}}$ 
is $\delta R_r(\bdeltaA_{RT},\bdeltaA_{RS},\bdeltaA_{RST}) \ztimes 
\textstyle\sum S_l(\bdeltaA_{RS},\inst{a}_{ST},\bdeltaA_{RST}) \ztimes 
      T_l(\inst{a}_{ST},\bdeltaA_{RT},\bdeltaA_{RST})$. 
      In order to compute the sum, the strategy dictates to 
iterate over all    
      $\inst{A}_{ST}$-tuples $\inst{a}_{ST}$ paired with 
      $(\bdeltaA_{RS}, \bdeltaA_{RST})$
      in $S_l$ and for each such tuple, to look up the multiplicity of
      $(\inst{a}_{ST},\bdeltaA_{RT},\bdeltaA_{RST})$ in $T_l$.
We give the pseudocode of the strategy:
 
\begin{center}
\renewcommand{\arraystretch}{1.0}
\begin{tabular}{l}
\toprule 
{\textsc{ComputeDelta}}($Q_{R_rS_lT_l}$,  
$\delta R_r = \{(\bdeltaA_{RT},\bdeltaA_{RS},\bdeltaA_{RST}) \mapsto \p\}$)\\
\midrule
$sum = 0$\\
\FOREACH\STAB $\inst{a}_{ST}$ such that $(\bdeltaA_{RS},\inst{a}_{ST}, 
\bdeltaA_{RST}) \in S_l$ \\
\TAB  $sum \mathrel{+{=}} S_l(\bdeltaA_{RS},\inst{a}_{ST},\bdeltaA_{RST}) \ztimes 
     T_l(\inst{a}_{ST},\bdeltaA_{RT},\bdeltaA_{RST})$ \\
$\delta Q_{R_rS_lT_l}() = \delta R_r(\bdeltaA_{RT},\bdeltaA_{RS},
\bdeltaA_{RST}) \ztimes sum$\\
\OUTPUT $\delta Q_{R_rS_lT_l}$\\
\bottomrule
\end{tabular}
\end{center}
   
\item In case of $Q_{R_rS_hT_l}$, the evaluation strategy is 
$\delta R_r(\bdeltaA_{RT},\bdeltaA_{RS},\bdeltaA_{RST}) \ztimes 
V_{S_hT_l}(\bdeltaA_{RT},\inst{a}_{RS},\bdeltaA_{RST})$, which amounts to a 
look-up of the multiplicity of $(\bdeltaA_{RT},\inst{a}_{RS},\bdeltaA_{RST})$ in 
$V_{S_hT_l}$: 

\begin{center}
\renewcommand{\arraystretch}{1.0}
\begin{tabular}{l}
\toprule 
{\textsc{ComputeDelta}}($Q_{R_rS_hT_l}$,  
$\delta R_r = \{(\bdeltaA_{RT},\bdeltaA_{RS},\bdeltaA_{RST}) \mapsto \p\}$)\\
\midrule
$\delta Q_{R_rS_hT_l}() =  \delta R_r(\bdeltaA_{RT},\bdeltaA_{RS},
\bdeltaA_{RST}) \ztimes 
V_{S_hT_l}(\bdeltaA_{RT},\bdeltaA_{RS},\bdeltaA_{RST})$ \\
\OUTPUT $\delta Q_{R_rS_hT_l}$ \\
\bottomrule
\end{tabular}
\end{center}

\item The evaluation strategy for the delta of 
$V_{R_hS_l}(\inst{a}_{ST},\inst{a}_{RT},\inst{a}_{RST})$
under the update $\delta R_{h} = 
\{(\bdeltaA_{RT}, \bdeltaA_{RS}, \bdeltaA_{RST}) \mapsto \p \}$ 
is given by  $\delta R_{h} (\bdeltaA_{RT}, \bdeltaA_{RS}, \bdeltaA_{RST})$ 
$\ztimes$  
$S_{l}(\bdeltaA_{RS}, \inst{a}_{ST}, \bdeltaA_{RST})$.
This means that we iterate over each $\inst{A}_{ST}$-tuple 
$\inst{a}_{ST}$ paired with 
$(\bdeltaA_{RT}, \bdeltaA_{RST})$ in $S_l$ and 
for each such tuple, we add 
$\delta R_{h} (\bdeltaA_{RT}, \bdeltaA_{RS}, \bdeltaA_{RST})$ 
$\ztimes$  
$S_{l}(\bdeltaA_{RS}, \inst{a}_{ST}, \bdeltaA_{RST})$ to 
the multiplicity of 
$\delta V_{R_hS_l}(\inst{a}_{ST},\bdeltaA_{RT},\bdeltaA_{RST})$.

\begin{center}
\renewcommand{\arraystretch}{1.0}
\begin{tabular}{l}
\toprule 
{\textsc{ComputeDelta}}($V_{R_hS_l}(\inst{a}_{ST},\inst{a}_{RT},\inst{a}_{RST})$,  
$\delta R_h = \{(\bdeltaA_{RT},\bdeltaA_{RS},\bdeltaA_{RST}) \mapsto \p\}$)\\
\midrule
$\delta V_{R_hS_l} = \emptyset$ \\
\FOREACH\STAB $\inst{a}_{ST}$ such that $(\bdeltaA_{RS},\inst{a}_{ST}, 
\bdeltaA_{RST}) \in S_l$ \\
\TAB  $\delta V_{R_hS_l}(\inst{a}_{ST},\bdeltaA_{RT},\bdeltaA_{RST})  \mathrel{+{=}} 
\delta R_h(\bdeltaA_{RT},\bdeltaA_{RS},\bdeltaA_{RST}) \ztimes 
S_l(\bdeltaA_{RS},\bdeltaA_{ST},\bdeltaA_{RST})$ \\
\OUTPUT $\delta V_{R_hS_l}$ \\
\bottomrule
\end{tabular}
\end{center}
\end{itemize}

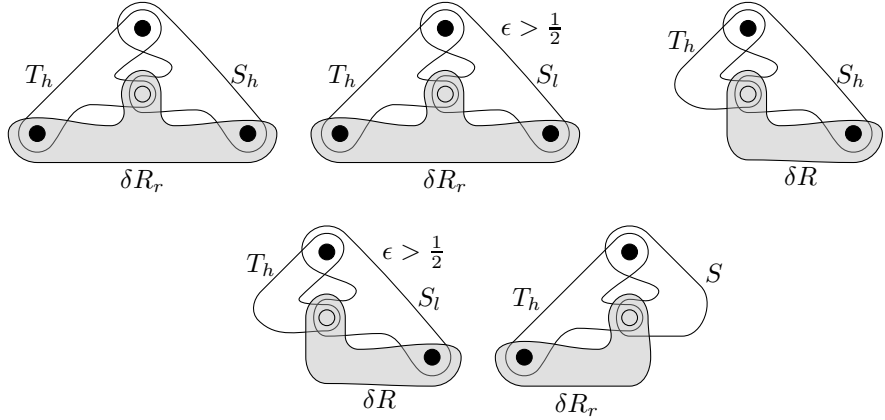
\begin{figure}[hbtp]
\begin{center}
\begin{tikzpicture}[scale=0.35]


\node at (19, 2.3) (strategy) {Delta skew-aware 
views with delta evaluation strategy 
$\delta R_r  \cdot \textstyle\sum S_s \cdot T_t$};
  \node at (0, 0) (C) {};
  \node at (-4, -4) (A) {};
  \node at (4, -4) (B) {};
  \node at (0, -4) (D) {};
  \node at (2, -2) (E) {};
  \node at (-2, -2) (F) {};
 \node at (0, -2.5) (X) {};  
       
         \draw
       ($(C)+(0.5,0.5)$) 
        to[out=135,in=45] ($(C) + (-0.5,0.5)$)
        to[out=225,in=45] ($(A) + (-0.5,0.5)$)
         to[out=225,in=135] ($(A) + (-0.5,-0.5)$)
         to[out=315,in=225] ($(A) + (0.5,-0.5)$)
        to[out=50,in=190] ($(F) + (0,-0.9)$)    
        to[out=0,in=180] ($(X) + (0,-0.5)$) 
        to[out=0,in=270] ($(X) + (0.5,0)$)
          to[out=90,in=0] ($(X) + (0,0.5)$)               
          to[out=180,in=225] ($(F) + (1,0.3)$)         
          to[out=45,in=225] ($(C) + (0.5,-0.5)$)  
          to[out=45,in=315] ($(C) + (0.5,0.5)$);

  \draw
       ($(C)+(-0.7,0.7)$) 
        to[out=45,in=135] ($(C) + (0.7,0.7)$)
        to[out=315,in=130] ($(B) + (0.5,0.5)$)
         to[out=315,in=45] ($(B) + (0.5,-0.5)$)
        to[out=225,in=315] ($(B) + (-0.5,-0.5)$)
        to[out=135,in=0] ($(E) + (0,-1.1)$)  
        to[out=180,in=0] ($(X) + (0,-0.7)$)
        to[out=180,in=270] ($(X) + (-0.5,0)$)
      to[out=90,in=180] ($(X) + (0,0.7)$) 
       to[out=0,in=225] ($(E) + (-1,0.3)$)        
       to[out=45,in=315] ($(C) + (-0.7,-0.5)$)  
       to[out=135,in=225] ($(C) + (-0.7,0.7)$);

  \draw[fill= light-gray,fill opacity = 0.4]
       ($(A)+(-1.1,0)$) 
        to[out=270,in=180] ($(A) + (0,-1.1)$)
         to[out=0,in=180] ($(B) + (0,-1.1)$)
         to[out=0,in=270] ($(B) + (1.1,0)$)
         to[out=90,in=0] ($(B) + (0,0.5)$)
         to[out=180,in=0] ($(D) + (1.2,0.3)$)  
         to[out=180,in=270] ($(D) + (0.7,0.8)$)  
         to[out=90,in=270] ($(X) + (0.7,0)$)
        to[out=90,in=0] ($(X) + (0,0.9)$)
        to[out=180,in=90] ($(X) + (-0.8,0)$)      
         to[out=270,in=90] ($(D) + (-0.7,0.8)$)  
         to[out=270,in=0] ($(D) + (-1.2,0.3)$)  
         to[out=180,in=0] ($(A) + (0,0.6)$)  
       to[out=180,in=90] ($(A)+(-1.1,0)$);

\draw [color=black, fill=black] (0, 0) circle (.3);
\draw [color=black, fill=black] (-4, -4) circle (.3);
\draw [color=black, fill=black] (4, -4) circle (.3);
\draw [color=black] (0, -2.5) circle (.3); 
   
   \node at (-3.9, -1.8) (T) {$T_l$};
  \node at (3.9, -1.8) (S) {$S_l$};
  \node at (0, -5.8) (R) {$\delta R_r$}; 


\begin{scope}[xshift = 11.5cm]
  \node at (0, 0) (C) {};
  \node at (-4, -4) (A) {};
  \node at (4, -4) (B) {};
  \node at (0, -4) (D) {};
  \node at (2, -2) (E) {};
  \node at (-2, -2) (F) {};
 \node at (0, -2.5) (X) {};  

   \draw
       ($(C)+(0.5,0.5)$) 
        to[out=135,in=45] ($(C) + (-0.5,0.5)$)
        to[out=225,in=45] ($(A) + (-0.5,0.5)$)
         to[out=225,in=135] ($(A) + (-0.5,-0.5)$)
         to[out=315,in=225] ($(A) + (0.5,-0.5)$)
        to[out=50,in=190] ($(F) + (0,-0.9)$)    
        to[out=0,in=180] ($(X) + (0,-0.5)$) 
        to[out=0,in=270] ($(X) + (0.5,0)$)
          to[out=90,in=0] ($(X) + (0,0.5)$)               
          to[out=180,in=225] ($(F) + (1,0.3)$)         
          to[out=45,in=225] ($(C) + (0.5,-0.5)$)  
          to[out=45,in=315] ($(C) + (0.5,0.5)$);

  \draw
       ($(C)+(-0.7,0.7)$) 
        to[out=45,in=135] ($(C) + (0.7,0.7)$)
        to[out=315,in=130] ($(B) + (0.5,0.5)$)
         to[out=315,in=45] ($(B) + (0.5,-0.5)$)
        to[out=225,in=315] ($(B) + (-0.5,-0.5)$)
        to[out=135,in=0] ($(E) + (0,-1.1)$)  
        to[out=180,in=0] ($(X) + (0,-0.7)$)
        to[out=180,in=270] ($(X) + (-0.5,0)$)
      to[out=90,in=180] ($(X) + (0,0.7)$) 
       to[out=0,in=225] ($(E) + (-1,0.3)$)        
       to[out=45,in=315] ($(C) + (-0.7,-0.5)$)  
       to[out=135,in=225] ($(C) + (-0.7,0.7)$);

  \draw[fill= light-gray,fill opacity = 0.4]
       ($(A)+(-1.1,0)$) 
        to[out=270,in=180] ($(A) + (0,-1.1)$)
         to[out=0,in=180] ($(B) + (0,-1.1)$)
         to[out=0,in=270] ($(B) + (1.1,0)$)
         to[out=90,in=0] ($(B) + (0,0.5)$)
         to[out=180,in=0] ($(D) + (1.2,0.3)$)  
         to[out=180,in=270] ($(D) + (0.7,0.8)$)  
         to[out=90,in=270] ($(X) + (0.7,0)$)
        to[out=90,in=0] ($(X) + (0,0.9)$)
        to[out=180,in=90] ($(X) + (-0.8,0)$)      
         to[out=270,in=90] ($(D) + (-0.7,0.8)$)  
         to[out=270,in=0] ($(D) + (-1.2,0.3)$)  
         to[out=180,in=0] ($(A) + (0,0.6)$)  
       to[out=180,in=90] ($(A)+(-1.1,0)$);                         
      
\draw [color=black, fill=black] (0, 0) circle (.3);
\draw [color=black, fill=black] (-4, -4) circle (.3);
\draw [color=black, fill=black] (4, -4) circle (.3);
\draw [color=black] (0, -2.5) circle (.3);

\node at (-3.9, -1.8) (T) {$T_h$};
\node at (3.9, -1.8) (S) {$S_l$};
\node at (0, -5.8) (R) {$\delta R_r$}; 
\node at (3.3, 0) (eps) {$\eps \leq \frac{1}{2}$};
\end{scope}

\begin{scope}[xshift = 23cm]
  \node at (0, 0) (C) {};
  \node at (4, -4) (B) {};
  \node at (0, -4) (D) {};
  \node at (2, -2) (E) {};
  \node at (-2, -2) (F) {};
 \node at (0, -2.5) (X) {};  

  \draw
       ($(C)+(0.5,0.5)$) 
        to[out=135,in=45] ($(C) + (-0.5,0.5)$)
         to[out=225,in=45] ($(F) + (-0.5,0.5)$)
        to[out=225,in=135] ($(F) + (-0.6,-0.6)$)   
        to[out=315,in=180] ($(X) + (0,-0.5)$) 
        to[out=0,in=270] ($(X) + (0.5,0)$)
          to[out=90,in=0] ($(X) + (0,0.5)$)               
          to[out=180,in=225] ($(F) + (1,0.3)$)         
          to[out=45,in=225] ($(C) + (0.4,-0.4)$)  
          to[out=45,in=315] ($(C) + (0.5,0.5)$);

  \draw
       ($(C)+(-0.7,0.7)$) 
        to[out=45,in=135] ($(C) + (0.7,0.7)$)
        to[out=315,in=130] ($(B) + (0.5,0.5)$)
         to[out=315,in=45] ($(B) + (0.5,-0.5)$)
        to[out=225,in=315] ($(B) + (-0.5,-0.5)$)
        to[out=135,in=0] ($(E) + (0,-1.1)$)  
        to[out=180,in=0] ($(X) + (0,-0.7)$)
        to[out=180,in=270] ($(X) + (-0.5,0)$)
      to[out=90,in=180] ($(X) + (0,0.7)$) 
       to[out=0,in=225] ($(E) + (-1,0.3)$)        
       to[out=45,in=315] ($(C) + (-0.7,-0.5)$)  
       to[out=135,in=225] ($(C) + (-0.7,0.7)$);

  \draw[fill= light-gray,fill opacity = 0.4]
       ($(D)+(-0.8,0)$) 
        to[out=270,in=180] ($(D) + (0,-1)$)  
        to[out=0,in=180] ($(B) + (0,-1.1)$)
         to[out=0,in=270] ($(B) + (1.1,0)$)
         to[out=90,in=0] ($(B) + (0,0.5)$)
         to[out=180,in=0] ($(D) + (1.2,0.3)$)  
         to[out=180,in=270] ($(D) + (0.7,0.8)$)  
         to[out=90,in=270] ($(X) + (0.7,0)$)
        to[out=90,in=0] ($(X) + (0,0.9)$)
        to[out=180,in=90] ($(X) + (-0.8,0)$) 
        to[out=270,in=90] ($(D) + (-0.8,0)$);       
      
\draw [color=black, fill=black] (0, 0) circle (.3);
\draw [color=black, fill=black] (4, -4) circle (.3);
\draw [color=black] (0, -2.5) circle (.3); 
   
   \node at (-2.5, -0.5) (T) {$T_l$};
  \node at (3.9, -1.8) (S) {$S_l$};
  \node at (2, -5.6) (R) {$\delta R$}; 

\end{scope}

\begin{scope}[xshift = 34.5cm]
  \node at (0, 0) (C) {};
  \node at (4, -4) (B) {};
  \node at (0, -4) (D) {};
  \node at (2, -2) (E) {};
  \node at (-2, -2) (F) {};
 \node at (0, -2.5) (X) {};  

 \draw
       ($(C)+(0.5,0.5)$) 
        to[out=135,in=45] ($(C) + (-0.5,0.5)$)
         to[out=225,in=45] ($(F) + (-0.5,0.5)$)
        to[out=225,in=135] ($(F) + (-0.6,-0.6)$)   
        to[out=315,in=180] ($(X) + (0,-0.5)$) 
        to[out=0,in=270] ($(X) + (0.5,0)$)
          to[out=90,in=0] ($(X) + (0,0.5)$)               
          to[out=180,in=225] ($(F) + (1,0.3)$)         
          to[out=45,in=225] ($(C) + (0.4,-0.4)$)  
          to[out=45,in=315] ($(C) + (0.5,0.5)$);

  \draw
       ($(C)+(-0.7,0.7)$) 
        to[out=45,in=135] ($(C) + (0.7,0.7)$)
        to[out=315,in=130] ($(B) + (0.5,0.5)$)
         to[out=315,in=45] ($(B) + (0.5,-0.5)$)
        to[out=225,in=315] ($(B) + (-0.5,-0.5)$)
        to[out=135,in=0] ($(E) + (0,-1.1)$)  
        to[out=180,in=0] ($(X) + (0,-0.7)$)
        to[out=180,in=270] ($(X) + (-0.5,0)$)
      to[out=90,in=180] ($(X) + (0,0.7)$) 
       to[out=0,in=225] ($(E) + (-1,0.3)$)        
       to[out=45,in=315] ($(C) + (-0.7,-0.5)$)  
       to[out=135,in=225] ($(C) + (-0.7,0.7)$);

  \draw[fill= light-gray,fill opacity = 0.4]
       ($(D)+(-0.8,0)$) 
        to[out=270,in=180] ($(D) + (0,-1)$)
        to[out=0,in=180] ($(B) + (0,-1.1)$)
         to[out=0,in=270] ($(B) + (1.1,0)$)
         to[out=90,in=0] ($(B) + (0,0.5)$)
         to[out=180,in=0] ($(D) + (1.2,0.3)$)  
         to[out=180,in=270] ($(D) + (0.7,0.8)$)  
         to[out=90,in=270] ($(X) + (0.7,0)$)
        to[out=90,in=0] ($(X) + (0,0.9)$)
        to[out=180,in=90] ($(X) + (-0.8,0)$) 
        to[out=270,in=90] ($(D) + (-0.8,0)$);       
      
\draw [color=black, fill=black] (0, 0) circle (.3);
\draw [color=black, fill=black] (4, -4) circle (.3);
\draw [color=black] (0, -2.5) circle (.3);

   \node at (-2.5, -0.5) (T){$T_h$};
  \node at (3.9, -1.8) (S) {$S_l$};
  \node at (2, -5.6) (R) {$\delta R$}; 
  \node at (3.3, 0) (eps) {$\eps \leq \frac{1}{2}$};
\end{scope}



\begin{scope}[yshift = -6.5cm]
  \node at (-4, -4) (A) {};
  \node at (4, -4) (B) {};
  \node at (0, -4) (D) {};
  \node at (2, -2) (E) {};
  \node at (-2, -2) (F) {};
 \node at (0, -2.5) (X) {};  

      \draw
       ($(F)+(-0.5,0.5)$) 
        to[out=225,in=45] ($(A) + (-0.5,0.5)$)
         to[out=225,in=135] ($(A) + (-0.5,-0.5)$)
         to[out=315,in=225] ($(A) + (0.5,-0.5)$)
        to[out=50,in=190] ($(F) + (0,-0.9)$)    
        to[out=0,in=180] ($(X) + (0,-0.5)$) 
        to[out=0,in=270] ($(X) + (0.5,0)$)
          to[out=90,in=0] ($(X) + (0,0.5)$)               
          to[out=180,in=315] ($(F) + (1,0.3)$)
          to[out=135,in=45] ($(F) + (-0.5,0.5)$);

  \draw
       ($(E)+(0.5,0.5)$) 
        to[out=315,in=135] ($(B) + (0.5,0.5)$)
         to[out=315,in=45] ($(B) + (0.5,-0.5)$)
        to[out=225,in=315] ($(B) + (-0.5,-0.5)$)
        to[out=135,in=0] ($(E) + (0,-1.1)$)  
        to[out=180,in=0] ($(X) + (0,-0.7)$)
        to[out=180,in=270] ($(X) + (-0.5,0)$)
      to[out=90,in=180] ($(X) + (0,0.7)$) 
       to[out=0,in=225] ($(E) + (-0.8,0.3)$)
       to[out=45,in=135] ($(E) + (0.5,0.5)$);

\draw[fill= light-gray,fill opacity = 0.4]
       ($(A)+(-1.1,0)$) 
        to[out=270,in=180] ($(A) + (0,-1.1)$)
         to[out=0,in=180] ($(B) + (0,-1.1)$)
         to[out=0,in=270] ($(B) + (1.1,0)$)
         to[out=90,in=0] ($(B) + (0,0.5)$)
         to[out=180,in=0] ($(D) + (1.2,0.3)$)  
         to[out=180,in=270] ($(D) + (0.7,0.8)$)  
         to[out=90,in=270] ($(X) + (0.7,0)$)
        to[out=90,in=0] ($(X) + (0,0.9)$)
        to[out=180,in=90] ($(X) + (-0.8,0)$)      
         to[out=270,in=90] ($(D) + (-0.7,0.8)$)  
         to[out=270,in=0] ($(D) + (-1.2,0.3)$)  
         to[out=180,in=0] ($(A) + (0,0.6)$)  
       to[out=180,in=90] ($(A)+(-1.1,0)$);

\draw [color=black, fill=black] (-4, -4) circle (.3);
\draw [color=black, fill=black] (4, -4) circle (.3);
\draw [color=black] (0, -2.5) circle (.3); 
   
   \node at (-3.9, -1.8) (T) {$T$};
  \node at (3.9, -1.8) (S) {$S_s$};
  \node at (0, -5.8) (R) {$\delta R_r$}; 
\end{scope}


\begin{scope}[yshift = -6.5cm, xshift = 11.5cm]
  \node at (-4, -4) (A) {};
  \node at (0, -4) (D) {};
  \node at (2, -2.5) (E) {};
  \node at (-2, -2) (F) {};
 \node at (0, -2.5) (X) {};  

      \draw
       ($(F)+(-0.5,0.5)$) 
        to[out=225,in=45] ($(A) + (-0.5,0.5)$)
         to[out=225,in=135] ($(A) + (-0.5,-0.5)$)
         to[out=315,in=225] ($(A) + (0.5,-0.5)$)
        to[out=50,in=190] ($(F) + (0,-0.9)$)    
        to[out=0,in=180] ($(X) + (0,-0.5)$) 
        to[out=0,in=270] ($(X) + (0.5,0)$)
          to[out=90,in=0] ($(X) + (0,0.5)$)               
          to[out=180,in=315] ($(F) + (1,0.3)$)
          to[out=135,in=45] ($(F) + (-0.5,0.5)$);

  \draw
       ($(E)+(0,0.7)$) 
        to[out=0,in=90] ($(E) + (0.7,0)$)  
        to[out=270,in=0] ($(E) + (0,-0.7)$)  
        to[out=180,in=0] ($(X) + (0,-0.7)$)
        to[out=180,in=270] ($(X) + (-0.5,0)$)
      to[out=90,in=180] ($(X) + (0,0.7)$) 
       to[out=0,in=180] ($(E) + (0,0.7)$);

  \draw[fill= light-gray,fill opacity = 0.4]
       ($(A)+(-1.1,0)$) 
        to[out=270,in=180] ($(A) + (0,-1.1)$)
         to[out=0,in=180] ($(D) + (0,-1.1)$)  
         to[out=0,in=270] ($(D) + (0.8,0)$)  
         to[out=90,in=270] ($(X) + (0.7,0)$)
        to[out=90,in=0] ($(X) + (0,0.9)$)
        to[out=180,in=90] ($(X) + (-0.8,0)$)      
         to[out=270,in=90] ($(D) + (-0.7,0.8)$)  
         to[out=270,in=0] ($(D) + (-1.2,0.3)$)  
         to[out=180,in=0] ($(A) + (0,0.6)$)  
       to[out=180,in=90] ($(A)+(-1.1,0)$);

\draw [color=black, fill=black] (-4, -4) circle (.3);
\draw [color=black] (0, -2.5) circle (.3); 
   
   \node at (-3.9, -1.8) (T) {$T$};
  \node at (3.6, -1.8) (S) {$S$};
  \node at (-2, -5.8) (R) {$\delta R_r$}; 
\end{scope}


\begin{scope}[xshift = 23cm,yshift = -6.5cm]
  \node at (4, -4) (B) {};
  \node at (0, -4) (D) {};
  \node at (2, -2) (E) {};
  \node at (-2, -2.5) (F) {};
 \node at (0, -2.5) (X) {};  

      \draw
 
            ($(F)+(-0.6,0)$) 
         to[out=90,in=180] ($(F) + (0,0.6)$)
        to[out=0,in=180] ($(X) + (0,0.6)$)          
          to[out=0,in=90] ($(X) + (0.6,0)$)         
          to[out=270,in=0] ($(X) + (0,-0.6)$)     
         to[out=180,in=0] ($(F) + (0,-0.6)$)         
         to[out=180,in=270] ($(F) + (-0.6,0)$);                   

  \draw
       ($(E)+(0.5,0.5)$) 
        to[out=315,in=135] ($(B) + (0.5,0.5)$)
         to[out=315,in=45] ($(B) + (0.5,-0.5)$)
        to[out=225,in=315] ($(B) + (-0.5,-0.5)$)
        to[out=135,in=0] ($(E) + (0,-1.1)$)  
        to[out=180,in=0] ($(X) + (0,-0.7)$)
        to[out=180,in=270] ($(X) + (-0.5,0)$)
      to[out=90,in=180] ($(X) + (0,0.7)$) 
       to[out=0,in=225] ($(E) + (-0.8,0.3)$)
       to[out=45,in=135] ($(E) + (0.5,0.5)$);

  \draw[fill= light-gray,fill opacity = 0.4]
       ($(D)+(-0.8,0)$) 
        to[out=270,in=180] ($(D) + (0,-1)$)  
        to[out=0,in=180] ($(B) + (0,-1.1)$)
         to[out=0,in=270] ($(B) + (1.1,0)$)
         to[out=90,in=0] ($(B) + (0,0.5)$)
         to[out=180,in=0] ($(D) + (1.2,0.3)$)  
         to[out=180,in=270] ($(D) + (0.7,0.8)$)  
         to[out=90,in=270] ($(X) + (0.7,0)$)
        to[out=90,in=0] ($(X) + (0,0.9)$)
        to[out=180,in=90] ($(X) + (-0.8,0)$) 
        to[out=270,in=90] ($(D) + (-0.8,0)$);       .

\draw [color=black, fill=black] (4, -4) circle (.3);
\draw [color=black] (0, -2.5) circle (.3);

   \node at (-3.3, -1.8) (T) {$T$};
  \node at (3.9, -1.8) (S) {$S_s$};
  \node at (2, -5.8) (R) {$\delta R$}; 
\end{scope}

\begin{scope}[xshift=34.5cm,yshift=-6.5cm]
  \node at (0, -4) (D) {};
  \node at (2.5, -2) (E) {};
  \node at (-2.5, -2) (F) {};
  \node at (0, -2) (X) {};
  
  \draw
       ($(F)+(-0.6,0)$) 
         to[out=90,in=180] ($(F) + (0,0.6)$)
        to[out=0,in=180] ($(X) + (0,0.6)$)          
          to[out=0,in=90] ($(X) + (0.6,0)$)         
          to[out=270,in=0] ($(X) + (0,-0.6)$)     
         to[out=180,in=0] ($(F) + (0,-0.6)$)         
         to[out=180,in=270] ($(F) + (-0.6,0)$);        
    
  \draw[fill= light-gray,fill opacity = 0.4]
       ($(X)+(0,1)$) 
         to[out=0,in=90] ($(X) + (0.9,0)$)
        to[out=270,in=90] ($(D) + (0.9,0)$)          
          to[out=270,in=0] ($(D) + (0,-0.9)$)         
          to[out=180,in=270] ($(D) + (-0.9,0)$)     
         to[out=90,in=270] ($(X) + (-1,0)$)         
         to[out=90,in=180] ($(X) + (0,1)$);

   \draw
       ($(X)+(-0.8,0)$) 
         to[out=90,in=180] ($(X) + (0,0.8)$)
        to[out=0,in=180] ($(E) + (0,0.8)$)          
          to[out=0,in=90] ($(E) + (0.8,0)$)         
          to[out=270,in=0] ($(E) + (0,-0.8)$)     
         to[out=180,in=0] ($(X) + (0,-0.8)$)         
         to[out=180,in=270] ($(X) + (-0.8,0)$);

\draw [color=black] (0, -2) circle (.3); 
   
      \node at (-2.5, -0.6) (T) {$T$};
  \node at (3, -0.6) (S) {$S$};
  \node at (0, -5.6) (R) {$\delta R$}; 
\end{scope}

 \node at (19, -16.6) (strategy) {Delta skew-aware 
  views with delta evaluation strategy 
$\delta R_r  \cdot V_{S_sT_t}$};;
 
 \begin{scope}[xshift = 3cm,yshift=-19cm]
 
  \node at (0, 0) (C) {};
  \node at (-4, -4) (A) {};
  \node at (4, -4) (B) {};
  \node at (0, -4) (D) {};
  \node at (2, -2) (E) {};
  \node at (-2, -2) (F) {};
 \node at (0, -2.5) (X) {};  

   \draw
       ($(C)+(0.5,0.5)$) 
        to[out=135,in=45] ($(C) + (-0.5,0.5)$)
        to[out=225,in=45] ($(A) + (-0.5,0.5)$)
         to[out=225,in=135] ($(A) + (-0.5,-0.5)$)
         to[out=315,in=225] ($(A) + (0.5,-0.5)$)
        to[out=50,in=190] ($(F) + (0,-0.9)$)    
        to[out=0,in=180] ($(X) + (0,-0.5)$) 
        to[out=0,in=270] ($(X) + (0.5,0)$)
          to[out=90,in=0] ($(X) + (0,0.5)$)               
          to[out=180,in=225] ($(F) + (1,0.3)$)         
          to[out=45,in=225] ($(C) + (0.5,-0.5)$)  
          to[out=45,in=315] ($(C) + (0.5,0.5)$);

  \draw
       ($(C)+(-0.7,0.7)$) 
        to[out=45,in=135] ($(C) + (0.7,0.7)$)
        to[out=315,in=130] ($(B) + (0.5,0.5)$)
         to[out=315,in=45] ($(B) + (0.5,-0.5)$)
        to[out=225,in=315] ($(B) + (-0.5,-0.5)$)
        to[out=135,in=0] ($(E) + (0,-1.1)$)  
        to[out=180,in=0] ($(X) + (0,-0.7)$)
        to[out=180,in=270] ($(X) + (-0.5,0)$)
      to[out=90,in=180] ($(X) + (0,0.7)$) 
       to[out=0,in=225] ($(E) + (-1,0.3)$)        
       to[out=45,in=315] ($(C) + (-0.7,-0.5)$)  
       to[out=135,in=225] ($(C) + (-0.7,0.7)$);

  \draw[fill= light-gray,fill opacity = 0.4]
       ($(A)+(-1.1,0)$) 
        to[out=270,in=180] ($(A) + (0,-1.1)$)
         to[out=0,in=180] ($(B) + (0,-1.1)$)
         to[out=0,in=270] ($(B) + (1.1,0)$)
         to[out=90,in=0] ($(B) + (0,0.5)$)
         to[out=180,in=0] ($(D) + (1.2,0.3)$)  
         to[out=180,in=270] ($(D) + (0.7,0.8)$)  
         to[out=90,in=270] ($(X) + (0.7,0)$)
        to[out=90,in=0] ($(X) + (0,0.9)$)
        to[out=180,in=90] ($(X) + (-0.8,0)$)      
         to[out=270,in=90] ($(D) + (-0.7,0.8)$)  
         to[out=270,in=0] ($(D) + (-1.2,0.3)$)  
         to[out=180,in=0] ($(A) + (0,0.6)$)  
       to[out=180,in=90] ($(A)+(-1.1,0)$);

\draw [color=black, fill=black] (0, 0) circle (.3);
\draw [color=black, fill=black] (-4, -4) circle (.3);
\draw [color=black, fill=black] (4, -4) circle (.3);
\draw [color=black] (0, -2.5) circle (.3);    
   
   \node at (-3.9, -1.8) (T) {$T_l$};
  \node at (3.9, -1.8) (S) {$S_h$};
  \node at (0, -5.8) (R) {$\delta R_r$}; 

\end{scope}

\begin{scope}[xshift = 13cm, yshift=-19cm]
  \node at (0, 0) (C) {};
  \node at (4, -4) (B) {};
  \node at (0, -4) (D) {};
  \node at (2, -2) (E) {};
  \node at (-2, -2) (F) {};
 \node at (0, -2.5) (X) {};  

  \draw
       ($(C)+(0.5,0.5)$) 
        to[out=135,in=45] ($(C) + (-0.5,0.5)$)
         to[out=225,in=45] ($(F) + (-0.5,0.5)$)
        to[out=225,in=135] ($(F) + (-0.6,-0.6)$)   
        to[out=315,in=180] ($(X) + (0,-0.5)$) 
        to[out=0,in=270] ($(X) + (0.5,0)$)
          to[out=90,in=0] ($(X) + (0,0.5)$)               
          to[out=180,in=225] ($(F) + (1,0.3)$)         
          to[out=45,in=225] ($(C) + (0.4,-0.4)$)  
          to[out=45,in=315] ($(C) + (0.5,0.5)$);

  \draw
       ($(C)+(-0.7,0.7)$) 
        to[out=45,in=135] ($(C) + (0.7,0.7)$)
        to[out=315,in=130] ($(B) + (0.5,0.5)$)
         to[out=315,in=45] ($(B) + (0.5,-0.5)$)
        to[out=225,in=315] ($(B) + (-0.5,-0.5)$)
        to[out=135,in=0] ($(E) + (0,-1.1)$)  
        to[out=180,in=0] ($(X) + (0,-0.7)$)
        to[out=180,in=270] ($(X) + (-0.5,0)$)
      to[out=90,in=180] ($(X) + (0,0.7)$) 
       to[out=0,in=225] ($(E) + (-1,0.3)$)        
       to[out=45,in=315] ($(C) + (-0.7,-0.5)$)  
       to[out=135,in=225] ($(C) + (-0.7,0.7)$);

  \draw[fill= light-gray,fill opacity = 0.4]
       ($(D)+(-0.8,0)$) 
        to[out=270,in=180] ($(D) + (0,-1)$)
        to[out=0,in=180] ($(B) + (0,-1.1)$)
         to[out=0,in=270] ($(B) + (1.1,0)$)
         to[out=90,in=0] ($(B) + (0,0.5)$)
         to[out=180,in=0] ($(D) + (1.2,0.3)$)  
         to[out=180,in=270] ($(D) + (0.7,0.8)$)  
         to[out=90,in=270] ($(X) + (0.7,0)$)
        to[out=90,in=0] ($(X) + (0,0.9)$)
        to[out=180,in=90] ($(X) + (-0.8,0)$) 
        to[out=270,in=90] ($(D) + (-0.8,0)$);

\draw [color=black, fill=black] (0, 0) circle (.3);
\draw [color=black, fill=black] (4, -4) circle (.3);
\draw [color=black] (0, -2.5) circle (.3);

   \node at (-2.5, -0.5) (T) (T) {$T_l$};
  \node at (3.9, -1.8) (S) {$S_h$};
  \node at (2, -5.6) (R) {$\delta R$}; 

\end{scope}


\begin{scope}[xshift=25.5cm, yshift=-19cm]
  \node at (0, 0) (C) {};
  \node at (-4, -4) (A) {};
  \node at (0, -4) (D) {};
  \node at (2, -2) (E) {};
  \node at (-2, -2) (F) {};
 \node at (0, -2.5) (X) {};  

      \draw
       ($(C)+(0.5,0.5)$) 
        to[out=135,in=45] ($(C) + (-0.5,0.5)$)
        to[out=225,in=45] ($(A) + (-0.5,0.5)$)
         to[out=225,in=135] ($(A) + (-0.5,-0.5)$)
         to[out=315,in=225] ($(A) + (0.5,-0.5)$)
        to[out=50,in=190] ($(F) + (0,-0.9)$)    
        to[out=0,in=180] ($(X) + (0,-0.5)$) 
        to[out=0,in=270] ($(X) + (0.5,0)$)
          to[out=90,in=0] ($(X) + (0,0.5)$)               
          to[out=180,in=225] ($(F) + (1,0.3)$)         
          to[out=45,in=225] ($(C) + (0.5,-0.5)$)  
          to[out=45,in=315] ($(C) + (0.5,0.5)$);

 \draw
       ($(C)+(-0.7,0.7)$) 
        to[out=45,in=135] ($(C) + (0.7,0.7)$) 
        to[out=315,in=135] ($(E) + (0.7,0.7)$)  
      to[out=315,in=0] ($(E) + (0,-1.2)$)      
        to[out=180,in=0] ($(X) + (0,-0.7)$)
        to[out=180,in=270] ($(X) + (-0.5,0)$)
      to[out=90,in=180] ($(X) + (0,0.7)$) 
       to[out=0,in=225] ($(E) + (-1,0.3)$)         
       to[out=45,in=315] ($(C) + (-0.6,-0.6)$)  
       to[out=135,in=225] ($(C) + (-0.7,0.7)$);

  \draw[fill= light-gray,fill opacity = 0.4]
       ($(A)+(-1.1,0)$) 
        to[out=270,in=180] ($(A) + (0,-1.1)$)
         to[out=0,in=180] ($(D) + (0,-1.1)$)  
         to[out=0,in=270] ($(D) + (0.8,0)$)  
         to[out=90,in=270] ($(X) + (0.7,0)$)
        to[out=90,in=0] ($(X) + (0,0.9)$)
        to[out=180,in=90] ($(X) + (-0.8,0)$)      
         to[out=270,in=90] ($(D) + (-0.7,0.8)$)  
         to[out=270,in=0] ($(D) + (-1.2,0.3)$)  
         to[out=180,in=0] ($(A) + (0,0.6)$)  
       to[out=180,in=90] ($(A)+(-1.1,0)$);

\draw [color=black, fill=black] (0, 0) circle (.3);
\draw [color=black, fill=black] (-4, -4) circle (.3);
\draw [color=black] (0, -2.5) circle (.3);  
   
   \node at (-3.9, -1.8) (T) {$T_l$};
  \node at (3.2, -0.8) (S) {$S$};
  \node at (-2, -5.8) (R) {$\delta R_r$}; 

\end{scope}


\begin{scope}[xshift=34cm, yshift=-19cm]
  \node at (0, 0) (C) {};
  \node at (0, -4) (D) {};
  \node at (2, -2) (E) {};
  \node at (-2, -2) (F) {};
 \node at (0, -2.5) (X) {};  

 \draw
       ($(C)+(0.5,0.5)$) 
        to[out=135,in=45] ($(C) + (-0.5,0.5)$)
         to[out=225,in=45] ($(F) + (-0.5,0.5)$)
        to[out=225,in=135] ($(F) + (-0.6,-0.6)$)   
        to[out=315,in=180] ($(X) + (0,-0.5)$) 
        to[out=0,in=270] ($(X) + (0.5,0)$)
          to[out=90,in=0] ($(X) + (0,0.5)$)               
          to[out=180,in=225] ($(F) + (1,0.3)$)         
          to[out=45,in=225] ($(C) + (0.4,-0.4)$)  
          to[out=45,in=315] ($(C) + (0.5,0.5)$);            
                              
 \draw
       ($(C)+(-0.7,0.7)$) 
        to[out=45,in=135] ($(C) + (0.7,0.7)$) 
        to[out=315,in=135] ($(E) + (0.7,0.7)$)  
      to[out=315,in=0] ($(E) + (0,-1.2)$)      
        to[out=180,in=0] ($(X) + (0,-0.7)$)
        to[out=180,in=270] ($(X) + (-0.5,0)$)
      to[out=90,in=180] ($(X) + (0,0.7)$) 
       to[out=0,in=225] ($(E) + (-1,0.3)$)         
       to[out=45,in=315] ($(C) + (-0.6,-0.6)$)  
       to[out=135,in=225] ($(C) + (-0.7,0.7)$);

  \draw[fill= light-gray,fill opacity = 0.4]
       ($(X)+(0,0.9)$) 
         to[out=0,in=90] ($(X) + (0.7,0)$)  
         to[out=270,in=90] ($(D) + (0.7,0)$)  
         to[out=270,in=0] ($(D) + (0,-0.7)$)
         to[out=180,in=270] ($(D) + (-0.7,0)$)         
         to[out=90,in=270] ($(X) + (-0.7,0)$)
         to[out=90,in=180] ($(X) + (-0,0.9)$);

 \draw [color=black, fill=black] (0, 0) circle (.3);
\draw [color=black] (0, -2.5) circle (.3); 
   
   \node at (-2.5, -0.5) (T) {$T_t$};
  \node at (3.2, -0.8) (S) {$S$};
  \node at (0, -5.3) (R) {$\delta R$}; 

\end{scope}

\node at (19, -29.1) (strategy) {Delta skew-aware 
views with delta evaluation strategy
$\delta R_r  \cdot \textstyle\sum T_t \cdot S_s$};;
      
\begin{scope}[xshift = 5cm, yshift=-31.5cm]
 \node at (0, 0) (C) {};
  \node at (-4, -4) (A) {};
  \node at (4, -4) (B) {};
  \node at (0, -4) (D) {};
  \node at (2, -2) (E) {};
  \node at (-2, -2) (F) {};
 \node at (0, -2.5) (X) {};  

    \draw
       ($(C)+(0.5,0.5)$) 
        to[out=135,in=45] ($(C) + (-0.5,0.5)$)
        to[out=225,in=45] ($(A) + (-0.5,0.5)$)
         to[out=225,in=135] ($(A) + (-0.5,-0.5)$)
         to[out=315,in=225] ($(A) + (0.5,-0.5)$)
        to[out=50,in=190] ($(F) + (0,-0.9)$)    
        to[out=0,in=180] ($(X) + (0,-0.5)$) 
        to[out=0,in=270] ($(X) + (0.5,0)$)
          to[out=90,in=0] ($(X) + (0,0.5)$)               
          to[out=180,in=225] ($(F) + (1,0.3)$)         
          to[out=45,in=225] ($(C) + (0.5,-0.5)$)  
          to[out=45,in=315] ($(C) + (0.5,0.5)$);

  \draw
       ($(C)+(-0.7,0.7)$) 
        to[out=45,in=135] ($(C) + (0.7,0.7)$)
        to[out=315,in=130] ($(B) + (0.5,0.5)$)
         to[out=315,in=45] ($(B) + (0.5,-0.5)$)
        to[out=225,in=315] ($(B) + (-0.5,-0.5)$)
        to[out=135,in=0] ($(E) + (0,-1.1)$)  
        to[out=180,in=0] ($(X) + (0,-0.7)$)
        to[out=180,in=270] ($(X) + (-0.5,0)$)
      to[out=90,in=180] ($(X) + (0,0.7)$) 
       to[out=0,in=225] ($(E) + (-1,0.3)$)        
       to[out=45,in=315] ($(C) + (-0.7,-0.5)$)  
       to[out=135,in=225] ($(C) + (-0.7,0.7)$);

  \draw[fill= light-gray,fill opacity = 0.4]
       ($(A)+(-1.1,0)$) 
        to[out=270,in=180] ($(A) + (0,-1.1)$)
         to[out=0,in=180] ($(B) + (0,-1.1)$)
         to[out=0,in=270] ($(B) + (1.1,0)$)
         to[out=90,in=0] ($(B) + (0,0.5)$)
         to[out=180,in=0] ($(D) + (1.2,0.3)$)  
         to[out=180,in=270] ($(D) + (0.7,0.8)$)  
         to[out=90,in=270] ($(X) + (0.7,0)$)
        to[out=90,in=0] ($(X) + (0,0.9)$)
        to[out=180,in=90] ($(X) + (-0.8,0)$)      
         to[out=270,in=90] ($(D) + (-0.7,0.8)$)  
         to[out=270,in=0] ($(D) + (-1.2,0.3)$)  
         to[out=180,in=0] ($(A) + (0,0.6)$)  
       to[out=180,in=90] ($(A)+(-1.1,0)$);

\draw [color=black, fill=black] (0, 0) circle (.3);
\draw [color=black, fill=black] (-4, -4) circle (.3);
\draw [color=black, fill=black] (4, -4) circle (.3);
\draw [color=black] (0, -2.5) circle (.3); 
   
   \node at (-3.9, -1.8) (T) {$T_h$};
  \node at (3.9, -1.8) (S) {$S_h$};
  \node at (0, -5.8) (R) {$\delta R_r$}; 
\end{scope}

\begin{scope}[xshift=16.5cm,yshift=-31.5cm]
 \node at (0, 0) (C) {};
  \node at (-4, -4) (A) {};
  \node at (4, -4) (B) {};
  \node at (0, -4) (D) {};
  \node at (2, -2) (E) {};
  \node at (-2, -2) (F) {};
 \node at (0, -2.5) (X) {};  

   \draw
       ($(C)+(0.5,0.5)$) 
        to[out=135,in=45] ($(C) + (-0.5,0.5)$)
        to[out=225,in=45] ($(A) + (-0.5,0.5)$)
         to[out=225,in=135] ($(A) + (-0.5,-0.5)$)
         to[out=315,in=225] ($(A) + (0.5,-0.5)$)
        to[out=50,in=190] ($(F) + (0,-0.9)$)    
        to[out=0,in=180] ($(X) + (0,-0.5)$) 
        to[out=0,in=270] ($(X) + (0.5,0)$)
          to[out=90,in=0] ($(X) + (0,0.5)$)               
          to[out=180,in=225] ($(F) + (1,0.3)$)         
          to[out=45,in=225] ($(C) + (0.5,-0.5)$)  
          to[out=45,in=315] ($(C) + (0.5,0.5)$);

  \draw
       ($(C)+(-0.7,0.7)$) 
        to[out=45,in=135] ($(C) + (0.7,0.7)$)
        to[out=315,in=130] ($(B) + (0.5,0.5)$)
         to[out=315,in=45] ($(B) + (0.5,-0.5)$)
        to[out=225,in=315] ($(B) + (-0.5,-0.5)$)
        to[out=135,in=0] ($(E) + (0,-1.1)$)  
        to[out=180,in=0] ($(X) + (0,-0.7)$)
        to[out=180,in=270] ($(X) + (-0.5,0)$)
      to[out=90,in=180] ($(X) + (0,0.7)$) 
       to[out=0,in=225] ($(E) + (-1,0.3)$)        
       to[out=45,in=315] ($(C) + (-0.7,-0.5)$)  
       to[out=135,in=225] ($(C) + (-0.7,0.7)$);

  \draw[fill= light-gray,fill opacity = 0.4]
       ($(A)+(-1.1,0)$) 
        to[out=270,in=180] ($(A) + (0,-1.1)$)
         to[out=0,in=180] ($(B) + (0,-1.1)$)
         to[out=0,in=270] ($(B) + (1.1,0)$)
         to[out=90,in=0] ($(B) + (0,0.5)$)
         to[out=180,in=0] ($(D) + (1.2,0.3)$)  
         to[out=180,in=270] ($(D) + (0.7,0.8)$)  
         to[out=90,in=270] ($(X) + (0.7,0)$)
        to[out=90,in=0] ($(X) + (0,0.9)$)
        to[out=180,in=90] ($(X) + (-0.8,0)$)      
         to[out=270,in=90] ($(D) + (-0.7,0.8)$)  
         to[out=270,in=0] ($(D) + (-1.2,0.3)$)  
         to[out=180,in=0] ($(A) + (0,0.6)$)  
       to[out=180,in=90] ($(A)+(-1.1,0)$);

\draw [color=black, fill=black] (0, 0) circle (.3);
\draw [color=black, fill=black] (-4, -4) circle (.3);
\draw [color=black, fill=black] (4, -4) circle (.3);
\draw [color=black] (0, -2.5) circle (.3); 
   
   \node at (-3.9, -1.8) (T) {$T_h$};
  \node at (3.9, -1.8) (S) {$S_l$};
  \node at (0, -5.8) (R) {$\delta R_r$}; 
    \node at (3.3, 0) (eps) {$\eps > \frac{1}{2}$};
\end{scope}

\begin{scope}[xshift = 28cm, yshift=-31.5cm]
  \node at (0, 0) (C) {};
  \node at (4, -4) (B) {};
  \node at (0, -4) (D) {};
  \node at (2, -2) (E) {};
  \node at (-2, -2) (F) {};
 \node at (0, -2.5) (X) {};  

 \draw
       ($(C)+(0.5,0.5)$) 
        to[out=135,in=45] ($(C) + (-0.5,0.5)$)
         to[out=225,in=45] ($(F) + (-0.5,0.5)$)
        to[out=225,in=135] ($(F) + (-0.6,-0.6)$)   
        to[out=315,in=180] ($(X) + (0,-0.5)$) 
        to[out=0,in=270] ($(X) + (0.5,0)$)
          to[out=90,in=0] ($(X) + (0,0.5)$)               
          to[out=180,in=225] ($(F) + (1,0.3)$)         
          to[out=45,in=225] ($(C) + (0.4,-0.4)$)  
          to[out=45,in=315] ($(C) + (0.5,0.5)$);

  \draw
       ($(C)+(-0.7,0.7)$) 
        to[out=45,in=135] ($(C) + (0.7,0.7)$)
        to[out=315,in=130] ($(B) + (0.5,0.5)$)
         to[out=315,in=45] ($(B) + (0.5,-0.5)$)
        to[out=225,in=315] ($(B) + (-0.5,-0.5)$)
        to[out=135,in=0] ($(E) + (0,-1.1)$)  
        to[out=180,in=0] ($(X) + (0,-0.7)$)
        to[out=180,in=270] ($(X) + (-0.5,0)$)
      to[out=90,in=180] ($(X) + (0,0.7)$) 
       to[out=0,in=225] ($(E) + (-1,0.3)$)        
       to[out=45,in=315] ($(C) + (-0.7,-0.5)$)  
       to[out=135,in=225] ($(C) + (-0.7,0.7)$);

  \draw[fill= light-gray,fill opacity = 0.4]
       ($(D)+(-0.8,0)$) 
        to[out=270,in=180] ($(D) + (0,-1)$)
        to[out=0,in=180] ($(B) + (0,-1.1)$)
         to[out=0,in=270] ($(B) + (1.1,0)$)
         to[out=90,in=0] ($(B) + (0,0.5)$)
         to[out=180,in=0] ($(D) + (1.2,0.3)$)  
         to[out=180,in=270] ($(D) + (0.7,0.8)$)  
         to[out=90,in=270] ($(X) + (0.7,0)$)
        to[out=90,in=0] ($(X) + (0,0.9)$)
        to[out=180,in=90] ($(X) + (-0.8,0)$) 
        to[out=270,in=90] ($(D) + (-0.8,0)$);

\draw [color=black, fill=black] (0, 0) circle (.3);
\draw [color=black, fill=black] (4, -4) circle (.3);
\draw [color=black] (0, -2.5) circle (.3);

   \node at (-2.5, -0.5) (T) {$T_h$};
  \node at (3.9, -1.8) (S) {$S_h$};
  \node at (2, -5.6) (R) {$\delta R$}; 

\end{scope}


\begin{scope}[xshift = 12cm, yshift=-40cm]
  \node at (0, 0) (C) {};
  \node at (4, -4) (B) {};
  \node at (0, -4) (D) {};
  \node at (2, -2) (E) {};
  \node at (-2, -2) (F) {};
 \node at (0, -2.5) (X) {};  

 \draw
       ($(C)+(0.5,0.5)$) 
        to[out=135,in=45] ($(C) + (-0.5,0.5)$)
         to[out=225,in=45] ($(F) + (-0.5,0.5)$)
        to[out=225,in=135] ($(F) + (-0.6,-0.6)$)   
        to[out=315,in=180] ($(X) + (0,-0.5)$) 
        to[out=0,in=270] ($(X) + (0.5,0)$)
          to[out=90,in=0] ($(X) + (0,0.5)$)               
          to[out=180,in=225] ($(F) + (1,0.3)$)         
          to[out=45,in=225] ($(C) + (0.4,-0.4)$)  
          to[out=45,in=315] ($(C) + (0.5,0.5)$);

  \draw
       ($(C)+(-0.7,0.7)$) 
        to[out=45,in=135] ($(C) + (0.7,0.7)$)
        to[out=315,in=130] ($(B) + (0.5,0.5)$)
         to[out=315,in=45] ($(B) + (0.5,-0.5)$)
        to[out=225,in=315] ($(B) + (-0.5,-0.5)$)
        to[out=135,in=0] ($(E) + (0,-1.1)$)  
        to[out=180,in=0] ($(X) + (0,-0.7)$)
        to[out=180,in=270] ($(X) + (-0.5,0)$)
      to[out=90,in=180] ($(X) + (0,0.7)$) 
       to[out=0,in=225] ($(E) + (-1,0.3)$)        
       to[out=45,in=315] ($(C) + (-0.7,-0.5)$)  
       to[out=135,in=225] ($(C) + (-0.7,0.7)$);

  \draw[fill= light-gray,fill opacity = 0.4]
       ($(D)+(-0.8,0)$) 
        to[out=270,in=180] ($(D) + (0,-1)$)
        
        to[out=0,in=180] ($(B) + (0,-1.1)$)
         to[out=0,in=270] ($(B) + (1.1,0)$)
         to[out=90,in=0] ($(B) + (0,0.5)$)
         to[out=180,in=0] ($(D) + (1.2,0.3)$)  
         to[out=180,in=270] ($(D) + (0.7,0.8)$)  
         to[out=90,in=270] ($(X) + (0.7,0)$)
        to[out=90,in=0] ($(X) + (0,0.9)$)
        to[out=180,in=90] ($(X) + (-0.8,0)$) 
        to[out=270,in=90] ($(D) + (-0.8,0)$);

\draw [color=black, fill=black] (0, 0) circle (.3);
\draw [color=black, fill=black] (4, -4) circle (.3);
\draw [color=black] (0, -2.5) circle (.3);

   \node at (-2.5, -0.5) (T){$T_h$};
  \node at (3.9, -1.8) (S) {$S_l$};
  \node at (2, -5.6) (R) {$\delta R$}; 
      \node at (3.3, 0) (eps) {$\eps > \frac{1}{2}$};
\end{scope}


\begin{scope}[xshift=23.5cm, yshift=-40cm]
  \node at (0, 0) (C) {};
  \node at (-4, -4) (A) {};
  \node at (0, -4) (D) {};
  \node at (2, -2) (E) {};
  \node at (-2, -2) (F) {};
 \node at (0, -2.5) (X) {};  

      \draw
       ($(C)+(0.5,0.5)$) 
        to[out=135,in=45] ($(C) + (-0.5,0.5)$)
        to[out=225,in=45] ($(A) + (-0.5,0.5)$)
         to[out=225,in=135] ($(A) + (-0.5,-0.5)$)
         to[out=315,in=225] ($(A) + (0.5,-0.5)$)
        to[out=50,in=190] ($(F) + (0,-0.9)$)    
        to[out=0,in=180] ($(X) + (0,-0.5)$) 
        to[out=0,in=270] ($(X) + (0.5,0)$)
          to[out=90,in=0] ($(X) + (0,0.5)$)               
          to[out=180,in=225] ($(F) + (1,0.3)$)         
          to[out=45,in=225] ($(C) + (0.5,-0.5)$)  
          to[out=45,in=315] ($(C) + (0.5,0.5)$);

 \draw
       ($(C)+(-0.7,0.7)$) 
        to[out=45,in=135] ($(C) + (0.7,0.7)$) 
        to[out=315,in=135] ($(E) + (0.7,0.7)$)  
      to[out=315,in=0] ($(E) + (0,-1.2)$)      
        to[out=180,in=0] ($(X) + (0,-0.7)$)
        to[out=180,in=270] ($(X) + (-0.5,0)$)
      to[out=90,in=180] ($(X) + (0,0.7)$) 
       to[out=0,in=225] ($(E) + (-1,0.3)$)         
       to[out=45,in=315] ($(C) + (-0.6,-0.6)$)  
       to[out=135,in=225] ($(C) + (-0.7,0.7)$);

  \draw[fill= light-gray,fill opacity = 0.4]
       ($(A)+(-1.1,0)$) 
        to[out=270,in=180] ($(A) + (0,-1.1)$)
         to[out=0,in=180] ($(D) + (0,-1.1)$)           
         to[out=0,in=270] ($(D) + (0.8,0)$)  
         to[out=90,in=270] ($(X) + (0.7,0)$)
        to[out=90,in=0] ($(X) + (0,0.9)$)
        to[out=180,in=90] ($(X) + (-0.8,0)$)      
         to[out=270,in=90] ($(D) + (-0.7,0.8)$)  
         to[out=270,in=0] ($(D) + (-1.2,0.3)$)  
         to[out=180,in=0] ($(A) + (0,0.6)$)  
       to[out=180,in=90] ($(A)+(-1.1,0)$);

\draw [color=black, fill=black] (0, 0) circle (.3);
\draw [color=black, fill=black] (-4, -4) circle (.3);
\draw [color=black] (0, -2.5) circle (.3);

   \node at (-3.9, -1.8) (T) {$T_h$};
  \node at (3.2, -0.8) (S) {$S$};
  \node at (-2, -5.8) (R) {$\delta R_r$}; 

\end{scope}	
\end{tikzpicture}
 \end{center}
 \caption{Hypergraphs of delta skew-aware views categorized with respect 
 to the type of the evaluation strategy under an update to (a part of) relation 
 $\delta R$.
It holds 
$R_r \in \{R_l,R_h\}$, 
$S_s \in \{S_l,S_h\}$, and 
$T_r \in \{T_l,T_h\}$.}
\label{fig:hypergraphs_evaluation_strategies}
\end{figure}

\subsubsection{Optimization Phase for $\threerpj{0} \cup \threerpj{1}$ Queries}
\label{sec:3r_opt_phase}
For a $\threerpj{0} \cup \threerpj{1}$ query, the relation partitioning, view definitions and 
delta evaluation strategies developed 
in the previous sections undergo
an optimization phase. 
The aim of this phase is to 
improve the maintenance strategy 
by discarding redundant partitions.    

Given a view $V$ and two relations 
$K$ and $K'$, we denote by $V[K \mapsto K']$
the view that results from $V$ by replacing $K$
by $K'$. The partitioning of a relation $K$ is called redundant
if for each pair $V_1, V_2$ of views with $V_2 = V_1[K_l \mapsto K_h]$ 
and each update $\delta K'$,
the strategy to compute the delta of $V_1$ under 
$\delta K'$ is the same
as  the strategy to compute the delta of $V_2$ under 
$\delta K'$.  

In the optimization phase, \ivme discards 
redundant partitions $\{K_l,K_h\}$
as follows. 
It defines $\parts{K} = \{K\}$ and replaces each pair $V_1,V_2$ of views with $V_2 = V_1[K_l \mapsto K_h]$
by a single view $V = V_1[K_l \mapsto K]$. The 
delta evaluation strategy for $V_1$ (or $V_2$) under updates to $K_l$ becomes 
the delta evaluation strategy for $V$ under updates to $K$.   

\subsubsection{\ivme States}
\label{sec:3r_ivme_states}
We introduce \ivme states for \threer queries. 
Given a database $\db = \{R, S, T\}$, an \ivme state of $\db$ is a tuple 
$(\eps, \dbeps, \inst{V}, N)$ where 
$\dbeps = \parts{R} \cup \parts{S} \cup \parts{T}$, 
$\inst{V}$ is the set of materialized views as given in 
Figure \ref{fig:meta_view_definitions}, and $N \in \mathbb{N}$ is the 
threshold base
such that 
$\floor{\frac{1}{4}N} \leq |\db| < N$. Each partitioned relation 
 $K$ is partitioned 
on its partition variable tuple with threshold $N^\eps$. 
Recall 
that the partition variable 
tuples of $R$, $S$, and $T$ are 
$\inst{A}_{RT}$, $\inst{A}_{RS}$, and $\inst{A}_{ST}$, 
respectively.
 
We derive two upper bounds for each partition $\{K_h, K_l\}\subseteq \dbeps$.
In case $K = R$,  it holds  that for any tuple 
$\inst{a}_{RT}$ over $\inst{A}_{RT}$, the number 
of tuples $(\inst{a}_{RS}, \inst{a}_{RST})$  paired with $\inst{a}_{RT}$
 in $R_l$ is less than $\frac{3}{2}N^{\eps}$. Moreover, 
the number of distinct tuples over $\inst{A}_{RT}$ in $R_h$ is at most 
$\frac{N}{\frac{1}{2}N^{\eps}} = 2N^{1-\eps}$. The bounds for the other partitions are analogous.

\subsubsection{Space}
\label{sec:3r_space}
We explain the space complexities of the views given in Figure 
  \ref{fig:meta_view_definitions}. Due to symmetry it suffices 
  to consider the first five views. 
  We assume that the views 
  are part of  an \ivme state of a database $\inst{D}$ with threshold base $N$.
  
\begin{itemize}
\item $Q$: Since the view consists of the empty tuple mapped to a single 
 value, its size is constant. 


\item $V_{S_hT_l}(\inst{a}_{RT}, \inst{a}_{RS}, \inst{a}_{RST})$:
The size of this view is upper-bounded
by the size of the join of 
$S_h$ and  $T_l$.
The size of this join 
can be upper-bounded in two ways. 
Since $T_l$ is light,  
it can be bounded 
by $|S_h| \ztimes  \frac{3}{2}N^{\eps} = 
\bigO{N \ztimes N^{\eps}} = \bigO{N^{1 + \eps}}$.
Since $S_h$ is heavy, it can also be bounded
by $|T_l| \ztimes  2N^{1-\eps} =
 \bigO{N \ztimes N^{1 - \eps}} = \bigO{N^{2 - \eps}}$.
It follows that the size of the view is 
$\bigO{N^{1 + \min\{\eps, 1-\eps\}}}$.
By using $N = \Theta(|\inst{D}|)$, we derive that 
the size of the view is $\bigO{|\inst{D}|^{1 + \min\{\eps, 1-\eps\}}}$.

\item $V_{S_hT_l}(\inst{a}_{RS}, \inst{a}_{RST})$:
As $\inst{A}_{RS} \cup  \inst{A}_{RST}$ is contained in 
the schema of $S$, the size of the view is linear.  
\end{itemize}
For the remaining views 
$V_{ST_l}(\inst{a}_{RT}, \inst{a}_{RST})$
and
$V_{ST_t}(\inst{a}_{RST})$
the analysis is similar  
to the latter case. 
Since the group-by variables 
of both views are contained in 
the schema of one of the relations defining
the views, 
 the sizes of both views 
 are linear.

\subsubsection{Update Time}
\label{sec:3r_time}
We explain the time complexities of the delta evaluation strategies 
given in Figure \ref{fig:best_update_strategy}.
Assume that the partitions and views
are part of an \ivme state of a database $\inst{D}$
with some threshold base $N$. Let 
$\delta R_r = \{(\inst{a}_{RT}, \inst{a}_{RS} , \inst{a}_{RST}) \mapsto \p\}$
be a single-tuple update. 

\begin{itemize}
\item Delta computation for the views 
$Q_{R_rS_sT}$, $\delta Q_{R_rS_hT_l}$, 
$Q_{R_rST_l}$, 
$Q_{RST_h}$, 
$V_{R_hS_l}(\inst{a}_{RT} , \inst{a}_{RST})$,
$V_{RS_s}(\inst{a}_{RST})$,
$V_{TR_l}(\inst{a}_{RS},\inst{a}_{RST})$, and $V_{TR_r}(\inst{a}_{RST})$:
In these cases, all variables in the delta views are fixed to the constants 
given by the update $\delta R_r$. Hence, delta computation 
amounts to constant-time look-ups.  

\item Views $Q_{R_rS_lT_l}$, $Q_{R_rS_lT_h}$ with $\eps \leq \frac{1}{2}$, 
$V_{R_hS_l}(\inst{a}_{ST}, \inst{a}_{RT}, \inst{a}_{RST})$,
and $V_{RS_l}(\inst{a}_{ST},\inst{a}_{RST})$:
In these cases the delta evaluation strategies dictate 
to iterate over all $\inst{A}_{ST}$-tuples
$\inst{a}_{ST}$ paired with $(\bdeltaA_{RS},\bdeltaA_{RST})$ in 
$S_l$ and,
if the strategy contains a relation part $T_t$, to look up the multiplicity of
$(\inst{a}_{ST},\bdeltaA_{RT}, \bdeltaA_{RST})$ in 
$T_h$ for each $\inst{a}_{ST}$.
Since
$S_l$ is light, $S_l$ 
contains less than $\frac{3}{2} N^{\eps}$ distinct  
$\inst{A}_{ST}$-tuples paired with 
$\bdeltaA_{RS}$.
Moreover, since 
$\inst{a}_{ST}$, $\bdeltaA_{RT}$, and $\bdeltaA_{RST}$
fix all variable values of 
$T_t$, the latter view  can contain at most one tuple
with these values.
Hence, the overall 
computation time is $\bigO{N^{\eps}}$, thus, by $N = \Theta(|\inst{D}|)$, 
 $\bigO{|\inst{D}|^{\eps}}$.

\item Views $Q_{R_rS_hT_h}$, $Q_{R_rS_lT_h}$ with $\eps > \frac{1}{2}$, $Q_{R_rST_h}$ with $R_r \in \{R_l,R_h\}$,
$V_{T_hR_l}(\inst{a}_{RS}, \inst{a}_{ST}, \inst{a}_{RST})$,
and 
$V_{T_hR_l}(\inst{a}_{ST}, \inst{a}_{RST})$:
In these cases the evaluation strategies dictate to iterate
over all $\inst{A}_{ST}$-tuples $\inst{a}_{ST}$ paired with 
$(\bdeltaA_{RT}, \bdeltaA_{RST})$ in $T_h$ and, if the strategy 
contains a view $S_s$, to look up the multiplicity of
$(\bdeltaA_{RS}, \inst{a}_{ST},\bdeltaA_{RST})$ 
in $S_s$ for each $\inst{a}_{ST}$.
Since $T_h$ is heavy, $T_h$ can contain at most      
$2 N^{1-\eps}$ distinct  
$\inst{A}_{ST}$-tuples. Furthermore, 
there can be at most one tuple in $S_s$ with values
from $\bdeltaA_{RS}$, $\inst{a}_{ST}$, and $\bdeltaA_{RST}$.
This implies that the computation time is 
$\bigO{N^{1-\eps}}$. By $N = \Theta(|\inst{D}|)$, it follows that the
computation time is
 $\bigO{|\inst{D}|^{1-\eps}}$.
\end{itemize}

\subsection{Incremental Maintenance for \threerpj{3} Queries}\label{sec:maintenance_3-rel_3-pair-join} 
We show that the general 
\ivme strategy 
 introduced in Appendix \ref{sec:meta_strategy} 
 admits the complexitiy results for \threerpj{3}
queries as given in 
Theorem \ref{theo:maintain_3rel_3pair-join}. 
In case of \threerpj{3} queries, 
all pair-join variable tuples, i.e., $\inst{A}_{RT}$, $\inst{A}_{RS}$, and 
$\inst{A}_{ST}$, are nonempty, which means
that \ivme decides to partition 
all three relations.
Figure 
\ref{fig:view_definitions_non_empty_join_pair}
gives the restriction of the views and 
delta evaluation strategies in 
Figures 
  \ref{fig:meta_view_definitions}
  and
  \ref{fig:best_update_strategy}
to the case where all pair-join variable tuples 
are nonempty.

 \begin{figure}[t]
   \begin{center}
    \renewcommand{\arraystretch}{1.2}  
    \begin{tabular}{@{\hskip 0.0in}l@{\hskip 0.3in}l@{\hskip 0.0in}}
      \toprule
      Materialized View & Space\\    
      \midrule \\[-0.5cm]
   
   $Q()=\sum\limits_{Q' \in SAV}Q'()$ 
   &$\bigO{1}$\\

$V_{S_hT_l}(\inst{a}_{RT},\inst{a}_{RS},\inst{a}_{RST} ) = 
\textstyle\sum 
S_{h} ( \inst{a}_{RS}, \inst{a}_{ST}, \inst{a}_{RST}) \cdot 
T_{l}(\inst{a}_{ST}, \inst{a}_{RT}, \inst{a}_{RST})$ 
  &  $\bigO{|\inst{D}|^{1 + \min\{\eps, 1- \eps\}}}$\\ 
        


$V_{T_hR_l}(\inst{a}_{RS},\inst{a}_{ST},\inst{a}_{RST} ) = 
\textstyle\sum 
T_{h} (\inst{a}_{ST}, \inst{a}_{RT}, \inst{a}_{RST}) \cdot 
R_{l}(\inst{a}_{RT}, \inst{a}_{RS}, \inst{a}_{RST})$ 
  &  $\bigO{|\inst{D}|^{1 + \min\{\eps, 1- \eps\}}}$ \\
 
   
  
$V_{R_hS_l}(\inst{a}_{ST},\inst{a}_{RT},\inst{a}_{RST} ) = 
\textstyle\sum 
R_{h} (\inst{a}_{RT}, \inst{a}_{RS}, \inst{a}_{RST}) \cdot 
S_{l}(\inst{a}_{RS}, \inst{a}_{ST}, \inst{a}_{RST})$ 
  &  $\bigO{|\inst{D}|^{1 + \min\{\eps, 1- \eps\}}}$ \\
      \bottomrule   
    \end{tabular}
  \end{center}
  \begin{center}
    \renewcommand{\arraystretch}{1.2}  
    \begin{tabular}{@{\hskip 0.0in}l
    @{\hskip 0.1in}l@{\hskip 0.1in}l@{\hskip 0.0in}}
      \toprule
      Delta evaluation strategy & Note & Time\\    
      \midrule

      $\delta Q_{R_rS_lT_l}() = 
      \delta R_r(\bdeltaA_{RT},\bdeltaA_{RS},\bdeltaA_{RST}) \cdot
      \textstyle\sum S_l(\bdeltaA_{RS},\inst{a}_{ST},\bdeltaA_{RST}) \cdot 
      T_l(\inst{a}_{ST},\bdeltaA_{RT},\bdeltaA_{RST})$ & & $\bigO{|\inst{D}|^{\eps}}$ \\
 

      $\delta Q_{R_rS_lT_h}() = \delta R_r(\bdeltaA_{RT},\bdeltaA_{RS},\bdeltaA_{RST}) \cdot
       \textstyle\sum S_l(\bdeltaA_{RS},\inst{a}_{ST},\bdeltaA_{RST}) \cdot 
      T_h(\inst{a}_{ST},\bdeltaA_{TR},\bdeltaA_{RST})$ & $\eps \leq \frac{1}{2}$ & $\bigO{|\inst{D}|^{\eps}}$\\


   $\delta Q_{R_rS_hT_h}() = 
   \delta R_r(\bdeltaA_{RT},\bdeltaA_{RS},\bdeltaA_{RST}) \cdot
   \textstyle\sum T_h(\inst{a}_{ST},\bdeltaA_{TR},\bdeltaA_{RST}) \cdot 
      S_h(\bdeltaA_{RS},\inst{a}_{ST},\bdeltaA_{RST})$ & & $\bigO{|\inst{D}|^{1-\eps}}$ \\


    $\delta Q_{R_rS_lT_h}() = \delta R_r(\bdeltaA_{RT},\bdeltaA_{RS},\bdeltaA_{RST}) \cdot
       \textstyle\sum T_h(\inst{a}_{ST},\bdeltaA_{TR},\bdeltaA_{RST}) \cdot 
      S_l(\bdeltaA_{RS},\inst{a}_{ST},\bdeltaA_{RST})$ & $\eps > \frac{1}{2}$ & $\bigO{|\inst{D}|^{1-\eps}}$ \\


    $\delta Q_{R_rS_hT_l} ()= 
    \delta R_r(\bdeltaA_{RT},\bdeltaA_{RS},\bdeltaA_{RST}) \cdot
    V_{S_hT_l}(\bdeltaA_{TR},\bdeltaA_{RS},\bdeltaA_{RST})$ & & $\bigO{1}$ \\


$\delta V_{R_hS_l}(\inst{a}_{ST},\bdeltaA_{TR},\bdeltaA_{RST}) = 
\delta R_{h} (\bdeltaA_{TR}, \bdeltaA_{RS}, \bdeltaA_{RST}) \cdot  
 S_{l}(\bdeltaA_{RS}, \inst{a}_{ST}, \bdeltaA_{RST})$  & & 
      $\bigO{|\inst{D}|^{\eps}}$\\
      

       $\delta V_{T_hR_l}(\bdeltaA_{RS},\inst{a}_{ST},\bdeltaA_{RST}) =
       \delta R_{l}(\bdeltaA_{RT}, \bdeltaA_{RS}, \bdeltaA_{RST}) \cdot 
T_{h} (\inst{a}_{ST}, \bdeltaA_{TR}, \bdeltaA_{RST})$ & &
      $\bigO{|\inst{D}|^{1-\eps}}$\\
 
     
      \bottomrule 
    \end{tabular}
  \end{center}
  \vspace{-0.3cm}
  \caption{(top table) The definitions and space complexities 
  of the views materialized by \ivme for the 
  maintenance of a \threerpj{3} query. 
  The views are those views from  
  Figure 
  \ref{fig:meta_view_definitions}
  where all pair-join variable tuples 
  $\inst{A}_{RS}$, $\inst{A}_{ST}$, and $\inst{A}_{RT}$
   are nonempty.  $\SAV$ is the set of all skew-aware views of such queries.
   (bottom table) The delta evaluation 
      strategies from Figure 
 \ref{fig:best_update_strategy} for computing the 
 deltas of the views in the top table. 
 It holds 
 $R_r \in \{R_l,R_h\}$.}
  \label{fig:view_definitions_non_empty_join_pair}
\end{figure}

\nop{
\paragraph{Effect of the Optimization Phase.}
Observe that 
$Q_{R_rS_hT_l} = Q_{R_rS_lT_l}[S_l \mapsto S_h]$
and the delta evaluation strategies for both views 
under updates to $R_r$ are different.
This implies that  the partitioning of relation $S$
is not redundant. Since in case of 
\threerpj{3} queries the delta evaluation strategies 
for updates to $S_s$ and $T_t$ are symmetric to 
the strategies for updates to $R_r$, 
it follows that the partitioning of the other two 
relations is not redundant either.   
Thus, in the optimization phase \ivme does not discard 
any partitioning.  
}

\paragraph{Maintenance Complexities.}
We first show that the preprocessing time is
$\bigO{|\inst{D}|^{\frac{3}{2}}}$.
For an initial database $\inst{\inst{D}}$,
the preprocessing stage consists 
of setting the threshold base, 
which requires constant time, the strict partitioning 
of the relations, which can be done in 
linear time, and the computation of the views
in Figure \ref{fig:view_definitions_non_empty_join_pair}, 
which we consider next.
Since the FAQ-width of $Q$ is 
  $\frac{3}{2}$ \cite{FAQ:PODS:2016}, it can be computed 
in time  $\bigO{|\inst{D}|^{\frac{3}{2}}}$.  
Following the reasoning in the size analysis,
it can be  shown that 
the view 
$V_{S_hT_l}(\inst{a}_{RT}, \inst{a}_{RS}, \inst{a}_{RST})$
can be computed 
in time
  $\bigO{|\inst{D}|^{1 + \min\{\eps, 1-\eps\}}}$.
 The analysis of the computation times 
for the other views 
in Figure \ref{fig:view_definitions_non_empty_join_pair}
 are analogous. It follows that the overall preprocessing time is 
$\bigO{|\inst{D}|^{\frac{3}{2}}}$.

The space complexity of \ivme is dominated by the size of the partitions
and the views. From the top table in Figure  
\ref{fig:view_definitions_non_empty_join_pair} and the fact that 
the sizes of partitions are linear, we derive that 
\ivme needs 
$\bigO{|\inst{D}|^{1 + \min\{\eps, 1-\eps\}}}$
space. 

It follows from 
 the complexity results in the bottom table in Figure
\ref{fig:view_definitions_non_empty_join_pair} 
that the time to process a single-tuple update is 
$\bigO{|\inst{D}|^{\max\{\eps, 1-\eps\}}}$.
Following the same 
reasoning as in the triangle count case, we can show
that the minor and major  rebalancing times are 
$\bigO{|\inst{D}|^{\eps + \max\{\eps, 1-\eps\}}}$
and 
$\bigO{|\inst{D}|^{1 + \min\{\eps, 1-\eps\}}}$, respectively.
The minor rebalancing time is amortized 
over $\Omega(|\inst{D}|^{\eps})$ and 
the major rebalancing time over 
$\Omega(|\inst{D}|)$ updates.
This implies that the amortized update time 
is $\bigO{|\inst{D}|^{\max\{\eps, 1-\eps\}}}$.
Since the final count $Q$ is included in the set of materialized 
views, the answer time is constant.

\subsection{Incremental Maintenance for \threerpj{2} Queries}\label{sec:maintenance_3-rel_2-pair-join}
\begin{figure}[t]
  \begin{center}
    \renewcommand{\arraystretch}{1.2}  
    \begin{tabular}{@{\hskip 0.0in}l@{\hskip 0.3in}l@{\hskip 0.0in}}
      \toprule
      Materialized View & Space\\    
      \midrule \\[-0.5cm]
   
   $Q()=\sum\limits_{Q' \in SAV}Q'()$ & $\bigO{1}$\\


$V_{R_lT}(\inst{a}_{RS},\inst{a}_{RST}) = 
\textstyle\sum 
T (\inst{a}_{RT}, \inst{a}_{RST}) \cdot 
R_{l}(\inst{a}_{RT}, \inst{a}_{RS}, \inst{a}_{RST})$ 
      &    $\bigO{|\inst{D}|}$ \\

    
$V_{R_hS_l}(\inst{a}_{RT},\inst{a}_{RST} ) = 
\textstyle\sum  
R_{h} (\inst{a}_{RT}, \inst{a}_{RS}, \inst{a}_{RST}) \cdot 
S_{l}(\inst{a}_{RS}, \inst{a}_{RST})$ 
  &  $\bigO{|\inst{D}|}$ \\
 
      \bottomrule   
    \end{tabular}
  \end{center}
  \vspace{-0.5cm}
  \begin{center}
    \renewcommand{\arraystretch}{1.2}  
    \begin{tabular}{@{\hskip 0.0in}l@{\hskip 0.1in}l
    @{\hskip 0.1in}l@{\hskip 0.0in}}
      \toprule
      Delta evaluation strategy & Note &Time\\    
      \midrule


      $\delta Q_{R_rS_sT}() = 
      \delta R_r(\bdeltaA_{RT},\bdeltaA_{RS},\bdeltaA_{RST}) \cdot
      S_s(\bdeltaA_{RS},\bdeltaA_{RST}) \cdot 
      T(\bdeltaA_{RT},\bdeltaA_{RST})$ & & $\bigO{1}$ \\


$\delta V_{R_hS_l}(\inst{a}_{ST},\bdeltaA_{RT},\bdeltaA_{RST}) = 
\delta R_{h} (\bdeltaA_{RT}, \bdeltaA_{RS}, \bdeltaA_{RST}) \cdot  
 S_{l}(\bdeltaA_{RS}, \inst{a}_{ST}, \bdeltaA_{RST})$  & &
      $\bigO{|\inst{D}|^{\eps}}$\\


       $\delta V_{TR_l}(\bdeltaA_{RS},\bdeltaA_{RST}) = 
       \delta R_{l}(\bdeltaA_{TR}, \bdeltaA_{RS}, \bdeltaA_{RST}) \cdot 
T (\bdeltaA_{TR}, \bdeltaA_{RST})$ & &
      $\bigO{1}$\\
      
     \hline\\[-0.5cm]   
  

      $\delta Q_{R_lS_sT} ()= 
      \delta S_s(\bdeltaA_{RS},\bdeltaA_{RST}) \cdot
      V_{R_lT}(\bdeltaA_{RS},\bdeltaA_{RST}) $ & & $\bigO{1}$  \\

        $\delta Q_{R_hS_sT} ()= 
        \delta S_s(\bdeltaA_{RS},\bdeltaA_{RST}) \cdot
        \textstyle\sum 
      R_h(\inst{a}_{RT},\bdeltaA_{RS}, \bdeltaA_{RST})
      \cdot 
      T(\inst{a}_{RT},\bdeltaA_{RST})$ & & $\bigO{|\inst{D}|^{1-\eps}}$ \\

          

$\delta V_{R_hS_l}(\inst{a}_{RT},\bdeltaA_{RST}) = 
\delta S_{l} (\bdeltaA_{RS},  \bdeltaA_{RST}) \cdot  
 R_{h}(\inst{a}_{RT}, \bdeltaA_{RS}, \bdeltaA_{RST})$  & & 
      $\bigO{|\inst{D}|^{1 -\eps}}$\\
      
  \hline\\[-0.5cm]
      
   $\delta Q_{R_lS_lT} ()= 
   \delta T(\bdeltaA_{RT},\bdeltaA_{RST}) \cdot
   \textstyle\sum R_l(\inst{a}_{RT},\bdeltaA_{RS},\bdeltaA_{RST}) \cdot 
      S_l(\bdeltaA_{RS}, \bdeltaA_{RST})$ & & $\bigO{|\inst{D}|^{\eps}}$ \\

      
         $\delta Q_{R_lS_hT} ()= 
         \delta T(\bdeltaA_{RT},\bdeltaA_{RST}) \cdot
         \textstyle\sum R_l(\bdeltaA_{RT},\inst{a}_{RS},\bdeltaA_{RST}) \cdot 
      S_h(\inst{a}_{RS},\bdeltaA_{RST})$ & $\eps \leq \frac{1}{2}$ & $\bigO{|\inst{D}|^{\eps}}$ \\

         $\delta Q_{R_hS_lT} ()= 
         \delta T(\bdeltaA_{RT},\bdeltaA_{RST}) \cdot  
      V_{R_hS_l}(\bdeltaA_{RT},\bdeltaA_{RST})$
       &&  $\bigO{1}$ \\

      
         $\delta Q_{R_lS_hT} ()= 
         \delta T(\bdeltaA_{RT},\bdeltaA_{RST}) \cdot
         \textstyle\sum  S_h(\inst{a}_{RS}, \bdeltaA_{RST}) \cdot R_l(\bdeltaA_{RT},\inst{a}_{RS},\bdeltaA_{RST})$
     & $\eps > \frac{1}{2}$ & $\bigO{|\inst{D}|^{1 -\eps}}$ \\

      
         $\delta Q_{R_hS_hT}() = 
         \delta T(\bdeltaA_{RT},\bdeltaA_{RST}) \cdot
         \textstyle\sum  S_h(\inst{a}_{RS}, \bdeltaA_{RST}) \cdot R_h(\bdeltaA_{RT},
      \inst{a}_{RS},\bdeltaA_{RST})$ & & $\bigO{|\inst{D}|^{1 -\eps}}$ \\

       $\delta V_{R_lT}(\inst{a}_{RS},\bdeltaA_{RST}) = 
       \delta T(\bdeltaA_{RT}, \bdeltaA_{RST}) \cdot 
R_{l} (\bdeltaA_{RT}, \inst{a}_{RS}, \bdeltaA_{RST})$ & &
      $\bigO{|\inst{D}|^{\eps}}$ \\
      
      \bottomrule 
    \end{tabular}
  \end{center}
   \vspace{-0.3cm}
   \caption{
 (top table) The restriction of the views in Figure 
  \ref{fig:meta_view_definitions} to the case of 
 a \threerpj{2} queries where $\inst{A}_{ST}$ is empty.
  $\SAV$ is the set of all skew-aware views of such queries.
 (bottom table)   The delta evaluation 
      strategies for computing the 
 deltas of the views in the top table under updates 
 to all three relations. 
 It holds $R_r \in \{R_l,R_h\}$ and $S_s \in \{S_l,S_h\}$.
 The strategies for updates to $R_r$ 
 are from Figure
 \ref{fig:best_update_strategy}. The strategies for updates to
 the other relations follow from Figure
 \ref{fig:best_update_strategy} by symmetry.
 } 
  \label{fig:view_definitions_2_non_empty_join_pair}
\end{figure}

We show that the restriction of the 
general \ivme strategy from Appendix \ref{sec:meta_strategy}
to  \threerpj{2} queries results in an 
\ivme strategy that maintains \threerpj{2} queries
with the complexities as given in
Theorem \ref{theo:maintain_3rel_3pair-join}. 
Without loss of generality, we consider
\threerpj{2} queries where the pair-join variable tuples 
$\inst{A}_{RT}$ and $\inst{A}_{RS}$ are nonempty
and $\inst{A}_{ST}$ is empty. 
The cases where one of the 
variable tuples $\inst{A}_{RT}$ and $\inst{A}_{RS}$
is empty, and the other tuples are nonempty are handled analogously.

\ivme partitions all relations besides 
$T$, since $\inst{A}_{ST}$, the partition variable tuple of 
$T$, is empty. 
Figure
\ref{fig:view_definitions_2_non_empty_join_pair}
gives the restrictions of the views and 
delta evaluation strategies in Figures  
\ref{fig:meta_view_definitions}
and
\ref{fig:best_update_strategy}
to the case where $\inst{A}_{ST}$
is empty and the other pair-join variable tuples are nonempty .  
Since the views of \threerpj{2} queries are not symmetric, the bottom table in
Figure 
 \ref{fig:view_definitions_2_non_empty_join_pair} gives the evaluation strategies
for updates to all three relations.

\nop{
\paragraph{Effect of the Optimization Phase.}
Observe that it holds 
$Q_{R_hS_lT} = Q_{R_lS_lT}[R_l \mapsto R_h]$ 
and $Q_{R_lS_hT} = Q_{R_lS_lT}[S_l \mapsto S_h]$.
According to Figure   
 \ref{fig:view_definitions_2_non_empty_join_pair},
under updates to $T$ 
the delta evaluation strategies  
for both $Q_{R_hS_lT}$ and $Q_{R_lS_hT}$ differ 
from the delta evaluation strategy
 for  $Q_{R_lS_lT}$. It follows
 that the partitions of $R$ and $S$
 are non-redundant, which means 
 that the optimization phase does not discard 
 these partitions.  
}

\paragraph{Maintenance Complexities.}
All maintained views admit hypertree decompositions
with the group-by variables on top of all other variables
and fractional hypertree width one \cite{GroheM14}.
This implies that the FAQ-width of the views 
is one. Hence, the views can be computed in linear time. 
It follows that the  preprocessing time is linear in the
 database size.  

It follows from the space complexities in the top table in Figure  
   \ref{fig:view_definitions_2_non_empty_join_pair}
that \ivme needs linear space.  
By the time complexities in the bottom table of Figure 
  \ref{fig:view_definitions_2_non_empty_join_pair}
and an amortization analysis along the lines of the 
proof  of Theorem~\ref{theo:main_result},
it follows that the amortized 
update time is $\bigO{\inst{|D|}^{\max\{\eps, 1-\eps\}}}$.
The materialization of the result of $Q$  
guarantees constant answer time.

\subsection{Incremental Maintenance for
$\threerpj{0} \cup \threerpj{1}$ Queries}\label{sec:maintenance_3-rel_1-pair-join} 
In this section we show that 
$\threerpj{0} \cup \threerpj{1}$ queries can 
be maintained 
with the complexities given in 
Theorem \ref{theo:maintain_3rel_3pair-join}:
$\bigO{|\inst{D}|}$ preprocessing time,
(non-amortized) constant update time,
constant answer time, and 
$\bigO{|\inst{D}|}$ space.
  We first prove that 
this class of  queries 
is equal to the class of non-hierarchical \threer 
queries. From this, the above complexity 
results follow immediately 
\cite{BerkholzKS17}. We then 
show that the 
general \ivme strategy presented in 
Appendix \ref{sec:meta_strategy}
recovers the same complexity results
when restricted to $\threerpj{0} \cup \threerpj{1}$ queries
and optimized as described in Appendix 
\ref{sec:3r_opt_phase}.

\paragraph{Hierarchical \threer Queries.}
We recall the definition of hierarchical queries
\cite{DalviS07a}. Given a variable $A$, we denote by 
$\atoms{A}$, the set of all relation symbols 
containing  $A$ in their schemas. 
A query is called hierarchical if for each pair 
of variables $A$ and $B$, it holds 
$\atoms{A} \subseteq \atoms{B}$,
$\atoms{B} \subseteq \atoms{A}$, or
$\atoms{A} \cap \atoms{B} = \emptyset$.

\begin{proposition}\label{prop:hierarchical_3_rel_queries}
A \threer query is hierarchical if and only if it is 
a $\threerpj{0} \cup \threerpj{1}$ query.
\end{proposition}

\begin{proof}
\underline{\textit{The ``if''-direction.}}
 We consider a $\threerpj{0} \cup \threerpj{1}$
 query of the form
 $$Q() = \sum
R(\inst{a}_R,\inst{a}_{RS},\inst{a}_{RST}) \ztimes 
S(\inst{a}_S,\inst{a}_{RS},\inst{a}_{RST}) \ztimes 
T(\inst{a}_T, \inst{a}_{RST}),$$ where 
the pair-join variable tuples 
$\inst{A}_{RT}$ and $\inst{A}_{ST}$ are empty 
and all other variable tuples are possibly nonempty.
Cases where other pair-join variable tuples are set to empty 
are handled along the same lines. 
 
We show that $Q$ is hierarchical.
Given two variables 
$A$ and $B$, we need to prove that 
$\atoms{A} \subseteq \atoms{B}$,
$\atoms{B} \subseteq \atoms{A}$, or
$\atoms{A} \cap \atoms{B} = \emptyset$.
In case both variables are included in 
$\inst{A}_{RS}$, $\inst{A}_{RST}$ or the 
same non-join variable tuple,  
the first two conditions 
obviously hold. 
We further distinguish the following cases:  
\begin{itemize}
\item $A$ and $B$ are from distinct non-join variable tuples:
It immediately follows that 
$\atoms{A} \cap \atoms{B} = \emptyset$.

\item $A$ is from a non-join variable tuple 
and $B$ is from $\inst{A}_{RS}$:
This means 
that $\atoms{B} = \{R,S\}$. If $A$ is included 
in $\inst{A}_R$ or $\inst{A}_S$, it holds 
$\atoms{A} \subseteq \atoms{B}$. If
$A$ is included in $\inst{A}_T$,
we have
$\atoms{A} \cap \atoms{B} = \emptyset$. 

\item $A$ is from a non-join variable tuple and $B$
is from  $\inst{A}_{RST}$:
We have
$\atoms{B} = \{R,S,T\}$. Each non-join variable, and hence $A$,  occurs in the schema 
of exactly one relation. 
Thus $\atoms{A} \subseteq \atoms{B}$.

\item $A$ is from  $\inst{A}_{RS}$ and $B$ is from $\inst{A}_{RST}$:
We have $\atoms{A} = \{R,T\} \subseteq   
\{R,S,T\} = \atoms{B}$.  
\end{itemize}

\smallskip
\underline{\textit{The ``only if''-direction.}}
Let 
$$Q() = \sum
R(\inst{a}_R,\inst{a}_{RT},\inst{a}_{RS},\inst{a}_{RST}) \ztimes 
S(\inst{a}_S,\inst{a}_{RS},\inst{a}_{ST},\inst{a}_{RST}) \ztimes 
T(\inst{a}_T,\inst{a}_{ST},\inst{a}_{RT},\inst{a}_{RST})$$
be a hierarchical \threer query. 
 We show that 
at most one of the pair-join variable tuples 
 $\inst{A}_{RT}$, $\inst{A}_{RS}$, and $\inst{A}_{ST}$
 can be nonempty.
For the sake of 
contradiction assume that two pair-join variable tuples, say 
$\inst{A}_{RS}$ and $\inst{A}_{ST}$
are nonempty. All other choices 
are handled analogously. 
Let $A$ be from $\inst{A}_{RS}$ and $B$ from $\inst{A}_{ST}$. 
This means that $\atoms{A} = \{R,S\}$ and $\atoms{B} = \{S,T\}$. Hence,
neither $\atoms{A} \subseteq \atoms{B}$ nor $\atoms{B} \subseteq \atoms{A}$
nor $\atoms{A} \cap \atoms{B} = \emptyset$.  
We conclude that 
$Q$ cannot be hierarchical, which contradicts our initial
assumption. 
\end{proof}

We consider $\threerpj{0} \cup \threerpj{1}$ queries where 
the pair-join variable tuples $\inst{A}_{RT}$ and $\inst{A}_{ST}$ are 
 empty. 
Cases where other pairs of pair-join variable tuples are empty 
 are handled completely analogously. 
 Figure 
 \ref{fig:view_definitions_1_non_empty_join_pair}
shows the restrictions of the views and strategies 
in Figures 
\ref{fig:meta_view_definitions}
and
\ref{fig:best_update_strategy}
to $\threerpj{0} \cup \threerpj{1}$ queries.
Since the views of this kind of queries 
are not symmetric, we give in 
the bottom table in Figure
\ref{fig:view_definitions_1_non_empty_join_pair}
the delta evaluation strategies under  updates to all relations.

\begin{figure}[t]
  \begin{center}
    \renewcommand{\arraystretch}{1.2}  
    \begin{tabular}{@{\hskip 0.0in}l@{\hskip 0.3in}l@{\hskip 0.0in}}
      \toprule
      Materialized View & Space\\    
      \midrule \\[-0.5cm]
   
 $Q()=\sum\limits_{Q' \in SAV}Q'()$  &$\bigO{1}$\\
          

$V_{RS_{s}}(\inst{a}_{RST}) = 
\textstyle\sum
R (\inst{a}_{R}, \inst{a}_{RS}, \inst{a}_{RST}) \cdot 
S_{s}(\inst{a}_{S}, \inst{a}_{RS}, \inst{a}_{RST})$ 
     &    $\bigO{|\inst{D}|}$ \\

      \bottomrule   
    \end{tabular}
  \end{center}
\vspace{-0.5cm}
\begin{center}
    \renewcommand{\arraystretch}{1.2}  
    \begin{tabular}{@{\hskip 0.0in}l@{\hskip 0.1in}l
    @{\hskip 0.1in}l@{\hskip 0.0in}}
      \toprule
      Delta evaluation strategy & Time \\    
      \midrule
       

      $\delta Q_{RS_sT} ()= 
      \delta R(\bdeltaA_R, \bdeltaA_{RS},\bdeltaA_{RST}) \cdot
      V_{S_s}(\bdeltaA_{RS},\bdeltaA_{RST}) \cdot 
      V_{T}(\bdeltaA_{RT},\bdeltaA_{RST})$ & $\bigO{1}$ \\

$\delta V_{RS_s}(\bdeltaA_{RST}) = 
\delta R (\bdeltaA_{R}, \bdeltaA_{RS}, \bdeltaA_{RST}) \cdot 
V_{S_{s}}(\bdeltaA_{RS}, \bdeltaA_{RST})$  & 
      $\bigO{1}$\\   
      
      \hline\\[-0.5cm]
       $Q_{RS_sT} = 
       \delta S_s(\bdeltaA_S, \bdeltaA_{RS},\bdeltaA_{RST}) \cdot
       V_{T}(\bdeltaA_{RST}) \cdot 
      V_{R}(\bdeltaA_{RS}, \bdeltaA_{RST})$ & $\bigO{1}$ \\
           
  $\delta V_{RS_s}(\bdeltaA_{RST}) =
  \delta S_s(\bdeltaA_S, \bdeltaA_{RS},\bdeltaA_{RST}) \cdot
      V_R(\bdeltaA_{RS},\bdeltaA_{RST})$ &     
      $\bigO{1}$ 
      \\

      \hline\\[-0.5cm]
      
       $\delta Q_{RS_sT} () = 
       \delta T(\bdeltaA_T, \bdeltaA_{RST}) \cdot 
     V_{RS_s}(\bdeltaA_{RST})$ & $\bigO{1}$\\
      \bottomrule   
    \end{tabular}
\end{center}
\vspace{-0.3cm}
\caption{
 (top table) The restriction of the views in Figure 
  \ref{fig:meta_view_definitions} to the case of   
 $\threerpj{0} \cup \threerpj{1}$ queries where $\inst{A}_{RT}$
 and $\inst{A}_{ST}$ are empty.
  $\SAV$ is the set of all skew-aware views of such queries. 
    (bottom table) The delta evaluation 
      strategies for computing the 
 deltas of the views in the top table. It
 holds $S_s \in \{S_l,S_h\}$. The strategies for updates to relation $R$
 are from Figure 
 \ref{fig:best_update_strategy}. The strategies for updates to the 
 other relations follow from Figure 
 \ref{fig:best_update_strategy} by symmetry. 
 The optimization phase 
 discards the partitioning of relation $S$ and replaces 
 in all views and update strategies the part  
 $S_s$ by $S$.} 
  \label{fig:view_definitions_1_non_empty_join_pair}
  \end{figure}

\paragraph{Effect of the Optimization Phase.}
It can easily be derived from the bottom table in Figure 
\ref{fig:view_definitions_1_non_empty_join_pair}
that for any view $V$ including $S_l$ and for any update 
$\delta K$, the evaluation strategies to compute the deltas 
of $V$ and $V[S_l \mapsto S_h]$ are the same, which means 
that the partitioning of S is redundant. Hence, in the
optimization phase, \ivme discards 
the partitioning of relation $S$.

\paragraph{Maintenance Complexities.}
The analysis of the preprocessing time is analogous 
to the case of 
$\threerpj{2}$ queries.
Since for any maintained view, the FAQ-width is one,
the views can be computed in linear time, thus, the preprocessing time is linear. 

It follows from the top table in Figure  
\ref{fig:view_definitions_1_non_empty_join_pair}
that the space needed by \ivme to incrementally maintain a 
$\threerpj{0} \cup \threerpj{1}$ query is 
linear. We derive from the bottom table in the same Figure
that the update time is constant. 
As a result of the optimization phase,   
\ivme decides to not partition any relation. Hence,  
the maintenance procedure does not include  
rebalancing steps, which means that no more 
than constant time is needed at any update step. 
It follows that the constant update time is non-amortized.

\subsection{Worst-Case Optimality of \ivme for \threer Queries}
The worst-case optimality of the update time of the 
\ivme strategy for $\threerpj{2} \cup \threerpj{3}$ queries, 
conditioned on the \OMv conjecture (Conjecture \ref{conj:omv}) 
follows from Proposition 
\ref{prop:lower_bound_3rel_count}. The proof of the proposition is a 
straightforward adaption of the proof of Proposition 
\ref{prop:lower_bound_triangle_count}.

\begin{proof}[Proof of Proposition \ref{prop:lower_bound_3rel_count}] 
We reduce the \OuMv problem given in Definition 
\ref{def:OuMv} to the incremental maintenance of 
$\threerpj{2} \cup \threerpj{3}$ queries. The reduction 
is along the same lines as in the proof of 
Proposition
\ref{prop:lower_bound_triangle_count}.
We explain the main idea of the reduction  to the maintenance 
of $\threerpj{2}$ queries. The case for 
$\threerpj{3}$ queries is a simple extension.  

Let 
$$Q() = \sum
R(\inst{a}_R,\inst{a}_{RT},\inst{a}_{RS},\inst{a}_{RST}) \ztimes 
S(\inst{a}_S,\inst{a}_{RS},\inst{a}_{ST},\inst{a}_{RST}) \ztimes 
T(\inst{a}_T,\inst{a}_{ST},\inst{a}_{RT},\inst{a}_{RST})$$
be a $\threerpj{2}$ query where the variable tuples 
$\inst{A}_{RS}$ and $\inst{A}_{ST}$ are nonempty.
Assume that variable $A_{RS}$ is included in
$\inst{A}_{RS}$ and the variable $A_{ST}$ is included 
in  $\inst{A}_{ST}$.    
Assume also that there is a dynamic algorithm 
that maintains the result of $Q$
with arbitrary preprocessing time, amortized update time 
$\bigO{|\inst{D}|^{\frac{1}{2}-\gamma}}$, 
and answer time $\bigO{|\inst{D}|^{1-\gamma}}$.
We can use this algorithm to solve
the $\OuMv$ problem in subcubic time, which contradicts the \OuMv conjecture.

Let  $(\vecnormal{M}, (\vecnormal{u}_1,\vecnormal{v}_1), \ldots ,(\vecnormal{u}_n,\vecnormal{v}_n))$ be an input to the $\OuMv$ problem.
After the construction of the initial \ivme state from an empty database 
$\db = \{R,S,T\}$, we execute at most $n^2$ updates to relation $S$ such that 
$S = \{\, (A_{RS}:i,A_{ST}:j, a , \ldots , a) \mapsto \vecnormal{M}(i,j) \,\mid\, i,j \in \{1,\ldots, n\} \,\}$ for some constant $a$. 
In each round $r \in \{1, \ldots , n\}$, we execute at most $2n$ updates to the relations $R$ and $T$ 
such that $R = \{\, (A_{RS}:i, a , \ldots , a) \mapsto \vecnormal{u}_r(i) \,\mid\, i \in \{1,\ldots, n\} \,\}$
and $T = \{\, (A_{ST}:i, a , \ldots , a) \mapsto \vecnormal{v}_r(i) \,\mid\, i \in \{1,\ldots, n\} \,\}$.
The algorithm outputs 
$1$ at the end of round $r$ if and only if  $Q()$ is nonzero. 
The time analysis of the reduction is   exactly the same as in 
the proof of Proposition \ref{prop:lower_bound_triangle_count}. 
\end{proof} 

 \end{document}